\documentclass[11pt]{article}
\usepackage{amssymb,latexsym,amsmath,enumerate,bbm,xcolor}     % Standard packages            
\usepackage{amsfonts, amsthm, mathabx, enumitem, caption,  multirow}
\usepackage[letterpaper, portrait, margin=1in]{geometry}
\usepackage{graphicx}
\makeatletter
\renewcommand\paragraph{\@startsection{paragraph}{4}{\z@}%
            {-2.5ex\@plus -1ex \@minus -.25ex}%
            {1.25ex \@plus .25ex}%
            {\normalfont\normalsize\bfseries}}
\makeatother
\newtheorem{assumption}{Assumption}[section]

\newtheorem{result}{Result}[section]
\newtheorem{lemma}{Lemma}[section]
\newtheorem{corollary}{Corollary}[section]
\theoremstyle{remark}
\newtheorem{note}{Note}[section] 

\newtheoremstyle{AssumptionCopy}
        {\topsep}{\topsep}              %%% space between body and thm
        {\itshape}                      %%% Thm body font
        {-0.5pt}                              %%% Indent amount (empty = no indent)
        {\bfseries}                     %%% Thm head font
        {.}                             %%% Punctuation after thm head
        { }                             %%% Space after thm head
        {\thmname{#1}\thmnote{ \bfseries #3}}%%% Thm head spec
\theoremstyle{AssumptionCopy}

\allowdisplaybreaks

\DeclareMathOperator*{\argmax}{argmax}
\DeclareMathOperator*{\Var}{Var}

\DeclareMathOperator*{\VecM}{Vec}
\DeclareMathOperator*{\tr}{tr}
\DeclareMathOperator*{\wpa1}{w.p.a \ 1}
\DeclareMathOperator*{\Int}{Int}
\DeclareMathOperator*{\Span}{Span}
\DeclareMathOperator*{\rank}{rank}

\newcommand{\R}{\mathbb{R}}
\newcommand{\1}{\mathbbm{1}}
\newcommand{\N}{\mathbb{N}}
\newcommand{\E}{\mathbb{E}}
\newcommand{\indep}{\raisebox{0.05em}{\rotatebox[origin=c]{90}{$\models$}}}
\newcommand{\oP}{o_{\mathbb{P}}}
\newcommand{\OP}{O_{\mathbb{P}}}
\newcommand{\overbar}[1]{\mkern 1.5mu\overline{\mkern-1.5mu#1\mkern-1.5mu}\mkern 1.5mu}

\usepackage{natbib}
\begin{document}

\title{\fontsize{25}{30}\selectfont \vspace{-50pt} A Correlated Random Coefficient Panel Model  \\ with Time-Varying Endogeneity \vspace{5pt}}
%the second number in font size indicates the size of the skip - should be around 1.2 times the font size
\author{\fontsize{15}{19}\selectfont Louise Laage\thanks{louise.laage@georgetown.edu, Department of Economics, Georgetown University.
		\newline
		 I am grateful to Donald W. K. Andrews, Xiaohong Chen and especially Yuichi Kitamura for their guidance and support. I  thank Anna Bykhovskaya, Philip A. Haile, Yu Jung Hwang, John Eric Humphries, Ivana Komunjer, Rosa Matzkin, Patrick Moran, Peter C. B. Phillips, Alexandre Poirier, Pedro Sant'Anna, Masayuki Sawada, Francis Vella, Edward Vytlacil as well as participants at the Yale econometrics seminar for helpful conversations and comments on this project. I thank an associate editor and two anonymous referees for suggestions that greatly improved the paper.  I thank the Toulouse School of Economics for hosting me during the academic year 2020-2021 and acknowledge financial support from the grant ERC POEMH 337665, from the ANR under grant ANR-17-EURE-0010 (Investissements d’Avenir program), and from the IAAE travel grant. All errors are mine.}  \vspace{5pt}}

\date{This version : June 12, 2022}

\maketitle

\begin{abstract}
This paper studies a class of linear panel models with random coefficients. We do not restrict the joint distribution of the time-invariant unobserved heterogeneity and the covariates. We investigate identification of the average partial effect (APE) when fixed-effect techniques cannot be used to control for the correlation between the regressors and the time-varying disturbances. Relying on control variables, we develop a constructive two-step identification argument. The first step identifies nonparametrically the conditional expectation of the disturbances given the regressors and the control variables, and the second step uses ``between-group'' variations, correcting for endogeneity, to identify the APE. We propose a natural semiparametric estimator of the APE, show its $\sqrt{n}$ asymptotic normality and compute its asymptotic variance. The estimator is computationally easy to implement, and Monte Carlo simulations show favorable finite sample properties. Control variables arise in various economic and econometric models, and we propose applications of our argument  in several models. As an empirical illustration, we estimate the average elasticity of intertemporal substitution in a labor supply model with random coefficients.
\end{abstract}

\newpage

\section{Introduction}\label{sec:intro}

Correlated Random Coefficient (CRC) models are linear models with random coefficients where the joint distribution of the random coefficients and the regressors is left unspecified. With panel data, panel CRC models relate an outcome variable $y_{it}$ to  regressors $x_{it}$ through the following equation\footnote{We purposefully do not allow for the constant term to be included in $x_{it}$. More details are provided in Section \ref{sec:id_model}.}
\begin{align}\label{eq:model}
	y_{it}  = \, x_{it}' \, \mu_i   + \alpha_i + \epsilon_{it}, \qquad \ i \leq n, \ t \leq T,
\end{align}
where $T$ is treated as fixed, $\epsilon_{it}$ is a time-varying residual and the unobserved random coefficients $(\mu_i', \alpha_i)$ can be arbitrarily related to the regressors $x_{it}$.  An important and empirically relevant task is to recover properties of the random coefficients and the literature has focused on this question while imposing strict exogeneity of the regressors, that is,  $\E(\epsilon_{it}|X_i)= 0$ where $X_i=(x_{i1}',...,x_{iT}')$ \citep{ch92,ab12}. This paper instead examines identification of $\E(\mu_i', \alpha_i)$ in situations where the regressors and the residuals are correlated.

Linear random coefficient models have many empirical applications, for instance to model heterogeneous returns to schooling \citep{c01},  heterogeneous production functions (\cite{mg88}, \cite{gu22}, or see \cite{s11} for an example of heterogeneous technology adoption),  heterogeneous taxable income responses to tax changes \citep{kl20}, or in medical studies  \citep{lw82}. With panel data, the absence of restrictions on the joint distribution of the random coefficients and the regressors corresponds to the \textit{fixed effect approach}. It depicts situations in which the researcher does not know which factors drive the heterogeneity in the impact of $x_{it}$, or does not observe them. Loosely speaking, the fixed-effect approach together with the strict exogeneity condition imply that the endogeneity of the model can be ``captured by a fixed effect''.
However strict exogeneity fails to hold in many cases, for instance whenever omitted variables correlated with the regressors $x_{it}$ vary over time, e.g., unobserved determinants of productivity in a production function. This leads to contemporaneous endogeneity. 
Since strict exogeneity conditions on past and future values of the regressors, it also fails when the regressors are \textit{predetermined}, i.e., correlated with residuals of prior periods, a well-known phenomenon in the panel data literature. 
To allow for such prevalent correlations, this paper focuses on situations in which there is \textit{time-varying endogeneity}. A rigorous definition of time-varying endogeneity here is that there is no time-invariant random vector $c_i$ such that the correlation between the regressors $x_{it}$ and the entire vector of unobserved heterogeneity $(\mu_i', \alpha_i,\epsilon_{it})'$ can be written as the correlation between $x_{it}$ and $c_i$. 

We adopt a \textit{control function approach} (CFA) and assume that instruments $Z_i = (z_{i1}',...,z_{iT}')$ and potentially time-varying control variables  are available such that once control variables are conditioned on, the residual is mean independent of the regressors.
The main contribution of this paper is to  develop a two-step method which identifies $\E(\mu_i)$, the average partial effect\footnote{As defined in \cite{w05}, the APE is $\E_q[ \partial \E(y_t|x_t,q_t) / \partial x|_{x_t = \bar{x}}]$ where the outer expectation is over the vector of unobservables $q_t = (\mu, \alpha, \epsilon_t)$.} (APE),  in (\ref{eq:model}) under time-varying endogeneity when the data satisfy such a control function approach assumption.
The identification argument proceeds as follows. The CFA implies that the conditional mean of $\epsilon_{it}$ given the collection of regressors and control variables at all periods is in fact a function of the control variables only. 
Therefore the first step identifies nonparametrically the first differences of these control functions using the residuals of individual-specific linear regressions of the vector of differenced outcomes on differenced regressors. %This procedure differences out the random coefficients %use of within group variation?
An \textit{invertibility} assumption must be imposed to recover the control functions from these residuals. Once the control functions are known, individual-specific linear regressions are run and a cross-section average identifies the APE. Additional steps extend the argument to identify  $\E(\alpha_i)$ as well as higher order moments of the random coefficients, using existing results in the literature.

Since this identification argument relies on several key assumptions, we study them in details and connect them to existing conditions in the literature. The most important and least standard ones are the control function approach and the invertibility assumptions.  Under the CFA, control variables control for the correlation between the residual at period $t$ and the collection of regressors at all time periods: we %compare this assumption to existing ones in the literature and 
argue that it can be used to capture contemporaneous endogeneity as well as predeterminedness. Additionally, CFA does not restrict the joint distribution of the instruments and the random coefficients but only that of $(X_i,Z_i,(\epsilon_{i \, t})_{t \leq T})$: it remains comparable to exogeneity restrictions in models without random coefficients. The second assumption is a high-level assumption   related to both support and rank conditions, as is common in the CFA literature \citep{in09}. However we show with examples that sufficient conditions include cases where the instrument has small support, as in \cite{fhmv08}, and even cases where the instrument is discretely distributed. We provide an example  where a regressor following a Markov process is not strictly exogenous but predetermined and show that all assumptions hold.

The proof of identification of $\E(\mu_i)$ is constructive and  suggests a natural multi-step estimator, taking sample analogs of each of the identification steps after estimating the control variables. The second  contribution of this paper is the analysis of  this estimator under a particular specification of the control variables. We show that it is consistent and asymptotically normal. The challenge in deriving the asymptotic properties of the estimator comes from the nonparametric regression estimators constructed using nonparametrically estimated regressors. % This relates to a broad literature on estimation with generated covariates %in which to our knowledge, there are no results directly applicable to our estimator on the asymptotic distribution of sample moments depending on nonparametric two-step sieve estimators.
%and as we use nonparametric two-step sieve estimators, our results \textcolor{blue}{resemble} \cite{hlr18}. %but differ in some ways that we explain in Section ?

%for struct of this paragraph, used "Teacher-to-classroom assignment and student achievement" from Graham et al
The results presented in this paper have various limitations. %, several of which could be addressed in future work. 
First throughout the paper we maintain  $T \geq d_x + 2$, where $d_x$ is the dimension of the regressor $x_{it}$ (without the constant term). To identify the APE when the regressor is scalar, the panel must therefore include 3 periods. As explained below, this is common in the analysis of CRC panel models \citep{ab12, ch92}, however an approach developed in \cite{gp12} may be extended so as to obtain identification when $T = d_x +1$. The associated estimation procedure   being outside the scope of this paper, we leave it for future work. Furthermore, a crucial implication from a modeling perspective comes from the use of the control function approach.  Indeed,
first-stage equations for which we can identify valid control variables typically have scalar unobserved heterogeneity when they model a scalar endogenous regressor. This highlights an underlying imbalance in the degree of heterogeneity between the outcome equation and the implicit selection (first-stage) equation. We further comment on this at the end of Section \ref{sec:CFAcomment}. 
%Finally, the control function approach requires some minimum variation in the instrument. When the instrument is discrete, we show that the minimum number of points on the support of the instrument conditional on the control variables is $(T-1)/(T-1 - d_x)$ which \textcolor{blue}{precludes the use of a binary instrument when there are as little as 2 endogenous regressors}.

\medskip

\textit{\textbf{Related literature.}} 
This paper directly contributes to the literature on CRC panel  models such as (\ref{eq:model}) placing no restrictions on the distribution of $(X_i,\mu_i', \alpha_i)$. A seminal paper in this literature  is \cite{ch92} studying a model similar to  (\ref{eq:model}) in a fixed $T$ setting with an additional additively separable parametric term. Under the condition\footnote{The discussion on the number of periods is for (\ref{eq:model}) where we impose the inclusion of a constant regressor.} $ T \geq d_x + 2$, the author derives the semiparametric efficiency bound for $\E(\mu_i', \alpha_i)$ and provides an efficient estimator. Important recent papers  studying comparable models include \cite{ab12} and \cite{gp12}. \cite{ab12}  focus on identifying the conditional variance and distribution of $(\mu_i', \alpha_i)'$ under various restrictions on the serial dependence of $\epsilon_{it}$. \cite{gp12} relax a regularity assumption in \cite{ch92} which imposed sufficient variations of the regressors over time, and, as mentioned above, allow for $T=d_x+1$. More recent contributions are \cite{v20} and \cite{su21}. We differ from all these papers in that we allow for time-varying endogeneity. 

A wide panel data literature studies models with time-invariant unobserved heterogeneity $(\mu_i', \alpha_i)$ and time-varying residuals $\epsilon_{it}$, allowing for time-varying endogeneity and adopting a fixed effect approach. % on the joint distribution of $(\mu_i', \alpha_i)$, $X_i$ and potential instruments. 
A well-known example is a linear regression model with an additive fixed effect, in which case the unobserved heterogeneity is scalar. One may use the fixed-effect instrumental variable estimator to consistently estimate the slope, see \cite{w05c}. An example using the control function approach when there is sample selection is \cite{w95}. Unlike these papers, we allow for the slope to be heterogeneous as well.
More recently, results have been developed for identification of nonseparable models in panel data with time-invariant unobserved heterogeneity and time-varying residuals, see, e.g., \cite{ab16}, \cite{abb17}.  Using results from  \cite{hs08}, these papers recover the distribution of the unobserved heterogeneity. The number of time periods required for identification is typically higher than we need: for instance \cite{ab16} need 5 time periods to identify the distribution of a bivariate fixed effect.  Focusing on the particular functional specification of the model (\ref{eq:model}), we need fewer time periods and our identification argument suggests a computationally simpler estimator. See also \cite{fglv22}.
%not citing Cherno et al - no time varying endogeneity

An existing  alternative approach for  panel models with time-invariant unobserved heterogeneity  restricts its joint distribution with the regressors. This corresponds to the  \textit{correlated random effect approach} (CRE).
It has been used in analyses of linear random coefficients models, see \cite{w05b}, but also of nonseparable panel data models. For instance \cite{am05} identify the local average response function assuming that the vector of both time-invariant unobserved heterogeneity and residual at time $t$ is independent of the regressor at time $t$ conditional on an instrument. %The assumptions they impose seem to allow for the distribution of the epsilon to vary with t although there is no t index on the distribution function. but the reasoning is t by t. and the exchangeability also looks only at the distribution of that time period, there is no reasoning across t. A little strange...
Some papers have developed a CRE approach in linear random coefficient models with time-varying endogeneity as well. For instance, \cite{mw08} study a linear panel random coefficient model and assume that the random coefficients are conditionally mean independent of the detrended instrument to show that the fixed-effect instrumental variables estimator is consistent to the average partial effect. See also \cite{mw16} and \cite{l21}.
Other recent examples of the CRE approach include \cite{bh09}, %actually bh09 does not have a model. thus the marginal effect they estimate is the usual marginal effect if there is a model with time varying stuff and this stuff has distrib (or mean depending on the model) independent of x_t
\cite{mkv11}, % \cite{ab16}, 
\cite{ghpp18}.
In contrast to these papers, we adopt a fixed-effect approach and do not restrict the joint distribution of the unobserved heterogeneity $(\mu_i', \alpha_i)$, the regressors and the instruments. The multiplicative specification in (\ref{eq:model}) allows us to ``difference out'' $(\mu_i', \alpha_i)$ entirely in one step and  recover the term capturing the time-varying endogeneity.

%Another approach to the identification problem in \ref{eq:model} when $x$ is correlated with $\epsilon$ could be to use consider the panel data as repeated corss-section of data, and use results in the existing literature on identification in RC models with endogenous regressors, or even nonsepaable models with endogeneous regressors. The second one suggests estimators that are also more complex to implement.
An alternative approach to identification of the APE in (\ref{eq:model}) would be to consider each cross-section separately and to apply existing results on identification of cross-section random coefficient models with endogenous regressors. For instance, \cite{w97}, \cite{w03} and \cite{hv98} identify the average partial effect imposing an exclusion restriction on the random coefficients and  a first-stage equation with homogeneous impact of the instruments on the regressors. More recently, \cite{mt16} use a nonseparable first-stage equation similar to \cite{in09} %, thus allowing for heterogeneity in the impact of the instrument, 
and retrieve the conditional APE  imposing a control function approach on the entire vector of random coefficients: in particular this vector is assumed independent of the instruments. A similar approach can be  found in \cite{ns21} who study a set of regressions which encompasses a random coefficient model where the control function approach is applied on the random coefficients. %A related result is in \cite{hhm17} analyze a triangular model with random coefficients in both stages, independent of instruments and exogenous regressors: they show nonidentification of the distribution of the random coefficients in general and that an independence condition between the random coefficients is required for identification.
In these papers, the control function approach restricts the joint distribution of all random coefficients, the regressors and the instruments. In contrast, using panel data and imposing the random coefficients multiplying the regressors to be time-invariant allows us to obtain identification while imposing the control function approach on the residual $\epsilon_{i \, t}$ only. 

%Our approach to endogeneity is to use a control function, thus this paper also relates to the control function literature. See e.g \cite{npv99,bp03,in09}.
	
\medskip

The structure of the paper is as follows. Section \ref{sec:id_model} defines  the model, explains briefly the two-step argument proving identification of $\E(\mu_i)$ and introduces our assumptions. Section \ref{sec:assn} discusses at length these assumptions and their empirical content.
Section \ref{sec:identification} 
provides the formal identification result and discusses other moments of the random coefficients. % as well as other models closely related to (\ref{eq:model}).
In Section \ref{sec:estim}, we define our proposed estimator and provide its asymptotic properties.
Finally, Section \ref{sec:illus} turns to an empirical illustration of a labor supply model with random coefficients. We give conditions under which our main assumptions hold and with a dataset from \cite{z97}  estimate the average elasticity of intertemporal substitution.
The proofs for all sections are in the Appendix, as well as some Monte Carlo simulations showing favorable finite sample properties of the estimator, and additional comments.

\section{Model \& Intuition}\label{sec:id_model}

We first introduce some notations used throughout the paper. Random variables are indexed with $i$ or $it$. We denote with $\mathcal{S}_{A_i}$ the support of a random variable $A_i$, $\mathcal{S}_{A_i|B_i=B}$ the support of $A_i$ given that the random variable $B_i$ is equal to $B$. $I$ is the identity matrix and its dimension is usually clear from the context.

For a sample of units indexed by $i$, for $i = 1,..,n$, the outcome variable $y_{it} \in \R$ in period $t=1,..,T$ is given by
\begin{equation}
		y_{it}  = \, x_{it}' \, \mu_i   + \alpha_i + \epsilon_{it}, \quad \E(\epsilon_{it}) = 0, \tag{\ref{eq:model}}
\end{equation}
where $x_{it} \in \R^{d_x}$ is a vector of observed variables which does not include the constant regressor, $\epsilon_{it}$ is a time-varying disturbance, and $(\alpha_i,\mu_i')' \in \R^{d_x+1}$ is a time-invariant vector capturing individual unobserved heterogeneity. The reason why the constant term is not included in $x_{it}$ is discussed at the beginning of Section \ref{sec:FD}. We focus on short panels, where $T$ is fixed and $n$ large. The constraint $T \geq d_x + 2$ will be maintained throughout the paper. More details on this condition are given in Section \ref{sec:assn_CRC_matrix}. Denoting by $y_i = (y_{i1},..,y_{iT})'$ the vector of outcomes of unit $i$, $X_i = (x_{i1},..,x_{iT})'$ the matrix of regressors, and $\epsilon_i = (\epsilon_{i1},..,\epsilon_{iT})'$ the vector of error terms, we can rewrite (\ref{eq:model}) as 
$$y_i = X_i \mu_i + \alpha_i 1_{T} + \epsilon_i,$$ 
where $1_{T}$ is the vector of size $T$ where each component is equal to $1$. The parameters of interest we focus on are the average effects $\E(\mu_i)$ and $\E(\alpha_i)$.

A standard assumption in the panel CRC literature is strict exogeneity of the regressors, that is, $\E(\epsilon_{i \, t}|X_i, \alpha_i, \mu_i) = 0$. %  in Model (\ref{eq:model}). 
As pointed out in the introduction, this assumption does not allow for the presence of time-varying omitted variables or other time-varying sources of correlation between the regressors and the residual $\epsilon_{it}$. We seek to relax this condition.
Consider for instance an education production function with school class size as  input. The impact of teacher's attention, thus of  class size, on students' achievements is likely heterogeneous as it potentially depends on teacher's quality. For data observed at the class level across schools and cohorts, many omitted variables are correlated with class size, e.g., 
through school budget, such as students, parents and community characteristics. These omitted factors vary over cohort and school and thus cannot be captured by an additive fixed effect.
This framework also applies when $t$ captures alternative  group structures. To study the impact of mother smoking habit on infant birth weight with a panel of mothers with multiple births  (see, e.g., \cite{a06}), data indexed by $it$ describe the $t^{\text{th}}$ child of mother $i$. The impact of a mother's level of smoking is likely heterogeneous across mothers but moreover some omitted variables related to a mother's health behavior may impact the birthweight and vary across children with the smoking level.

In both examples, omitted variables are correlated with the regressor of interest yet can not be captured by additive fixed effects. To relax the strict exogeneity condition and allow for such correlations, we assume instead that there exist control variables $v_{it}\in\R^{d_v}$ which are known or identified functions of regressors and  instrumental variables $z_{it} \in \R^{d_z}$, and satisfy a conditional mean independence condition in line with the control function approach. This condition is our first assumption, Assumption \ref{assn:idCFA}, and will be formally stated in Section \ref{sec:CFAcomment} together with discussions of the serial correlations it allows for and of how to construct control variables. This assumption implies that for $V_i = (v_{i1}',.. v_{iT}')'$, for all $0\leq t \leq T$ there exists a  random variable $u_{it} $ such that
\begin{equation}\label{eq:cfa_bf_assn}
	\epsilon_{it} = f_t(V_i) + u_{it}, \text{ for some function } f_t  \text{ where } \E(u_{i \, t}|X_i, V_i) = 0 \text{ and } \E(f_t(V_i)) = 0.
\end{equation}
%i.e., $u_{it}$ is strictly exogenous. %Note that we allow  for $u_{it}$ to be correlated with $(\alpha_i, \mu_i)$.
To identify the average effects $\E(\mu_i)$ and $\E(\alpha_i)$, we suggest a stepwise procedure which will isolate the unobserved heterogeneity from the regressor and the control variables. We describe this procedure below and define $f(V_i)\, = \, (f_1(V_i),.. \, , f_T(V_i))'$ and $Z_i = (z_{i1},...,z_{iT})$.

\subsection{First-Differencing}\label{sec:FD}

In (\ref{eq:model}), we deliberately treat the additive fixed effect $\alpha_i$ separately, that is, the constant regressor is not included in $x_{it}$.
Two alternative approaches could be used instead, or, specifically, the model could have been written in two other ways.
The first is to change (\ref{eq:model}) to write instead $y_{it}  = \, \tilde{x}_{it}' \, \tilde{\mu}_i + \epsilon_{it}$ where $\tilde{\mu}_i = (\alpha_i,\mu_i')'$ and $\tilde{x}_{it} = (1,x_{it}')'$. However one of the main assumptions needed for identification, Assumption \ref{assn:id_M_GLn}, does not hold if the regressor is replaced with $\tilde{x}_{it}$. This assumption could be rewritten using $\tilde{x}_{it}$ so as to guarantee identification when it holds, but its formulation and analysis would be more complex. We thus prefer to treat the constant regressor separately from the other regressors.
Note that with our use of the control function approach, Model (\ref{eq:model}) can be written $y_{it} = x_{it}' \, \mu_i  + \alpha_i  + f_t(V_i) + u_{it}.$ Consequently, a second approach would be to define instead $\tilde{\alpha} = \E(\alpha_i)$, $\tilde{f}_t(V_i)= \E(\alpha_i|V_i) - \E(\alpha_i) + f_t(V_i)$ and $\tilde{u}_{it} = [\alpha_i - \E(\alpha_i|V_i)] +  u_{it}$%write $\alpha_i = \E(\alpha_i) + [ \E(\alpha_i|V_i) - \E(\alpha_i)] + [\alpha_i - \E(\alpha_i|V_i)]$, merge $[ \E(\alpha_i|V_i) - \E(\alpha_i)]$  together with $f_t(V_i)$, and merge $[\alpha_i - \E(\alpha_i|V_i)]$ together with $u_{it}$
, so as to obtain  $y_{it} = x_{it}' \, \mu_i  + \tilde{\alpha}  + \tilde{f}_t(V_i) + \tilde{u}_{it}$. However, the control function approach mentioned below imposes the condition  $\E(\alpha_{i}|X_i, V_i) =  \E(\alpha_i|V_i) $, which is conflicting with the fixed effect approach we adopt.  We thus treat $\alpha_i$ separately from the other unobserved heterogeneity terms $f_t(V_i)$ and $u_{it}$, i.e., from $\epsilon_{it}$.
It implies that the functions $f(V_i)$ and $\E(\alpha_i|V_i)$ are not separately identifiable. But the normalization $\E(f(V_i)) \, = \, 0$  can be leveraged to identify $\E(\alpha_i)$ once we identify first differences of $f$, as will be clear in Section \ref{sec:maindid}.

Exploiting the time invariance of $\alpha_i$, we take time differences to eliminate this term and focus on identification of $\E(\mu_i)$. % We will then later obtain identification of $\E(\alpha_i)$ using $\E(f(V_i)) \, = \, 0$.
Thus we now focus on a first-differencing transformation of the model,
\begin{equation}
\dot{y}_{i t}   = \, \dot{x}_{it}^{\, \prime} \, \mu_i  + \, g_t(V_i) + \, \dot{u}_{it}, \label{eq:modelwTdiff}
\end{equation}
with %\footnote{We use the notation $\dot{y}$ rather than the standard $\Delta y$ to lighten notations in the equations.}  
$\dot{y}_{i t} = y_{i \, t + 1} - y_{i t}$, $\dot{x}_{it} = x_{i \, t + 1} - x_{i t}$, $g_t(V_i) = f_{t+1}(V_i) - f_{t}(V_i)$, and $\dot{u}_{it} = u_{i \, t+1} - u_{it}$ for $ t \leq T-1$.
In vector form, define the $(T-1) \, \times d_x$ matrix $\dot{X}_i = ( \dot{x}_{i1},.. , \dot{x}_{i \, T-1})'$ and the $(T-1) \times 1$ vectors $\dot{y}_i = (\dot{y}_{i1},.. \, , \dot{y}_{i \, T-1})'$, $g(V_i)\, = \, (g_1(V_i),.. \, , g_{T-1}(V_i))'$ and $\dot{u}_i = (\dot{u}_{i1},.. \, , \dot{u}_{i \, T-1})'$. By assumption, $\E(\dot{u}_i | X_i , V_i) =0$ and
Equation (\ref{eq:modelwTdiff}) can then be rewritten 

\begin{equation}\label{eq:modelwTdiffVector}
\dot{y}_i = \,\dot{X}_i  \mu_i + \, g(V_i)  \, + \, \dot{u}_i.
\end{equation}

%\textcolor{blue}{Note that other types of differencing can be used. More generally, one can define a matrix $\Pi$ such that $\Pi \, 1_T = 0$, the transformation of the model to consider is now $\dot{y}_i = \Pi y_i$, $\dot{X}_i = \Pi X_i$ etc. Say more? what about not getting the same vector size}

\subsection{Two-Step Identification}\label{sec:intuition_id}

We now discuss identification of $\E(\mu_i)$ and introduce additional assumptions. Formal results are provided in Section \ref{sec:identification}. Since the random coefficients $\mu_i$ are heterogeneous, standard estimators such as TSLS are not consistent without further restrictions. 
However an intuitive idea is to consider each unit $i$ separately: for each we observe several time periods and can run a unit-specific linear regression of differenced outcomes $\dot{y}_{it} =  y_{i \, t + 1} - y_{i t}$ on differenced regressors $\dot x_{it} =  x_{i \, t + 1} - x_{i t}$ over observations $t=1, .., T-1$. This produces a unit-specific estimator $\beta_{i}^{OLS}$.
first-differencing eliminates the additive fixed effects $\alpha_i$ and if the regressors are strictly exogenous, the average of  $\beta_{i}^{OLS}$ across units  is $\E(\mu_i)$. Note that to run such a regression, we need more observations than regressors, i.e., $T-1 \geq d_x$  as otherwise, if the matrix of observations $\dot X_i$ is of full rank, there is not a unique solution to the least squares optimization problem. % For regularity reasons, i.e., to ensure that the expectations exist, we need actually $T-1 >d_x$)

If there is time-varying endogeneity and (\ref{eq:cfa_bf_assn}) holds, % that is, for each unit the correlation between $x_{it}$ and $\epsilon_{it}$ varies across observations, 
the argument described above does not identify $\E(\mu_i)$. 
In fact, since the  residual $\epsilon_{it}$  is composed of the two terms $u_{it}$ and $f_t(V_i)$, the unit-specific $\beta_{i}^{OLS}$  is the sum of three terms: $\mu_i$, the individual slope, $\beta_{i}^{OLS, \, g}$, the coefficient of the regression of $g_t(V_i)$ on  $\dot x_{it}$ and  $\beta_{i}^{OLS, \, u}$,  the coefficient of the regression of $\dot u_{it}$ on $\dot x_{it}$. The regressors are strictly exogenous for the new residual $u_{it}$, thus $\beta_{i}^{OLS, \, u}$  averages to $0$ across units. However this does not hold for  $\beta_{i}^{OLS, \, g}$. 
Thus to recover $\E(\mu_i)$, we  suggest a two-step identification argument. It is not uncommon to use two-step approaches to handle endogeneity, as in the construction of the TSLS estimator for a linear regression with endogenous regressors. Our two steps are more involved due to $\mu_i$ varying across units and the nonparametric take on the endogeneity. The first step will focus on recovering the functions $g_t = f_{t+1} - f_t$ and the second on recovering $\E(\mu_i)$. We now provide  an intuitive explanation of what these steps are.

\textbf{\textit{Step 1:}} To isolate $g$, we use variations across units. Let us first fix a unit $i$ and compute the residuals of the unit-specific linear regression of  $\dot y_{it} $ on $\dot x_{it}$. Denote the vector of these residuals with $R_i$, it is of size $T-1$ and obtained by multiplication of the vector $\dot{y}_i$ by the residual-maker matrix. There are two components in this residual: $R_{i}^u$, the vector of residuals of the regression of $\dot u_{it}$ on $\dot x_{it}$, and $R_{i}^g$, the vector of residuals of the regression of $g_t(V_i)$ on $\dot x_{it}$. We have
$$R_i = R_{i}^g+ R_{i}^u.$$
Note that $T-1 >d_x$ is now needed. Indeed, if $T-1=d_x$  and $\dot X_i$ is of full rank, the regression fit is perfect and $R_i = 0$ independently of $g$ and $u$. There is no residual variation left that we can exploit because any unit-specific vector is a linear combination of the regressors.

Both $g_t(V_i)$ and $\dot u_{it}$ vary over time, i.e., across observations. However $g_t$ is a function of $V_i$ while $\dot u_{it}$ is mean independent of $X_i$ given $V_i$.  
Thus focusing on the units  $i$ such that $V_i = V$ with $V$ a fixed value, by mean independence the residuals of the regression of $\dot u_{it}$, that is, $R_{i}^u$, average to $0$ in this sub-population.
%Thus we use these cross-sectional variations to separately identify the $g$ term. 
On the other hand, the average over this same sub-population of the vector of residuals of the regression of $g_t(V_i)$ on $\dot x_{it}$, that is, of $R_{i}^g$, is an average of regression residuals over regressions which all have the same vector of outcomes $g(V)$, but have different values of regressors  thus  of the residual-maker matrices. Therefore, this vector of outcomes $g(V)$ can be isolated by multiplication of the average of $R_{i}^g$ in the subpopulation of interest, by the inverse of the average of the residual-maker matrices in the same population. Lastly, this is also true if we replace $R_{i}^g$ with the observed $R_{i}$ since  $R_{i}^u$ averages to $0$.

\textbf{\textit{Step 2:}} Once the function $g$ is identified on the support of $V_i$, we can recover $\E(\mu_i)$ using the intuition described at the beginning of this section. Recall that the difference between $\beta_{i}^{OLS}$ and $\mu_i$ is $\beta_{i}^{OLS, \, g} + \beta_{i}^{OLS, \, u}$. Since $\beta_{i}^{OLS, \, g}$ is the coefficient of the regression of $g_t(V_i)$ on  $\dot x_{it}$ for unit $i$ and $g$ is identified, it is known. Since $\beta_{i}^{OLS, \, u}$ averages to $0$ across units, we consider once again cross-section averages:  $\E(\mu_i)$  is the average of the difference between $\beta_{i}^{OLS}$ and $\beta_{i}^{OLS, \, g}$  across units.

\medskip

We now formalize the steps in the argument described above. Crucial to this argument is the use of two matrices constructed with the regressors: the unit-specific regression matrix $Q_i= (\dot{X}_i' \dot{X}_i)^{-1} \dot{X}_i'$, and the unit-specific residual-maker matrix $M_i = I_{T-1} - \dot{X}_i (\dot{X}_i' \dot{X}_i)^{-1} \dot{X}_i'$  if $\dot{X}_i$ is of full rank or $ M_i = I - \dot{X}_i  \dot{X}_i^{ +}$ if not, where $\dot{X}_i^{ +}$ is the Moore Penrose inverse. % (implying $ \dot{X}_i \dot{X}_i^{ +} \dot{X}_i = \dot{X}_i $).
The matrix $Q_i$ is defined only if $\dot{X}_i$ has full column rank and if so, is of size $d_x \times (T-1)$ while $M_i$ is of size $(T-1) \times (T-1)$. 
These unit-specific matrices  have been used in analysis of CRC panel models with short $T$, see e.g \cite{ab12}, \cite{gp12}, but also by a literature focusing more on the large $T$ case, see e.g \cite{s70,ps95,h14}. Note that $R_i = M_i \dot y_i$ and $\beta_{i}^{OLS} = Q_i \dot y_i$.
%The matrix $Q_i$ is defined only if $\dot{X}_i$ has full column rank. If so, $Q_i = (\dot{X}_i' \dot{X}_i)^{-1} \dot{X}_i'$ and is of size $d_x \times (T-1)$. The matrix $M_i$ satisfies  $ M_i = I_{T-1} - \dot{X}_i (\dot{X}_i' \dot{X}_i)^{-1} \dot{X}_i' \, $ if $\dot{X}_i$ is of full rank or $ M_i = I - \dot{X}_i  \dot{X}_i^{ +}$ if not, where $\dot{X}_i^{ +}$ is the Moore Penrose inverse (implying $ \dot{X}_i \dot{X}_i^{ +} \dot{X}_i = \dot{X}_i $). $M_i$ is the orthogonal projection matrix projecting onto the space orthogonal to the columns of $\dot{X}_i$. It is of size $(T-1) \times (T-1)$. By definition $M_i \dot{X}_i \mu_i = 0$ and $Q_i \dot{X}_i \mu_i = \mu_i$. Note too that these matrices are observed.
The identification argument relies on averages of quantities functions of these matrices: we assume that such moments exist and this corresponds to our second assumption, Assumption \ref{assn:id_finiteE}. See more details in Section \ref{sec:assn_CRC_matrix}.

\textbf{\textit{Step 1:}}  Left multiplication of (\ref{eq:modelwTdiffVector}) with $M_i$  gives
\begin{align}
	R_i = M_i \dot{y}_i \,  =  M_i g(V_i) + \, M_i \dot{u}_i , \qquad  \E( M_i \dot{u}_i | X_i , V_i) =0, \label{eq:model_id_g} 
\end{align}
where the conditional expectations are $0$ because $M_i$ is a function of $\dot{X}_i$. The discussion above suggests focusing on the subpopulation with the same value of $V_i$, which gives
\begin{equation}\label{eq:almost_closedform_g}
	\E(M_i \dot{y}_i |  V_i = V) \, = \, \E( \, M_i | V_i = V) \, g(V) \, = \, \mathcal{M}(V) g(V),
\end{equation}
where we write $\mathcal{M}(V) = \E( M_i | V_i = V)$.
If $\mathcal{M}(V)$ is invertible for a given value $V$ on the support of $V_i$, then (\ref{eq:almost_closedform_g}) gives the following closed-form expression, $g(V) = \mathcal{M}(V)^{-1}	\E(M_i \dot{y}_i |  V_i = V)$, that is, we retrieve $g(V) = \mathcal{M}(V)^{-1}	\E(R_i |  V_i = V)$. Thus our last assumption, Assumption \ref{assn:id_M_GLn} (see Section \ref{sec:IA}), will be invertibility of $\mathcal{M}(V_i)$ almost surely. Note that writing $\E(M_i \dot{y}_i | X_i , V_i) = M_i g(V_i)$  does not identify $g(V_i)$ because $M_i$ is a projection matrix and  singular whenever $\dot X_i\neq 0$.

\textbf{\textit{Step 2:}}  Left multiplication of (\ref{eq:modelwTdiffVector}) with $Q_i$  gives 
\begin{align}
	&\beta_{i}^{OLS} = Q_i \dot{y}_i  =   \mu_i + \, Q_i g(V_i) +  Q_i  \dot{u}_i, \text{ and }\E( Q_i  \dot{u}_i | X_i , V_i)  =0, \label{eq:model_id_mu} \\
	& \Rightarrow  \E(\mu_i)  = \E\left(Q_i \dot{y}_i  -  Q_i g(V_i) \right), \label{eq:intuitionmu}
\end{align}
which identifies $\E(\mu_i)$. Note that (\ref{eq:intuitionmu}) can be rewritten $\E(\mu_i) = \E(\beta_{i}^{OLS} - \beta_{i}^{OLS, \, g})$.

This short exposition gave the main intuition of the identification strategy. We postpone the formal statement of the  results to Section \ref{sec:identification}  and first discuss the assumptions we introduced. %As one of these assumptions is new in the literature, we provide some intuition and examples discussing relatability of the setup of this paper.

\section{Empirical Content of the Assumptions}\label{sec:assn}

%\textcolor{blue}{Title: Relatability /  Assessing fitness/appropriateness/suitability of the framework}

This section discusses the main assumptions on which the identification results are based. As some of these assumptions are nonstandard, we provide intuition and introduce situations in which they might or might not hold.  We start by discussing our use of the control function approach.

\subsection{Control Function Approach}\label{sec:CFAcomment}

%\textcolor{blue}{I might want to give title to the comments or ``points'', or enumerate them.}

%point 1 of the list
To identify structural parameters of models with endogenous regressors, two well-known approaches are the instrumental variable approach and the control function approach. The assumption we use is in the tradition of the control function approach (CFA) literature.  Under CFA assumptions, conditional on the control variables, the endogenous regressor is either independent or mean independent of the residual. In nonlinear models this has been leveraged to identify various parameters of interest, %such as the average structural function, the quantile structural function, or the outcome function, depending on the model, 
see e.g \cite{npv99}, \cite{in09}, \cite{fhmv08}, \cite{df15}, \cite{t15}, etc. 
%One known limitation of control function techniques is that they do not apply to cases where the endogeneity arises because of simultaneity, i.e is due to the simultaneous determination of $y_{it}$ and $x_{it}$, unless the system satisfies a restrictive \textit{control function separability} condition described in \cite{bm14}.
The identification results of this paper similarly rely on control variables to control for the dependence between the regressors and the residual but the exact assumption that we maintain, %on the control variables 
although it remains a conditional mean independence assumption,
differs from usual versions of CFA assumptions due to the panel aspect of the data. % It is however a conditional mean independence assumption and can be seen as a panel data version of the assumption in e.g \cite{npv99}. 
It is formally stated below in Assumption \ref{assn:idCFA} (\ref{assn:idCFA_CFA}). %We also assume that the data is an i.i.d sample.

\begin{assumption}\label{assn:idCFA}
	\
	\begin{enumerate}[noitemsep,nolistsep]
		\item \label{assn:idCFA_continuity}
		$\left(X_i, Z_i, \epsilon_{i}, \mu_i, \alpha_i \right)$ is i.i.d.\ across $i$, and (\ref{eq:model}) holds with $\E(\epsilon_{it})=0$,
		
		\item \label{assn:idCFA_CFA}
		There exist scalar-valued functions $(f_t)_{t \leq T}$ and identified functions $(C_t)_{t \leq T}$ such that, defining $v_{
			it} = C_t(X_i, Z_i) \in \R^{d_v}$,
		\vspace{-14pt}
		\begin{equation}\label{eq:assnCFA}
			\E(\epsilon_{it} \, | \, x_{i1},.. x_{iT}, v_{i1},.. v_{iT}) = f_t(v_{i1},.. v_{iT})
		\end{equation}
	\end{enumerate}
\end{assumption}

Note that (\ref{eq:assnCFA}) can be rewritten $\E(\epsilon_{it}|X_i, V_i) = f_t(V_i)$. It can be compared in particular to an assumption maintained in \cite{npv99}, which studies a triangular simultaneous equation model in a cross-sectional setting where the regressors $x_{it}$ in the outcome equation include endogenous regressors $x_{it}^{en}$ and exogenous ones $x_{it}^{ex}$. Using our notations and applying their framework to a fixed period $t$ for the sake of the comparison, the control variable $v_{it}$ used by the authors is the residual of a nonparametric reduced form first-stage regression of  $x_{it}^{en}$ on instruments $z_{it}$ and  $x_{it}^{ex}$.  %, i.e, $v_{it} = x_{it} - \E(x_{it} |z_{it})$, 
They assume that $\E(\epsilon_{it}  | x_{it}^{en}, z_{it},v_{it})= f_t(v_{it})$, which we call CS-CFA. CS-CFA implies $\E(\epsilon_{it}  | x_{it},v_{it})= f_t(v_{it})$. The main difference between this equality and (\ref{eq:assnCFA}) is that the latter conditional expectation conditions on the vector $(X_i,V_i)$ and not $(x_{it},v_{it})$: we discuss later how this brings additional flexibility and describe now how the practitioner can construct control variables. %Note that we normalized $\E(\epsilon_{it})=0$ so $\E(f_t(V_i))=0$ for all $t \leq T$. This is without loss of generality since a constant is not separately identifiable from $\E(\alpha_i)$. 

%point 2 of the list
In the CFA literature, control variables are typically provided by a first-step selection equation. Although we do not rely on a particular specification for the control variables, a natural approach is to construct,  for each period, control variables studied in the cross-section CFA literature using cross-sectional data for this period so as to obtain $v_{it} = C_t(x_{it}, z_{it})$. For instance, following \cite{npv99} % $v_{it}$ would be the residuals of the nonparametric  cross-section  regression of $x_{it}$ on $z_{it}$. 
with a linear first stage, we  obtain the following two-equation model,

\vspace{-25pt}

\begin{align}
	y_{it} & = \, x_{it}^{ex \ \prime} \, \mu_i^{ex} + x_{it}^{en \ \prime} \, \mu_i^{en} + \epsilon_{it}, \label{eq:modelX1X2} \\
	x_{it}^{en} & =  \pi_t x_{it}^{ex} + \gamma_t z_{it} + v_{it}, \qquad \E(v_{it}| x_{it}^{ex},  z_{it} )=0 \label{eq:modelX1X2lin1stage}  
\end{align}
where $ x_{it}^{ex} \in \R^{d_1}$ is a vector of exogenous regressors and $x_{it}^{en} \in \R^{d_2}$ a vector of potentially endogenous regressors. Define  $x_{it}=(x_{it}^{ex \ \prime},x_{it}^{en \ \prime})'$, $d_x = d_1+d_2$. If the instrument satisfies a natural extension of CS-CFA, i.e, $\E(\epsilon_{it} | Z_i, X_i^{ex}, V_i) = h_t(v_{it})$, then (\ref{eq:assnCFA}) holds with $C_t(x,z) = x^{en} - \E(x_{it}^{en}|x_{it}^{ex} = x^{ex},z_{it} = z)$ and $f_t(V_i)=h_t(v_{it})$.
% This model will be our leading example when discussing estimation and deriving asymptotic properties of estimators in Section \textcolor{blue}{????}. %\ref{sec:estim}.
To study the impact of class size on students' achievements in a panel of class level data, a well-known instrument is  random variations in cohort size, see \cite{h00}, and control variables are the vector of residuals of a first-stage regression of class size or its logarithm on an instrument.
\cite{er99} use state cigarette taxes on birth data to instrument for mother level of smoking in an analysis of its impact on birth weight: the instrument can be used   on a panel of mothers with multiple births to construct $v_{it}$ as the residual of the first-stage linear regression.
For nonseparable first stages $x_{it}^{en} = h(x_{it}^{ex}, z_{it},\eta_{it}) $ where $(\epsilon_{it},\eta_{it}) \indep z_{it}$ and $\eta_{it}$ is scalar-valued, another well-known control variable studied in \cite{bp03} and \cite{in09}  is  $v_{it} = F_{x^{en}|x^{ex},z}(x_{it}^{en}|x_{it}^{ex}, z_{it})$.
%More recently in an analysis of treatment effects, \cite{n21} shows that when confounding factors are unobserved, the distribution function of proxies conditional on the treatment assignment and instruments satisfying an exclusion restriction is a control variable that satisfies a CFA type assumption. 
Aside from first-stage equations, the panel data literature has suggested various assumptions under which variables can be constructed without  instrumental variables yet satisfy conditional independence conditions similar to (\ref{eq:assnCFA}). Exchangeability conditional on individual-level heterogeneity is one such condition, see  \cite{am05}, or the use of sufficient statistics, as in \cite{ai19}. See also \cite{lps21} for a discussion of these conditions. 
%to compare w \cite{ai19}, the exchangeability we would need to obtain $\epsilon_{it} \indep x_{it} |v_{it}$ would be $(\epsilon_{it},x_{it}) \indep C_{it} | (\alpha_i, \mu_i) $
Assuming exchangeability in (\ref{eq:model}) would amount to assume that $\alpha_i,\epsilon_{it}|(x_{i1},...,x_{iT}) $ has the same distribution as $\alpha_i,\epsilon_{it}|(x_{i\tau(1)},...,x_{i\tau(T)})$ for any permutation $\tau$: this in general is contrary to our goal of addressing time-varying endogeneity where timing matters. However this motivates  strengthening  Assumption \ref{assn:idCFA} (\ref{assn:idCFA_CFA}) to $\E(\epsilon_{it}|X_i, V_i) = f_t(s(V_i))$ where $s(V_i)$ is scalar. For instance, one could have $s(V_i) = \frac{1}{T} \sum_{t=1}^T v_{it}$. An obvious benefit of this assumption is that it will alleviate the curse of dimensionality as all nonparametric estimations will be conditional  on $s(V_i)$ instead of $V_i$, a vector of dimension $d_V \times T$.

One important difference between (\ref{eq:assnCFA}) and  well-known applications of the CFA in nonseparable cross-sectional models is that we do not impose restrictions on the distribution of $(\mu_i, \alpha_i, (\epsilon_{it})_{t \leq T}, Z_i, X_i)$ but only on that of  $((\epsilon_{it})_{t \leq T}, Z_i,X_i)$. The CFA typically includes the entire vector of unobserved heterogeneity, see e.g \cite{npv99}, \cite{in09}, \cite{mt16}, etc.
Under Assumption \ref{assn:idCFA} (\ref{assn:idCFA_CFA}), the instrument need not be independent of the impact of the regressor $\mu_i$, and of $\alpha_i$.
Despite the added dimensions of unobserved heterogeneity in (\ref{eq:model}) when compared to linear panel data model with constant slope, additive fixed effects and time-varying endogeneity, the requirement on the instrument is only on its correlation with $\epsilon_{it}$ thus remains comparable to exogeneity assumptions of the CFA type in the latter model. Thus the practitioner, when choosing instruments, need not worry about additional exogeneity requirements.

%point 3 of the list
We now explore the serial correlations allowed in (\ref{eq:assnCFA}), noting that $\E(\epsilon_{it}|X_i,V_i)$ can depend on past and future values of $v$. %, that is, the control variable is not itself strictly exogenous. 
We consider the conceptually easier case where $v_{it} = C_t(x_{it}, z_{it})$. In the extension of CS-CFA mentioned above, i.e, $\E(\epsilon_{it} | Z_i, X_i^{ex}, V_i)= h_t(v_{it})$, the control variables are strictly exogenous with respect to the residual $\epsilon_{it} - h_t(v_{it})$. This approach is used for instance in \cite{w95}, which studies a linear panel data model with additive unobserved heterogeneity and sample selection issues. %In this paper the function $f_t$ is assumed to be linear and time-invariant. %should say what the CV are? it seems residual of regression for tobit case
However letting the conditional expectation $\E(\epsilon_{it}|X_i, V_i)$ take the exact form $f_t(V_i)$  allows for the error terms $\epsilon_{it}$ to be correlated with  past and future values of the regressors as long as this dependence is captured by $V_i$.  Of particular relevance in panel data is the correlation between contemporaneous error terms and future regressors, also called \textit{feedback}. To allow for a feedback effect, the assumption of strict exogeneity is often replaced with \textit{sequential exogeneity}. Sequentially exogenous or \textit{predetemined} regressors are uncorrelated with present and future values of the error term $\epsilon$ but may be correlated with past values of the error term, see e.g \cite{ah01}, \cite{a03}. In Section \ref{sec:relax_SE}, %\ref{sec:applic}
we show that we can formally leverage the flexibility of the functional form $f_t(V_i)$ to identify $\E(\mu_i)$ in a model with contemporaneously exogenous but predetermined regressors, which to the best of our knowledge is an open question in CRC panel models.
Moreover feedback can also occur when there is contemporaneous time-varying endogeneity. In the context of evaluating the impact of smoking on birthweight, a mother may change her smoking level after giving birth to a low birthweight infant ``$t$''. This can be captured by Assumption \ref{assn:idCFA} since in (\ref{eq:assnCFA})  $\epsilon_{it}$ is allowed to be correlated with the residuals of the first-stage regressions for the following births. Another example is in Section \ref{sec:illus}.
Note also that assuming $f_t(V_i) = h_t(v_{it})$ would not simplify the identification argument of Section \ref{sec:id_model} and cannot alleviate the curse of dimensionality as the left multiplication with $M_i$ in (\ref{eq:almost_closedform_g}) involves linear combinations of all periods. %term closed-form expression for $\E(\mu_i)$ would still involve conditional expectations given the entire vector $V_i$.

%point 4 of the list
A drawback of using the control function approach is that recovering valid control variables from first-stage equations is not always possible if there are additional unobserved sources of heterogeneity. Looking at a specific selection equation with two-dimensional unobserved heterogeneity, \cite{i07} shows that the conditional quantile of the regressor given the instrument is not a valid control variable anymore. There is consequently an apparent imbalance in the degree of heterogeneity between the implicit first stage and the main model equation (\ref{eq:model}).
Such imbalance can be found in other analyses of random coefficient models with endogenous regressors, see e.g \cite{hv98}, \citeauthor{w97} (\citeyear{w97}, \citeyear{w03}). As in \cite{mt16}, our approach allows for nonlinear first-stage equations which guarantee a certain degree of heterogeneity as partial effects vary across agents. It may be possible to use panel techniques to identify control variables in first-stage equations with more heterogeneity: for instance if a first stage is $x_{it} = m(z_{it})+ \xi_i + \eta_{it}$ where $\E(\eta_{it}|z_{it})=0$ then $v_{it} = \xi_i + \eta_{it}$ is identified if $T \geq 2$ and may be a valid control variable. However generalizing this approach to more complex models is outside the scope of this paper and is left for future research.

%point 5 of the list
%As a last comment on the control function approach assumption, we note that Model (\ref{eq:model}) and our general definition of the control variables as $v_{it} = C_t(x_{it}, z_{it})$ in (\ref{eq:assnCFA}) do not make explicit which of the regressors are endogenous. Nor does the condition $\E(\epsilon_{it} \, | \, X_i,V_i) = f_t(V_i).$ Furthermore, our general identification results do not depend on the specific definition of the control variables, which is determined outside of the model. 
%This deliberate lack of precision is offset by a gain in flexibility and our general identification argument applies to a variety of specifications, one of which will be presented in Section \textcolor{blue}{????}.%\ref{sec:applic}.

\subsection{Within and Between Matrices}\label{sec:assn_CRC_matrix}

To separately identify the function $g$ from $\mu_i$, the identification argument relies on the matrices 
$ M_i = I_{T-1} - \dot{X}_i (\dot{X}_i' \dot{X}_i)^{-1} \dot{X}_i' \, $ and  $Q_i = (\dot{X}_i' \dot{X}_i)^{-1} \dot{X}_i'$, called within- and between-matrices respectively, and considers expectations of vectors left-multiplied with these matrices. First, we require $\dot{X}_i$ to have full column rank with probability $1$. This section explains how the use of $M_i$ and $Q_i$ further imposes constraints on the data. 
Note that if $T = d_x + 1$ then $M_i = 0$: as discussed in Section \ref{sec:intuition_id}, if $\dot{X}_i$ is of full rank then there is no variation left in the residual $R_i$ to identify the function $g$. This fact motivates the constraint $T \geq d_x + 2$,  standard in the analysis of CRC panel models, see \cite{ch92,ab12}. 
A second issue is that the matrix $Q_i$ and functions of it may not have finite moments. Indeed the norm of $ (\dot{X}'\dot{X})^{-1}\dot{X}$ increases without bound as $\det(\dot{X}'\dot{X})$ approaches $0$, that is, as the columns of $\dot{X}$ are ``close to'' being linearly dependent. Thus if the regressors are continuous and for instance the density of $\dot{X}$ is positive at values such that $\det(\dot{X}'\dot{X}) = 0$, % which is possible even though it is assumed that $\dot{X}_i$ has full column rank with probability $1$, 
simple moments involving $Q_i$ such as $\E(||Q_i \dot{u}_{i}||)$ may not exist. % This may also be the case even if the density is equal to $0$, as it must converge to $0$ at a fast enough rate.
For these reasons, we  assume that the needed moments do exist. The formal assumption is stated below. Note that this does not affect $M_i$: it is an orthogonal projection matrix and thus bounded. 
\begin{assumption}\label{assn:id_finiteE}
	\
	\begin{enumerate}[noitemsep,nolistsep]
		\item \label{assn:full_rk} $\dot{X}_i$ is of full column rank a.s,
		
		\item $\E(||\dot{u}_{i}||) < \infty$, $\E(||Q_i \dot{u}_{i}||) < \infty$, $\E(||Q_i g(V_i)||) < \infty$, and 
		$\E(||(\mu_i', \alpha_i)||) < \infty$.
	\end{enumerate}
\end{assumption}

Assumption \ref{assn:id_finiteE} implies that almost surely, units have strictly positive $\det(\dot{X}_i'\dot{X}_i)$, i.e., regressors variations must be linearly independent and thus each regressor must have variations over time of its first-differences. But Assumption \ref{assn:id_finiteE} also implies that the proportion of units with low values of $\det(\dot{X}_i'\dot{X}_i)$ must be sufficiently low.  Such units with high persistence are called \textit{stayers}. %here the idea is that linear dpdce between the regressors is not an issue, it becomes one only if the regressors do not vary enough over time
This issue also arises when there is no time-varying endogeneity and is discussed in details in \cite{gp12}. 
The authors illustrate it with an extreme example where  they consider the case $d_x = 1$ and $x_{it} = s_i w_{it}$, with $s_i \sim \mathcal{U}[a,b]$ and $w_{it} \sim \text{i.i.d }\mathcal{N}(0,1)$. They show that for  fixed $b>0$ and independently of the number of periods, the moment $\E\left((X_i'X_i)^{-1}\right)$ is not finite if $a\leq 0$, i.e., if regressors with perfect persistence (when $s_i = 0$) have positive density despite the matrix of regressors being of full rank almost surely.

%As for this paper, the interpretation is the same.  Instead of $Q_i$, they use $\tilde{Q}_i = (X_i'X_i)^{-1}X_i$ as they do not take first differences. The example : in this case, for a fixed $b>0$ and independently of the number of time periods, the moment $\E(X_i'X_i)$ is not finite if $a\leq 0$, that is, if perfect persistence of the regressor  occurs with positive density (when $s_i = 0$). 
%The practical implication of this discussion is that for the desired moments to be finite, the sample should not have units with little within-variation of the regressor $\dot{X}_i$,  as for these units,  the term $X_i \mu_i$ may not be distinguishable from the term $\alpha_i$.
%In particular,  the norm of $q(X)$ is not necessarily bounded when $X$ lies in a compact subset of the set of matrices of size $T \times d_x$.
%The identification argument will involve expectations of the product  $Q_i Y_i$ but this expectation will be properly defined only if $\E(||Q_i \dot{u}_{i}||) < \infty$. Given the properties of $Q_i$ that we highlighted, this is a strong condition. 

If Assumption \ref{assn:id_finiteE} does not hold and there are not sufficient  regressor time variations, i.e., \textit{within}-variations of the regressors for all units in the sample, an alternative is to focus on units for which there is sufficient variation, i.e., such that $\det(\dot{X}_i' \dot{X}_i) > \delta$ for some threshold $\delta \geq 0$. We show in Section \ref{sec:maindid} that $\E(\mu_i | \det(\dot{X}_i' \dot{X}_i) > \delta)$ is also identified under standard assumptions and focus on an estimator of this parameter\footnote{This  approach is pursued in \cite{ab12}.} in Section \ref{sec:estim} since samples with substantial shares of stayers are more likely. To identify the APE for the entire population in cases where Assumption \ref{assn:id_finiteE} does not hold or where $T = d_x + 1$, \cite{gp12} provide an alternative procedure: letting the threshold $\delta$ go to $0$. %Focusing on the case %\footnote{Again, the authors do not need to take first differences. The corresponding case in the present paper would be $T = d_x + 1$.} $T = d_x+1$,  the additional issue is that $\tilde{M}_i = I - X_i(X_i'X_i)^{-1} X_i'= 0$. The identification method uses the subgroup of stayers to recover the common parameters of the models, i.e, the parameters that do not vary across units. Since moments involving $\tilde{Q}_i$ are likely not finite due to the existence of stayers,  the authors provide instead a closed-form expression for the average effect as the limit of a conditional expectation conditional on $\det(X_i'X_i) > h$, where the limit is taken as $h \to 0$. 
Although it is beyond the scope of this paper, this approach can be used here.

%As they point out, even when $T > d_x + 1$ and $\E(||Q_i \dot{u}_{i}||) < \infty$ as is the case here, for similar reasons the semiparametric variance bound for  $\E(\mu_i)$ computed by \cite{ch92} might not be bounded. That is, $\E(\mu_i)$ is not regularly identified. Their limit closed-form equation can be instead used to identify the average effect. 

\subsection{Invertibility Assumption}\label{sec:IA}

This last assumption is perhaps the least standard but is crucial for the derivation of the main results. For $\mathcal{M}(V) = \E( \, M_i \, | V_i = V) $, the assumption is stated below.
\begin{assumption}\label{assn:id_M_GLn}
	The matrix $\mathcal{M}(V_i)$ is invertible, $\mathbb{P}_{V_i} \, \text{a.s}$.
\end{assumption}

%point 1 of the list
%1)
We call  Assumption \ref{assn:id_M_GLn} the invertibility assumption (IA). Imposing invertibility of a matrix for identification is not surprising: in the linear regression model in cross-section, $y_i = w_i'\beta+u_i$ with $w \in \R^{d_w}$,  a necessary condition for identification of $\beta$ is nonsingularity of $\E(w w')$. It may be surprising however that the expectation of the singular matrices $M_i$ can be assumed to be nonsingular, but taking the same linear regression example, $w_iw_i'$ is also singular for each unit $i$.

%It is used in the first step of the identification procedure combined with Equation (\ref{eq:almost_closedform_g}):   left-multiplication with $M_i $ projects the vector of outcomes  on the space orthogonal to that spanned by the columns of $\dot{X}_i$, denoted $\Span(\dot{X}_i)$. On the right-hand side, this isolates the term $g(V_i)$. Aggregating across all individuals with the same value $V_i=V$, $g(V)$ is recovered if  the aggregated projection operator is invertible. 
%2)

Importantly, IA is an assumption  on the conditional expectation of a function of  $\dot{X}_i$ given $V_i$, which  itself is  a function of $X_i$ and $Z_i$. Thus IA  is an assumption on the joint distribution of $(X_i,Z_i)$: it assesses relevance of the instrument. %, to be compared, for instance, to the rank condition in instrumental variable identification of the linear regression with endogenous regressors.
However due to its unusual form, what IA imposes on $Z_i$ is not immediate and not easily comparable with standard conditions. Therefore we now have a closer look at its empirical content and obtain conditions on the distribution of $(Z_i,X_i,V_i)$ that are sufficient for IA to hold.  Section \ref{sec:minsupport} provides a geometric intuition on why the instrument can have discrete support and derives the minimum number of points on its support.
Section \ref{sec:sufficientcondIA} specifies a generic first-stage equation and under various restrictions obtains sufficient conditions.
These conditions are not restrictive: if the instrument is continuous, we obtain a small support condition similar to \cite{fhmv08} and  we show that IA holds for a binary instrument $Z_i$ when $x_{it}$ is scalar even if $Z_i$ impacts the regressor at only one period. We note here that our discussions of IA mostly assume that  $\dot{X}_i$ is of full column rank $\mathbb{P}_{\dot{X}_i|V_i=V} \, \text{a.s}$, which relates with Assumption \ref{assn:id_finiteE}. We return to this point when discussing Result \ref{rslt:evdok_z}.  % and generates variations of the regressor at only one time period
We start by showing that IA is necessary and sufficient to identify $g$.

\subsubsection{A Necessary Condition} 

%point 2 of the list: intuition
Under IA, if a function $\dot{h}: \mathcal{S}_{V_i} \mapsto \R^{T-1}$ is such that $\dot{h}(V_i) = \dot{X}_i \beta_i \text{ a.s}$, then $M_i \dot{h}(V_i) = 0 \Rightarrow \mathcal{M}(V_i) \dot{h}(V_i) = 0 \Rightarrow \dot{h}(V_i) = 0\text{ a.s}$. % by nonsingularity of $\mathcal{M}(V_i)$.  
IA precludes any nonzero function $\dot{h}$ of $V_i$ from being of the form $\dot{X}_i \beta_i$ or, to rephrase it,  precludes any time-varying function $h: \mathcal{S}_{V_i} \mapsto \R^{T}$ of $V_i$ from being of the form $X_i \beta_i$. This separately identifies $g$ from $\dot{X}_i \mu_i $: we show that IA is a necessary and sufficient condition for identification of the function $g$.

%Left-multiplication with $M_i$ in (\ref{eq:modelwTdiffVector}) substracts from $g(V_i)$ of the form $\dot{X}_i \beta_i$ for some nonzero time-invariant $\beta_i \in \R^{d_x}$. Thus intuitively, IA precludes any function of $V_i$ from being of the form $\dot{X}_i \beta_i$, or, to rephrase it, it precludes the existence of a random vector $\beta$ such that conditional on $V$, $X\beta$ is constant. Indeed if a function $h$ with domain $\mathcal{S}_{V_i}$ is such that $h(V_i) = \dot{X}_i \beta_i \text{ a.s}$, then $M_i h(V_i) = 0 \Rightarrow \mathcal{M}(V_i) h(V_i) = 0 \Rightarrow h(V_i) = 0\text{ a.s}$ by nonsingularity of $\mathcal{M}(V_i)$. That is, by IA,  the term $\dot{X}_i \mu_i$ is separately identifiable from $g(V_i)$. As a matter of fact, 

\begin{result}\label{rslt:cns_g}
	 Let Assumptions \ref{assn:idCFA} and \ref{assn:id_finiteE} hold. If $V_i$ is continuously distributed, then $g$ is identified $\mathbb{P}_{V_i} \, \text{a.s}$  on $\mathcal{S}_{V_i}$ if and only if Assumption \ref{assn:id_M_GLn} holds.
\end{result}
%The proof does not guarantee that IA is necessary and sufficient for identification of the APE, however it indicates that it might be the case. Exploring this question is beyond the scope of this paper. However to the best of our knowledge, there are no other method to identify the APE when strict exogeneity is relaxed, than first identifying the control function $g$.

%\vspace{-10pt}

\subsubsection{Minimum Number of Support Points for $Z$}\label{sec:minsupport}

Going back to the vocabulary used in Section \ref{sec:intuition_id}, identification of $g(V)$ focuses on 
%the residuals of the unit-specific linear regressions removes the linear term $\dot X_i \mu_i$ and aggregating over 
the subpopulation with fixed $V_i=V$.
% then removes the term $\dot u_i$. 
For this subpopulation we obtain the average vector of residuals of the regression of $g_t(V_i)$ on $\dot x_{it}$, i.e., $\E(R_i^g|V_i=V)$.
For a unit $i$ such that $V_i=V$, $g(V_i)$ cannot be separated from $\dot X_i \mu_i$ if $g_t(V_i)$ is linear in the regressors $\dot x_{it}$, i.e., if $g(V_i)$ is a linear combination of the columns in $\dot X_i$, as in the standard ``no perfect multicollinearity assumption''. % This writes as $g(V) \in \Span \dot X_i$. 
Let $\Span \dot{X}_i$ denote the set spanned by the columns of $\dot{X}_i$. The argument above implies that we must have $g(V) \notin \Span \dot X_i$ for any $i$ such that $V_i=V$. Since $g(V)$ is nonparametrically specified and could take any value, we need to make sure that no nonzero vector can be in the intersection of all $\Span \dot X_i$ for  $i$ such that $V_i=V$.  Result \ref{rslt:equiv_span} below follows from this intuition and its interpretation is that within the subpopulation with fixed $V_i=V$, the matrices of regressor variations must vary sufficiently.

\begin{result}\label{rslt:equiv_span}
	Fix $V \in \mathcal{S}_{V_i}$, the matrix $\mathcal{M}(V)$ is invertible if and only if for all subsets $\tilde{\mathcal{S}}\subset \mathcal{S}_{X_i|V_i=V}$ such that $\mathbb{P}_{X_i|V_i=V}\left(\tilde{\mathcal{S}}\right)=1$, then $\ \bigcap_{X \in \tilde{\mathcal{S}}} \Span \dot{X} = \{0\} $.
\end{result}
In the subpopulation with fixed $V_i=V$, cross-sectional variations of the matrix of regressor variations are generated by cross-sectional variations of the instrument. Thus Result \ref{rslt:equiv_span} provides some evidence that IA can hold when the instrument has discrete support.  Consider the case where $d_x=1$  and  the instrument is discrete. We maintain that $\dot{X}_i$ is of full column rank $\mathbb{P}_{\dot{X}_i|V_i=V} \, \text{a.s}$. Then $ \bigcap_{X \in \tilde{\mathcal{S}}} \Span\dot{X} $ is an intersection of straight lines. Since variations of $\dot{X}_i$ holding $V_i$ fixed are generated by variations of $Z_i$ holding $V_i$ fixed, it is a finite intersection of lines. Such an intersection  is always the trivial set $ \{0\} $ unless these straight lines are all equal: thus IA imposes that for some two different values of the vector of instruments, the corresponding vectors of regressors $\dot{X}$ are not collinear. This example gives a better understanding of the empirical content of IA: variations of $Z_i$ across individuals need to create  time variations of the regressors that are sufficiently diverse in the population $V_i=V$. 
%Point 3bis
Now if $d_x =2$ and  $T=4$, then $\Span \dot{X}$ is a plane in a 3-dimensional set. The intersection of two planes is either a plane or a line. But the intersection of three planes can be $\{0\}$: for $\mathcal{M}(V)$ to be nonsingular for a given $V$, the vector of instruments must have at least 3 points on its support. This points to a tension between the dimension of the random coefficients $d_x$, the number of time periods $T$ and the number of points in the support of $Z_i$, which is formalized in the following result.
\begin{result}\label{rslt:discrete_IV}
	Fix $V \in \mathcal{S}_{V_i}$. Assume that $\dot{X}_i$ is of full column rank $\mathbb{P}_{\dot{X}_i|V_i=V} \, \text{a.s}$ and that $\mathcal{S}_{Z_i|V_i=V}$ has $k_V$ points. Then 
	\vspace{-5pt}
	\begin{equation}\label{eq:minZdiscr}
		\mathcal{M}(V) \text{is nonsingular } \Rightarrow k_V \geq  \frac{T-1}{T-1 - d_x} .
	\end{equation}
\end{result}

%Note that assuming that $\dot{X}$ is of full column rank $\mathbb{P}_{\dot{X}|V} \, \text{a.s}$ for all values of $V$ is a stronger version of Assumption \ref{assn:id_finiteE}.\ref{assn:full_rk}.
%The inequality in Result \ref{rslt:discrete_IV}  is informative whenever $2 d_x > T-1$, %in the sense that it asks at least 2 IV and not just says that I need at least one IV.
The minimum number of points on $\mathcal{S}_{Z_i|V_i=V}$  increases with $d_x$ and decreases with $T$. When $T$ is the minimum number of periods required for the framework of this paper to apply, i.e, $T=d_x+2$, then (\ref{eq:minZdiscr}) becomes $k_V \geq d_x+1 $, $d_x+1 $ being the dimension of the unobserved heterogeneity. %if put epsilon and alpha together actually
Recent contributions show that discrete instruments can be used with the control function approach in cross-sectional data. For instance, in a nonseparable model with scalar unobserved heterogeneity,  \cite{df15} and \cite{t15} obtain nonparametric identification with a binary instrument. As the dimension of the unobserved heterogeneity increases, existing results document stricter requirements on the support of the instrument. To identify a nonseparable model  with unobserved heterogeneity of unknown dimension, \cite{in09} requires the instrument to be continuous and have large support. On the other hand, \cite{ns18} studies an outcome equation polynomial in an endogenous variable with random coefficients and shows that the number of points needed on the support of the instrument conditional on the control variable is at least as large as the dimension of the unobserved heterogeneity. 
See also \cite{ns21}, and \cite{mt16} for a similar result.
This is comparable to (\ref{eq:minZdiscr}) but for its dependence in $T$:  higher $T$ translates into more observations for the  same draw of $(\mu_i, \alpha_i)$, thus more information.
Note that the constraint (\ref{eq:minZdiscr})  is on the support of $Z_i$ and not $z_{it}$, and the cardinality of the support of $z_{it}$ can potentially be substantially lower.

\subsubsection{Sufficient Conditions for Triangular Systems}\label{sec:sufficientcondIA}

%point 5 of the list
%The above discussion on the invertibility condition did not specify a particular first-stage model nor functional form for the control variables $V$. Thus it is of interest to look at 
Although applications of the control function approach are often associated with large support conditions on the instruments, the previous discussion implies that IA does not require large support. We now look at specific first-stage equations and illustrate this point with various sufficient conditions on $Z_i$, when $\dot{X}_i$ is of full  rank $\mathbb{P}_{\dot{X}_i|V_i=V} \, \text{a.s}$.
%5.a
We impose the following general equation,
\begin{equation}\label{eq:1ststagenp}
	x_{it} = m_t(Z_i,V_i) \in \R^{d_x}, \, \text{ or in vector form, } \, X_i = m(Z_i,V_i) \in \R^{(T-1) \times d_x},
\end{equation}
and its first-difference version $\dot{X}_i = \dot{m}(Z_i,V_i)$. 
%Note that the examples below provide sufficient conditions: some of them are particularly stronger than Result \ref{rslt:discrete_IV}. They will generally illustrate that IA requires the matrix of instruments, as it varies in the population with the same value of the vector of control variables, to generate variations of the regressors over time which are sufficiently different across individuals.

The discussion preceding Result \ref{rslt:equiv_span} states that to recover $g(V)$, we must ensure that it is not possible for any nonzero vector to be linearly generated by the matrix of regressor time variations $\dot X_i$, for all units $i$ such that $V_i = V$. Thus the regressors variations must be sufficiently different from each other among these units. Let us consider a simple example.
Let us write the components of $x_{it}$ as $x_{it,k}$  where $k=1,..,d_x$ and define similarly $\dot x_{it,k}$. If one regressor $x_{it,k}$  has constant variations among  units $i$ such that $V_i = V$, i.e., if 
$$ \text{for } t=1,..,T-1, \text{ there exists } \bar{x}_{t,V} \text{ such that  } \dot x_{it,k} = \bar{x}_{t,V}, \, \mathbb{P}_{X_i|V_i=V}\text{ a.s.},$$
defining the vector $\bar x_V = (\bar{x}_{1,V},..,\bar{x}_{T-1,V})'$, then 
$\bar x_V \in \Span \dot X_i, \, \mathbb{P}_{X_i|V_i=V}\text{ a.s.}$ Result \ref{rslt:equiv_span} implies therefore  that $\mathcal{M}(V)$ is singular if for some $t$,  $\bar{x}_{t,V} \neq 0$. This counterexample shows that as $Z_i$ varies on $\mathcal{S}_{Z_i|V_i=V}$, it must generate variations of the regressor time variations in the cross-sectional dimension, that is, of the within-variations of the regressors.
But these cross-sectional variations need not be very rich.
%As discussed in Section \ref{sec:minsupport}, continuous variations may not be not necessary: we explore this now. % Since extending an analysis of first stage beyond an additively separable form and time invariant impact of the instrument is very tedious, we look at the simpler  case $d_x=1$. 
To illustrate this, we consider (\ref{eq:1ststagenp}) in the simpler case $d_x = 1$, and
$Z_i | V_i = V$ has discrete distribution.
Result \ref{rslt:counterex_1cstt} shows that, conditional on  $V_i=V$, even when there is so little cross-sectional variation that the instrument is binary and impacts the regressor at only one period $t_0$,  $\mathcal{M}(V)$ is nonsingular  as long as one of the regressors is not of the form $(b,..b,a,b,..b)'$. 

\begin{result}\label{rslt:counterex_1cstt}
	Let (\ref{eq:1ststagenp}) hold with $d_x =1$, and fix $V \in \mathcal{S}_{V_i}$, $t_0 \leq T$. Assume that $\dot{X}_i$ is of full column rank $\mathbb{P}_{\dot{X}_i|V_i=V} \, \text{a.s}$, $\mathcal{S}_{Z_i|V_i=V}$ is finite and  %$b_{t}$ is constant for all $t \leq T$ but $t_0$, and that 
	there exist  $(Z^{(1)},Z^{(2)}) \in \mathcal{S}_{Z_i|V_i=V}$ such that $ m_{t_0}(Z^{(1)},V) \neq  m_{t_0}(Z^{(2)},V)$ and  $ m_{t}(Z^{(1)},V) = m_{t}(Z^{(2)},V)$ for $t \neq t_0$. Then $\mathcal{M}(V)$ is nonsingular if and only if 
	
	\vspace{-18pt}
	
	$$ \exists (t_1,t_2), \ t_1 \neq t_2, \, t_1 \neq t_0, t_2 \neq t_0, \ \text{such that }  m_{t_1}(Z^{(1)},V) \neq m_{t_2}(Z^{(1)},V). $$
	%$X^{(1)}=b(Z^{(1)},V)$ is of the form\footnote{Note that $X^{(2)}$ will be of that form too.} $(b,..b,a,b,..b)'$, where $a$ is the $t_0^{\text{th}}$ component. % and $a\neq b$.
\end{result}

This result combines variations in the population (at $t_0$), and   uniform-over-$i$ variations over time (between $t_1$ and $t_2$).
Thus for  $\mathcal{M}(V)$ to be nonsingular, if the instrument generates cross-sectional variations of the regressor at $t_0$ only as it  varies on $\mathcal{S}_{Z_i|V_i=V}$, then the regressor must vary over time at least between two periods that are not $t_0$.
Another evidence for the need of both variations over time and across units comes from looking at models where a binary instrument does not vary over time. 
Such an example with ``structural breaks'', i.e., time variations in the function $m_t$, is detailed  in Appendix \ref{sec:app_struc_break}. It
illustrates that one structural break is in this case not enough, but IA can hold when there are two. Another example with a time-invariant instrument will be given in Section \ref{sec:relax_SE}  below.

\vspace{8pt}

IA is a condition on the within-variations of $X_i$ and further insight can also be gained looking at units $i$ such that $x_{it+1}=x_{it}$. Let us call these units \textit{consecutive stayers}.
Fixing $V$, note that  Results  \ref{rslt:discrete_IV} and \ref{rslt:counterex_1cstt}   as well as the discussions around these results examined $\mathcal{M}(V)$ while imposing for $\dot{X}_i$ to be of full column rank $\mathbb{P}_{\dot{X}_i|V_i=V} \, \text{a.s}$. This simplified some of the proofs but is also particularly relevant since we impose Assumption \ref{assn:id_finiteE} for identification of $\E(\mu_i)$. Recall that the second step of the identification argument runs unit-specific linear regressions. %, and the IV is discrete.
Without this additional restriction, we can consider a dgp such that $\dot X_i = 0,$ $\mathbb{P}_{\dot{X}_i|V_i=V} \, \text{a.s}$. Then $M_i$ is the identity matrix $\mathbb{P}_{\dot{X}_i|V_i=V} \, \text{a.s}$ and  $\mathcal{M}(V)$ is invertible. %There is no need for cross-sectional variation anymore. 
It is easy to understand why $g(V)$ is identified with this dgp. Indeed $\dot X_i = 0$ implies that $ \forall t \leq T-1, \ x_{it+1}=x_{it}$:  thus for a fixed $ t \leq T-1 $,
$$
\E(\dot y_{it} | V_i=V) = \E((x_{it+1} - x_{it})'\mu_i + g_t(V_i) + \dot{u}_{it} | V_i=V) = g_t(V),
$$
which identifies $g_t(V)$. We cannot not consider such a dgp but consecutive stayers can  still be leveraged when regressor matrices are of full column rank and the instrument is discrete, as the following result states.
%Define $\mathcal{B}(\dot{X},\epsilon)$ the open ball of center $\dot{X} \in \R^{T-1 \times d_x}$  and radius $\epsilon>0$.
\begin{result}\label{rslt:evdok_z}
	Fix $V \in \mathcal{S}_{V_i}$. Assume that $\dot{X}_i$ is of full column rank $\mathbb{P}_{\dot{X}_i|V_i=V} \, \text{a.s}$, $\mathcal{S}_{Z_i|V_i=V}$ is finite and for all $t \leq T-1$, $x_{it} = m_t(z_{it},v_{it})$. 	
	Assume also that for all $t \leq T-1$, there exists $Z^{(t)} \in \mathcal{S}_{Z_i|V_i=V}$ such that $z_{t+1}^{(t)}=z_{t}^{(t)}$.  Then $\mathcal{M}(V)$ is invertible.
\end{result}

The full rank condition  implies that it is not possible that $z_{s}^{(t)}=z_{s'}^{(t)}$ for all $(s',s)$. There needs to be cross-sectional variation of the period at which consecutive stayers ``stay'',  as $Z_i$ varies on $\mathcal{S}_{Z_i|V_i=V}$. This generates cross-sectional variation of the within-variations.
%
%When (\ref{eq:1ststagenp}) holds and is of the form $x_{it} = m_t(z_{it},V_i)$, and $Z \indep V$,  Result \ref{rslt:evdok} implies the following corollary on existence of stayers in terms of the instrument.
%\begin{corollary}
	%Suppose $Z \indep V$. Assume  that for all $t \leq T$, $x_{it} = m_t(z_{it},v_{it})$ and for all $t \leq T-1$, there exists $Z^{(t)} \in \mathcal{S}_{Z}$ such that a) $z_{t+1}^{(t)}=z_{t}^{(t)}$, b) $\dot{m}(Z,V)$ has full rank a.s, and c) for all $\epsilon >0$ and $V \in \mathcal{S}_{V_i}$, $\mathbb{P}_{Z}\left(\mathcal{B}( \dot{m}(Z^{(t)},V) ,\epsilon)\right)>0$. Then Assumption \ref{assn:id_M_GLn} holds.
%\end{corollary}
Leveraging these consecutive stayers is common in the panel data literature and important papers using this approach are for instance \cite{hk00,evdo10,hw12,gp12}.

\vspace{8pt}

Although existence of consecutive stayers guarantees invertibility, it is surely not a necessary condition. Indeed, if the instrument is continuous, we obtain a sufficient condition comparable to a well-known result in \cite{fhmv08}. This paper studies a cross-section outcome equation which is polynomial  in a scalar endogenous regressor $w_i$ and with random coefficients. The authors impose a nonseparable first stage $w_i = m(z_i,v_i)$ with an instrument $z_i$ independent of the control variable $v_i$ and $m(z_i,v_i)$ assumed to be a continuous function of $v_i$. Then existence of an open interval contained in the support of $m(z_i,v_0)$ for all fixed $v_0$ is shown to be a sufficient condition for identifiability of some average effects. The condition we obtain is a panel version of this restriction, instead on the entire vector of regressors, and implies too that a large support is not needed.
%stayers impose large support but we don't need such a strong condition? similar to a relaxation of IN09 by Fetal08. Condition very similar to those of Theorem 3 of Fetal08.
\begin{result}\label{rslt:florenstype}
	Let (\ref{eq:1ststagenp}) hold and fix $V \in \mathcal{S}_{V_i}$, assume that $\dot{X}_i$ is of full column rank $\mathbb{P}_{\dot{X}_i|V_i=V}\text{a.s}$ and that the support of $m(Z_i,V)$  contains an open ball, as $Z_i$ varies conditional on $V_i=V$. 
	Then $\mathcal{M}(V)$ is invertible.
\end{result}

\vspace{-5pt}
IA and the rank condition require the vector of instruments to generate variations of the regressors both over time and across individuals. Result \ref{rslt:florenstype} obtains these variations with a stronger than needed condition that the instrument has $X_i$ vary ``in all directions''. %This guarantees that for any vector $g$, there will always be a fraction of the population where such that the variations of $x$ will not linearly generate $g$.
Note that since we do not impose independence between $Z$ and $V$, the assumption is on the support of  $Z$ conditional on a fixed value of $V$.
A direct application of this result is to a linear first stage where the impact of the instrument does not vary with time. Write
\vspace{-5pt}
\begin{equation}\label{eq:lin1ststage_csttA}
 \forall \ t \leq T, \	x_{it} = A z_{it} + v_{it}, 
\end{equation}

\vspace{-7pt}
where $A$ is of size $d_x \times d_z$ and $\E(v_{it}|z_{it}) = 0$. % We denote the $k$th component of $z_t$ with $z_{k,t}$ and the vector of remaining components with $z_{-k,t}$.% and we use dots to denote first differences.%$\dot{z}_{k,t} = z_{k,t+1} - z_{k,t} $. 

\begin{corollary}\label{rslt:lin1ststage_csttA}
	Let (\ref{eq:lin1ststage_csttA}) hold, where $A$ is of full row rank and fix $V \in \mathcal{S}_{V_i}$. Assume that $\dot{X}_i$ is of full column rank $\mathbb{P}_{\dot{X}_i|V_i=V}\text{a.s}$ and that the support $\mathcal{S}_{Z_i|V_i=V} $ contains an open ball.  Then $\mathcal{M}(V)$ is invertible.
\end{corollary}

If $Z \indep V$, one can replace the support condition with ``$\mathcal{S}_{Z_i} $ contains an open ball''.
%Note that the assumption of Result \ref{rslt:lin1ststage_csttA} implies that $\dot{Z}$ has a distribution absolutely continuous with respect to the Lebesgue measure. maybe invert and state this as the assumption then explains what it means -- not exactly true, could be a mix of continuous and discrete
The condition that $A$ is of full row rank implies  that $d_z \geq d_x$.
% We note here that the sufficient conditions studied in the two previous results are stronger that what Result \ref{rslt:discrete_IV} suggests.
% \textcolor{blue}{I could replace the support assumption writing the support of $\dot{Z}|V$ has nonempty interior}
%5.c

%point 7 of the list
%not sure I will write anything here
%\textcolor{blue}{Pt 7: Deterministic relationship between regressors}

%point 8 of the list
\begin{note}\label{note:test}
	The invertibility assumption is a condition on observable and estimable random variables which implies that it can be tested. A formal test would require testing for the value of the function $V \mapsto\det(\mathcal{M}(V))$ being nonzero for all values of $V$ on $\mathcal{S}_{V_i}$.  One can consider for instance the minimum of this function over the support of $V_i$, that is, $\min_{V \in \mathcal{S}_{V_i}} \det(\mathcal{M}(V))$. A test can be constructed using the nonparametric estimator for $\mathcal{M}(V)$, however such a test is nonstandard and studying its test statistic is beyond the scope of this paper.
	%check: functionals of nonparametric estimation that are only Hadamard differentiable (Lehman Romano p570) / confidence bounds for cdf (Lehman Romano p255 ) nonparametric estimators, this should be related. 
	%papers on that topic:
	%Estimation of the Minimum of a Function Using Order Statistics, De Haan (but the unfction is deterministic)
	%Confidence Bands in Nonparametric Regression
	%Note that for given specific first stages, more standard procedures can be used. \textcolor{blue}{HERE: what if the first stage is linear/struc. break etc}
	%considering the proposed first stages: only the linear first stage here ? anything else? I think if I can prove anything with a structural break, would be great. Or in that spirit.
\end{note}

\subsection{Example: Sequential Exogeneity}\label{sec:relax_SE}

%Instead of contemporaneous endogeneity, that is, a conditional mean dependence between $\epsilon_{it}$ and $x_{it}$, 
A particular failure of the strict exogeneity condition in  (\ref{eq:model}) is sequential exogeneity, which we separate from contemporaneous endogeneity. Sequentially exogenous regressors satisfy the assumption $\E(\epsilon_{i \, t}|x_{i1},..,x_{it}) = 0$ for all $t$. To the best of our knowledge, there does not exist a proof of identification of $\E(\mu_i)$ when only this assumption holds.
In this section, we show that the two-step identification argument obtains identification under certain conditions. In particular, we assume that the regressors follow a Markov process and show that Assumption \ref{assn:idCFA} holds. The instrument is then the regressor at the first period: an interesting question is whether Assumption \ref{assn:id_M_GLn} holds when there are so few variations. It will be shown in a particular example that it is the case, underlying the minimal variations the invertibility assumption imposes on the regressor as the instrument varies.

Equation (\ref{eq:model}) holds and we look at the case where  $x_{it}$ is a Markov process. We assume 
\vspace{-5pt}
\begin{align}
	 \text{for }1 \leq t \leq T-1, \ x_{i \, t+1} \, = \, m_t(x_{it}) + \, \eta_{i \, t+1}, \text{ and } (\epsilon_{it}, \eta_{i\, t+1},.. \, , \eta_{iT}) \indep (x_{i1}, \eta_{i2}, \, .. \, , \eta_{it}), \label{eq:markovmodel} 
\end{align}

\vspace{-5pt}
where $m_t$ is unknown. This implies that $(\epsilon_{it}, \eta_{i\, t+1},.. \, , \eta_{iT}) \indep (x_{i1}, \, .. \, , x_{it})$, thus sequential exogeneity holds but $x_{it+1}$ may be correlated with $\epsilon_{i \, t}$. %The idea of modeling a predetermined regressor as function of past values of the data and an independent shock has been used in panel data. See e.g Section 4 of \cite{ab16}, where the authors study a nonseparable dynamic panel model with time-invariant unobserved heterogeneity.
We define $V_i \, = \, (\eta_{i2},..\, , \eta_{iT})$, it is identified as the vector of residuals of nonparametric regressions of $x_{it+1}$ on $x_{it}$. Then for all $t \leq T$,
\vspace{-5pt}
\begin{align*}
	\E(\epsilon_{it} \, | \, x_{i1},..\, , x_{iT}, \eta_{i2},..\, , \eta_{iT}) \, & = \, \E(\epsilon_{it}|x_{i1}, \eta_{i2},..\, , \eta_{iT}) \, = \, \E(\epsilon_{it}| \eta_{i\, t+1},..\, , \eta_{iT})   : = f_t(V_i).
\end{align*}

\vspace{-5pt}
Assumption \ref{assn:idCFA} holds\footnote{Note that
	% as for the sample selection example, 
	independence is stronger than needed. Conditional mean independence of $\E(\epsilon_{it}|x_{i1}, \eta_{i2},..\, , \eta_{iT})$ with respect to $x_{i1}$ is sufficient.} and $V_i$ is a valid vector of control variables.
%It is not as clear what restrictions these assumptions impose on the joint distribution of $(X_i, \mu_i)$. If for instance $x_{i1}$ is correlated with $\mu_i$, then these assumptions are more easily obtained if $\eta_{it} \indep \mu_i$ for $t \geq 2$ which restricts the joint distribution of $(x_{it}, \mu)$ substantially (can only be dependent through $x_{i1}$. On the other hand it is not needed that $\eta_{it} \indep \mu_i$: both $\eta_{it}$ and $x_{i1}$ can be correlated with $\mu_{i}$ while be independent at the same time.
Note that one could alternatively consider $x_{i \, t+1} = m_{t+1}(x_{it},\eta_{i \, t+1})$ with $\eta_{i \, t+1}$ scalar and $m_t$ strictly monotonic in $\eta_{i \, t+1}$ and use the control variable suggested in \cite{in09}. Defining $M_i \,  = \, I - \dot{X}_i (\dot{X}_i' \dot{X}_i)^{-1} \dot{X}_i' \, $, $\mathcal{M}(V_i) \, = \, \E(M_i|V_i)$, $u_{it} = \epsilon_{it} - f_t(V_i)$ and $g_t(V_i) \, = \, f_{t+1}(V_i) - f_t(V_i)$,
\vspace{-5pt}
\begin{equation}\label{eq:idgwithSeqExog}
	\E(M_i \dot{y}_{i} | V_i) \,  = \mathcal{M}(V_i) g(V_i), \text{ and } \E(Q_i \dot{y}_i) = \E(\mu_i) + \E(Q_i g(V_i)).
\end{equation}

\vspace{-5pt}
It follows that the two-step identification of the main model applies here: a first step identifies the vector of functions $g$ and the second step identifies the average effect. We thus need to assume invertibility of $\mathcal{M}(V)$ on the support of $V$. This condition is nontrivial here because $V_i = (\eta_{i2},.. \eta_{iT})$ while $M_i$ is a function of the vector of variables $X_i$: this implies that the expectation of $M_i$ conditional on $V_i$ is an expectation over $x_{i1}$,
% which is the instrument here. should navoid calling it too often an instrument here, it is confusing with the role of z_it as a way to construct v_it.  
that is, the only source of variation once $V_i$ is fixed is $x_{i1}$. To show that Assumption \ref{assn:id_M_GLn} can hold with so little variation, let us look more closely at the case where
$x_{it}$ is a scalar AR(1) process.

\begin{assumption}\label{assn:relax_SE}
	\
	\begin{enumerate}[noitemsep,nolistsep]
		\item (\ref{eq:model}) and (\ref{eq:markovmodel})  hold, and $m_t(x) = \rho x$, $\rho \notin \{0,1\}$
		
		\item $(X_i, \mu_i, \alpha_i, \epsilon_i, \eta_i)$ is i.i.d across $i = 1..n$,
		\item Either $x_{i1}$ has a discrete distribution with at least two support points, or $x_{i1}$ is continuously distributed. % and $\Int(\mathcal{S}_{x_1}) \neq \emptyset$. 
		Moreover, $(\eta_{i2},.. ,\eta_{iT})$ is continuously distributed on $\R^{T-1}$.%, and $f_t$ is continuous.
	\end{enumerate}
\end{assumption}
In this case, the instrument $x_{i1}$ does not vary over time but its impact on each period $x_{it}$, i.e., $\rho^{t-1}$, creates sufficient time variation of the regressor to ensure that IA holds. We obtain the following result.
\begin{result}\label{rslt:relax_SE}
	Under Assumptions \ref{assn:id_finiteE} and \ref{assn:relax_SE}, IA holds and $\E(\mu_i', \alpha_i)$ is identified.
\end{result}
  %Assuming independence between $\eta_{it}$ for $t \geq 2$ and $x_{i1}$  might restrict the joint distribution of $(x_{it}, \mu_i)$ if $x_{i1}$ is correlated with $\mu_i$ but it does not necessarily require independence between $x_{it}$ and $\mu_i$.
We note  that if there is contemporaneous endogeneity and if the control variables $v_{it}$ are identified by a cross-section regression of $x_{it}$ on $z_{it}$, the approach developed in this subsection can be applied to cases where $z_{it}$ is not strictly exogenous but sequentially exogenous, in particular if it follows a similar Markovian structure. %The vector $V_i$ would then be the vector of residuals of both the regressions of $x_{it}$ on $z_{it}$ and the regressions of $z_{it+1}$ on $z_{it}$. 
The instrument would then be $z_{i1}$ instead of $Z_i$. We use this approach in Section \ref{sec:illus} and in a structural model developed in Section \ref{sec:app_struct_model} of the Appendix.
\begin{note}
	While the model exposed above allows for the regressor to be correlated with past disturbances, the assumptions are not compatible with a lagged dependent variable as regressor in the outcome equation (\ref{eq:model}). Indeed writing  $x_{i \, t+1} \, = \, m_t(x_{it}) + \, \eta_{i \, t+1}$ implicitly imposes a homogeneous dependence on the past value that is not consistent with the specification $y_{i \, t+1} = \mu_i y_{it} + \epsilon_{it+1}$ and a fixed-effect approach.  %Identification of the APE of lagged dependent variable is in fact an open question and \cite{ch93} shows that it is generally not identified. 
	%It could however be possible to adapt an extension mentioned in \cite{ab12} to allow for lagged dependent variable with homogeneous impact.
\end{note}

\begin{note}
	The control variables introduced in Assumption \ref{assn:idCFA} may come from various types of first-stage equations: this is leveraged here to obtain identification. It can be done in other panel models with random coefficients similar to Model (\ref{eq:model}), where other types of violations of the strict exogeneity condition occur. For example, in a previous version of this paper, we applied the idea to a panel CRC model with sample selection where the propensity score is used as the control variable. This will now appear in a separate paper.
\end{note}

\section{Identification Results}\label{sec:identification} 

\subsection{Main Result}\label{sec:maindid}

Under Assumptions \ref{assn:idCFA}, \ref{assn:id_finiteE} and  \ref{assn:id_M_GLn}, we obtain 
\begin{equation}\label{eq:closedform_g}
g(V_i) = \mathcal{M}(V_i)^{-1} \E(M_i \dot{y}_i | V_i), \ \mathbb{P}_{V_i} \ \text{a.s}.
\end{equation}
Under Assumption \ref{assn:id_finiteE},  the matrix $Q_i$ is well-defined with probability $1$. 
Equation (\ref{eq:model_id_mu}) implies 
\begin{align}
	\mu_i  \, & =  \, Q_i \dot{y}_i \, - \, Q_i g(V_i) \, - \, \, Q_i  \dot{u}_i, \label{eq:idMiou}, \\
	 \E( \mu_i ) \, & = \, \E (Q_i \dot{y}_i   -  \, Q_i g(V_i) ), \label{eq:Miou}
\end{align}
where the second line holds since by the law of iterated expectations and Assumption \ref{assn:id_finiteE}, $\E(Q_i \dot{u}_i)=0$. Equation (\ref{eq:Miou}) identifies $\E(\mu)$ since all elements on the right hand side are observed. 
\begin{result}
Under Assumptions \ref{assn:idCFA}, \ref{assn:id_finiteE} and \ref{assn:id_M_GLn}, the average effect $\E(\mu_i)$ is identified and $\E( \mu_i )  =  \E (Q_i \dot{y}_i   -  \, Q_i g(V_i) )$.
\end{result}

As mentioned in Section \ref{sec:assn_CRC_matrix},  the conditions $\E(||Q_i \dot{u}_{i}||) < \infty$ and $\E(||Q_i g(V_i)||) < \infty$ of Assumption \ref{assn:id_finiteE} may not hold. In this case, one can instead focus on $\E( \mu | \delta) := \E(\mu_i | \det(\dot{X}_i' \dot{X}_i) > \delta)$.
Define $\delta_i = \1(\det(\dot{X}_i' \dot{X}_i) > \delta_0)$ and  $Q_i^{\delta} = \delta_i Q_i $. Then
\begin{equation}\label{eq:Miou_delta}
 \E( \mu | \delta) \, = \, \frac {\E (\delta_i Q_i \dot{y}_i   -  \, \delta_i Q_i g(V_i))}{\mathbb{P}(\det(\dot{X}_i' \dot{X}_i) > \delta_0)} \, = \, \frac {\E ( Q_i^\delta \dot{y}_i   -  \, Q_i^\delta g(V_i) )}{\mathbb{P}(\det(\dot{X}_i' \dot{X}_i) > \delta_0)},
\end{equation}
which identifies $\E( \mu | \delta) $ as all the terms on the right hand side of (\ref{eq:Miou_delta}) are identified. The considered expectations exist as long as $\E(||Q_i^\delta \dot{u}_{i}||) < \infty$ and $\E(||Q_i^\delta g(V_i)||) < \infty$, which holds under standard conditions: for instance if $X_i$ has bounded support, $\E(||\dot{u}_i||) < \infty$ and $\E(||g(V_i)||) < \infty$. Bounded support can be replaced with finite moments conditions.
%\footnote{Proof to come - deltabdd}

It remains to identify $\E(\alpha_i)$, which we obtain using the variables in period $1$. We multiply (\ref{eq:idMiou}) with $x_{i1}$ and subtract $y_{i1}$, obtaining $y_{i1} - x_{i1}' \, \mu_i \,  = \, y_{i1} \, - \, x_{i1}' \left[ Q_i \dot{y}_i \, - \, Q_i g(V_i) \, - \,  Q_i  \dot{u}_i \right].$
Since $y_{i1} - x_{i1}' \, \mu_i = \alpha_i + \epsilon_{it}$ where $\E(\epsilon_{it}) \, = \, 0$, we obtain the following identifying equation for $\E(\alpha_i)$,
\vspace{-5pt}
$$\E(\alpha_i) = \E(y_{i1} \, - \, x_{i1}' \left[ Q_i \dot{y}_i \, - \, Q_i g(V_i)  \right] ).$$

\subsection{Average Effect of an Exogenous Intervention}\label{sec:exogIntervention}

Consider a policy intervention on $x_{it}$ for each unit $i$ in a given period $t$. The average effect of this exterior intervention is an object of interest and \cite{bp03} studies its identifiability in different models when the change in covariates is exogenous, i.e, independent of the unobservable error terms. The unobservables in the CRC model  are $(\mu_i, \alpha_i, (\epsilon_{it})_{t \leq T})$, and if the exogenous shift is  a variation $\Delta_{it}$ independent of $(\alpha_i, \mu_i, (\epsilon_{it})_{t \leq T})$, the average impact of the policy is $\E(\mu_i) \E(\Delta_{it})$ which is identified. However some policy interventions may impose a dependence between $\Delta_{it}$ is correlated with $x_{it}$, hence $\Delta_{it}$ will be correlated with $(\mu_i, \alpha_i)$ while being exogenous in the sense that it is independent of $(\epsilon_{it})_{t \leq T}$.
Consider an exogenous intervention that shifts $x_{it}$ to $l(x_{it})$, i.e, $\Delta_{it} = l(x_{it}) - x_{it}$. The average outcome after this intervention is $\E(l(x_{it})' \mu_i + \alpha_i + \epsilon_{it})$ and depends on the joint distribution of $(\mu_i,x_{it})$ where $\mu_i$ is unobservable and this joint distribution is unrestricted. Using Equation (\ref{eq:idMiou}) expressing $\mu_i$ as a function of the primitives,
the change in expected outcome is 
\vspace{-7pt}
\begin{align*}
\E(l(x_{it})' \mu_i + \alpha_i + \epsilon_{it} - [x_{it}'\mu_i + \alpha_i + \epsilon_{it}]) &\, = \, \E([l(x_{it}) - x_{it}]'\mu_i ), \\
&  \, = \, \E ( [l(x_{it}) - x_{it} ]' [Q_i y_i  -  \, Q_i g(V_i) - Q_i \dot{u}_i]),\\
&  \, = \, \E ( [l(x_{it}) - x_{it} ]' [Q_i y_i  -  \, Q_i g(V_i)] ),
\end{align*}
where the second equality holds by exogeneity of the change in regressors. All elements in the last expectation are identified,  thus so is the average change in outcome.

\subsection{Identifying Higher-Order Properties of $\mu_i$}\label{sec:mu_distrib}

The focus of this paper being on allowing for time-varying endogeneity in correlated random coefficient panel models, the parameter of interest is the simplest one, the average effect $\E(\mu_i)$. However more properties of the unobserved heterogeneity can be obtained. In particular, by $\E(\dot{u}_i |X_i)=0$ and (\ref{eq:idMiou})  we have $\E( \mu_i | X_i ) \, = \, \E (Q_i \dot{y}_i   -  \, Q_i g(V_i) |X_i)$.

In a model with strict exogeneity, \cite{ab12} extend the method in \cite{ch92} to identify the variance matrix and the distribution of $(\alpha_i,\mu_i)$ conditional on $X_i$ under various restrictions on the time-dependence of $\epsilon_{it}$ conditional on $X_i$ and on the joint distribution of $(\epsilon_i, \alpha_i, \mu_i)$ conditional on $X_i$. The argument first identifies the common parameters and subtracting the common part from the outcome variables, higher order moments of $\mu_i$ are separated from those of $\epsilon_{i}$ using the above-mentioned restrictions.
We note here that their argument can be combined with the assumptions made in the present paper so as to allow for endogeneity of the regressors. Indeed   $g(V_i)$ being recovered using the method described in Section \ref{sec:maindid}, the analysis of \cite{ab12} can be conducted on $y_i - g(V_i) \, = \, X_i \mu_i + u_{i}$ which takes the same form as in their paper. We refer to the paper for more details on the procedure to recover these moments.

\section{Estimation}\label{sec:estim}

The  proof of identification of $\E(\mu)$ is constructive as we obtained closed form expressions (\ref{eq:closedform_g}) and (\ref{eq:Miou}). The estimator we suggest follows the identification steps, replacing population moments with their sample analogs. Estimation is thus a multi-steps procedure. 
First, for all $(i,t)$, $v_{it}$ is estimated as $\hat{v}_{it} = \hat{C}_t(X_i, Z_i)$ for $\hat{C}_t$ an estimator of the function $C_t$.
Second, the conditional expectation functions $\mathcal{M}(V) = \E( \, M_i | V_i = V)$ and $k(V) = \E(M_i \dot{y}_i |  V_i = V)$ are estimated nonparametrically using the generated values $\hat{V}_i = (\hat{v}_{it})_{t \leq T}$ as regressors and the function $g = \mathcal{M}^{-1} k$ is estimated by plug in of the estimators $\hat{\mathcal{M}}(V)$ and $\hat{k}(V)$. Finally, the estimator for $\E(\mu | \delta)$ will be a sample analog of Equation (\ref{eq:Miou}), plugging in the estimator of $g$ and $V$.

The asymptotic properties of this estimator will depend on the definition of the control variables, that is, on $C_t$. We  thus choose focus on the following model
\vspace{-5pt}
\begin{align}
	& y_{it}  = \, x_{it}^{ex \prime} \, \mu_i^{ex} + x_{it}^{en \, \prime} \, \mu_i^{en}  + \alpha_i + \epsilon_{it}, \label{eq:model_for_estim}\\
	& x_{it}^{en}  = \, b_t(x_{it}^{ex}, z_{it}) + v_{it}, \qquad \E( v_{it} | x_{it}^{ex},z_{it}) = 0, \label{eq:modelCV_for_estim}
\end{align}
\vspace{-5pt}
where $x_{it}^{ex} \in \R^{d_1}$, $x_{it}^{en} \in \R^{d_2}$, $z_{it} \in \R^{d_z}$, and where Assumption \ref{assn:idCFA} holds. The regressors $x^{ex}$ are exogenous while $x^{en}$ can be endogenous. The control variables in this model are the residuals of the nonparametric regression of the endogenous regressors on the exogenous regressors and the instruments and will be estimated as residuals of the nonparametric regression estimation of $x_{it}^{en}$. All estimators of the nonparametric regressions will be series estimators.  As highlighted in Section \ref{sec:assn_CRC_matrix}, the condition $\E(||Q_i||^2) < \infty$ may not hold thus out of caution we use (\ref{eq:Miou_delta}) to estimate $\E(\mu | \delta)$, where we fix $\delta_0$ and define $\delta_i = \1(\det(\dot{X}_i' \dot{X}_i) > \delta_0)$.  We proceed   with an explicit definition of the estimators and a stepwise proof of asymptotic normality of $\hat{\mu}$. All proofs are in the Appendix, Section \ref{sec:appConsistency}.

\subsection{Definition of the estimators}\label{sec:def_estim}
%reference for the previous subsection this section had: \label{sec:def_hatV} \label{subsec:seriesofGC} and nothing

We introduce some notations. For a vector $a \in \R^p$, $||a||$ is its Euclidean norm. We also denote by $||.||_F$ the Frobenius norm (the canonical norm) in the space of matrices $\mathcal{M}_p(\R)$, and $||.||_2$ the matrix norm induced by $||.||$ on $\R^p$ (the spectral norm). By an abuse of notation, in Section \ref{sec:estim} we drop the index $i$ for random variables whenever it does not hinder clarity.
For $g$ a vector of functions of $x \in \mathcal{S}_x \subset R^{k}$, $||g||_{\infty}$ is $ \sup_{\mathcal{S}_x} ||g(.)||$. For $l = (l_1,\, .. \, , l_{k}) \in \mathbb{N}^{k}$, we define $|l| = \sum_{j = 1}^{k} l_j $, and the partial derivative $\partial^{l} g(x) = \partial^{|l|} g(x) / \partial^{l_1}x_1 ... \partial^{l_{k}}$. We will use the norm $|g|_d = \max_{|l| \leq d} \sup_{x \in \mathcal{S}_x} || \partial^{l} g(x) ||$ when $g$ is $d$ times differentiable.
We denote by $\partial g(x)$ the Jacobian matrix $( \partial g(x) / \partial x_1, ... ,  \partial g(x) / \partial x_k )$. 
For a sequence $(c_n)_{n \in \N} \in \R^{\mathbb{N}}$, the notation $c_n \to 0$ should be understood as $c_n \to_{n \to \infty} 0$. For a random variable $x$, $f_x$ denotes its density.

%\subsubsection{Series estimator of $V_i$}\label{sec:def_hatV}

By (\ref{eq:modelCV_for_estim}), we have $v_{it} = x_{it}^{en} - \E(x_{it}^{en} | x_{it}^{ex}, z_{it})$. 
We write $\xi_{it} = (x_{it}^{ex},z_{it})$. Consider the $L \times 1$ vector of approximating functions $r^L(\xi_t)=(r_{1L}(\xi_t),.. \, ,r_{LL}(\xi_t))'$ and $r_{it} = r^L(\xi_{it})$. We define the series estimators of the regression function $\E(x_{it}^{en} | \xi_{it}=\xi_t) = b_t(\xi_t)$ to be $ \hat{\beta}_t'  r^L(\xi_t)$ where $\hat{\beta}_t$ is $L \times d_2$, and
\begin{equation}\label{eq:estimV}
	\hat{\beta}_t \, =  \, (R_t R_t')^{-1} \sum_{i}  r^L(\xi_{it}) x_{it}^{en \, \prime} \, = \, (R_t R_t')^{-1} R_t X_t^{en \, \prime},
\end{equation} 
where $R_t  \, = \,  (r_{1t},.. \, ,r_{nt})$ is $L \times n$ and $X_t^{en} \, = \, (x_{1t}^{en},.. \, ,x_{nt}^{en})$ is $d_2 \times n$.
Write $r_{it} = r^L(\xi_{it})$ and $\hat{b}_{it} =  \hat{\beta}_t' \, r_{it}$, define $b_{it} = b_t(\xi_{it})$, $\hat{b}_t = \hat{\beta}_t'  r^L$ and the residuals $\tilde{v}_{it} \, = \, x_{it}^{en} - \hat{b}_{it} $ with $\tilde{V}_i = (\tilde{v}_{i1}',\, .. \, ,\tilde{v}_{iT}' )'$.
Later, the support $\mathcal{S}_V$ of $V$ will be assumed bounded. However, the values obtained using the estimated residuals might not be in $\mathcal{S}_V$: it will be convenient for the asymptotic analysis to introduce a transformation $\tau$ of the generated variables such that their transformed values lie in $\mathcal{S}_V$. Specifically, we assume that the support of $v_t$ is of the form $\bigtimes_{d=1}^{d_2} \, [\underline{v}_{td},\bar{v}_{td}]$ and that the support of $V$ is $\mathcal{S}_V = \bigtimes_{d \leq d_2, \, t \leq T} \,  [\underline{v}_{td},\bar{v}_{td}]$.
We define $\tau$ such that for $V = (v_1',\, .. \, ,v_T')' \in \R^{Td_2}$, then $\tau(V)$ is the projection of $V$ onto $\mathcal{S}_V$. %That is, if $V$ lies outside of $\mathcal{S}_V$, then $\tau(V)$ is the point on the boundaries of the support that is the closest to $V$. 
The exact definition of $\tau$ is given in the Appendix \ref{ref:trimmingapp} by (\ref{eq:def_tau_consist}). Our estimator for $V_i$ will then be $\hat{V}_i = \tau(\tilde{V}_i)$.  Note that for all draws of $V_i$, $\tau(V_i)=V_i$ and $|| \hat{V}_i - V_i || \leq ||\tilde{V}_i - V_i||$.  We also assume that $\mathcal{S}_{\xi t}$ is of the form $\bigtimes_{d=1}^{d_\xi} \, [\underline{\xi}_{td};\bar{\xi}_{td}]$, where we use the notation $\xi_{td}$ for the d$^{th}$ component of $\xi_t$.  

%\subsubsection{Series estimator of $\mathcal{M}$ and $k$ }\label{subsec:seriesofGC}
Let $p^K(V)=(p_{1K}(V),.. \, ,p_{KK}(V))'$ denote a $K \times 1$ vector of approximating functions, $p_i = p^K(V_i)$ and $\hat{p}_i = p^K(\hat{V}_i)$. An estimator of $h^W(V) =  \E(w_i | V_i = V)$ for a generic scalar random variable $w_i$ using the generated $\hat{V}$ is $p^K(V)' \, \hat{\pi}^W$ where $\hat{\pi}^W$ is a vector of size $K$ given by
\begin{equation}\label{eq:def_2stepestim}
	\hat{\pi}^W \, =  \, (\hat{P} \hat{P}')^{-1} \sum_{i}  p^K(\hat{V}_i) \, W_i' \, = \, (\hat{P} \hat{P}')^{-1} \hat{P} W,
\end{equation}
where $\hat{P}  \, = \,  ( \hat{p}_1,.. \, , \hat{p}_n)$ is $K \times n$ and $W \, = \, (w_1,.. \, ,w_n)'$ is a vector of size $ n$.
Using this general definition, we construct component by component estimators $\hat{\mathcal{M}}$ and $\hat{k}$ for the matrix and vector valued functions $\mathcal{M}$ and $k$. We obtain $p^K(V)' \, \hat{\pi}^{M,st}$ an estimator of the $(s,t)$ component of the matrix $\mathcal{M}$, taking $w_i$ to be $(M_{i})_{s,t}$. Similarly, an estimator of the $s$th component of $k$ will be $p^K(V)' \, \hat{\pi}^{k,s}$, taking $w_i = (M_i \dot{y}_i)_s$.
%\subsubsection{Construction of $\hat{g}$ and $\hat{\mu}$}
By (\ref{eq:closedform_g}) and (\ref{eq:Miou_delta}),
a plug-in estimator of $g$ is 
$\hat{g}(V) \, = \,  \hat{\mathcal{M}}(V)^{-1} \, \hat{k}(V)$
and a plug-in estimator of  $\E(\mu|\delta)$ is
$$\hat{\mu}  \, = \, \frac{\sum_{i=1}^n \delta_i Q_i \, [ \dot{y}_i - \, \hat{g}(\hat{V_i}) ]}{\sum_{i=1}^n \delta_i} \, = \, \frac{\sum_{i=1}^n Q_i^{\delta} \, [ \dot{y}_i - \, \hat{g}(\hat{V_i}) ]}{\sum_{i=1}^n \delta_i}.$$

The multi-step estimation procedure only uses closed form expressions: its ease of implementation comes with a layered asymptotic analysis as each step needs to be analyzed one by one to eventually obtain the asymptotic behavior of $\hat{\mu}$. This type of asymptotic analysis is the subject of a wide literature on nonparametric and semiparametric estimation with generated covariates. Before laying out the main results of our asymptotic analysis, we give here a brief overview of this literature.
Papers studying asymptotic normality of semiparametric estimators, such as \cite{n94}, \cite{clvk03}, \cite{ac03} and \cite{hl10} among many other references, have a level of generality which encompasses the case where the regressors are themselves estimated. However, the conditions given in these papers are ``high-level'' conditions and are not easily applied to the composition of nonparametrically estimated infinite-dimensional nuisance parameters.
Examples of asymptotic derivations in specific models with generated regressors are papers already cited such as \cite{npv99}, \cite{dnv03}, and \cite{in09}, as the use of a nonparametric control function approach naturally suggests an estimator with generated covariates. Others are, e.g, \cite{ap93}, \cite{bp04}, \cite{n09} and \cite{ejl16}.
Moreover, recent contributions have focused on obtaining general asymptotic results for such semiparametric estimators.
Among important recent contributions, \cite{hr13,hr19} derive in the spirit of \cite{n94} a general formula of the asymptotic variance of semiparametric estimators using generated regressors. 
However they do not provide results on how to obtain asymptotic normality for particular classes of estimators. For estimators with generated regressors depending on a nonparametrically estimated function, this type of analysis can be found for instance in \cite{ejl14}, \cite{mrs16} and \cite{hlr18}. %CONTINUE THE LIST - PAPERS BY ESCANCIANO & SEE HE LIST IN ELIA PAPER
\cite{ejl14} obtain a uniform expansion of a weighted sample average of residuals obtained from kernel-estimated nonparametric regressions with generated covariates, which can  then be used to prove asymptotic normality of a class of semiparametric estimators.
%THEN DETAILS ABOUT THE MOST IMPORTANT PAPERS
\cite{mrs16} study the asymptotic normality of a general class of semiparametric GMM estimators depending on a nonparametric nuisance parameter, also constructed with generated covariates. Our estimator of the APE  $\hat{\mu}$ belongs to this class of estimators, although of a simpler form since it has a closed-form expression. Moreover we use series to construct the nonparametric estimates while the infinite dimensional nuisance parameter in \cite{mrs16} is a conditional expectation estimated with local polynomial estimator and they do not specify an estimator for the generated covariates. 
Estimators in \cite{hlr18} have a structure closer to that of $\hat{\mu}$: they study nonparametric two-step sieve M estimators, but focus on known functionals. They show asymptotic normality of their estimator when standardized by a finite sample variance and give a practical estimator of this variance. They do not however provide an explicit formula of the asymptotic variance.
The estimator we analyze in this section is instead an estimated functional of the two-step nonparametric estimators. Using a different type of proof techniques with lower level conditions on the primitives of a more specific class of models, we show asymptotic normality and obtain the asymptotic variance of a generic class of estimators to which ours belongs.
See, e.g, \cite{mrs16} for a literature review on semiparametric estimation with generated covariates and explanation on the specificity of this type of estimation.

\subsection{Consistency}\label{sec:consistency}
%reference for the previous subsection this section had: nothing, \label{sec:reg_GalCase} and nothing

%We now proceed step by step to derive uniform and MSE convergence rates of all nonparametric estimators.
%\medskip
%\textit{\textbf{Sample mean square error for estimator of the control variables :}}
Consistency of $\hat{\mu}$ is obtained under the following assumptions. %Details and proofs are in Appendix 

\begin{assumption}\label{assn:NPV1stStep0Dsty}
	There exists $\alpha_1 > 0$ and $\gamma_1 > 0$ such that $\sqrt{L/n} \  L^{\alpha_1 + 1} \xrightarrow[n \to \infty]{} 0$ and $\forall \, t \leq T$,
	\begin{enumerate}[noitemsep,nolistsep]
		\item $(x_{it}^{en},\xi_{it})$ is i.i.d over $i$, continuously distributed and $\Var(x_t^{en}|\xi_t)$ is bounded,
		
		\item $r^L(.)$ is the power series basis, and $\forall \, \xi_t \in \mathcal{S}_{\xi t}, \ f_{\xi t}(\xi_t) \, \geq \, \Pi_{d=1}^{d_\xi} (\xi_{dt} - \underline{\xi}_{dt})^{\alpha_1}(\bar{\xi}_{dt} - \xi_{dt})^{\alpha_1}$,
		
		\item \label{assn:NPV1stStep0Dsty_approx}  There exists $\beta_t^L$ such that $\sup_{\mathcal{S}_{\xi t}} || b_t(\xi_t) - \beta_t^{L \prime} \,  r^L(\xi_t)|| \leq C L^{- \gamma_1}$.
	\end{enumerate}
\end{assumption}

Assumption \ref{assn:NPV1stStep0Dsty} allows us to obtain a convergence rate for the sample mean-squared error of the generated covariates $\hat{V}_i$ following results in \cite{npv99}. Note that we use the power series basis so as to allow for the density of the regressors to be $0$ on the boundary of their support. A version of these results with arbitrary sieve basis but for a density bounded away from $0$ can be found in Appendix \ref{sec:appConsistency}.
If for instance $b_t$ is continuously differentiable up to order $p$, writing $d_{\xi} = d_1 + d_z$, then Assumption \ref{assn:NPV1stStep0Dsty} (\ref{assn:NPV1stStep0Dsty_approx}) holds with $\gamma_1 \, = \, p / d_{\xi}$ for different choices of sieve basis.

\begin{assumption}\label{assn:IN2stStep0Dsty_model}
	\
	\begin{enumerate}[noitemsep,nolistsep]
		\item $\E(\dot{u} |  X, Z)$ and $\Var(|| \dot{u}|| \, | X, Z)$ are bounded on $\mathcal{S}_{X,Z}$,
		
		\item  \label{assn:IN2stStep0Dsty_model_lipschitz} $\mathcal{M}$ and $k$ are Lipschitz and $g$ is bounded on $\mathcal{S}_{V}$,
		
		\item  \label{assn:IN2stStep0Dsty_model_density} $p^K(.)$ is the power series basis, and $\forall V \in \mathcal{S}_{V}, \ f_{V}(V) \, \geq \, \Pi_{d \leq k_2, \, t \leq T} \, (v_{t,d} - \underline{v}_{td})^{\alpha_2}(\bar{v}_{td} - v_{t,d})^{\alpha_2}$,
		
		\item \label{assn:IN2stStep0Dsty_model_approx} There exists $\gamma_2$ and $(\pi_{M,st}^{K})_{s,t}$ and $(\pi_{k,t}^{K})_{s,t}$ such that $\sup_{\mathcal{S}_{V}} | \mathcal{M}_{st}(V) -  \, \pi_{M,st}^{K \prime} p^K(V) | \leq C K^{- \gamma_2}$ and $\sup_{\mathcal{S}_{V}} | k_t(V) -  \, \pi_{k,t}^{K \prime} p^K(V) | \leq C K^{- \gamma_2}$ for all $s,t \leq T-1$, 
		
		\item $K^{\alpha_2 + 7/2} \Delta_n \to_{n \to \infty} 0$ and $  K^{\alpha_2 + 3/2} / \sqrt{n} \, \to_{n \to \infty} 0$. 
	\end{enumerate}
\end{assumption}

The second step of the proof of consistency is the derivation of a convergence rate for $\hat{\mathcal{M}}$ and $\hat{k}$ in sup norm and mean square norms, under Assumptions \ref{assn:IN2stStep0Dsty_model}. Convergence rates of nonparametric estimators of conditional expectations $\E(w_i| V_i=V) = h^{W}(V)$ using generated regressors are also derived in \cite{npv99}. However, they impose an orthogonality condition which by definition does not hold for our specific choices of $w_i$ and this rules out a direct application of their asymptotic results.
Specifically, writing $  e^W = w - h^W(V)$,
an additional assumption required to apply directly \cite{npv99} is $\E(e^W |X^{ex},V,Z)= 0$, that is, $e^W$ must be conditionally mean-independent of all variables involved in the first step, i.e, in the construction of the control variables.
This condition does not hold when $w_i$ is either a component of the matrix $ M= I - \dot{X} (\dot{X}' \dot{X})^{-1} \dot{X}' $ or of the vector $M \dot{y}$, because 
$\E(M|X^{ex},X^{en}, Z)= M \neq \E(M|V)=\mathcal{M}(V)$ and $\E(M \dot{y}|X^{ex},X^{en}, Z) =M g(V) + M \E(\dot{u}|X^{ex},X^{en}, Z) \neq  \E(M \dot{y}|V) = k(V) = \mathcal{M}(V) \, g(V)$.
To account for the difference $ \E(e^W|X, V) \neq \E(e^W|V)$, we write
\vspace{-5pt}
\begin{equation}\label{eq:reg_GalCase}
	w  = h^W(V) + \rho^W(X,Z) + e^{W*}, \ \E(e^{W*} | X,Z)  = 0,
\end{equation}

\vspace{-5pt}
where $\rho^W(X,Z) =  \E(e^W |X,Z) = \E(w | X,Z) - \E(w|V)$. For our choices of $w$, $\rho \neq 0$ and this will add an extra term to the convergence rate of the two-step estimator as well as, as will be clear in a later part of the paper, the asymptotic variance of a linear functional of this estimator.
This difference has been documented for instance in \cite{hr13} and \cite{mrs16}. 
\begin{assumption}\label{assn:cv_g}
	$\mathcal{M}$ and $g$ are continuous on $ \mathcal{S}_V$, $\mathcal{S}_V$ is a compact set and $\mathcal{M}(V)$ is invertible for all values $V \in \mathcal{S}_V$.
\end{assumption}
Recall that $\hat{g}(V) \, = \,  \hat{\mathcal{M}}(V)^{-1} \, \hat{k}(V).$
We use Assumption \ref{assn:cv_g} to obtain the rate of convergence of $\hat{g}(.)$ with continuity arguments. 
%Note that the continuity assumption somewhat overlaps with Assumption \ref{assn:IN2stStep0Dsty_model} (\ref{assn:IN2stStep0Dsty_model_approx}) and (\ref{assn:IN2stStep0Dsty_model_lipschitz}) as the existence of a linear approximation relies on smoothness assumptions. Moreover 
Note that while Assumption \ref{assn:id_M_GLn} requires the matrix to be invertible only $\mathbb{P}_V \, \text{a.s}$, under Assumption \ref{assn:cv_g} $\mathcal{M}(V)$ is invertible for all values of $V$ on the support $\mathcal{S}_V$. 
\begin{assumption}\label{assn:finiteE}
	Assume  $\E(||Q^{\delta}||) < \infty$ and $\E(||Q^{\delta} \, \dot{y}||) < \infty$.
\end{assumption}
Under the above assumptions, and defining $\gamma_n =  b_1(K) (K/n + K^{-2 \gamma_2} +  \Delta_n^2 b_2(K)^2)^{1/2}$, we obtain the following consistency result.
\begin{result}\label{rslt:consistency}
	Suppose Assumptions \ref{assn:NPV1stStep0Dsty}, \ref{assn:IN2stStep0Dsty_model}, \ref{assn:cv_g}, and \ref{assn:finiteE} hold. Assume also $\gamma_n \to 0$, $a_1(L) \Delta_n  \to 0$ and that $g$ is continuously differentiable on $\mathcal{S}_V$. Then $ \hat{\mu}\, \to_{\mathbb{P}} \, \E (\mu | \delta).$
\end{result}
%continuous differentiability is definitely overkill. Could do Cauchy Schwartz in the proof and use the MSE cv rate for $\hat{v} - v$ instead of the sup rate. But the assumption \gamma_n going to zero is already very constraining on \Delta_n, probably more than what the sup rate on $\hat{v} - v$ imposes.

\subsection{Asymptotic normality}\label{sec:asymptoticnorm}

\subsubsection{Definitions}\label{sec:trimming}

Recall that $\hat{V}_i = \tau(\tilde{V}_i)$, where $\tilde{V}_i =(\tilde{v}_{it})_{t \leq T}$ is the vector of residuals from the sieve regression of $x_{it}^{en}$ on $\xi_{it} = (x_{it}^{ex}, z_{it})$ and where $\tau$ projects onto $\mathcal{S}_V = \bigtimes_{d \leq d_2, \, t \leq T} \,  [\underline{v}_{td};\bar{v}_{td}]$.
The proof of asymptotic normality  will require $\tau$ to be twice differentiable. Thus we change the definition of $\tau$ from (\ref{eq:def_tau_consist}) to (\ref{eq:def_tau_asnorm}), defined in Appendix \ref{ref:trimmingapp}.
The new trimming function is twice differentiable and  projects onto a bounded superset of $\mathcal{S}_V$. For $\varsigma>0$, this superset is $\mathcal{S}_V^\varsigma := \bigtimes_{d \leq d_2, \, t \leq T} \, [\underline{v}_{td} - \varsigma ;\bar{v}_{td} + \varsigma]$. We will refer to $ \mathcal{S}_V^\varsigma $ as the ``extended support''. 
On this extended support, we  use extensions of the various regression functions. 
To be precise, let $p \in \N$ and a  twice-continuously differentiable function $m : \mathcal{S}_V \to \R^p$.
We define an extension of $m$ as $m^\varsigma : \mathcal{S}_V^\varsigma \to \R^p$ such that for all $V$ in $\mathcal{S}_V$, $m^\varsigma(V) = m(V)$, and $m^\varsigma$ is twice continuously differentiable on the extended support $ \mathcal{S}_V^\varsigma$.
We previously used, for $g$ a function of the variable $V$, the norm $|g|_d = \max_{|l| \leq d} \sup_{V \in \mathcal{S}_V} || \partial^{l} g(.) ||$. A corresponding norm for the extended functions  will change the supremum to a supremum over the extended support, i.e., $|g|_d^\varsigma = \max_{|l| \leq d} \sup_{V \in \mathcal{S}_V^\varsigma} || \partial^{l} g(.) ||$.
We now use bases of functions $p^K(.)$ defined on the extended support.
We point out here that unlike in Section \ref{sec:consistency}, we will not allow for the density of the regressors to go to zero on the boundaries of their support. We are thus silent on the choice of the basis.

\subsubsection{Linearization}\label{sec:lineariz}

To study asymptotic normality of $\sqrt{n} (\hat{\mu} - \E(\mu))$, we define
$
\hat{\mu}^\delta = \sum_{i=1}^n Q_i^{\delta} \, [ \dot{y}_i - \, \hat{g}(\hat{V_i}) ] /n,
$
then $\hat{\mu}  = \hat{\mu}^\delta / (\sum_{i=1}^n \delta_i/n)$. We thus focus first on $\hat{\mu}^\delta - \E(\mu \delta)$.  Define $W_i = (X_i, V_i, Z_i, u_i, \mu_i, \alpha_i)$ the vector of primitive variables where we now  write the variables as column vectors, e.g $V_i \, = \, (v_{i1}',... \, , v_{iT}')'$. Define also $\mathcal{G} \,  = \, ((b_t)_{t \leq T}, k, \mathcal{M})$  a vector of generic functions with $b_t : \mathcal{S}_{\xi t} \mapsto \R^{d_2}$, $k : \mathcal{S}_{V}^\varsigma \mapsto \R^{T-1}$ and $\mathcal{M} : \mathcal{S}_{V}^\varsigma \mapsto \ \mathcal{M}_{T-1}(\R)$. We now write $ \mathcal{G}_0 \, = \, ((b_{0t})_{t \leq T}, k_0, \mathcal{M}_0), $ for the true values of these functions, that is, for the nonparametric primitives of the model. Note that the functions we consider here are functions on the extended support. We dropped the exponent $\varsigma$ and will display it to avoid confusion whenever necessary.
We decompose
\begin{align}
	\sqrt{n} & (\hat{\mu}^\delta - \E(\mu_i \delta_i)) \,  = \  \frac{1}{\sqrt{n}}  \sum_{i=1}^n  Q_i^\delta \, [ \dot{y}_i - \ \hat{g}(\hat{V_i}) ] - \E(\mu_i \delta_i), \nonumber \\
	= \, & \frac{1}{\sqrt{n}}  \left[ \sum_{i=1}^n [ \delta_i \mu_i - \E(\mu_i \delta_i)]  + \, \sum_{i=1}^n  Q_i^\delta  \dot{u}_i + \, \frac{1}{\sqrt{n}} \sum_{i=1}^n [ Q_i^\delta g_0(V_i) - \E(Q^\delta g_0(V))]  -  \sum_{i=1}^n [ Q_i^\delta \hat{g}(\hat{V_i}) - \E(Q^\delta g_0(V))] \right] , \nonumber \\
	= \, & \frac{1}{\sqrt{n}}  \sum_{i=1}^n  [\delta_i \mu_i - \E(\mu_i \delta_i)]  + \, \frac{1}{\sqrt{n}}  \sum_{i=1}^n  Q_i^\delta  \dot{u}_i \, - \, \sqrt{n}  \big[ \mathcal{X}_n(\hat{\mathcal{G}}) -  \mathcal{X}_n(\mathcal{G}_0) \big], \label{eq:dvlptasNorm}
\end{align}
where we define
\begin{align*}
	\chi(W_i, \mathcal{G})    & =   Q_i^\delta \, \mathcal{M}\left( \tau \left[ (x_{it}^{en} - b_t(\xi_{it}))_{t \leq T} \right] \right)^{-1} \, k\left( \tau \left[ (x_{it}^{en} - b_t(\xi_{it}))_{t \leq T} \right] \right),\\
	\mathcal{X}_n(\mathcal{G})   & = \frac{1}{n} \sum_{i=1}^n [\chi(W_i, \mathcal{G}) - \E(\chi(W_i, \mathcal{G}_0))] \text{ and } \mathcal{X}(\mathcal{G}) = \E(\chi(W_i, \mathcal{G})) - \E(\chi(W_i, \mathcal{G}_0)),
\end{align*}
where $\tau$ %is as defined in Section \ref{sec:trimming} and 
ensures that the argument of $\mathcal{M}$ and $k$ lies in $\mathcal{S}_V^\varsigma$. 
Note that $\mathcal{X}(\mathcal{G}_0) = 0$.   
The use of generated covariates 
%\footnote{GC refersto the variables or their estimates ?} 
in place of the true value of the variables has a twofold impact on semiparametric estimators such as $\hat{\mu}$. First the nuisance parameters $\mathcal{M}$ and $k$ are estimated using the generated values. Second the estimators $\hat{\mathcal{M}}$ and $\hat{k}$ are evaluated at the generated values when plugged in in the sample average that defines $\hat{\mu}$. The dependence of $\mathcal{X}$ on $(b_t)_{t \leq T}$ highlights the latter aspect.

The two first terms in Equation (\ref{eq:dvlptasNorm}) are normalized sums of i.i.d random variables. Their asymptotic normality can be established by a standard CLT argument.
We focus on the last term, $\sqrt{n} \big[ \mathcal{X}_n(\hat{\mathcal{G}}) -  \mathcal{X}_n(\mathcal{G}_0) \big]$, which includes a composition of estimated infinite dimensional nuisance parameters. Typically, under restrictions detailed  in Appendix \ref{sec:lineq_app} its asymptotic distribution will be that of 
$ \sqrt{n} \,  \mathcal{X}_0^{(G)} [\hat{\mathcal{G}} - \mathcal{G}_0]$ where   $\mathcal{X}_0^{(G)}$ is the pathwise derivative of $\mathcal{X}$ at $\mathcal{G}_0$ and it is evaluated at $\hat{\mathcal{G}} - \mathcal{G}_0$. It is a linear functional of a vector of nonparametric estimators.
The structure of our asymptotic analysis is thus as follows. First we derive the asymptotic distribution of  $ \sqrt{n} \,  \mathcal{X}_0^{(G)} [\hat{\mathcal{G}} - \mathcal{G}_0]$ and obtain its influence function, and second we show that  $\sqrt{n} \big[ \mathcal{X}_n(\hat{\mathcal{G}}) -  \mathcal{X}_n(\mathcal{G}_0) \big]$ is asymptotically equivalent to its linearization. Finally these results are used together with (\ref{eq:dvlptasNorm}) to obtain the influence function of $\hat{\mu}^\delta $ and thus asymptotic normality of $\hat{\mu}$.

Focusing on the linearized term, the pathwise derivative applied to the estimators can be decomposed as the sum of $T + 2$ partial pathwise derivatives applied to each nonparametrically estimated function. We denote with  $v_t$  the t$^{th}$ component of $V$ and with $\frac{\partial g}{\partial v_t}(V_i)$  the Jacobian matrix of size $(T-1) \times d_2$. %Note that the function $\tau$ does not appear in the above formula, nor does any of its partial order derivatives. This is because when evaluated at the true value of $V$, by design $\tau$ simplifies to the identity function on $\mathcal{S}_V$ and its Jacobian is the identity matrix.
We define $\mathcal{X}_0^{(k)}[\tilde{k}]$ the partial pathwise derivative of $\mathcal{X}$ with respect to $k$ at $\mathcal{G}_0$ and evaluated at $\tilde{k}$, and similarly $\mathcal{X}_0^{(M)}[\tilde{\mathcal{M}}]$ and $(\mathcal{X}_0^{(bt)}[\tilde{b}_t])_{t \leq T}$. Assuming one can interchange expectation and differentiation, we follow \cite{mrs16} and write these derivatives as
\begin{align}
	\mathcal{X}_0^{(G)} [\mathcal{G} - \mathcal{G}_0] \,  = &  \, \mathcal{X}_0^{(k)} [k - k_0] + \mathcal{X}_0^{(M)}[\mathcal{M} - \mathcal{M}_0 ] +   \sum_{t = 1}^T \mathcal{X}_0^{(bt)}[b_t - b_{0,t}], \nonumber \\
	\,  = &\,  \int_V \left[ \lambda_{M}(v) \right.\, \left. \VecM \left((\mathcal{M} - \mathcal{M}_0)(v)\right) \right] d F_{V}(v) + \int_V \left[ \lambda_{k}(v) \, (k - k_0)(v) \right] d F_{V}(v) \nonumber \\
	& \  + \ \ \sum_{t = 1}^T \int_{\xi_t} \lambda_{bt}(\xi_t) \, (b_t - b_{0,t})(\xi_t) d F_{\xi_t}(\xi_t), \label{eq:linearX_int}
\end{align}
where the functions $\lambda_{.}$ are obtained by differentiation (see details in Appendix \ref{sec:lineq_app}),
\begin{align*}
	\lambda_{M}(v)   \, & = \, - \E \left( g(V_i)' \otimes (Q_i^\delta \mathcal{M}_0(V_i)^{-1}) \,  | \,  V_i = v \right) \, = \,  -  g(v)' \otimes [\E(Q_i^\delta | V_i = v) \mathcal{M}_0(v)^{-1}] ,\\
	\lambda_{k}(v) \, & = \, \E( Q_i^\delta \mathcal{M}_0(V_i)^{-1} \,  | \,  V_i = v )  \,  = \,  \E(Q_i^\delta \,  | \,  V_i = v ) \mathcal{M}_0(v)^{-1}, \\
	\lambda_{bt}(\xi_t)  \, & = \, - \E\left(Q_i^\delta \frac{ \partial g(V_i)}{\partial v_t}  \,  | \, \xi_{it} = \xi_t \right) . \vspace{-2pt}
\end{align*}

%\subsubsection{Asymptotic normality of $\hat{\mu}$}\label{sec:as_norm}

\subsubsection{Asymptotics of the linear part}\label{sec:asympt_lin_model}

The term $\sqrt{n} \mathcal{X}_0^{(G)} [\mathcal{G} - \mathcal{G}_0]$ is  a sum of linear functionals applied to the components of $\mathcal{G} -\mathcal{G}_0$, where $\mathcal{G} = ((b_{t})_{t \leq T}, k, \mathcal{M})$. To obtain its asymptotic distribution, we need to define the following random variables and functions,
\begin{align*}
	& e_i^M  = M_i - \E(M_i|V_i) = M_i - \mathcal{M}_0(V_i),\ e_i^{M*} = M_i - \E(M_i|X_i,Z_i) = 0, \\
	& \rho^M(X_i, Z_i)  = \E(M_i|X_i, Z_i) - \E(M_i|V_i) = M_i - \mathcal{M}_0(V_i),\\
	& e_i^{k}  = M_i \dot{y}_i - \E(M_i \dot{y}_i|V_i) = [M_i - \mathcal{M}_0(V_i)] g(V_i) + M_i \dot{u}_i,\\
	& e_i^{k*} = M_i \dot{y}_i - \E(M_i \dot{y}_i|X_i, Z_i) = M_i[\dot{u}_i  - \E(\dot{u}_i |X_i, Z_i)], \\ 
	& \rho^{k}(X_i, Z_i)  = \E(M_i \dot{y}_i|X_i, Z_i) - \E(M_i \dot{y}_i|V_i) = [M_i - \mathcal{M}_0(V_i)] g(V_i) + M_i \E(\dot{u}_i |X_i, Z_i).
\end{align*}
For a given $v$, $\lambda_M^j(v)$ is defined as the $j^{th}$ column of the matrix $\lambda_M(v)$ and
\begin{align*}
	\Lambda^M =  \int_v \left[ \, \lambda_M^1(v) p_1^K(v) , \ ... \ , \lambda_M^{1}(v) p_K^K(v) ,  \ ... \ , \lambda_M^{(T-1)^2}(v) p_1^K(v), \ ... \ ,  \lambda_M^{(T-1)^2}(v) p_K^K(v) \right]dF_V(v) ,
\end{align*}
of dimension $d_x \times K(T-1)^2$. Define similarly the matrix $\Lambda^{k}$. We will interchangeably index the columns of $\lambda_M$ as $\lambda_M^d$ with $d \leq (T-1)^2$ and as $\lambda_M^{st}$ with $1 \leq s,t \leq T-1$.
We also define
\begin{align*}
	H_t^M \, & = \, \E \left[ \frac{\partial \VecM( \mathcal{M}_0(V_i))}{\partial v_t} \otimes  p_i \otimes r_{it}' \right], \ H_t^{k} \,  = \, \E \left[\frac{\partial k_0(V_i)}{\partial v_t} \otimes p_i \otimes  r_{it}' \right], \\
	\text{d}P_t^M  & = \E \left[ \VecM(\rho_i^M)  \otimes \frac{\partial p^K (V_i)}{\partial v_t} \otimes  r_{it}' \right], \  \text{d}P_t^{k}   = \E \left[  \rho_i^{k}  \otimes \frac{\partial p^K (V_i)}{\partial v_t} \otimes r_{it}' \right].
\end{align*}
%. For those, we define
For a given $\xi_t$, $\lambda_{bt}^j(\xi_t)$ is defined as the $j^{th}$ column of the matrix $\lambda_{bt}(\xi_t)$ and for each $t$,
\begin{align*}
	\Lambda^{bt} \, = \,
	\int_{\xi_t} \left[ \, \lambda_{bt}^1(\xi_t) r_1^L(\xi_t) , \ ... \ ,  \lambda_{bt}^{1}(\xi_t) r_L^L(\xi_t) ,  \ ... \ , \lambda_{bt}^{T-1}(\xi_t) r_1^L(\xi_t) , \ ... \ ,  \lambda_{bt}^{T-1}(\xi_t) r_L^L(\xi_t)  \right]dF_{\xi_t}(\xi_t).
\end{align*}

\begin{assumption}\label{assn:LAE_model}
	\
	\begin{enumerate}[noitemsep,nolistsep]
		\item \label{assn:LAE_model_fctn} $\mathcal{M}_0, g_0$ are twice continuously differentiable with bounded first and second order derivatives,
		
		\item \label{assn:LAE_model_mmt} $\E(||Q_i^\delta||) < \infty$,  $\E(||\dot{u}|| \, | X, Z)$ is bounded on $\mathcal{S}_{X,Z}$, and $\mathcal{S}_{\xi t}$ for all $t \leq T$ and $\mathcal{S}_V$ are bounded,

		\item \label{assn:LAE_model_approx} There exists $\gamma_1$ and $\beta_{t}^L$ such that for all $t \leq T$, $\sup_{\mathcal{S}_{\xi_t}} || b_{0t}(\xi_t) - \beta_{t}^{L \, \prime} r^L(\xi_t) || \leq C L^{- \gamma_1}$. There exists $\gamma_2$, $\pi_{M}^K$ and $\pi_{k}^K$ such that $\sup_{\mathcal{S}_{V}^\varsigma} || \mathcal{M}_0^\varsigma(v) - p^K(v)' \, \pi_{M}^{K} || \leq C K^{- \gamma_2}$ and $\sup_{\mathcal{S}_{V}^\varsigma} || k_0^\varsigma(v) - p^K(v)' \, \pi_{k}^{K} || \leq C K^{- \gamma_2}$,
		
		\item \label{assn:LAE_model_sievebasis}   For all $t \leq T$, there exists $\Gamma_{1t}$, a $L \times L$ nonsingular matrix such that for $R_t^L(\xi_t) = \Gamma_{1t} r_t^L(\xi_t)$, $\E(R_t^L(\xi_t) R_t^L(\xi_t)')$ has smallest eigenvalue bounded away from $0$ uniformly in $L$.
		There exists $\Gamma_2$, a $K \times K$ nonsingular matrix such that for $P^K(V) = \Gamma_2 p^K(V)$, $\E(P^K(V) P^K(V)')$ has smallest eigenvalue bounded away from $0$ uniformly in $K$, 
		
		\item \label{assn:LAE_model_matrixbounded} $|| \Lambda^M||$, $|| \Lambda^{k}||$, and $|| \Lambda^{bt}||$ are bounded,
		
		\item \label{assn:LAE_model_rate}For $|R_t^L(\xi_t)|_0 \leq a_1(L)$, $|P^K(V)|_0^\varsigma \leq b_1(K)$, $|P^K(V)|_1^\varsigma \leq b_2(K)$, $|P^K(V)|_2^\varsigma \leq b_3(K)$, we have
		$\sqrt{n} K^{- \gamma_2} = o(1)$, $\max(\sqrt{K}, \sqrt{L} b_2(K)) \, a_1(L) \sqrt{L/n}  = o(1)$, 
		$b_2(K) \sqrt{n} L^{-\gamma_1} \allowbreak = o(1)$, $b_2(K)[\sqrt{L/n} + L^{- \gamma_1}][K + L]  = o(1)$, $ b_3(K) [L /\sqrt{n} + \sqrt{n} L^{- 2\gamma_1}] = o(1)$, $b_2(K)^2 \sqrt{K} [\sqrt{L/n} + L^{- \gamma_1}] = o(1)$,
		
		\item \label{assn:LAE_model_varbounded} \label{assn:MSE_model_extrammt} $\E(\dot{u} |X, Z) = 0$, and $\E(||v_{t}||^2|\xi_{t})$ and $\E(||\dot{u}||^2|X,Z)$ are bounded on $\mathcal{S}_{\xi t}$ and $\mathcal{S}_{X,Z}$ respectively.
	\end{enumerate}
\end{assumption}

\begin{result}\label{rslt:LAE_model}
	Under Assumptions \ref{assn:idCFA}, \ref{assn:cv_g} and \ref{assn:LAE_model},
	\begin{align}
		\ \sqrt{n}  & \mathcal{X}_0^{(G)}  [\hat{\mathcal{G}} - \mathcal{G}_0] \,  =  \, o_{\mathbb{P}}(1) \nonumber \\
		&  +  \ \frac{1}{\sqrt{n}}  \sum_{i=1}^n \Lambda^M (I_{(T-1)^2} \otimes \Theta) \,  \VecM(e_i^M)  \otimes p_i +  \frac{1}{\sqrt{n}}  \sum_{i=1}^n \Lambda^{k} (I_{T-1} \otimes \Theta) \, e_i^{k} \otimes p_i \label{eq:linearEq}\\
		&  +  \ \frac{1}{\sqrt{n}}  \sum_{i=1}^n  \sum_{t = 1}^T \left[\Lambda^M (I_{(T-1)^2} \otimes \Theta) (H_t^M - \text{d}P_t^M) + \Lambda^{k} (I_{T-1} \otimes \Theta) (H_t^{k} - \text{d}P_t^{k})  + \Lambda^{bt}\right] (I_{k_2} \otimes \Theta_1) \, v_{it} \otimes r_{it}, \nonumber
	\end{align}
	where we define $\Theta = \E(p_i p_i^{\prime})$ and $\Theta_1 = \E(r_i r_i^{\prime})$.
\end{result}

Result \ref{rslt:LAE_model} is obtained using Appendix \ref{app:lin_galcase}, where we look at Model (\ref{eq:reg_GalCase}) and obtain the influence function in the general case of a linear functional of nonparametric two-step series estimators, with some comments comparing our results to the existing literature. Under Assumption \ref{assn:LAE_model}, these results can be applied to the functionals in (\ref{eq:linearX_int}), choosing $w_i$ to be either a component of $M_i \dot{y}_i$ or of $M_i$. The regression functions $b_{0t}$ are estimated nonparametrically and the asymptotic distribution of functionals of such objects is studied in \cite{n97}.
Note that Assumption \ref{assn:LAE_model} (\ref{assn:MSE_model_extrammt}) is imposed to simplify computations. It amounts to strengthening the control function assumption, that is, Assumption \ref{assn:idCFA} (\ref{assn:idCFA_CFA}). 
Note also that in the influence function,  the terms $\text{d}P_t^M$ and $\text{d}P_t^k$ depend on $\rho^M$ and $\rho^k$. As discussed in Appendix \ref{app:lin_galcase}, this extra term compared to other applications of the CFA is due to the failure of the orthogonality condition discussed above Equation (\ref{eq:reg_GalCase}).

Under Result \ref{rslt:LAE_model}, we can  write $\frac{1}{\sqrt{n}}  \sum_{i=1}^n  [\delta_i \mu_i - \E(\mu \delta)]  + \, \frac{1}{\sqrt{n}}  \sum_{i=1}^n  Q_i^\delta  \dot{u}_i \, - \,  \sqrt{n}  \mathcal{X}_0^{(G)}  [\hat{\mathcal{G}} - \mathcal{G}_0]$ as $\frac{1}{\sqrt{n}} \sum_{i=1}^n s_{i,n} + o_{\mathbb{P}}(1)$ where
\begin{align*}
	s_{i,n} =  & [\delta_i \mu_i - \E(\mu \delta)] + Q_i^\delta  \dot{u}_i - \Lambda^M (I_{(T-1)^2} \otimes \Theta) \,  \VecM(e_i^M)  \otimes p_i -   \Lambda^{k} (I_{T-1} \otimes \Theta) \, e_i^{k} \otimes p_i  \\
	& + \sum_{t = 1}^T \left[\Lambda^M (I_{(T-1)^2} \otimes \Theta) (\text{d}P_t^M - H_t^M) + \Lambda^{k} (I_{T-1} \otimes \Theta) (\text{d}P_t^{k} - H_t^{k} )  - \Lambda^{bt}\right] (I_{k_2} \otimes \Theta_1) \, v_{it} \otimes r_{it}.
\end{align*}
We define $\Omega = \Var(s_{i,n})$ and the following objects, $\tilde{Q}_i^\delta =  Q_i^\delta  -  \, \E(Q_i^\delta | V_i)  \mathcal{M}_0(V_i)^{-1} M_i$, and 
$\Omega_0 = \Var \Big( [\delta_i \mu_i - \E(\mu \delta)] + \tilde{Q}_i^\delta   \dot{u}_i + \allowbreak \sum_{t = 1}^T  \E\left( \tilde{Q}_i^\delta \frac{ \partial g(V_i)}{\partial v_t}  | \, \xi_{it} \right) v_{it} \Big)$.

\begin{assumption}\label{assn:MSE_model}
	\
	\begin{enumerate}[noitemsep,nolistsep]
		\item \label{assn:MSE_model_approx} For each function $\lambda_a(.)$, column of $\lambda_{M}(.)$ or $\lambda_{k}(.)$, $\lambda_a(.)$ is continuously differentiable and
		there exists $\iota_{a}^K$ such that $|\lambda_a(V) - \iota_{a}^K p^K(V)|_1 =O(K^{-\gamma_3})$, and $(\iota_{aht}^L,\iota_{a\rho t}^L)$ such that as $L \to \infty$,
		$\E\left(||\E \left[ \lambda_a(V) \frac{\partial h^W(V)}{\partial v_{td}} | \xi_t \right] - \iota_{aht}^L r^L(\xi_t)||^2 \right) \to 0$ and
		$\E\left(||\E \left[ \rho_i^W \frac{\partial \lambda_a(V)}{\partial v_{td}} | \xi_t \right] - \iota_{a\rho t}^L r^L(\xi_t)||^2 \right) \to 0$.
		
		For all $t \leq T$, $b_t$ is continuous and there exists $\iota_{2t}^L$ such that $\E\left(|| \lambda_{bt}(\xi_t) - \iota_{bt}^L r^L(\xi_t)||^2 \right) \to 0$ as $L \to \infty$,
		
		\item \label{assn:MSE_model_rate} $b_2(K) K^{- \gamma_3} = o(1)$,
		
		\item \label{assn:MSE_model_var} $\E(Q_i^\delta||^2) < \infty$, $\E(|| \mu_i||^2) < \infty$, and there exists $C>0$ such that $\Omega_0 \geq C I_{d_x}$.
	\end{enumerate}
\end{assumption}

\begin{result}\label{rslt:as.var}
	Under Assumptions \ref{assn:idCFA}, \ref{assn:cv_g}, \ref{assn:LAE_model}' and \ref{assn:MSE_model}, $\Omega^{-1/2} \to_{n \to \infty} \Omega_0^{-1/2} $. Moreover, $|| \Lambda^M||$, $|| \Lambda^{k}||$, $||\Lambda^{bt}||$, $|| \Lambda^M (I_{(T-1)^2} \otimes \Theta) (H_t^M - \text{d}P_t^M)||$, and $|| \Lambda^{k} (I_{T-1} \otimes \Theta) (H_t^{k} - \text{d}P_t^{k})||$ are bounded. 
\end{result}

%\textit{IMPORTANT : Justify in mode details why $\E([s_i'c]^2) < \infty$, check that $\lambda_{2t}^j(.)$, $\lambda_{M}^j(.)$ $\lambda_{MY}^j(.)$ $  \E\left[\lambda_a(V) \frac{\partial g^W(V)}{\partial v_{td}}| \xi_t \right]$, and $  \E \left[ h_i^W \frac{\partial \lambda_a(V)}{\partial v_{td}} | \xi_t \right]$ are bounded functions.} : obtained imposing extra smoothness conditions.
Result \ref{rslt:as.var} states that $\Omega_0$ is the asymptotic variance of $\sqrt{n} \mathcal{X}_0^{(G)} [\mathcal{G} - \mathcal{G}_0]$. It additionally guarantees that  Assumption  \ref{assn:LAE_model} (\ref{assn:LAE_model_matrixbounded}) holds.  The boundedness of the two last matrices is added for later results on asymptotic normality.  We thus write Assumption \ref{assn:LAE_model}' for Assumption \ref{assn:LAE_model} without its condition (\ref{assn:LAE_model_matrixbounded}). Thus, under Assumption \ref{assn:LAE_model}' and Assumption \ref{assn:MSE_model}, Equation (\ref{eq:linearEq}) on $\sqrt{n} \mathcal{X}_0^{(G)}  [\hat{\mathcal{G}} - \mathcal{G}_0] $ holds. Note that the condition $\Omega_0 \geq C I_{d_x}$ holds if for instance $\Var(\mu_i | X_i, Z_i, u_i, V_i) \geq C I_{d_x}$ for some $C > 0$, or if a similar condition holds on the conditional variance of $\dot{u_i}$, as is typically assumed.

\subsubsection{Asymptotic distribution of $\hat{\mu}$}\label{sec:as_miou_together}

%We now assemble the arguments of  Section \ref{sec:lineariz} and \ref{sec:asympt_lin_model}. Recall that if Assumption \ref{assn:linearization} holds, Result  \ref{rslt:X_linear} will guarantee that $\sqrt{n}  (\hat{\mu}^\delta - \E(\mu_i \delta_i)) =  \frac{1}{\sqrt{n}} \sum_{i=1}^n s_{i,n} + o_{\mathbb{P}}(1)$.
%Let us now show that Assumption \ref{assn:linearization} does hold.

Asymptotic normality of $\hat{\mu}^\delta$ is obtained from the decomposition in (\ref{eq:dvlptasNorm}) and asymptotic equivalence of  $\sqrt{n} \big[ \mathcal{X}_n(\hat{\mathcal{G}}) -  \mathcal{X}_n(\mathcal{G}_0) \big]$ to its linearization. The conditions under which the two terms are equivalent asymptotically are detailed in  Appendix \ref{sec:lineq_app}, one main condition being that of stochastic equicontinuity. To guarantee that it holds, we follow  \cite{clvk03} and impose smoothness conditions. In particular, 
for $\mathcal{S}_W$ a bounded subset of $\R^k$, we define for a function $g : \mathcal{S}_W \mapsto \R$, and $\varrho > 0$, the norm $||g||_{\infty,\varrho} = |g|_{ \left \lfloor{\varrho}\right \rfloor } + \max_{ |r| =  \left \lfloor{\varrho}\right \rfloor } \sup_{w \neq w'} \frac{|\partial^rg(w) - \partial^rg(w')|}{|| w - w' ||^{\varrho -  \left \lfloor{\varrho}\right \rfloor }}$. 

We define $\mathcal{C}_c^{\varrho}(\mathcal{S}_W)$ to be the set of continuous functions $g : \mathcal{S}_W \mapsto \R$ such that $||g||_{\infty,\varrho} \leq c$. The set $\mathcal{H}_{2t,c}^\varrho = \mathcal{C}_c^{\varrho}(\mathcal{S}_{\xi_t})^{k_2}$ will be the class of vector valued functions taking values in $\R^{k_2}$, each component of which lies in $\mathcal{C}_c^{\varrho}(\mathcal{S}_{\xi_t})$. Since $k$ and $\mathcal{M}$ are defined on the extended support $\mathcal{S}_V^\varsigma$, we define $\mathcal{H}_{\mathcal{M},c,c'}^\varrho = \mathcal{C}_c^{\varrho}(\mathcal{S}_{V}^\varsigma)^{(T-1)^2} \cap \{ g : \forall V \in \mathcal{S}_{V}^\varsigma , \, \lambda_{\min} (M_{g(V)})> c' \}$, where $M_{g(V)}$ is the matrix formed by the coefficients of $g(V)$, and $\mathcal{H}_{k,c}^\varrho = \mathcal{C}_c^{\varrho}(\mathcal{S}_{V}^\varsigma)^{T-1}$. Finally, for the entire vector of infinite dimensional parameters $\mathcal{G}$, we define the set $\mathcal{H}_{c,c'}^\varrho = \left( \times_{t \leq T}\mathcal{H}_{2t,c}^\varrho \right) \times \mathcal{H}_{\mathcal{M},c,c'}^\varrho \times \mathcal{H}_{k,c}^\varrho$.

\begin{assumption}\label{assn:SE}
	$\mathcal{G}_0 \in \mathcal{H}_{c,c'}^\varrho$ for some $(c,c') \in (\R_+^*)^2$ and $\varrho > \max(Td_2,d_z + d_1)/2$.
\end{assumption}

Condition (\ref{assn:LAE_model_approx}) of  Assumption \ref{assn:LAE_model}' needs to be modified to
\textit{``There exists $\gamma_1$ and $\beta_{t}^L$ such that for all $t \leq T$, $\sup_{\mathcal{S}_{\xi_t}} || g^{ent}(\xi_t) - \beta_{t}^{L \, \prime} r^L(\xi_t) || \leq C L^{- \gamma_1}$. There exists $\gamma_2$, $\pi_{M,st}^K$ and $\pi_{k,t}^K$ such that $ | \mathcal{M}_0^\varsigma(.)_{st} - p^K(.)' \, \pi_{M,st}^{K} |_1^\varsigma \leq C K^{- \gamma_2}$ and $| k_0^\varsigma(.)_t - p^K(.)' \, \pi_{k,t}^{K} |_1^\varsigma \leq C K^{- \gamma_2}$, for all $1 \leq s,t \leq T-1$''.}
This modification is a stronger assumption, changing the approximation rate to be over the $|.|_1$ norm instead of the sup norm. Assumption \ref{assn:LAE_model}'' is the modified version of Assumption \ref{assn:LAE_model}'. We can now state the main results of this section.

\begin{result}\label{rslt:linear_miou_num}
	Under Assumptions  \ref{assn:idCFA}, \ref{assn:cv_g}, \ref{assn:LAE_model}'', \ref{assn:MSE_model} and \ref{assn:SE}, assuming moreover that $a_1(L) \Delta_n = o(n^{-1/4})$ and $b_2(K) [K/n + K^{-2 \gamma_2} +  \Delta_n^2 b_2(K)^2]^{1/2} = o(n^{-1/4})$,  then 
	$$ [\hat{\mu}^\delta - \E(\mu_i \delta_i)] =  \frac{1}{\sqrt{n}} \sum_{i=1}^n s_{i,n} + o_{\mathbb{P}}(1).$$
\end{result}

\begin{assumption}\label{assn:as.normlt} 
	$\E[||\mu_i - \E(\mu)||^4] < + \infty$, $\E[||Q_i^\delta||^4] < \infty$ and $\Var(\delta_i) >0$. Also, $\E(|| v_{t}||^4|\xi_{t})$ and $\E(|| \dot{u}||^4|X,Z)$ are bounded on $\mathcal{S}_{\xi t}$ for all $t \leq T$ and on $\mathcal{S}_{X,Z}$ respectively.
\end{assumption}

\begin{result}\label{rslt:as.normlt}
	Under Assumptions  \ref{assn:idCFA}, \ref{assn:cv_g}, 
	\ref{assn:LAE_model}'', \ref{assn:MSE_model}, \ref{assn:SE} and \ref{assn:as.normlt}, assuming moreover that $a_1(L) \Delta_n = o(n^{-1/4})$ and $b_2(K)[K/n + K^{-2 \gamma_2} +  \Delta_n^2 b_2(K)^2]^{1/2} = o(n^{-1/4})$,
	$$\sqrt{n}  [\hat{\mu} - \E(\mu|\delta)] \to^d \mathcal{N}(0, \Phi^{-2} \Xi), $$
	where $\Xi = \Omega_0 +  \E( (\delta_i - \Phi) s_i) \E(\mu |\delta)' +  \E(\mu | \delta) \E( (\delta_i - \Phi) s_i)' + (\Phi - \Phi^2) \E(\mu | \delta)\E(\mu | \delta)'$.
\end{result}
\vspace{-12pt}

\section{Illustration}\label{sec:illus}

As an empirical exercise, we apply our method to a model of labor supply.  An important parameter of interest  is the elasticity of intertemporal substitution (EIS), or Frisch elasticity, which quantifies how labor supply responds to anticipated wage changes over time.
To estimate the EIS, the literature derives a labor supply response equation from a life-cycle labor supply model. Agents face uncertainty in future wage rates and interest rates, and choose paths of consumption, hours worked and possibly an asset distribution to maximize the expected discounted present value of lifetime utility under a dynamic budget constraint. A linear relationship is often obtained between the log of labor supply and the log of wage, see e.g \cite{mc85,f04}.
Within this framework, to estimate the EIS \cite{z97} studies the following model
\begin{equation}\label{eq:model_empirics1}
	\ln h_{it} = \alpha_i +  \mu \ln \omega_{it} + b'\chi_{it} + \epsilon_{it}, \qquad i=1..n, \ t=1..T,
\end{equation}
where $h_{it}$ is the number of hours worked annually by agent $i$, $\omega_{it}$ is the hourly wage, $\mu$ is the EIS and $\chi_{it}$ is composed of age and additional demographics (taste shifters). The additive fixed effect $\alpha_i$ depends on the initial marginal utility of wealth. %Equation (\ref{eq:model_empirics1})  can be obtained from assuming that the drifts in the marginal utility of wealth are constant (\cite{mc85}). not true: it would actually need a trend in it, and age in the data does not behave like a trend
The empirical exercise considers instead a version with heterogeneity and, using a reduced form approach, studies
\begin{equation}\label{eq:model_empirics}
	\ln h_{it} = \alpha_i +  \mu_i \ln \omega_{it} + b'\chi_{it} + \epsilon_{it}, \qquad i=1..n, \ t=1..T,
\end{equation}
where the parameter $\mu_i$ is now allowed to vary across agents and potentially covary with $\omega_{it}$ and $\chi_{it}$. In (\ref{eq:model_empirics1}) $\omega_{it}$ and $\epsilon_{it}$ are usually assumed correlated. The main reason is $\epsilon_{it}$ depends on unexpected shocks to the marginal utility of wealth that are realized at time $t$ and these shocks are correlated with realized wage at $t$. Another reason %, which we do not address here,  : I am deleting this because now I have a somewhat reduced form here, so if there is m.e, put the mu_i times error in the residual and assume CFA
is that  wage is usually observed with error in survey data. 
This thus justifies applying the method developed in this paper to (\ref{eq:model_empirics}), and to do so, we will use the data  in \cite{z97}.

%i dont mention the measurement error bias coming from the division by number of hours (mentioned in Keane 2011 and Altonji 1986) because Ziliak uses reported measure of wage - ie the individuals report themselves the wage.

We now discuss how to obtain control variables and provide some support for Assumptions \ref{assn:idCFA} and \ref{assn:id_M_GLn}. We can rewrite (\ref{eq:model_empirics}) with first differences, $
\dot{\ln} \, h_{it} =  \mu_i \, \dot{\ln} \, \omega_{it} + b'\dot{\chi}_{it} + \dot{\epsilon}_{it}$.
For a panel with periods $t=0,..,T$, we consider the following conditions,
\begin{equation}
	\text{for all } t \leq T, \ \ln \omega_{it} = \gamma_0 + \gamma_1 \ln \omega_{it-1} + \eta_{it}, \text{ and } \{\dot{\epsilon}_{it},(\eta_{is})_{s=t}^{T} \} \,\indep \, \{ \omega_{i0}, (\eta_{is})_{s=1}^{t-1}, (\chi_{is})_{s=1}^{t} \}. \label{eq:wageAR_argUtWealthdpdce}
\end{equation}
These conditions are a simplification and chosen for illustrative purposes: the purpose  is to provide an intuitive framework in which Assumptions \ref{assn:idCFA} and \ref{assn:id_M_GLn} hold. 
The assumption that $\ln \omega_{it}$ follows an autoregressive process can be found in, e.g, \cite{hnr88} : (\ref{eq:wageAR_argUtWealthdpdce}) uses a simple AR(1) model.
The independence condition  is based on the fact that the first-difference $\dot{\epsilon}_{it} = \epsilon_{it+1} - \epsilon_{it}$ is composed of shocks unexpected by the agent prior to $t+1$. Thus %it may be correlated with $\eta_{it}$, the unexpected shock to $\ln \omega_{it}$, however 
as a white noise can be assumed  independent of prior values of $\omega_{is}$, $s \leq t-1$, and of past and present values of the exogenous demographics. It is potentially correlated with future values of these exogenous demographics as for instance, if $\chi_{it}$ includes an indicator for bad health  at period $t$ then it will be positively impacted by a positive shock to hours of work at period $t-1$. % Thus (\ref{eq:wageAR_argUtWealthdpdce}) imposes independence $\chi_{it}$ and future shocks only.
 
Define $v_{it} = \ln \omega_{it} - \E(\ln \omega_{it}|  \ln \omega_{i0})$, $V_{i} = (v_{it})_{1 \leq t \leq T}$, $W_i = (\ln \omega_{i1},..,\ln \omega_{iT})'$. Note that $v_{it} = \sum_{k=0}^{t-1} \gamma_1^k [\eta_{it-k} + \gamma_0]$ and that this relationship can be inverted so as to write $\eta_{it}$ as a function of $v_{i1},..,v_{it}$. Since the impact of $\chi_{it}$, $b$, is not heterogeneous, we use the framework detailed in Appendix \ref{sec:mix_homog_heterog} and  condition on additional instruments for $\chi_{it}$, $Z_i^{\chi} = (\chi_{i1}, \chi_{i0})$. Then under (\ref{eq:wageAR_argUtWealthdpdce}), we have

\vspace{-30pt}

\begin{align*}
	\text{for all } t \leq T, \ \E(\dot{\epsilon}_{it} | V_i, W_i,   Z_i^{\chi})& = \E\left(\E (\dot{\epsilon}_{it} | \eta_{i1},.., \eta_{iT}, \ln \omega_{i0},  Z_i^{\chi}) | V_i, W_i, Z_i^{\chi} \right)\\
	&=  \E\left( \E (\dot{\epsilon}_{it} | \eta_{it},.., \eta_{iT}) | V_i, \ln \omega_{i0},  Z_i^{\chi} \right) = \E (\dot{\epsilon}_{it} | \eta_{it},.., \eta_{iT})=: g_t(V_i)
\end{align*} 
where the second-to-last equality holds because $( \eta_{it},.., \eta_{iT})$ is a deterministic function of $V_i$. Note that the control function approach is directly on the first-differenced residual which is slightly different from  Assumption \ref{assn:idCFA} but does not change the two-step identification procedure.
%the difference is that in that section, the regressor for \omega_{it} is \omega_{it-1}. plugging in the same reduced form equations for the past values of \omega, this becomes \omega_{it} = m_t(\omega_{i0}, v_{i1}, v_{it-1}) + v_{it}, for the same v_{it}. However regressing on \omega_{i0} does not give v_{it} but another residual $\tilde{v}$ function of \omega_{i0}, v_{i1}, v_{it-1} and v_{it} (except if linear regression). And if the indpdce assumptions are initially on v, then the CFA assumption with conditional expectation function of the tilde{v} should not necessarily hold. More rigorous application would be a first stage REGRESSING ON $\omega{i t-1}$ instead of $\omega_{i0}$
As for Assumption \ref{assn:id_M_GLn},  Result \ref{rslt:IA_illus} states that it holds under very mild conditions on the support of $\ln \omega_{i0}$.

The dataset constructed in \cite{z97}
is described in Section 2.1 of the paper. It is a selected sample from the Survey Research Center subsample of the Panel Study of Income Dynamics composed of $532$ men aged $22$ to $55$, married and working at all periods of the panel. The demographics  $\chi_{it}$ are number of children, age and a dummy variable for bad health. We take $T=3$ and since we use an initial period $t=0$ to construct the control variables through $\ln \omega_{i0}$, the total number of periods in the panel is $4$. We use a panel of years $1979$ to $1982$ where period $1$ is year $1980$, period $T$ is year $1982$.
The estimation steps are as follows.
First, we estimate the control variables $v_{it}$ using the following specification, slightly  richer than described above
\vspace{-5pt}
$$\ln \omega_{it} = \lambda_{0} + \gamma_{1t} \ln \omega_{i0} + \gamma_{2t}' \chi_{i}^{\text{GC}} + v_{it},$$

\vspace{-5pt}
where $\chi_{i}^{\text{GC}}$ includes $\chi_{i1}$ and age$_{i1}^2$. We choose this simple linear specification adding only a quadratic in age to avoid the curse of dimensionality which potentially has a strong impact given our small sample size.
The subsequent steps are estimation of the vector $b$, the functions $g_t(.)$ for $t \leq T-1$ and the average partial effect $\E(\mu_i)$: more details are provided in the Online Appendix \ref{sec:app_illus} together with the estimation results.

\vspace{-6pt}

\section{Conclusion}

This paper proved that in a correlated random coefficient panel model, the average partial effect $\E(\mu_i)$ is identified when the strict exogeneity condition standard in the literature is relaxed to allow for time-varying endogeneity. The approach imposes the existence of a control function to control for the dependence between the regressors and the time-varying disturbance with first-stage variables function of instruments and regressors. Within-group and between-group operations are then combined to disentangle the heterogeneous impact of the regressors from the time-varying disturbances. An estimator of $\E(\mu_i | \det(\dot{X}_i' \dot{X}_i) > \delta_0)$ can be naturally constructed from the identification argument: we showed its asymptotic normality and computed its asymptotic variance.

We highlight two directions for future research, following the contribution of \cite{gp12}. First, it would be of interest to relax the condition $T > d_x +1$: long enough panels might not be available to identify average partial effects in models with multiple covariates and heterogeneous impacts.
A second direction for future work would address the dependence of the limit of the proposed estimator, $\E(\mu| \delta)$, on the constant $\delta_0$: $\delta_0$ is arbitrarily fixed and choosing a value when implementing the estimator can be an issue in practice. One could instead study the asymptotic properties of $\E(\mu_i | \det(\dot{X}_i' \dot{X}_i) > \delta_n)$ as $\delta_n \to 0$. %This could give a sense of an optimal choice for $\delta_n$ as a function of the sample size $n$. Note that extending the asymptotic analysis of \cite{gp12} is nontrivial, as our estimation procedure involves nuisance parameters estimated with nonparametric two-step series estimators. 

\bibliography{bibCRC}

\appendix

\newpage

%\vspace{-0.5pt}

\begin{center}
	{\Large \textbf{PAPER APPENDIX}}
\end{center}

\vspace{-20pt}

\section{Combining Random Coefficients with Common Parameters}\label{sec:mix_homog_heterog}

If it is known to the researcher that the random coefficients associated to some covariates $l_{it} \in \R^{d_l}$ have a degenerate distribution, we propose a different procedure. 
Consider the model
\begin{align}\label{eq:modelwCsttCoeff}
	y_{it} = \alpha_i +  \, l_{it}^{ \prime} \, b + x_{it}^{\prime} \, \mu_i + \epsilon_{it},
\end{align}
where $ x_{it} \, = \, (x_{it}^{ex},x_{it}^{en}) \in \R^{d_x}$, $x_{it}^{ex} \in \R^{d_1}$ are exogenous regressors and $x_{it}^{en} \in \R^{d_2}$ allowed to be endogenous. The regressors $ l_{it} \, = \, (l_{it}^{ex},l_{it}^{en}) \in \R^{k_m}$, where $l_{it}^{ex} \in \R^{d_{l1}}$ is exogenous and $l_{it}^{en} \in \R^{d_{l2}}$ is endogenous,  have a homogeneous impact.
When the control variables are the residuals of the nonparametric regression of $x_{it}^{en}$, these extensions are useful for two reasons. First, if all coefficients are assumed heterogeneous, the procedure described in Section \ref{sec:maindid} requires $T$ to be at least $d_x + d_l + 2$ which is quite restrictive. Second,  the vector of control variables will be of dimension at least $(d_{l2} + d_2)(d_x + d_l + 2)$. Large dimension for $V$ is highly undesirable because of the curse of dimensionality, as the closed-form for $g$ involves nonparametric conditional expectations conditional on $V$.
We now present an alternative identification argument valid whenever  $T \geq d_x + 2$. The dimension of the conditioning set then does not have to exceed $d_2(d_x + 2)$ when $T = d_x + 2$.
We modify Assumption \ref{assn:idCFA} (\ref{assn:idCFA_CFA}) and impose 
\begin{equation}\label{assn:CFAmixHomog_heterog}
	\E( \epsilon_{it} | Z_i^L, X_i, V_i) \, = \, f_t(V_i),
\end{equation}
where $V_i$ is an identified function of the regressors $X_i$ and the instruments $Z_i$. $Z_i^L$ is composed of $L_i^{ex}$ and instruments for $L_i^{en}$. Defining $M_i$ and $Q_i$ are defined as before,
\begin{align}
	&M_i \dot{y}_i \, =  \, M_i \dot{L}_i b +\, M_i g(V_i)   + \, M_i \dot{u}_i , \textrm{ with } \E( M_i \dot{u}_i | Z_i^L, X_i , V_i) =0, \label{eq:id_b_whomog} \\
	&\Rightarrow \E(M_i \dot{y}_i | V_i) \,  = \, E(M_i \dot{L}_i | V_i) \,  b \, + \, \mathcal{M}(V_i) g(V_i) \nonumber \\
	&\Rightarrow M_i \, \mathcal{M}(V_i)^{-1} \E(M_i \dot{y}_i | V_i) \,  = \, M_i \, \mathcal{M}(V_i)^{-1}  E(M_i \dot{L}_i | V_i) \,  b \, + \, M_i g(V_i). \nonumber
\end{align}
 This suggests a modification of the procedure developed in \cite{r88} for the identification of $b$. Defining $\Delta \dot{y}_i \, = \, \dot{y}_i - \, \mathcal{M}(V_i)^{-1} \E(M_i \dot{y}_i | V_i)$ and $\Delta \dot{L}_i \, = \, \dot{L}_i -  \, \mathcal{M}(V_i)^{-1}  E(M_i \dot{L}_i | V_i)$, we obtain
$$M_i \, \Delta \dot{y}_i \, = \, M_i \, \Delta \dot{L}_i b + \, M_i \dot{u}_i.$$
Since $\E( Z_i^L \, M_i \dot{u}_i)\, = \, \E( Z_i^L \,  M_i \E(\dot{u}_i| Z_i^L, X_i , V_i)) \, = \, 0$ and using $M_i = M_i' = M_i^2$, we obtain
\begin{equation}\label{eq:idDeltaCsttCoeff_endog}
	b \, = \, \left(\E(   \Delta \dot{L}_i' M_i Z_i^{L}) \E(  Z_i^{L\, \prime}   M_i\Delta \dot{L}_i ) \right)^{-1} \E(   \Delta \dot{L}_i' M_i Z_i^{L}) \, \E( Z_i^{L \, \prime} \, M_i \Delta \dot{y}_i),
\end{equation}
under the assumption that $ \E(  Z_i^{L\, \prime}   M_i\Delta \dot{L}_i)$ is of full column rank.
As $b$ is now identified, identification of $g$ and $\E(\mu_i)$ is obtained by applying the results of Section \ref{sec:maindid} to $y_{it} - l_{it}' b$. 
%\textit{Note that a cleaner version of the condition $\E(  \dot{Z}_i^{L \, \prime} \, M_i \, \Delta \dot{L}_i )$ is non singular is mentioned in Newey (2009). Assumption 2.1}
If $d_x = 1$, i.e., only one regressor has heterogeneous impact, then we retrieve the condition $T \geq 3$. If  $v_{it}$ is scalar and $T = 3$, then the dimension of the conditioning set for the nonparametric regressions is $3$, independently of the number of regressors in $l_{it}$.

\vspace{-3pt}

\section{Another Approach if $T > d_x + 2$}\label{sec:don_idea}

The dimension of the conditioning set is also large if the number of periods is large (but fixed) even keeping the number of endogenous regressors  fixed. Since for identification,  all that one needs is $T \geq d_x + 2$, then if  $T > d_x + 2$ one can use a subset of $d_x + 2$ periods among the $T$ available. However, it is possible\footnote{I thank Donald Andrews for this suggestion.}
to use the $T$ periods without increasing the dimension of the conditioning set, under a slightly modified Assumption \ref{assn:idCFA}.
Assume  that $T > d_x + 2$ and denote $\mathcal{T}$ the set of subsets of $ \{ 1,\, .. \, , T \}$ of cardinality $d_x + 2$. The cardinality of $\mathcal{T}$ is $T \choose d_x + 2$. Consider $\tau \in \mathcal{T}$ a subset of $d_x + 2$ periods, write $\tau = (t_1, \, ... \, , t_{d_x + 2})$ where $t_1 < \, \cdots \, < t_{d_x + 2}$. We write with a superscript $\tau$ the vectors that are defined using only the periods in $\tau$. For instance, $V_i^{\tau} \, = \, (v_{i t_1}, \, .. \, , v_{i t_{d_x + 2}})$.
Then (\ref{eq:model}) implies
$ y_i^{\tau} = X_i^{\tau} \mu_i + \epsilon_i^{\tau}.$
\begin{assumption}\label{assn:CFAbigT}
	There exist a set of functions $(h_t^{\tau})_{\substack{ t \in \tau \\ \tau \in \mathcal{T}}}$ and identified functions $(C_t)_{t \leq T}$ such that, defining $v_{
		it} = C_t(x_{it}, z_{it}) \in \R^{d_v}$,
	$$ \forall \, \tau \in \mathcal{T}, \ \forall \,  t \in \tau, \  \E(\epsilon_{it} \, | \, X_i^{\tau}, V_i^{\tau}) \, = \, h_t^{\tau}(V_i^{\tau}).$$
\end{assumption}
Assumption  \ref{assn:idCFA}, i.e., $\E(\epsilon_{it} \, | \, X_i, V_i) \, = \, f_t(V_i)$, does not imply Assumption \ref{assn:CFAbigT}. %If Assumption \ref{assn:idCFA} (\ref{assn:idCFA_CFA}) holds, then by the law of iterated expectations $\E(\epsilon_{it} \, | \, X_i^{\tau}, V_i^{\tau}) \, = \,  \E(f_t(V_i) \, | \, X_i^{\tau}, V_i^{\tau})$ is not necessarily a function of $V_i^{\tau}$ only.
For a given $\tau$ in $\mathcal{T}$, changing the definition of $g_t$ to $g_t^{\tau} = h_{t+1}^{\tau} - h_t^{\tau}$, identification of the vector of functions $g^{\tau}$ follows from the same first step if $\mathcal{M}^{\tau}(V^{\tau})$ is invertible.
Identification of $\E(\mu_i)$ follows from (\ref{eq:Miou}), which becomes
$\E(\mu_i)  \, =  \, \E(Q_i^{\tau} \dot{y}_i^{\tau} \, - \, Q_i^{\tau} g^{\tau}(V_i^{\tau})).$ 
Since this holds for any subset $\tau$, we can also write
$$
\E(\mu_i)  \, =  \, \frac{1}{{T \choose d_x + 2}} \, \sum_{\tau \in \mathcal{T}} \E(Q_i^{\tau} \dot{y}_i^{\tau} \, - \, Q_i^{\tau} g^{\tau}(V_i^{\tau})).
$$

%\textit{Justify why this is good for asymptotics.}

\section{Complements to the Illustration}\label{sec:app_illus}

\medskip

\textbf{Theoretical Result: }

\medskip

The following result states that Assumption \ref{assn:id_M_GLn} holds in the model analyzed in Section \ref{sec:illus} under appropriate assumptions. Recall that the control variable is given by $v_{it} = \ln \omega_{it} - \E(\ln \omega_{it}|  \ln \omega_{i0})$.  The proof can be found at the end of Section \ref{sec:app_proofid}. 

\begin{result}\label{rslt:IA_illus}
	Let Assumption \ref{assn:id_finiteE} and (\ref{eq:wageAR_argUtWealthdpdce}) hold with	$\gamma_1 \notin \{0,1\}$. 
	Assume additionally that  $(v_{i1},.. ,v_{iT})'$ is continuously distributed on $\R^{T}$ and either:
	\begin{enumerate}[label=(\roman*)]
		\item $\ln \omega_{i0}$ has a discrete distribution with at least two support points,
		\item Or $\ln \omega_{i0}$ is continuously distributed and $\Int(\mathcal{S}_{\ln \omega_{i0}}) \neq \emptyset$.
	\end{enumerate}
	Then Assumption \ref{assn:id_M_GLn} holds.
\end{result}

\bigskip

\textbf{Estimation Results: }

\medskip

The estimator of $\E(\mu_i)$ uses nonparametric estimators of conditional expectation functions conditional on $V$. To estimate these nonparametric regression functions, we choose the same basis of approximating functions of $V$ (power series) and the same number of approximating terms for each of these functions. The exact choice of approximating functions is decided using a leave-one-out cross-validation (CV) criterion with criterion function the mean square forecast error of the random variable $M_i$. 
%The set of conditioning variables is $(v_{i1}, v_{i2}, v_{i3})$, hence the terms that can be included in the sieve basis are $v_{it}$ raised to various powers and interactions of those. 
We report the CV values for some specifications in Table \ref{table:cv_values}. Our choice will be the power series basis of degree $2$.

To estimate $b$, we follow the  method developed in Section \ref{sec:mix_homog_heterog} of the Appendix. The added set of instruments is $Z_i^{\chi} = (\chi_{i1}, \chi_{i0}$, age$_{i1}^2$, age$_{i0}^2)$. Defining the differences $\Delta \dot{\ln} h_{i}$ and $\Delta \dot{\chi}_{i}$ as in Appendix \ref{sec:mix_homog_heterog}, and their estimators as $\widehat{\Delta \ln \dot{h}}_{i}$ and $\widehat{\Delta \dot{\chi}}_{i}$, the estimator for $b$ is 
$$\hat{b} = \left(\sum_{i=1}^n  \widehat{\Delta \dot{\chi}}_{i}^{\ \prime} Z_i   \left[ \sum_{i=1}^n  Z_i' Z_i \right]^{-1} \sum_{i=1}^n Z_i'\widehat{\Delta \dot{\chi}}_{i} \right)^{-1} \left(\sum_{i=1}^n  \widehat{\Delta \dot{\chi}}_{i}^{\ \prime} Z_i  \left[  \sum_{i=1}^n  Z_i' Z_i  \right]^{-1} \sum_{i=1}^n Z_i' \, \widehat{\Delta \ln \dot{h}}_{i} \right).$$

The third step is estimation of $g$. We first estimate $\mathcal{M}(.)$ and $k(.) = \E(M_i [\dot{\ln} h_{i} - \dot{\chi}_{i}'b] \, |V_i = .)$ using a series approximation and plugging in the estimate $\hat{b}$. Using the estimators $\hat{\mathcal{M}}(.)$ and $\hat{k}( .)$, our estimate of $g$ is $\hat{g} = \hat{\mathcal{M}}^{-1} \hat{k}$. 

The fourth and final step to obtain the estimate of  $\E(\mu_i)$  entails computing the sample analog of  the moment  $\E(Q_i [\dot{\ln} h_{i} - \dot{\chi}_{i}'b - g(V_i)]) $. This gives $\hat{\mu} = \frac{1}{n} \sum
_{i=1}^n Q_i [\dot{\ln} h_{i} - \dot{\chi}_{i}' \hat{b} - \hat{g}(\hat{V}_i)] $ reported in Table \ref{table:empirics_rslt}. %We note that this value is in the range of those reported in \cite{z97}.

The estimation results are reported below. For comparison purposes, we also include the 2SLS estimator as well as the FF estimator which are among those favored by Ziliak. We use the stacked instrument matrix suggested by the author, where for each time period we use the instruments of all past time periods. We note that  our estimator $\hat{\mu}$ (reported as CRC) is in the range of the values reported in \cite{z97} and those computed for the specific time periods we consider.

\bigskip

\begin{center}
	\begin{tabular}{c c} 
		
		Terms included & CV values \\ [0.5ex] 
		\hline\hline
		$(v_{it}, v_{it}^2)_{t \leq T}$, $v_{i1}v_{i2}$, $v_{i2}v_{i3}$  & 252  \\ 
		
		$(v_{it}, v_{it}^2)_{t \leq T}$ & 262 \\
		
		$(v_{it}, v_{it}^2, v_{it}^3)_{t \leq T}$ & 332 \\
		
		$(v_{it}, v_{it}^2, v_{it}^3)_{t \leq T}$, $v_{i2}v_{i3}$ & 278 \\
		
		$(v_{it})_{t \leq T}$, $v_{i1}v_{i2}$, $v_{i2}v_{i3}$ & 269 \\ 
		
		$(v_{it})_{t \leq T}$ & 264 \\ [1ex] 
		\hline
	\end{tabular}
	\captionof{table}{Cross-validation values}\label{table:cv_values}
\end{center}
%I report specifications 1 2 3 4 6 8 of the code ziliak_CV

\begin{center}
	\begin{tabular}{l c c c} 
		
		Hours worked (log) &  CRC  &2SLS&FF \\ [0.5ex] 
		\hline\hline
		Hourly wage (log) &  0.253 & 0.209 & 0.254 \\
		
		Number of children  & -0.0287 & 0.0943 & 0.0815 \\ 
		
		Age & 0.211 & -0.130 & -0.134 \\
		
		Age$^2$ & 0.0001 & 0.0014 & 0.0014\\
		
		Bad health indicator & -0.0254 & 0.0644 & 0.0918\\
		\hline
	\end{tabular}
	\captionof{table}{Results for $\hat{\mu}$ and $\hat{b}$}\label{table:empirics_rslt}
\end{center}
%The code used for this table is ziliak_forpaper. It does not use W3.

\newpage

\section{Proofs of Results in Sections \ref{sec:assn} and \ref{sec:app_illus}}\label{sec:app_proofid}  %, \ref{sec:mix_homog_heterog} and \ref{sec:don_idea}}

We define $M(X)  = I_{T-1} - \dot{X} (\dot{X}' \dot{X})^{-1} \dot{X}'$  if $\dot{X}$ is of full rank or $ M = I - \dot{X}  \dot{X}^{ +}$ if not, where $\dot{X}^{ +}$ is the Moore Penrose inverse. In what follows, the $i$ indexing random variables is dropped whenever this does not hinder readability.

\begin{proof}[Proof of Result \ref{rslt:cns_g}:]
	We first state a more general statement of this result, then give a proof of this general result.
	\begin{result}
		Let Assumptions \ref{assn:idCFA} and \ref{assn:id_finiteE} hold.
		\begin{enumerate}
			\item If $V$ is continuously distributed, then $g$ is identified almost surely on $\mathcal{S}_V$ if and only if Assumption \ref{assn:id_M_GLn} holds.
			\item If $V$ has discrete support $\{\omega_1,\dots,\omega_l\}$ with corresponding probability weights $(p_1,\dots,p_l)$, then $g$ is identified on the support of $V$ if and only if a) for all $j \leq l$, $\mathcal{M}(\omega_j)$ is nonsingular, or b) for all $(\omega_{j_1},\dots,\omega_{j_k})$ satisfying  $\mathcal{M}(\omega_{j_s}) \neq \{0\}$ for all $s \leq k$, then for all set of nonzero vectors $(e(\omega_{j_1}),\dots,e(\omega_{j_k}))$ satisfying $e(\omega_{j_s}) \in \mathcal{M}(\omega_{j_s})$ for all $s \leq k$, $\sum_{s\leq k} p_{j_s} e(\omega_{j_s}) \neq 0$.
		\end{enumerate}
	\end{result}

	\textit{Proof of Case 1:}
	
	If $\mathcal{M}(V)$ is nonsingular $\mathbb{P}_V \, \text{a.s}$, then by the argument given in Section \ref{sec:maindid} the function $g$ is identified on a set of probability 1 on $\mathcal{S}_V$. We now show the other implication. We assume that IA does not hold and show that $g$ is not identified by constructing a function $\tilde{g}$ that differs from $g$ on a set of positive probability, and a random vector $\tilde{\mu}$ such that the dgp with  $\tilde{g}$ and $\tilde{\mu}$ as primitives generates the same distribution of observables as the dgp with the true value of these parameters does. Note that the only constraint on the function $g$ is the normalization $\E(g(V)) = 0$.
	
	If  IA does not hold, there exists a subset $N_V \subset \mathcal{S}_V$ such that $\mathbb{P}_V(N_V) >0$ and for all $V \in N_V$, $\mathcal{M}(V)$ is singular. 
	%Since $V$ is continuously distributed, The density of $V$, $p_V$ must be strictly positive for some point in $N_V$
	%$\mathcal{M}(V)$ is singular for some value $V_0 \in \mathcal{S}_V$. By continuity of $V\mapsto \mathcal{M}(V)$ and of the determinant function, this implies that there exists $\eta >0$ such that for all $V \in \mathcal{B}(0,\eta)$, $ \mathcal{M}(V)
	%shouldn't I be able to find a neighborhood of a point with noninvertibilitu on the whole neighborhhod ? will then need to replace N_V in what follows
	For all $V\in N_V$, there exists a vector  $e(V)\in \R^{T-1}$ such that $e(V) \neq 0$ and $\mathcal{M}(V)e(V)=0$. Given this function $e$, one can find a real-valued function $V \in N_V \mapsto \gamma(V) \in \R$ and a subset $\tilde{N}_V \subset N_V$ such that $\mathbb{P}_V(\tilde{N}_V) >0$, for all $V \in \tilde{N}_V$, $\gamma(V) \neq 0$ and $\E(\gamma(V)e(V))=0$. For instance, partition $N_V$ into $T$ subsets with all strictly positive $\mathbb{P}_V$-measure and on each of these $T$ subsets take $\gamma$ to be a constant to be determined. Write the vector of these constants $\vec{\gamma}$. The equality $\E(\gamma(V)e(V))=0$  is of the form $E\vec{\gamma}=0$ where $E$ is a matrix of size $T-1 \times T$: this linear equation is guaranteed to have a nonzero solution.
	We now define for all $V \in \tilde{N}_V$, $\tilde{e}(V) = \gamma(V) e(V)$. It satisfies $\tilde{e}(V)  \neq 0$ and $\mathcal{M}(V)\tilde{e}(V) =0$.
	\begin{align*}
		\mathcal{M}(V) \, \tilde{e}(V) \, = \, 0& \Rightarrow \,  \tilde{e}(V)' \, \mathcal{M}(V) \, \tilde{e}(V) \, = \, 0 \\
		& \Rightarrow \, \E( \, \tilde{e}(V)' \, M(X)'M(X)\tilde{e}(V) \, | \,V) \, = \, 0 \\
		& \Rightarrow \, \E( \, || M(X) \, \tilde{e}(V) ||^2 \, | \, V) \, = \, 0.
	\end{align*}
	Since $|| M(X) \, \tilde{e}(V) ||^2$ is  positive for all values of $X$ for a given $V$, this implies that  $ M(X) \, \tilde{e}(V)  = 0$ with probability $1$ on the support of $X$ conditional on $V$. That is, $M(X) \, \tilde{e}(V) \, = \, 0, \ \mathbb{P}_{X|V} \textrm{-a.s} \,$ which implies that $\tilde{e}(V) = \dot{X} (\dot{X}' \dot{X})^{-1} \dot{X}'\tilde{e}(V), \ \mathbb{P}_{X|V} \textrm{-a.s}$. 
	
	We define the distribution of a new random vector $(\tilde{Y},X)$ conditional on $V$. If $V \notin \tilde{N}_V$, take $(\tilde{Y},X)|V$ to have the same distribution as $(Y,X)|V$. If $V \in \tilde{N}_V$, define the random vector $\beta = (\dot{X}' \dot{X})^{-1} \dot{X}'\tilde{e}(V) \in \R^{T-1}$, then  $\tilde{e}(V) = \dot{X}\beta$. Define also the function $\tilde{g}(V) = g(V) + \tilde{e}(V)$ and the random vector $\tilde{\mu} = \mu - \beta$. Since $\tilde{e}(V)  \neq 0$, $\tilde{g}(V) \neq g(V)$ and $\tilde{\mu} \neq \mu$, and we also have  $\E(\tilde{g}(V)) = 0$. Define then  $\tilde{\dot{y}}_{t} = \dot{x}_{t}'\tilde{\mu} + \tilde{g}_t(V) + \dot{u}_t$, and the vector $(\tilde{Y},X)|V$ through $\tilde{y}_t = y_1 + \sum_{s\leq t-1} \tilde{\dot{y}}_s $. Note that $\tilde{\dot{y}}_{t} = \dot{y}_t$ and $\tilde{y}_t = y_t$, which implies that if $V \notin \tilde{N}_V$, the distribution of $(\tilde{Y},X)|V$  is also that of $(Y,X)|V$.
	
	\medskip
	
	\textit{Proof of Case 2:}
	
	We first assume that neither a) nor b) hold. Then there exists $(\omega_{j_1},\dots,\omega_{j_k})$ satisfying  $\mathcal{M}(\omega_{j_s}) \neq \{0\}$ for all $s \leq k$ and a set of nonzero vectors $(e(\omega_{j_1}),\dots,e(\omega_{j_k}))$ satisfying $e(\omega_{j_s}) \in \mathcal{M}(\omega_{j_s})$ for all $s \leq k$ and $\sum_{s\leq k} p_{j_s} e(\omega_{j_s}) = 0$. As in Case 1, it is thus possible to construct a function $\tilde{g}$ such that, if $j = j_s$ for some $s \leq k$, then $\tilde{g}(\omega_{j}) = g(\omega_{j}) + e(\omega_{j_s})$, otherwise $\tilde{g}(\omega_{j}) = g(\omega_{j})$. It is then guaranteed that $\E(\tilde{g}(V)) = 0$. As in the previous case also, we know that  for all  $s \leq k$, $e(\omega_{j_s}) = \dot{X}\beta, \ \mathbb{P}_{X|V = \omega_{j_s}} \textrm{-a.s}$ where $\beta$ is a random vector. Thus following the procedure described in Case 1, one can also construct $\tilde{\mu} \neq \mu$ and a vector $\tilde{Y}$ such that the distribution of $(\tilde{Y},X)|V=\omega_{j}$  is also that of $(Y,X)|V=\omega_{j}$ for all $j\leq l$. This implies that $g$ is not uniquely identified.

	We now look at the other implication. If a) holds, then as before by the argument given in Section \ref{sec:maindid}, the function $g$ is identified on a set of probability 1 on $\mathcal{S}_V=\{\omega_1,\dots,\omega_l\}$. If b) holds and $g$ is not identified, then there exists  $\tilde{g}(V)$ such that  $\tilde{g}(V) \neq g(V)  $, $\E(\tilde{g}(V)) \neq 0$ and $\mathcal{M}(\omega_s)g(\omega_s) = \E(M(X)Y|V=\omega_s)$ is equal to $\mathcal{M}(\omega_s)\tilde{g}(\omega_s)$ for all $s\leq k$. Define $e(\omega_s) = \tilde{g}(\omega_s) - g(\omega_s)$, then $\mathcal{M}(\omega_s)e(\omega_s) = 0$. Take $(\omega_{j_1},\dots,\omega_{j_k})$ the nonempty subset of support points such that $e(\omega_{j_s})$ is nonzero, then 
	\begin{align*}
		& \forall s \leq k, \ e(\omega_{j_s}) \neq 0, \  \mathcal{M}(\omega_s)e(\omega_s) = 0, \\
		& \text{and } \sum_{s\leq k} p_{j_s} e(\omega_{j_s}) = 0
	\end{align*}
	which is a contradiction with b).
\end{proof}

\begin{proof}[Proof of Result \ref{rslt:equiv_span}]
	Let us fix $V$. Recall that by the Proof of Result  \ref{rslt:cns_g}, $\mathcal{M}(V)$ is singular if and only if there exists a vector $e(V) \neq 0$ such that  $ M(X) \, e(V) = 0$ with probability $1$ on the support of $X$ conditional on $V$. Let us denote with $\tilde{\mathcal{S}}$ the set of $X$ such that  $ M(X) \, e(V) = 0$, then $\mathbb{P}_{X|V}(\tilde{\mathcal{S}})=1$ and $e(V) \in \bigcap_{X \in \tilde{\mathcal{S}}} \Span(\dot{X})$. This concludes the proof.
	%cannot do what follows because M(X) is not continuous when X does not have full rank
	%Let us fix $V$. Recall that by the Proof of Result  \ref{rslt:cns_g}, $\mathcal{M}(V)$ is singular if and only if there exists a vector $e(V) \neq 0$ such that  $ M(X) \, e(V) = 0$ with probability $1$ on the support of $X$ conditional on $V$. If $X$ given $V$ has a discrete discrete distribution, then we obtain directly that $ M(X) \, e(V) = 0$ for all $X\in \mathcal{S}_{X|V}$. If $X$ given $V$ is continuously distributed then by continuity of the function $X \mapsto  M(X) \, e(V)$, we also obtain $ M(X) \, e(V) = 0$ for all $X\in \mathcal{S}_{X|V}$. Thus this shows that  $\mathcal{M}(V)$ is singular if and only if there exists a vector $e(V) \neq 0$ such that  $ M(X) \, e(V) = 0$ for all $X\in \mathcal{S}_{X|V}$, that is,  if and only if $\bigcap_{X \in \mathcal{S}_{X|V}} \Span(\dot{X}) \neq \{0\}$.
\end{proof}

\begin{proof}[Proof of Result \ref{rslt:discrete_IV}]
	Let us write $\underline{k}_V$ the number of points in the support of $\dot{X}$ conditional on  $V$, noting that  $\underline{k}_V \leq k_V$. We write $\mathcal{S}_{\dot{X}|V} = \{\dot{X}^{(1)}, ..., \dot{X}^{(\underline{k}_V)}\}$. Then by Result \ref{rslt:equiv_span},  $\mathcal{M}(V)$ is invertible if and only if $\ \bigcap_{k \leq \underline{k}_V} \Span(\dot{X}^{(k)}) = \{0\} $. The dimension of this intersection can be computed as follows,
	\begin{align*}
		\dim\left(\bigcap_{k \leq \underline{k}_V} \Span(\dot{X}^{(k)}) \right) & =  \sum_{k \leq \underline{k}_V} \rank(\dot{X}^{(k)}) - \rank {\small\begin{pmatrix}
				\dot{X}^{(1)} & \dot{X}^{(2)} &  & &  \\
				\dot{X}^{(1)} &  & \dot{X}^{(3)} & &  \\
				\vdots & & & \ddots & \\
				\dot{X}^{(1)} & & & & \dot{X}^{(\underline{k}_V)}
		\end{pmatrix}}\\
		& \geq  \underline{k}_V d_x - \min((\underline{k}_V - 1)(T-1),\ \underline{k}_V d_x )
	\end{align*}
	where we used the main theorem in \cite{t02} to obtain the first line and  the full column rank assumption for the second line. By this second line, if $\min((\underline{k}_V - 1)(T-1),\ \underline{k}_V d_x ) < \underline{k}_V d_x$, then $\ \bigcap_{k \leq \underline{k}_V} \Span(\dot{X}^{(k)})$ cannot be the trivial set $ \{0\} $. Moreover, $\min((\underline{k}_V - 1)(T-1),\ \underline{k}_V d_x ) < \underline{k}_V d_x$  if and only if $(\underline{k}_V - 1)(T-1) < \underline{k}_V d_x$, that is, if and only if $\underline{k}_V < \frac{T-1}{T-1 - d_x}$. The result follows.
\end{proof}

%\mathcal{S}_{Z|V}

\begin{proof}[Proof of Result \ref{rslt:counterex_1cstt}]
	Consider first the case where $t_0 = 1$. $\mathcal{M}(V)$ is singular if and only if either there is only one point on the support of $\dot{X}$ given $V$, or if all draws are collinear. By assumption, there are at least two points on the support of $\dot{X}$ given $V$, $\dot{X}^{(1)}= \dot{m}(Z^{(1)},V)$ and $\dot{X}^{(2)}= \dot{m}(Z^{(2)},V)$ where $\dot{X}^{(1)} \neq \dot{X}^{(2)}$. Let us assume that they are collinear. By full rank assumption, %WRONG $X^{(j)} = b_{t_0}(Z^{(j)},V) \neq 0$ for $j=1,2$, that is,  
	$\dot{X}^{(1)} \neq 0$ and $\dot{X}^{(2)}\neq 0$. Let us  write the vector of time differences as $\dot{X}^{(1)} = (a_1,a_2,...,a_{T-1})'$ and $\dot{X}^{(2)} =(a_1 - \delta,a_2,...,a_{T-1})'$ with $\delta \neq 0$. Since $\dot{X}^{(1)}$ and $\dot{X}^{(2)}$ are collinear, there exists $\lambda \notin \{0,1\}$ such that 
	$(a_1,a_2,...,a_{T-1}) = \lambda (a_1 - \delta,a_2,...,a_{T-1})$. This implies that $a_2=...=a_{T-1} = 0$ which means that $X^{(1)} = (X_1^{(1)}, X_1^{(1)} + a_1, X_1^{(1)} + a_1,..., X_1^{(1)}+a_1)$ and is thus of the form  $(a,b,b,..b)'$. 	
	The reciprocal implication follows from noticing that if  $X^{(1)}$ is of this form, then $X^{(2)}$ is too, and $\dot{X}^{(1)}$ and $\dot{X}^{(2)}$ are collinear.	
	The case where $t_0 = T$ follows from the same argument. 
	
	We now consider the case where $1 < t_0 < T$. Let us  write $\dot{X}^{(1)} = (a_1,a_2,...,a_{T-1})'$ and $\dot{X}^{(2)} =(a_1,..,a_{t_0-2},a_{t_0 - 1} + \delta,a_{t_0}-\delta,a_{t_0+1},...,a_{T-1})'$ with $\delta \neq 0$. As for case a), if  $\dot{X}^{(1)}$ and $\dot{X}^{(2)}$ are collinear then there exists $\lambda \notin \{0,1\}$ such that 
	$\dot{X}^{(1)} = \lambda \dot{X}^{(2)}$. This implies that $a_1 = ..=a_{t_0-2} = a_{t_0+1}=...=a_{T-1}=0$. Plus, 
	$\lambda (a_{t_0 - 1} + \delta)=  a_{t_0 - 1}$ and $\lambda(a_{t_0}-\delta) =  a_{t_0}$  together imply that $a_{t_0 - 1} = - a_{t_0}$. In terms of the components of $X^{(1)}$, this means that $X_{t_0 + 1}^{(1)} - X_{t_0 }^{(1)} = -(X_{t_0}^{(1)} -X_{t_0- 1}^{(1)} )$ thus $X_{t_0 + 1}^{(1)}=X_{t_0- 1}^{(1)}$. Thus $X^{(1)}$ can be written  $(b,..b,a,b,..b)'$ where $a$ is the $t_0^{\text{th}}$ component and $a\neq b$. The  reciprocal implication follows as in the first case.
\end{proof}

The following result will be used in the proof of Result \ref{rslt:evdok_z}.

\begin{result}\label{rslt:evdok}
	Fix $V \in \mathcal{S}_V$. Assume that for all $t \leq T-1$, there exists $X^{(t)} \in \mathcal{S}_{X| V}$ such that a) $x_{t+1}^{(t)}=x_{t}^{(t)}$, b) $\dot{X}^{(t)}$ is of full rank, and c) for all $\epsilon >0$, $\mathbb{P}_{X|V}\left(\mathcal{B}(X^{(t)},\epsilon)\right)>0$. Then $\mathcal{M}(V)$ is invertible.
\end{result}

\begin{proof}[Proof of Result \ref{rslt:evdok}]
	For each $t\leq T-1$,  write $e_t \in \R^{T-1}$ the vector whose $t^{th}$ component is $1$ and has $0$ everywhere else. Then,
	\begin{equation*}
		x_{t+1}^{(t)}=x_{t}^{(t)}\ \Leftrightarrow \dot{x}_{t}^{(t)}=0\ \Leftrightarrow e_t'\dot{X}^{(t)}=0 \ \Leftrightarrow e_t'M(\dot{X}^{(t)})e_t =1.
	\end{equation*}
	Fix $t \leq T-1$. $\dot{X}^{(t)}$ is of full column rank $d_x$ thus $\det(\dot{X}^{(t) \prime}\dot{X}^{(t)}) \, \neq 0$. The determinant function being continuous, there exists $\overline{\epsilon}>0$ such that $\forall \, \dot{X} \in \mathcal{B}(\dot{X}^{(t)},\overline{\epsilon})$, $Rank (\dot{X})= d_x$. Note that for notational simplicity, we use the same notation $\mathcal{B}$ independently of the dimension of the vector space.
	
	Take $c \in \R^{T-1}$ such that $\mathcal{M}(V) c = 0$. Then by the argument given in the proof of Result \ref{rslt:cns_g}, $M(X) \, c \, = \, 0, \ \mathbb{P}_{\dot{X}|V} \textrm{-a.s} \, $. Thus $M(X) \, c \, = \, 0$ for all  $\dot{X}$ in $\mathcal{B}(\dot{X}^{(t)},\overline{\epsilon})$ except on a set of $\mathbb{P}_{\dot{X}|V}$-measure $0$.
	By assumption, for all $\epsilon >0$, $\mathbb{P}_{X|V}\left(\mathcal{B}(X^{(t)},\epsilon)\right)>0$. Using equivalence of norms in finite dimension, the triangular inequality implies that for all $\epsilon >0$, $\mathbb{P}_{\dot{X}|V}\left(\mathcal{B}(\dot{X}^{(t)},\epsilon)\right)>0$. Thus for all $\epsilon \leq \overline{\epsilon}$,  there exists $\dot{X}^{(t, \epsilon)} \in \mathcal{B}(\dot{X}^{(t)},\epsilon)$ such that $M(\dot{X}^{(t, \epsilon)}) \, c \, = \, 0$.	
	Additionally, $\dot{X}$ is of full rank  for all $\dot{X}$ in $\mathcal{B}(\dot{X}^{(t)},\overline{\epsilon})$, thus $M(X) \, c$ is a continuous function of $\dot{X}$ on $\mathcal{B}(\dot{X}^{(t)},\epsilon)$.
	By continuity, letting $\epsilon$ converge to zero guarantees that $M(X^{(t)}) \, c = 0$. Moreover  $\dot{X}^{(t) \, \prime} e_t \, = \, 0$ implies $M(X^{(t)}) e_t = M(X^{(t)})' e_t = e_t$. Thus,
	\begin{align*}
		\forall t \leq T-1, \ M(X^{(t)}) \, c \, = \, 0 \, & \Rightarrow \, \forall t \leq T-1, \ e_t^{\, \prime} \,   M(X^{(t)}) \, c \, = \, 0, \\
		& \Rightarrow \, \forall t \leq T-1, \ e_t^{\, \prime} \, c \, = \, c_t \, = \, 0, \\
		& \Rightarrow \, c \, = \, 0.
	\end{align*}
	Hence $\mathcal{M}(V)$ is invertible.
\end{proof}

\begin{proof}[Proof of Result \ref{rslt:evdok_z}]
	Under the assumptions of the result, the assumptions of Result \ref{rslt:evdok} hold if we define  $X^{(t)} = \dot m (Z^{(t)},V)$.	
\end{proof}

\begin{proof}[Proof of Result \ref{rslt:florenstype}]
	We consider a sub-multiplicative matrix norm $||.||$ (e.g the Frobenius norm) and for any matrix $B$ and any $\delta >0$ write $\mathcal{B}(B,\delta)$ the ball of radius $\delta$ centered around $B$.
	The first step of the proof is to show that for all $\tilde{B} \subset \mathcal{B}(\dot{X},\gamma)$ such that $\mathbb{P}_{\dot{X}|V}\left(\tilde{B}\right)=\mathbb{P}_{\dot{X}|V}\left(\mathcal{B}(\dot{X},\gamma)\right)$, then
	\begin{equation}\label{eq:proof_emptyset}
		\bigcap_{W \in \tilde{B}} \ \text{Span}(W) =  \{0\}.
	\end{equation}
	Let  $\tilde{B} \subset \mathcal{B}(\dot{X},\gamma)$ such that $\mathbb{P}_{\dot{X}|V}\left(\tilde{B}\right)=\mathbb{P}_{\dot{X}|V}\left(\mathcal{B}(\dot{X},\gamma)\right)$. To obtain (\ref{eq:proof_emptyset}), first take $b \in  \bigcap_{W \in \tilde{B}} \ \text{Span}(W)$. We show that for $W_0 \in \Int \mathcal{B}(\dot{X},\gamma)$ of full column rank,  $b \in \text{Span}(W_0)$. Indeed, since $\mathbb{P}_{\dot{X}|V}\left(\tilde{B}\right)=\mathbb{P}_{\dot{X}|V}\left(\mathcal{B}(\dot{X},\gamma)\right)\subset \mathcal{S}_{\dot{X}|V}$, there exists a sequence $(W_n)_{n \in \N} \subset \tilde{B}$ such that $W_n \to W_0$ as $n \to \infty$ and $W_n$ is of full column rank for all $n \in \N$. For $n \geq 0$, since $W_n \in \tilde{B}$ there exists $\lambda_n \in \R^{d_x}$ such that $b = W_n\lambda_n$. This implies that $\lambda_n = (W_n' W_n)^{-1} W_n' b \to \lambda_0 = (W_0 W_0')^{-1}W_0' b$ as $n \to \infty$ and $b=W_0 \lambda_0 \in \text{Span}(W_0)$, which is the desired result. We now show how this allows us to obtain (\ref{eq:proof_emptyset}).
	Take 
	$b \in \bigcap_{W \in \tilde{B}} \ \text{Span}(W)$.
	Since $\dot{X}$ is of full column rank, 
	$b \in \text{Span}(\dot{X})$ and there exists $\lambda \in \R^{d_x}$ such that $b = \dot{X} \lambda$.
	Write $(e_1, \, ... , \, e_{T-1} )$ the canonical basis of $\R^{T-1}$. Since $d_x < T-1$, there exists $k \leq T-1$ such that $e_k \not\in \text{Span}(\dot{X})$. For $j \leq d_x$, consider the matrix $E_j$ of size $(T-1) \times d_x$ with all columns set to $0$ except for the $j$th column equal to $e_k$, and define the matrix $W_j = \dot{X} + \tilde{\gamma} E_j $ where $\tilde{\gamma} \neq 0$ is such that $W_j \in \Int \mathcal{B}(\dot{X},\gamma)$ and $W_j$ is of full column rank. Then for all $j\leq d_x$, $b \in \text{Span}(W_j)$ and  there exists $\lambda^{(j)} \in R^{d_x}$ such that $b = W_j \lambda^{(j)}$.
	For all $j\leq d_x$, we have 
	\begin{align*}
		& \dot{X} \lambda  = W_j \lambda^{(j)} = (\dot{X} + \tilde{\gamma} E_j) \lambda^{(j)} \\
		& \Rightarrow  \dot{X} (\lambda - \lambda^{(j)})/\tilde{\gamma} = E_j \lambda^{(j)} = \lambda_j^{(j)} e_k
	\end{align*}
	where $\lambda_j^{(j)} $ is the $j$th component of the vector $\lambda^{(j)} $. Since  $e_k \not\in \text{Span}(\dot{X})$, the previous equality implies that $\lambda_j^{(j)} =0$ and that $\lambda = \lambda^{(j)}$. This implies in turn that the $j$th component of the vector $\lambda$, $\lambda_j$, is also $0$.
	Since $\lambda_j = 0$ for all $j\leq d_x$, then $\lambda =0$ which implies $b = \dot{X} \lambda = 0$. This proves (\ref{eq:proof_emptyset}).
	
	The second step concludes using Result \ref{rslt:equiv_span}. Take a subset $\tilde{S}\subset \mathcal{S}_{X|V}$ such that  $\mathbb{P}_{X|V}\left(\tilde{S}\right)=1$. Write $\dot{\tilde{S}} = \{\dot{X} \, | \, X \in \tilde{S}\}$, then $\mathbb{P}_{\dot{X}|V}\left(\dot{\tilde{S}}\right)=1$. Write $\tilde{B} = \dot{\tilde{S}} \cap  \mathcal{B}(\dot{X},\gamma) $, then $\mathbb{P}_{\dot{X}|V}\left(\tilde{B}\right)=\mathbb{P}_{\dot{X}|V}\left(\mathcal{B}(\dot{X},\gamma)\right)$. Thus,	
	$$\bigcap_{W \in \tilde{S}} \text{Span}(\dot{W}) = \bigcap_{W \in \dot{\tilde{S}}} \text{Span}(W) \subset \bigcap_{W \in \tilde{B}} \ \text{Span}(W)  = \{0\},$$
	by Equation (\ref{eq:proof_emptyset}).
	By Result \ref{rslt:equiv_span}, $\mathcal{M}(V)$ is invertible.
\end{proof}

\begin{proof}[Proof of Corollary \ref{rslt:lin1ststage_csttA}]
	Consider the given value $V \in \mathcal{S}_V$. Equation (\ref{eq:lin1ststage_csttA}) implies that 
	$$\dot{X} = \dot{Z} A' + \dot{V}.$$
	We consider a sub-multiplicative matrix norm $||.||$ (e.g the Frobenius norm) and for any matrix $B$ and any $\delta >0$ write $\mathcal{B}(B,\delta)$ the ball of radius $\delta$ centered around $B$.
	By the full column rank and support assumptions,  there exists $\dot{Z} \in \mathcal{S}_{\dot{Z}|V}$ such that $\dot{X}= \dot{Z} A + \dot{V}$ is of full column rank and for some $\delta >0$, $\mathcal{B}(\dot{Z},\delta) \subset \mathcal{S}_{\dot{Z}|V=\bar{V} }$. Recall that  $\mathcal{B}(\dot{Z},\delta)  \subset \R^{ (T-1)\times d_x }$. 
	
	Let us show that for some $\gamma >0$, $\mathcal{B}(\dot{X},\gamma) \subset \mathcal{S}_{\dot{X}|V}$, then Result \ref{rslt:florenstype} will conclude the proof. Since $A$ is of full row rank, let $a$ be the norm of the transpose of its Moore-Penrose inverse, denoted $A^* = (AA')^{-1}A'$. Take $\gamma = \delta /a$. Let $\dot{X} + h \in \mathcal{B}(\dot{X},\gamma) $ and define  $\dot{Z}_h = \dot{Z} + h A^*$. Then $\dot{Z}_h \in \mathcal{B}(\dot{Z},\delta) $ since $||\dot{Z}_h  - \dot{Z} || = ||h A^* || \leq \gamma a = \delta$. Moreover, $\dot{Z}_h A'  + \dot{V} = h A^* A' +  \dot{Z} A' + \dot{V} = \dot{X} + h$. This implies that $\dot{X} + h \in  \mathcal{S}_{\dot{X}|V}$. Thus  $\mathcal{B}(\dot{X},\gamma) \subset \mathcal{S}_{\dot{X}|V}$. 	
\end{proof}

\begin{proof}[Proof of Result \ref{rslt:relax_SE}:]
	Since $x_{it}$ is scalar,
	$M_i  \,  = \, I - \dot{X}_i \dot{X}_i'/(\dot{X}_i' \dot{X}_i)$. Under Assumption \ref{assn:relax_SE}, for $t \geq 2$,
	\begin{align*}
		x_{it} =  \rho^{t-1} x_{i1} + \sum_{s=2}^t \rho^{t-s} \eta_{is}  \Rightarrow \ \  x_{i \, t+1} - x_{it}  = \rho^{t-1} (\rho - 1) x_{i1} + (\rho - 1) \sum_{s=2}^t \rho^{t-s}  \eta_{is} + \eta_{i \, t+1},
	\end{align*}
	therefore defining the two vectors $C_1 = [\rho - 1] (1, \rho,\, .. \, , \rho^{T-2})' \in \R^{T-1}$ and $C_2(V) = (\eta_{i2}, \, [\rho - 1] \eta_{i2} + \eta_{i3}, \, ..\, , \, \allowbreak [\rho-1] \sum_{s=2}^{T-1} \rho^{T-s-1}  \eta_{is} + \eta_{i \, T})' \in \R^{T-1},$ we can write
	\begin{equation}\label{eq:dotXseqEx}
		\dot{X}_i = x_{i1} C_1 + C_2(V_i).
	\end{equation}
	Let us fix  a given value $\bar{V} \in \mathcal{S}_V$. The draws of $\dot{X}$ from $\mathbb{P}_{\dot{X}|V =\bar{V}}$, as can be seen in (\ref{eq:dotXseqEx}), differ only in the value of $x_{1}$: these draws are the sum of two vectors, $C_2(\bar{V})$ which is fixed since the draws are conditional on $V = \bar{V}$, and $x_{1} C_1$  proportional to the constant vector $C_1$. Note that $\mathcal{S}_{x_1 |V = \bar{V}} = \mathcal{S}_{x_1}$ since $x_1 \indep V$.
	We use Result \ref{rslt:equiv_span} to show that $\mathcal{M}(\bar{V})$ is nonsingular almost surely.
	
	First, if $x_{i1}$ has discrete distribution with at least two support points $\underline{x_1} \neq \overbar{x_1}$, a subset $\tilde{\mathcal{S}}$  such that $\mathbb{P}_{X|V}\left(\tilde{\mathcal{S}}\right)=1$ satisfies $\tilde{\mathcal{S}}=\mathcal{S}_{x_1 |V = \bar{V}} = \mathcal{S}_{x_1}$. Since $x$ is scalar,  $\mathcal{M}(\bar{V})$ is nonsingular as long as $\underline{\dot{X}} = \underline{x_1}C_1 + C_2(\bar{V})$ and $\overbar{\dot{X}} = \overbar{x_1} C_1 + C_2(\bar{V})$ are not collinear. For these vectors to be collinear, since $C_1 \neq 0$, then either $C_1$ and $C_2(\bar{V})$ are proportional, or $C_2(\bar{V}) = 0$. Note that $C_2(\bar{V}) = 0$ implies $\bar{V}=0$, and one can show that if $C_1$ and $C_2(\bar{V})$ are proportional, this implies that $\bar{V} \in  \left\lbrace  b . \left( \begin{smallmatrix}1 \\ \vdots \\ 1
	\end{smallmatrix} \right)| b \in \R \right\rbrace =: \mathcal{D}$. Hence,  $\mathcal{M}(\bar{V})$ singular  implies $\bar{V} \in \mathcal{D}$. $\mathcal{D}$ is a subset of $\R^{T-1}$ with $\mathbb{P}_V$ measure $0$ since $V_i$ is continuously distributed on $\R^{T-1}$ by assumption. 
	
	Second, in the case where  $x_{i1}$ is continuously distributed and $\Int(\mathcal{S}_{x_1}) \neq \emptyset$, for any subset $\tilde{\mathcal{S}}$  such that $\mathbb{P}_{X|V}\left(\tilde{\mathcal{S}}\right)=1$ one can find  $(\underline{x_1},\overline{x_1}) \in \tilde{\mathcal{S}}$ such that $\underline{x_1} \neq \overbar{x_1}$. Then $\underline{\dot{X}} = \underline{x_1}C_1 + C_2(\bar{V})$ and $\overbar{\dot{X}} = \overbar{x_1} C_1 + C_2(\bar{V})$ are not collinear as long as $\bar{V} \notin \mathcal{D}$. Thus  $\ \bigcap_{X \in \tilde{\mathcal{S}}} \Span(\dot{X}) = \{0\} $  as long as $\bar{V} \notin \mathcal{D}$: similarly,  $\mathcal{M}(\bar{V})$ singular  implies $\bar{V} \in \mathcal{D}$.
	
	In either case, $\mathcal{M}(\bar{V})$  not invertible implies $\bar{V} \in \mathcal{D}$. The function $g$ is then identified over $\mathcal{S}_V \setminus \mathcal{D}$. However, $\mathbb{P}_V(\mathcal{S}_V \setminus \mathcal{D})=1$ and the second identification step formalized in Section \ref{sec:maindid}  identifies $\E(\mu_i)$ and $\E(\alpha_i)$.
\end{proof}

\begin{proof}[Proof of Result \ref{rslt:IA_illus}:]
	Recall $v_{it} = \ln \omega_{it} - \E(\ln \omega_{it}|  \ln \omega_{i0})$, $V_{i} = (v_{it})_{1 \leq t \leq T}$, $W_i = (\ln \omega_{i1},..,\ln \omega_{iT})'$. We also have as before	$\mathcal{M}(V) = \E(M_i|V_i=V)$ where  $M_i = I_{T-1} - \dot W_i(\dot W_i'\dot W_i)^{-1}\dot W_i'$.
	Note that $v_{it} = \sum_{k=0}^{t-1} \gamma_1^k [\eta_{it-k} + \gamma_0]$ and $\ln \omega_{it} = \gamma_1^t \ln \omega_{i0} + v_{it}$.
	This implies that, as in Result \ref{rslt:relax_SE}, we can write 
	$$
	\dot W_i = \ln \omega_{i0} \, E_1 + E_2(V_i),
	$$
	where $E_1 = (\gamma_1 - 1)(\gamma_1,..,\gamma_1^{T-1})$ and $E_2(V_i) = \dot V_i$.	
	The remainder of the proof is similar to that of  Result \ref{rslt:relax_SE}.
	Let us fix  a given value $\bar{V} \in \mathcal{S}_V$.
	
	Let us first consider the case where  $\ln \omega_{i0}$ has two points on its support, $\underline{w} \neq \overbar{w}$. Then  $\mathcal{M}(\bar{V})$ is nonsingular as long as $\underline{\dot{W}} = \underline{w}E_1 + E_2(\bar{V})$ and $\overbar{\dot{W}} = \overbar{w} E_1 + E_2(\bar{V})$ are not collinear. 
	For these vectors to be collinear, since $E_1 \neq 0$, then either $E_1$ and $E_2(\bar{V})$ are proportional, or $E_2(\bar{V}) = 0$. Note that $E_2(\bar{V}) = 0$ implies $\bar{V}=c \, 1_T$ for some $c \in R$ and that if $C_1$ and $C_2(\bar{V})$ are proportional,  then $\bar{V} = (\gamma_1,..,\gamma_1^{T})' + c \, 1_T$ for some $c \in \R$.
	Hence,  $\mathcal{M}(\bar{V})$ singular  implies $\bar{V} \in \{ c \, 1_T \, : c \in \R \} \cup \{ (\gamma_1,..,\gamma_1^{T})' + c \, 1_T \} := \mathcal{E}$. 
	
	We now look at the case where  $\ln \omega_{i0}$ is continuously distributed and $\Int(\mathcal{S}_{\ln \omega_{i0}}) \neq \emptyset$, for any subset $\tilde{\mathcal{S}}$  such that $\mathbb{P}_{W_i|V_i}\left(\tilde{\mathcal{S}}\right)=1$ one can find  $(\underline{w},\overline{w}) \in \tilde{\mathcal{S}}$ such that $\underline{w} \neq \overbar{w}$. Then $\underline{\dot{W}} = \underline{w}E_1 + E_2(\bar{V})$ and $\overbar{\dot{W}} = \overbar{w} E_1 + E_2(\bar{V})$ are not collinear as long as $\bar{V} \notin \mathcal{E}$. Thus  $\ \bigcap_{W \in \tilde{\mathcal{S}}} \Span(\dot{W}) = \{0\} $  as long as $\bar{V} \notin \mathcal{E}$: similarly,  $\mathcal{M}(\bar{V})$ singular  implies $\bar{V} \in \mathcal{E}$.
	
	In either case, $\mathcal{M}(\bar{V})$  not invertible implies $\bar{V} \in \mathcal{E}$, which is a subset of $\R^{T}$ with $\mathbb{P}_V$ measure $0$ since $V_i$ is continuously distributed on $\R^{T}$ by assumption. Thus $\mathcal{M}(V_i)$ is nonsingular $\mathbb{P}_{V_i}$a.s.
\end{proof}

\section{Complements to Section \ref{sec:IA}}\label{sec:app_struc_break}

We now cover an example with ``structural breaks'', i.e., time-variations in the specification of $m_t$ but the instrument does not vary over time. We assume  $d_x=1$ and that the instrument at the first time period is binary. To simplify exposition, we assume that $z_{1i} \indep V_i$.
Consider the case with one structural break, 
\begin{equation}\label{eq:struct_break}
	\text{(\ref{eq:1ststagenp}) holds,  }\ \forall Z \in \mathcal{S}_Z, \, z_{t} = z_{1},\ \ \forall \, t\leq t_0, \, m_t = m_{t_0},\ \ \forall \, t > t_0, m_t = m_{t_0 + 1}.
\end{equation}
and the case with two structural breaks,
\begin{equation}\label{eq:struct_break2}
	\text{(\ref{eq:1ststagenp}) holds,  }\ \forall Z \in \mathcal{S}_Z,\, z_{t} = z_{1},\ \, \forall \, t\leq t_0, \, m_t = m_{t_0}, \ \, \forall \, t_0<t\leq t_1, \, m_t = m_{t_1},\ \ \forall \, t > t_1, m_t = m_{t_1 + 1}.
\end{equation}
\begin{result}\label{rslt:structbreak}
	Fix $V \in \mathcal{S}_V$ and assume that $\dot{X}$ is of full column rank $\mathbb{P}_{\dot{X}|V} \, \text{a.s}$. 
	\begin{itemize}
		\item[\textit{(i)}]  If (\ref{eq:struct_break}) holds, then $\mathcal{M}(V)$ is singular.
		\item[\textit{(ii)}] If (\ref{eq:struct_break2}) holds and $  \mathcal{S}_{z_1|V} = \{\underline{z}, \overline{z}\}$ with $\underline{z} \neq \overline{z}$, then  $\mathcal{M}(V)$ is nonsingular if and only if $(m_{t_1}(\underline{z},V) - m_{t_0}(\underline{z},V)\, , \, m_{t_1 + 1}(\underline{z},V) - m_{t_1}(\underline{z},V))$ and $(m_{t_1}(\overline{z},V) - m_{t_0}(\overline{z},V)\, , \, m_{t_1 + 1}(\overline{z},V) - m_{t_1}(\overline{z},V))$ are noncollinear.
	\end{itemize}
\end{result}

\begin{proof}[Proof of Result \ref{rslt:structbreak}]
	The proof follows arguments similar to the proof of Result \ref{rslt:counterex_1cstt}.
\end{proof}

\section{Example: Production Function}\label{sec:app_struct_model}

This section provides a simplifying structural model where the mains assumptions of the paper hold. It focuses on production functions.

Many factors can generate heterogeneity in economic agents' productivity such as education or ability for workers, or technology adoption for firms. A literature has studied models where this heterogeneity, or type, is not known by the agent. The agent has a prior on its distribution and learns as she observes output realizations over time. See, e.g., \cite{j79,m84} for analyses of heterogeneity in worker productivity, and \cite{pe98} for an analysis of heterogeneity in firms profits. We consider such a framework, where the agent does not know her type and we look at the extreme case where there is no learning by the agent. The output is given by the following outcome equation
\begin{equation}\label{eq:modelProdFct}
	y_{it} \, = \, x_{it}' \mu_i + \epsilon_{it}.
\end{equation}
where the output  $y_{it}$ is realized after the choice of inputs $x_{it}$ is made by the agent. 
The random coefficients $\mu_i$ capture the heterogeneity in the returns to inputs. 	
%Write $x_{is}^t = (x_{is},..,x_{it})$ for $t \geq s$.
The econometrician observes $(x_{it}, z_{it}, y_{it})_{t=1}^T$.
Outcome equations of the form (\ref{eq:modelProdFct}) are commonly used to describe production functions where returns to inputs are heterogeneous across agents but time invariant. Such production function can describe a firm or farm production function, with inputs being log of capital, labor and/or land. See for example \cite{mg88,gu22,s11}. It  can  also be used to model a child's education outcome, where inputs are any type of parental investments.

When $\mu_i$ has a degenerate distribution, the identification of production function parameters is challenged by the correlation between inputs $x_{it}$ and the residual $\epsilon_{it}$ as it includes a productivity shock observed by the agent prior to her choice of input. This problem is well-known and the subject of an extensive literature. It has been documented that for firms, where inputs are for instance capital and labor, this endogeneity cannot be captured by an additive fixed effect in models with constant returns, see a review in \cite{acf15}. 
We consider a simplified setup with the following assumptions:
\begin{enumerate}[label=(\roman*)]%[leftmargin=15pt]
	
	\itemsep-0.2em 
	\item The input $x_{it}$ is scalar.
	
	\item It is nondynamic: the realization $x_{it}$ does not impact the distribution of future outcomes. See \cite{acf15}. This excludes capital, which accumulates.
	
	\item Input cost is $C_t(x_{it},z_{it})$ where  $z_{it}$  is a cost shifter. The agent chooses $x_{it}$ to maximize the expected discounted value of future profits.

	\item $\epsilon_{it} = \omega_{it} + \eta_{it}$, where $\omega_{it}$ is a productivity term  observed by the agent prior to her choosing $x_{it}$ and  $\eta_{it}$ is a residual capturing either measurement error or unpredicted shock to productivity. It is not a state variable. Both are not observed by the econometrician.
	
	\item $\omega_{it}$ and $z_{it}$ follow Markov processes. We impose $\omega_{it} = h_t^\omega(\omega_{it-1},\xi_{it}^\omega)$ and $z_{it} = h_t^z(z_{it-1},\xi_{it}^z)$, where $(\xi_{it}^\omega, \xi_{it}^z) \indep  \{(\xi_{is}^\omega, \xi_{is}^z)_{s=2}^{t-1},\omega_{i1}, z_{i1}\} $. We also assume $\omega_{i1} \indep z_{i1}$ and for all $t \leq T$, $ \xi_{it}^\omega \indep \xi_{it}^z$. Thus, for all $t \leq T$, $\omega_{it} \indep z_{it}$.   We also impose for all $t \leq T$, $\eta_{it} \indep (z_{is},\omega_{is})_{s=1}^{T} $.
	
	%the markov assumption is not so much for the CFA but to ensure that the info set is only info at t
	
	\item We assume that for all $t \leq T$, $\omega_{it}$ is continuously distributed and its CDF is strictly increasing.
	
	\item The agent does not know  $\mu_i$. She has knowledge about its distribution but there is no learning. 
	
\end{enumerate}
By the independence assumptions above, the information set at  $t$ does not include past values of $y_{is}$ and need only be $(\omega_{it},z_{it})$. Write for simplicity $y_{it} = y(x_{it}, \mu_i,\epsilon_{it})$. We set the price to $1$.
At period $t$, the expected discounted value of future profits up to period $T$ %\footnote{There might be constraints, or other options such as exiting.} 
is 
\begin{align*}
	V_t(\omega_{it},z_{it}) & = \max_{x_t} \E\left(\sum_{r=t}^{T} y(x_{r}, \mu,\epsilon_{r}) - C(x_r,z_r) \, | \omega_{it},z_{it}\right).
\end{align*}
This implies that there exists a function $G_t$ such that $x_{it} = G_t(\omega_{it},z_{it})$. We assume that $G_t$ is strictly monotonic in $\omega_t$. %see discussion in IN09, I think it applies here.
For agent $i$, define $v_{it} = F_{x_t | z_{t}}(x_{it}|z_{it})$  where $F$ indicates the conditional cdf.  Since  $\omega_{it} \indep z_{it}$, it is known that  $v_{it}= H_t(\omega_{it})$, see \cite{in09}, where $H_t = F_{\omega_t}$ is the pdf of $\omega_{it}$ and is strictly increasing.
Define $V_i=(v_{i1},..,v_{iT})'$ and $X_i=(x_{i1},..,x_{iT})'$, then
\begin{align*}
	\E (\epsilon_{it} |V_i, X_i) &= \E\left( \omega_{it} + \eta_{it} |V_i, X_i\right)= \E\left( \omega_{it} + \eta_{it} | \omega_{i1}, ..,\omega_{iT}, X_i\right) \\
	&=  \omega_{it} + \E\left( \E\left( \eta_{it} | \omega_{i1}, ..,\omega_{iT}, z_{i1}, ..,z_{iT} \right) | \omega_{i1}, ..,\omega_{iT}, X_i\right) = \omega_{it} =H_t^{-1}(v_{it}) =: f_t(V_i).
\end{align*}	
This proves that the control function approach assumption, Assumption 3.1, holds. In this example, there is a feedback effect as $\epsilon_{it}$ is correlated with future values of the input $x_{is}$ for $s \geq t$ through $\omega_{is}$. However the control function at time $t$ has only $v_{it}$ as argument: by the independence conditions on the future shocks to productivity $\left(\xi_{is'}^\omega\right)_{t\leq s'\leq s}$, this dependence is entirely captured by $\omega_{it}$ which has a one-to-one relation with $v_{it}$.

We now look at Assumption 3.3. Write $X_i = G(Z_i,V_i)$ and  $\dot{X}_i = \dot{G}(Z_i,V_i)$. Using the discussion in Section \ref{sec:minsupport}, since $Z_i \indep V_i$ the invertibility assumption, Assumption 3.3 will hold if for instance 
\begin{itemize}
	\item[(i)] $Z_i$ is discrete and has at least two points on its support, $\bar{Z}$ and $\underline{Z}$,
	\item[(ii)] $\dot{G}(\bar{Z},V_i)$ and $\dot{G}(\underline{Z},V_i)$ are not collinear $\mathbb{P}_{V_i}$ almost surely.
\end{itemize}	
Under these assumptions, if $T \geq 3$, the identification results obtained in the paper apply and the average returns to input are identified. 

\medskip

To be able to use the framework developed in the paper, a crucial assumption is that  $\mu_i$ is not in the information set of the agent. This implies that the source of unobserved heterogeneity  across agents  in the choice of $x_{it}$  is solely the scalar random variable $\omega_{it}$. It allows us to recover $v_{it}$, a transformation of $\omega_{it}$, which is a valid control variable. %An additional limitation is that the input is not dynamic, i.e does not depend on its past values. This also applies to the output.
%Note that a direct application of the identification results in \cite{in09} would need the instruments to be independent of the random coefficients $\mu$. In contrast, we do not impose any conditions on the joint distribution of $(\mu_i, x_i, z_i)$. In a model of heterogeneous production function, finding instruments independent of the unobserved heterogeneity $\mu_i$ might be challenging. In this case and if the panel has enough time periods, our two-step approach would guarantee identification.
%deleting this comment on Gollin Udry 2017 - read the paper, important application.
%\iffalse
%For example, \cite{gu22} study misallocation in the agricultural sector in different African countries, using a different dataset for each country. They estimate agricultural production functions under different specifications. They can allow for heterogeneity in countries for which an instrument (e.g. cash grant amount) is available in the dataset, but do not have such an instrument for all countries.
%\fi
The timing and independence assumptions are also crucial. If instead $x_{it-1}$ is allowed to impact the distribution of variables at time $t$, then  $x_{it}$ would be a function of $x_{it-1}$, $\omega_{it}$ and $z_{it}$. To apply \cite{in09} and recover a one-to-one function of $\omega_{it}$, the condition $(x_{it-1},z_{it}) \indep \omega_{it}$ should replace $\omega_{it} \indep z_{it}$. Note that independence between $x_{it-1}$ and $\omega_{it}$ potentially does not allow for $\omega_{it}$ to be serially correlated, thus depending on the setting, either framework could be preferred.

%If now the instrument impacts two time periods (i.e. two nonzero derivatives) and is continuously distributed, IA holds. \textcolor{blue}{Should I put it ? Discrete case: a bit more complicated} 

\section{Definition of the trimming functions}\label{ref:trimmingapp}

The first trimming function we use is such that $\tau(V) \in \mathcal{S}_V$ and the $(d_2 (t-1) + d)^{\text{th}}$ component of $\tau(V)$ satisfies
\begin{equation}\label{eq:def_tau_consist}	
	\tau(V)_{(t-1)d_2 + d }  \, = \, \begin{cases} 
		v_{t,d}, & \text{ if } v_{t,d} \in [\underline{v}_{td};\bar{v}_{td}], \\
		\underline{v}_{td}, & \text{ if } v_{t,d} \leq \underline{v}_{td}, \\
		\bar{v}_{td}, & \text{ if } v_{t,d} \geq \bar{v}_{td},
	\end{cases}
\end{equation}
where $v_{t,d}$ is the d$^{th}$ component of $v_t$.

Define $\varsigma>0$, and $\tau_\varsigma : x \in \R \mapsto \varsigma (e^{-x^2/(2\varsigma^2) \, + \, x/\varsigma} - 1)$. Note that $\lim_{x \to - \infty} \tau_\varsigma(x) = - \varsigma$, $\lim_{x \to + \infty} \tau_\varsigma(x) = -  \varsigma$ and we also have $\tau_\varsigma(0) = 0$, $\tau_\varsigma'(0)=1$ and $\tau_\varsigma''(0)=0$.
For $V \in \R^{Td_2}$, the $(d_2 (t-1) + d)^{\text{th}}$ component of $\tau(V)$ is given by
\begin{equation}\label{eq:def_tau_asnorm}
	\tau(V)_{(t-1)k_2 + d }  \, = \, \begin{cases} 
		v_{td}, & \text{ if } v_{td} \in [\underline{v}_{td};\bar{v}_{td}], \\
		\underline{v}_{td} + \tau_\varsigma(v_{td} -  \underline{v}_{td}), &  \text{ if } v_{td} \leq \underline{v}_{td}, \\
		\bar{v}_{td} - \tau_\varsigma(\bar{v}_{td} - v_{td}), & \text{ if }  v_{td} \geq \bar{v}_{td}, 
	\end{cases}
\end{equation}
and define as before $\hat{V}_i = \tau(\tilde{V}_i)$. The support of $\tau$ is $\R^{Td_2}$ and we now have $\hat{V}_i \in \mathcal{S}_V^\varsigma = \bigtimes_{d \leq d_2, \, t \leq T} \, [\underline{v}_{td} - \varsigma ;\bar{v}_{td} + \varsigma]$. We will refer to $ \mathcal{S}_V^\varsigma $ as the ``extended support''.

Each component of $\tau$ is a twice differentiable function of $V$, implying that $\tau$ itself is twice continuously differentiable. Moreover for all $V \in \mathcal{S}_V$, $\partial \tau / \partial V = I_{Tk_2}$ which will imply that the derivative of a function $m$ composed with $\tau$ evaluated at $V$, $m(\tau(V))$, is equal to the derivative of $m(V)$ whenever $V \in \mathcal{S}_V$. On the extended support, that is for all $V \in \mathcal{S}_V^\varsigma,  |\partial \tau / \partial V| \leq C$ and $ |\partial^2 \tau / \partial V^2| \leq C$ for some constant $C$.

\clearpage

\section{Proof of Results in Section \ref{sec:consistency}}\label{sec:appConsistency}

We remind here of some notations and introduce new ones. We denote by $||.||_F$ the Frobenius norm (the canonical norm) in the space of matrices $\mathcal{M}_p(\R)$, and $||.||_2$ the matrix norm induced by $||.||$ on $\R^p$ (the spectral norm). 
We recall that for a given matrix $A \in \mathcal{M}_p(\R)$, $||A||_F \, = \, \left( \sum_{i,j \leq p} a_{ij}^2 \right)^{1/2} \, = \, \tr(A'A)^{1/2}$. To avoid tedious notations, we will regularly omit the subscript $F$, $||A||$ without index implies that the norm considered is the Frobenius norm. The index will be displayed when clarity requires it.
We define $\lambda_{\min}(A)$ to be the smallest eigenvalue of the matrix $A$ (when it has one), similarly $\lambda_{\max}(A)$, as well as $\lambda_1(A) \leq ... \leq \lambda_p(A)$ all the eigenvalues ranked by increasing order (when they exist).
We will use the following results. First, for all $A \in \mathcal{M}_p(\R), ||A||_2 \leq ||A||_F$. This inequality also holds for nonsquare matrices. Also, for $A$ a symmetric matrix, $||A||_2 \, = \, |\lambda_{\max}(A)|$ and $||A||_F^2 \, = \, \sum_{i = 1}^p \lambda_i(A)^2$.  By definition of $||.||_2$, $||A \, a|| \leq ||A||_2 \,  ||a||$.
In what follows, T will denote the triangular inequality, M the Markov inequality, CS indicates the use of the Cauchy Schwarz inequality, LLN the weak law of large numbers, $C$ a generic constant (whose value can change from one line to another), and we follow \cite{in09} in denoting with CM (for Conditional Markov) the result that if $\E\left(|a_n|\, \big| b_n \right) = \OP(r_n)$ then $|a_n| = \OP(r_n)$.

The proof of Result \ref{rslt:consistency} is split in 3 steps described each in Sections \ref{sec:cvMSE}, \ref{sec:proof2stepNP} and \ref{sec:proofconsist}.

\subsection{Sample mean square error for estimator of the control variables}\label{sec:cvMSE}

We use results from \cite{npv99}. We first obtain a bound on the convergence rate of the sample MSE of the control variables when their density is assumed bounded away from $0$.

\begin{assumption}\label{assn:NPV1stStep}
	There exists $\gamma_1 > 0$ and $a_1(L)$ such that $\sqrt{L/n} \ a_1(L) \xrightarrow[n \to \infty]{} 0$ and for all $t \leq T$,
	\begin{enumerate}[noitemsep,nolistsep]
		\item  $(x_{it}^{en},\xi_{it})$ is i.i.d over $i$, continuously distributed and $\Var(x_t^{en}|\xi_t)$ is bounded,
		
		\item \label{assn:NPV1stStep_matrix} There exists a $L \times L$ nonsingular matrix $\Gamma_{1 t}$ such that for $R^L(\xi_t) = \Gamma_{1 t} r^L(\xi_t)$, $\E(R^L(\xi_t) R^L(\xi_t)')$ has smallest eigenvalue bounded away from zero uniformly in $L$,
		
		\item \label{assn:NPV1stStep_approx} There exists $\beta_t^L$ such that $\sup_{\mathcal{S}_{\xi t}} || b_t(\xi_t) - \beta_t^{L \prime} \,  r^L(\xi_t)|| \leq C L^{- \gamma_1}$,
		
		\item \label{assn:NPV1stStep_bds} $\sup_{\mathcal{S}_{\xi t}} || R^L(\xi_t) || \leq a_1(L)$.
	\end{enumerate}
\end{assumption}

\begin{result}\label{rslt:NPV1stStep}
	Under Assumption  \ref{assn:NPV1stStep}, 
	\begin{align}
		& \frac{1}{n} \, \sum_{i=1}^n || \, V_{i} - \hat{V}_{i} \, ||^2  = O_{\mathbb{P}} \left( L / n + L^{- 2 \gamma_1}  \right) = O_{\mathbb{P}} (\Delta_n^2), \label{eq:V_sampleMSErate} \\
		& \max_i || \, V_{i} - \hat{V}_{i} \, || \, = \,  O_{\mathbb{P}}(a_1(L) \Delta_n). \label{eq:V_suprate}
	\end{align}
\end{result}

\begin{proof}[Proof of Result \ref{rslt:NPV1stStep}]
	Under Assumption \ref{assn:NPV1stStep}, using Theorem 1 of \cite{n97} and Lemma A1 of \cite{npv99}  (see e.g Equations A.3 and A.5), we have $\forall \, t \leq T$,
	%\footnote{Should I type the proof ?} 
	$
	\frac{1}{n} \, \sum_{i=1}^n || \, v_{it} - \tilde{v}_{it} \, ||^2 = O_{\mathbb{P}} (\Delta_n^2)$, 
	as well as $
	\max_i || \, v_{it} - \tilde{v}_{it} \, || \, = \,  O_{\mathbb{P}}(a_1(L) \Delta_n)$, where $\Delta_n = \sqrt{L/n} + L^{- \gamma_1}$, and 
	\begin{equation}\label{eq:bt_suprate}
		\sup_{\mathcal{S}_{\xi t}} ||b_t - \hat{b}_t || = O_{\mathbb{P}}(a_1(L) \Delta_n)
	\end{equation}
	
	Define $\tilde{V}_i = (\tilde{v}_{i1}, ... , \tilde{v}_{iT})$. Since $|| \hat{V}_i - V_i || \leq ||\tilde{V}_i - V_i||$, the result applies.
\end{proof}

If for instance $b_t$ is continuously differentiable up to order $p$, writing $d_{\xi} = d_1 + d_z$, then Assumption \ref{assn:NPV1stStep} (\ref{assn:NPV1stStep_approx}) holds with $\gamma_1 \, = \, p / d_{\xi}$ for different choices of sieve basis.
Conditions satisfying Assumption \ref{assn:NPV1stStep} (\ref{assn:NPV1stStep_matrix}) typically require the support of $\xi_t$ to be bounded and the density of $\xi_t$ to be bounded away from $0$ on its support. 
This restriction is not desirable. Indeed, in applications where the density of the regressors goes to $0$ at the boundaries, regressors will be trimmed to consider only a subset of $\mathcal{S}_{\xi}$ where the density is bounded away from $0$. However, we are interested here in the average effect $\E(\mu)$ and counterfactuals involving population means. Trimming arbitrarily on regressors to estimate a conditional effect is contrary to this goal.
We therefore give  in Assumption \ref{assn:NPV1stStep0Dsty}
%in Appendix \ref{app:density_0} 
a set of conditions guaranteeing that (\ref{eq:V_sampleMSErate}) and (\ref{eq:V_suprate}) hold and allows for the density of the regressor to go to $0$ at the boundary of its support when the support is bounded.
%We obtain these results following \cite{in09} which develops an argument of \cite{a91} in assuming a polynomial lower bound on the rate of decrease of the density.
%\footnote{Should I type the proofs - some modifications wrt IN09 ?}.
To obtain these results, we follow \cite{in09} which develops an argument of \cite{a91} in assuming a polynomial lower bound on the rate of decrease of the density.

\begin{result}\label{rslt:NPV1stStep0Dsty}	
	Under Assumption \ref{assn:NPV1stStep0Dsty}, (\ref{eq:V_sampleMSErate}) and (\ref{eq:V_suprate}) hold.
\end{result}

\begin{proof}[Proof of Result \ref{rslt:NPV1stStep0Dsty}]
	Define $a_1(L) = L^{\alpha_1 + 1}$.
	Under these conditions, Lemma S.3 of \cite{in09} can be modified using \cite{a91} (See Equations 3.14 or A.40) to account for the fact that $\xi_t$ is not scalar. One obtains that since $r^L$ is a the power series basis of functions, there exists a nonsingular $L \times L$ matrix $\tilde{\Gamma}_{\xi t}$ such that for 
	$\tilde{r}^L(\xi_t) = \tilde{\Gamma}_{\xi t} r^L(\xi_t)$ then $\E(\tilde{r}^L(\xi_t) \tilde{r}^L(\xi_t)') \, = \, I_{L}$, implying that Assumption \ref{assn:NPV1stStep} (\ref{assn:NPV1stStep_matrix}) holds. One also obtain $\sup_{\mathcal{S}_{\xi}} || \tilde{r}^L(\xi) || \leq a_1(L)$ with $a_1(L) \, = \, C L^{\alpha + 1} $ as is required in Assumption \ref{assn:NPV1stStep} (\ref{assn:NPV1stStep_bds}). Thus, Assumption \ref{assn:NPV1stStep} is satisfied and Result \ref{rslt:NPV1stStep} applies.	
\end{proof}

We note that allowing for unbounded support is a desirable extension as well and is made possible using a method similar to \cite{cht05} and \cite{cht08}, but is outside the scope of this paper.

\subsection{Convergence rates of two-step series estimators}\label{sec:proof2stepNP}

We focus on the generic estimator of $h^W(V) =  \E(w_i | V_i = V)$ as defined in Section \ref{sec:def_estim}.

\begin{assumption}\label{assn:IN2stStep}
	\
	\begin{enumerate}[noitemsep,nolistsep]
		\item $e_i^{W*}$ is i.i.d and $\E((e_i^{W*})^2 \, | \, X,Z)$ is bounded on $\mathcal{S}_{X,Z}$,
		
		\item $h^W$ is Lipschitz on $\mathcal{S}_V$ and $\rho^W$ is bounded on $\mathcal{S}_{X,Z}$,
		
		\item \label{assn:IN2stStep_matrix} There exists a $K \times K$ nonsingular matrix $\Gamma_2$, such that for $P^K(V) = \Gamma_2 p^K(V)$, $\E(P^K(V) P^K(V)')$ has smallest eigenvalue bounded away from zero uniformly in $K$,
		
		\item \label{assn:IN2stStep_approx} There exists $\gamma_2$ and $\pi^{K}$ such that $\sup_{\mathcal{S}_{V}} | h^W(V) - p^K(V)' \, \pi_W^{K} | \leq C K^{- \gamma_2}$,
		
		\item \label{assn:IN2stStep_limitzero} For $\, \sup_{\mathcal{S}_{V}} || P^K(V) || \leq b_1(K)$ and $ \, \sup_{\mathcal{S}_{V}} || \partial P^K(V)/\partial V || \leq b_2(K)$, $\sqrt{K} \ b_2(K) \Delta_n \xrightarrow[n \to \infty]{} 0$ and 
		$\sqrt{K/n} \, b_1(K) \allowbreak \xrightarrow[n \to \infty]{} 0$. 
	\end{enumerate}
\end{assumption}

\begin{result}\label{rslt:IN2stStep} Under Assumption \ref{assn:NPV1stStep} and \ref{assn:IN2stStep},
	\begin{align}
		\int \left| \hat{h}^W(V) - h^{W}(V) \right|^2 dF(V) & \, = \, O_{\mathbb{P}} (K/n + K^{-2 \gamma_2} + \Delta_n^2 b_2(K)^2), \label{eq:mnM_MSErate} \\
		\sup_{V \in \mathcal{S}_{V}} | \hat{h}^W(V) - h^{W}(V) |  &  \, = \,  O_{\mathbb{P}} \left( b_1(K) (K/n + K^{-2 \gamma_2} +  \Delta_n^2 b_2(K)^2)^{1/2} \right). \label{eq:mnM_suprate}
	\end{align}
\end{result}
The additional term in the mean-squared error convergence rate compared to, e.g, \cite{npv99} or \cite{hlr18} (Section 3 of their Online Appendix) is $\Delta_n^2 b_2(K)^2$. It comes from the correlation between $e^W$ and $\hat{P} - P$.

\begin{proof}[Proof of Result \ref{rslt:IN2stStep}]
	Instead of applying the general results of Section 5.2 of the Appendix of \cite{hlr18}, we directly extend the proof of Theorem 12 of \cite{in09} (IN09 thereafter) because it uses lower level conditions similar to the ones we seek to impose.
	We adapt some of their claims to our model where $\E(e_i^W|X_i^{ex},X_i^{en},Z_i) \neq 0$.
	
	Define  $P = (p_1,\,.. \, , p_n)$, $Q = P P'/n$,  $\hat{Q} = \hat{P} \hat{P}'/n$, $\rho_i^W = \rho^W(X_i^{ex}, X_i^{en}, Z_i)$, as well as the vectors $e^W = (e_1^W,\,.. \, , e_n^W)'$, $\vec{\rho}^W = (\rho_1^W,\,.. \, , \rho_n^W)'$ and $e^{W*} = (e_1^{W*}, \, .. \, , e_n^{W*})'$. Note that $e^W = \vec{\rho}^W + e^{W*}$.
	Because the series estimator is unchanged by a linear transformation of the basis of functions, we can assume that $p_i(V_i) = P^K(V_i)$. As argued in \cite{n97}, we can assume without loss of generality  that under Assumption \ref{assn:IN2stStep} $\E(p_i^K p_i^{K \prime}) = I_K$. By construction, $\hat{V}_i \in \mathcal{S}_V$, and under Assumption \ref{assn:NPV1stStep}, (\ref{eq:V_sampleMSErate}) holds. 
	Therefore, as in Lemma S.5 of \cite{in09}, we have
	\begin{align}
		||Q - I_K|| &= \OP(b_1(K)\sqrt{K/n}), \\
		||P'e^W /n || & = \OP(\sqrt{K/n}), \label{eq:ratePu} \\
		||\hat{P} - P||^2/n & = \OP(b_2(K)^2 \Delta_n^2), \\
		||\hat{Q} - Q|| & = \OP(b_2(K)^2 \Delta_n^2 + \sqrt{K}b_2(K)\Delta_n).
	\end{align}
	Hence by Assumption \ref{assn:IN2stStep} (\ref{assn:IN2stStep_limitzero}), $||\hat{Q} - I_K|| = \oP(1)$ and as in Lemma S.6 of \cite{in09}, with probability going to $1$, $\lambda_{\min}(\hat{Q}) \geq C$ and $\lambda_{\min}(Q) \geq C$.
	
	We now show how the rate of convergence is impacted by the conditional mean dependence of $e^W$ on $(X,Z)$ by deriving the rate of $|| \hat{\pi}^W - \pi_W^{K}||$, where we recall $ \hat{\pi}^W = \hat{Q}^{-1} \hat{P} W$. We define $H^W = (h^W(V_1),\,.. \, , h^W(V_n))'$, $\hat{H}^W = (h^W(\hat{V}_1),\,.. \, , h^W(\hat{V}_n))'$, $\tilde{\pi}^W = \hat{Q}^{-1} \hat{P} \hat{H}^W /n$, $\bar{\pi}^W = \hat{Q}^{-1} \hat{P} H^W /n $. We decompose
	$$|| \hat{\pi}^W - \pi_W^{K}|| \leq \underbrace{|| \hat{\pi}^W - \bar{\pi}^W||}_{(A)}  + \underbrace{|| \bar{\pi}^W - \tilde{\pi}^W||}_{(B)}  + \underbrace{|| \tilde{\pi}^W - \pi_W^{K}|| }_{(C)}. $$
	The first term can in turn be decomposed as
	$(A) = \hat{Q}^{-1} \hat{P} \, [\vec{\rho}^W/n \, + \, e^{W*}/n ].$
	Since $(X_i, Z_i, e_i^{W*})$ are i.i.d, we have $\E(e_i^{W*} | X_1,Z_1,..\, , X_n, Z_n) = 0$, $\E((e_i^{W*})^2 | X_1,Z_1,..\, , X_n, Z_n) = \E((e_i^{W*})^2 | X_i,Z_i) \leq C$, and $\E(e_i^{W*} e_j^{W*} | X_1,Z_1,..\, , X_n, Z_n) = 0$. This gives
	\begin{align*}
		\E(||\hat{Q}^{1/2} \, \hat{Q}^{-1} \hat{P} e^{W*}/n ||^2 | X_1,Z_1,..\, , X_n, Z_n) & = \tr(\hat{Q}^{-1/2} \hat{P} \, \E(e^{W*} e^{W* \, \prime} | X_1,Z_1,..\, , X_n, Z_n) \, \hat{P}' \hat{Q}^{-1/2} )/n^2, \\
		& \leq C \tr( \hat{P}' (\hat{P} \hat{P}')^{-1} \hat{P})/n \leq C K/n.
	\end{align*}
	This implies by M that $\hat{Q}^{1/2} \, \hat{Q}^{-1} \hat{P} e^{W*}/n = \OP(\sqrt{K/n})$, and by $\lambda_{\min}(\hat{Q}) \geq C \ \wpa 1$,   that 
	$\hat{Q}^{-1} \hat{P} e^{W*}/n = \OP(\sqrt{K/n}).$ 
	This rate is the same as Lemma S.7 (i) of IN09 since $e^{W*}$ is  by definition conditionally mean-independent of the regressors generating $V$. 
	As for the second term appearing in $(A)$, we write
	$\hat{Q}^{-1} \hat{P} \vec{\rho}^W/n  = \hat{Q}^{-1} P \vec{\rho}^W/n + \hat{Q}^{-1} (P - \hat{P}) \vec{\rho}^W/n.$
	Since $\E(\rho_i^W|V_i)=0$ and $(\rho_i^W,V_i)$ is i.i.d, we know that as in (\ref{eq:ratePu}), $|| P \vec{\rho}^W/n||^2 = \OP (K/n)$. Therefore, by $\lambda_{\min}(Q) \geq C \ \wpa 1$,  $||\hat{Q}^{-1} P \vec{\rho}^W/n|| = \OP(\sqrt{K/n})$.
	
	Moreover, $ (P - \hat{P}) \rho^W/n = \frac{1}{n} \sum_{i=1}^n (p_i^K - \hat{p}_i^K) \rho_i^W $ and
	\begin{align*}
		\frac{1}{n} || \sum_{i=1}^n (\hat{p}_i^K - p_i^K) \rho_i^W || & \leq \frac{1}{n} \sum_{i=1}^n || (\hat{p}_i^K - p_i^K) \rho_i^W || \leq C \left( \frac{1}{n} \sum_{i=1}^n || (\hat{p}_i^K - p_i^K)||^2 \right)^{1/2} \left( \frac{1}{n} \sum_{i=1}^n | \rho_i^W |^2 \right)^{1/2}, \\
		& \leq C b_2(K) \left(\frac{1}{n} \sum_{i=1}^n || (\hat{V}_i - V_i)||^2 \right)^{1/2} ||\rho||_{\infty} \, = \, \OP(b_2(K) \Delta_n).
	\end{align*}
	This implies by $\lambda_{\min}(\hat{Q}) \geq C \ \wpa 1$ that $\hat{Q}^{-1} (P - \hat{P}) \vec{\rho}^W/n = \OP(b_2(K) \Delta_n)$. This gives a convergence rate for $(A)$, $|| \hat{\pi}^W - \bar{\pi}^W||^2 = \OP(K/n + K/n + b_2(K)^2 \Delta_n^2) = \OP (K/n + b_2(K)^2 \Delta_n^2)$.
	
	By Lemma S.7 (ii) and (iii) of IN09, $(B) = \OP(\Delta_n)$ since $h^W(.)$ is Lipschitz, and $(C) = \OP(K^{-\gamma_2})$ using Assumption \ref{assn:IN2stStep} (\ref{assn:IN2stStep_approx}). %\footnote{Should I type these proofs ?}
	This implies that 
	$$|| \hat{\pi}^W - \pi_W^{K}||^2 = \OP(K/n + b_2(K)^2 \Delta_n^2 + \Delta_n^2 + K^{-2\gamma_2}) = \OP(K/n + K^{-2\gamma_2} + b_2(K)^2 \Delta_n^2),$$
	which differs from the rate $(K/n + K^{-2\gamma_2} + \Delta_n^2)$ obtained in IN09. The structure of the proof showed that the extra term $b_2(K) \Delta_n$ comes from the correlation between $e^W$ and $\hat{P} - P$, which is nonzero because $\hat{P}$ is constructed using the estimated $\hat{V}$ which themselves depend on the covariates $(X^{ex}, X^{en},Z)$. Deriving a rate of convergence required linearizing the term $\hat{P} - P$.
	
	We can now conclude that
	\begin{align*}
		\int \left| \hat{h}^W(V) - h^{W}(V) \right|^2 dF_V(V) &  \leq \int \left| \hat{h}^W(V) - p^K(V)' \pi_W^{K}  \right|^2 dF_V(V) + \int \left| p^K(V)'  \pi_W^{K} - h^{W}(V) \right|^2 dF_V(V)\\
		& \leq  (\hat{\pi}^W - \pi_W^{K})' \, [\int p^K(V)p^K(V)'dF_V(V)] \, (\hat{\pi}^W - \pi_W^{K}) + C K^{- 2 \gamma_2}\\
		&   = O_{\mathbb{P}} (K/n + K^{-2 \gamma_2} + \Delta_n^2 b_2(K)^2),
	\end{align*}
	where the last line holds by the normalization $\E(Q)=I_K$. Moreover,
	\begin{align*}
		\sup_{V \in \mathcal{S}_{V}} | \hat{h}^W(V) - h^{W}(V) |  &  \, \leq \, \sup_{V \in \mathcal{S}_{V}} \left| \hat{h}^W(V) - p^K(V)' \, \pi_W^{K}  \right|  + \OP(K^{- \gamma_2}) \\
		&\, \leq \,  O_{\mathbb{P}} \left( b_1(K) (K/n + K^{-2 \gamma_2} +  \Delta_n^2 b_2(K)^2)^{1/2} \right). 
	\end{align*} 
\end{proof}

As in Section \ref{sec:cvMSE}, these results can be extended to regressors with density allowed to go to $0$ at the boundary of the support.

\begin{assumption}\label{assn:IN2stStep0Dsty}
	\
	\begin{enumerate}[noitemsep,nolistsep]
		\item $e_i^{W*}$ is i.i.d and $\E((e^{W*})^2 \, | \, X,Z)$ is bounded on $\mathcal{S}_{X,Z}$, 
		
		\item $ h^W$ is Lipschitz on $\mathcal{S}_V$ and $\rho^W$ is bounded on $\mathcal{S}_{X,Z}$,
		
		\item  \label{assn:IN2stStep0Dsty_density} $p^K(.)$ is the power series basis, and $\, \forall V \in \mathcal{S}_{V}, \ f_{V}(V) \, \geq \, \Pi_{d \leq k_2, \, t \leq T} \, (v_{t,d} - \underline{v}_{td})^{\alpha_2}(\bar{v}_{td} - v_{t,d})^{\alpha_2}$,
		
		\item \label{assn:IN2stStep0Dsty_approx} There exists $\gamma_2$ and $\pi_W^{K}$ such that $\sup_{\mathcal{S}_{V}} | h^W(V) -  \, \pi_W^{K \prime} p^K(V) | \leq C K^{- \gamma_2}$,
		
		\item \label{assn:IN2stStep0Dsty_limitzero} $b_1(K) \, = \, K^{\alpha_2 + 7/2} \Delta_n \to_{n \to \infty} 0$ and $b_2(K) \, = \,   K^{\alpha_2 + 3/2} / \sqrt{n} \, \to_{n \to \infty} 0$. 
	\end{enumerate}
\end{assumption}

\begin{result}\label{rslt:IN2stStep0Dsty}
	Under Assumptions \ref{assn:NPV1stStep0Dsty} and \ref{assn:IN2stStep0Dsty}, (\ref{eq:mnM_MSErate}) and (\ref{eq:mnM_suprate}) hold.
\end{result}

\begin{proof}[Proof of Result \ref{rslt:IN2stStep0Dsty}]
	As in the proof of Result \ref{rslt:NPV1stStep0Dsty}, define $b_1(K) \, = \,  K^{\alpha_2 + 1} $ and  $b_2(K) \, = \, K^{\alpha_2 + 3}$. Then Assumption \ref{assn:IN2stStep0Dsty} (\ref{assn:IN2stStep0Dsty_limitzero}) implies Assumption \ref{assn:IN2stStep} (\ref{assn:IN2stStep_limitzero}), and together with Assumption \ref{assn:IN2stStep0Dsty} (\ref{assn:IN2stStep0Dsty_density}), implies Assumption \ref{assn:IN2stStep} (\ref{assn:IN2stStep_matrix}).
\end{proof}

And for these results to apply to our choice of $w_i = (M_i)_{s,t}$ and $w_i = (M_i \dot{y}_i)_{t}$ for $1 \leq s,t \leq T-1$, we adapt Assumption \ref{assn:IN2stStep0Dsty} to the model primitives to obtain Assumption \ref{assn:IN2stStep0Dsty_model}.
Under Assumptions \ref{assn:idCFA}, \ref{assn:NPV1stStep0Dsty} and \ref{assn:IN2stStep0Dsty_model}, the convergence rate of $\hat{\mathcal{M}}$ and $\hat{k}$ in sup norm and mean square norms are therefore given by  (\ref{eq:mnM_MSErate}) and (\ref{eq:mnM_suprate}). Recall that $\hat{g}(V) \, = \,  \hat{\mathcal{M}}(V)^{-1} \, \hat{k}(V).$
The rate of convergence of $\hat{g}(.)$ is obtained using continuity arguments. 
Indeed, Assumption \ref{assn:cv_g} implies that $k \, = \, \mathcal{M} \, g$ is continuous and that $||m||_{\infty}$, $||\mathcal{M}||_{\infty}$ and $||g||_{\infty}$ exist. 
%Note that the continuity assumption somewhat overlaps with Assumption \ref{assn:IN2stStep0Dsty_model} (\ref{assn:IN2stStep0Dsty_model_approx}) and (\ref{assn:IN2stStep0Dsty_model_lipschitz}) as the existence of a linear approximation relies on smoothness assumptions. Moreover 
Under these conditions, the following MSE and sup norm rates are obtained for $\hat{g}$.

\begin{result}\label{rslt:rate_g}
	Under Assumptions \ref{assn:NPV1stStep0Dsty}, \ref{assn:IN2stStep0Dsty_model}, \ref{assn:cv_g}, assuming $b_1(K)^2 (K/n + K^{-2 \gamma_2} +  \Delta_n^2 b_2(K)^2) \to 0$,
	\begin{align*}
		\int \left|\left| \hat{g}(V) - g(V) \right|\right|^2 dF(V)  \, & = \, O_{\mathbb{P}} (K/n + K^{-2 \gamma_2} + \Delta_n^2 b_2(K)^2), \\
		||\hat{g}(V) - g(V) ||_{\infty} \, & = \, O_{\mathbb{P}} ( b_1(K) (K/n + K^{-2 \gamma_2} +  \Delta_n^2 b_2(K)^2)^{1/2}).
	\end{align*}
\end{result}

To prove Result \ref{rslt:rate_g}, we will need the two following Lemmas.

\begin{lemma}\label{lem:lambdaMinCty}
	Under Assumption \ref{assn:cv_g}, the function $ V \in \mathcal{S}_V \mapsto \lambda_{\min}(\mathcal{M}(V))$ is continuous.
\end{lemma}

\begin{proof}[Proof of Lemma \ref{lem:lambdaMinCty}]
	For all $V$, $\mathcal{M}(V)$ is symmetric. Using the Weilandt and Hoffman inequality (see e.g Corollary 6.3.8 p408 in \cite{hj12}),
	%valid only for normal matrices}
	for two values $V$ and $V'$,
	$$\sum_{i=1}^{T-1} \left| \, \lambda_i (\mathcal{M}(V)) - \lambda_i(\mathcal{M}(V')) \, \right|^2 \leq || \mathcal{M}(V) - \mathcal{M}(V') ||_F^2, $$
	where we index the eigenvalues $(\lambda_i)_{i = 1}^{T-1}$ by increasing order.
	This implies 
	\begin{equation}\label{eq:lambdaMinBound}
		\left| \,  \lambda_{\min}(\mathcal{M}(V)) - \lambda_{\min}(\mathcal{M}(V')) \right| \leq || \mathcal{M}(V) - \mathcal{M}(V') ||_F^{1/2}.
	\end{equation}
	Since $\mathcal{M}(.)$ is a continuous function, this concludes the argument. Note that the Lipschitz inequality (\ref{eq:lambdaMinBound}) will be used in the proof for the convergence rates.
\end{proof}

\begin{lemma}\label{lem:lambdaMinBound}
	Under Assumption \ref{assn:cv_g}, there exists  $c > 0$, such that for all $V \in  \mathcal{S}_V$, $\lambda_{\min}(\mathcal{M}(V)) \geq c $.
\end{lemma}

\begin{proof}[Proof of Lemma \ref{lem:lambdaMinBound}]
	Under Assumption \ref{assn:cv_g}, $\mathcal{M}(V)$ is nonsingular for all  $V \in  \mathcal{S}_V$. This implies that for all  $V \in  \mathcal{S}_V$, $\lambda_{\min}(\mathcal{M}(V)) > 0 $. Since $\mathcal{S}_V$ is a compact set and the function $ V \mapsto \lambda_{\min}(\mathcal{M}(V))$ is continuous by Lemma \ref{lem:lambdaMinCty}, it has a minimum and reaches it. This minimum value cannot be $0$, hence $ \exists \, c > 0, \, \forall \, V \in  \mathcal{S}_V, \, \lambda_{\min}(\mathcal{M}(V)) \geq c $.
\end{proof}

\begin{proof}[Proof of Result \ref{rslt:rate_g}]
	Under Assumptions \ref{assn:NPV1stStep0Dsty} and \ref{assn:IN2stStep0Dsty}, writing $\Gamma_n^2$ and $\gamma_n$ respectively the mean square and sup norm rates of convergence, we have
	\begin{align*}
		\int \left|\left| \hat{k}(V) - k(V) \right|\right|^2 dF(V) \, & = \, O_{\mathbb{P}} (\Gamma_n^{2}), \ \sup_{V \in \mathcal{S}_{V}} ||  \hat{k}(V) - k(V) ||  \, = \,  O_{\mathbb{P}} \left( \gamma_n \right), \\
		\int \left|\left| \hat{\mathcal{M}}(V) - \mathcal{M}(V) \right|\right|_F^2 dF(V)  \, & = \, O_{\mathbb{P}} (\Gamma_n^{2}), \ \sup_{V \in \mathcal{S}_{V}}  \left|\left|  \hat{\mathcal{M}}(V) - \mathcal{M}(V)  \right|\right|_F \, = \,  O_{\mathbb{P}} \left( \gamma_n \right),
	\end{align*} 
	since the Frobenius norm and the Euclidean norm are square roots of the sum of squared elements, and this rate was obtained for each element of $\mathcal{M}(V)$ and $k(V)$.
	
	We write $\1_n = \1 \left( \min_{V \in \mathcal{S}_V}  \lambda_{\min}(\hat{\mathcal{M}}(V)) > \frac{c}{2} \right)$.
	Using (\ref{eq:lambdaMinBound}), we have
	\begin{align*}
		& \lambda_{\min}(\hat{\mathcal{M}}(V)) \,  > \, \lambda_{\min}(\mathcal{M}(V)) - \big |\big |  \, ||\hat{\mathcal{M}}(V) - \mathcal{M}(V)||_F  \big |\big |_{\infty}, \\
		& \Rightarrow \, \min_{V \in \mathcal{S}_V}  \lambda_{\min}(\hat{\mathcal{M}}(V)) \, > \, c - \big |\big |  \, ||\hat{\mathcal{M}}(V) - \mathcal{M}(V)||_F \big |\big |_{\infty} ,
	\end{align*}
	where the last implication uses Lemma \ref{lem:lambdaMinBound}.
	Hence  $\1_n \geq \1 \left( \big |\big |  \, ||\hat{\mathcal{M}}(V) - \mathcal{M}(V)||_F \big |\big |_{\infty} \leq c/2 \right)$. By Assumptions \ref{assn:NPV1stStep0Dsty} and \ref{assn:IN2stStep0Dsty},  $\gamma_n \to 0$ which implies
	$\1_n \, = \, 1 \ \wpa1$.
	
	To obtain the sup norm rate, we write
	\begin{align*}
		\forall \, V \in \mathcal{S}_V, \ g(V) - \hat{g}(V) & \, = \, \mathcal{M}(V)^{-1} \, k(V) - \hat{\mathcal{M}}(V)^{-1} \hat{k}(V) \\
		& \, = \, \mathcal{M}(V)^{-1} \,  \left[ \, \hat{\mathcal{M}}(V) -  \mathcal{M}(V) \, \right] \, \hat{\mathcal{M}}(V)^{-1} \, k(V) \, + \, \hat{\mathcal{M}}(V)^{-1} \, \left[ \, k(V) - \hat{k}(V) \, \right].
	\end{align*}
	Using T, the norm inequality and definition of induced norm, this gives
	\begin{align*}
		\1_n  ||  \hat{g}(V) - g(V) ||_{\infty} & \, \leq \, \1_n  \, \frac{1}{c}  \ \big | \big | \, || \hat{\mathcal{M}}(V) - \mathcal{M}(V)||_F \  \big |\big |_{\infty} \frac{2}{c} ||k||_{\infty} \, + \, \1_n \, \frac{2}{c} \ || k - \hat{k} ||_{\infty}.
	\end{align*}
	This implies that $||\hat{g}(V) - g(V) ||_{\infty} = O_{\mathbb{P}} (\gamma_n)$.
	To obtain the mean square error rate, we write
	\begin{align*}
		\1_n \, \int || & \hat{g}(V)  - g(V)||^2 dF(V) \\ & \, = \, \1_n \int \left|\left| \mathcal{M}(V)^{-1} \,  \left[ \, \hat{\mathcal{M}}(V) -  \mathcal{M}(V) \, \right] \, \hat{\mathcal{M}}(V)^{-1} \, k(V) \, + \, \hat{\mathcal{M}}(V)^{-1} \, \left[ \, k(V) - \hat{k}(V) \, \right] \right|\right|^2 dF(V) \\
		& \, \leq \,  \1_n  \ \frac{1}{c} \ \frac{2}{c} \  ||k||_{\infty} \, \int \left|\left| \hat{\mathcal{M}}(V) - \mathcal{M}(V) \right|\right|_F^2 dF(V) \, + \,  \1_n \, \frac{2}{c} \, \int \left|\left| \hat{k}(V) - k(V) \right|\right|^2 dF(V),
	\end{align*} 
	which implies 
	$ \int || \hat{g}(V)  - g(V)||^2 \, dF(V) \, = \, O_{\mathbb{P}} (\Gamma_n^{2})$.
\end{proof}

\subsection{Consistency of $\hat{\mu}$}\label{sec:proofconsist}

Equipped with the convergence rate results on the nonparametric estimators, we can now show consistency of $\hat{\mu}$.
Recall that, writing $\delta_i = \1(\det(\dot{X}_i' \dot{X}_i) > \delta_0)$ and  $Q_i^{\delta} = \delta_i Q_i $,
$$\hat{\mu}  \, = \, \frac{\sum_{i=1}^n Q_i^{\delta} \, [ \dot{y}_i - \, \hat{g}(\hat{V_i}) ]}{\sum_{i=1}^n \delta_i}.$$

\begin{proof}[Proof of Result \ref{rslt:consistency}]
	To prove consistency, we need to show that $\frac{1}{n} \sum_{i=1}^n  Q_i^{\delta} \, \hat{g}(\hat{V_i}) \, \to_{\mathbb{P}} \, \E (Q g(V) \delta)$. 
	Indeed then by the LLN,  we have $ \frac{1}{n} \sum_{i=1}^n \delta_i \to_{\mathbb{P}} \mathbb{P}(\det(\dot{X}_i' \dot{X}_i) > \delta)$, and by Assumption \ref{assn:finiteE} and the LLN, $\frac{1}{n} \sum_{i=1}^n  Q_i^{\delta} \,  \dot{y}_i \, \to_{\mathbb{P}} \, \E (Q \, \dot{y} \, \delta)$ also holds. Then consistency would follow from Equation (\ref{eq:Miou_delta}).
	
	To obtain  $\frac{1}{n} \sum_{i=1}^n  Q_i^{\delta} \, \hat{g}(\hat{V_i}) \, \to_{\mathbb{P}} \, \E (Q g(V) \delta)$, we decompose
	\begin{align*}
		\frac{1}{n}\sum_{i=1}^n  Q_i^{\delta}\, \hat{g}(\hat{V_i}) \, - \, \E (Q g(V) \delta)  & \, = \, \frac{1}{n} \sum_{i=1}^n Q_i^{\delta}  \, [ \hat{g}(\hat{V_i}) \, - \, g(V_i) ] \, + \, \frac{1}{n} \sum_{i=1}^n Q_i^{\delta} \, g(V_i) \, - \, \E (Q g(V) \delta) \\
		& \, := \, A_n \, + \, B_n.
	\end{align*}
	We have
	\begin{align*}
		|| A_n || \, & = \big | \big | \, \frac{1}{n} \sum_{i=1}^n  Q_i^{\delta}  \, [ \, \hat{g}(\hat{V_i}) - g(\hat{V_i}) \, ] + \frac{1}{n} \sum_{i=1}^n  Q_i^{\delta}  \, [ \, g(\hat{V_i}) - g(V_i) \, ] \, \big | \big | \\
		& \leq \, ||\hat{g}- g ||_{\infty} \, \frac{1}{n} \sum_{i=1}^n || Q_i^{\delta} || \, + \, C  \max_{i} || \hat{V}_i - V_i || \, \frac{1}{n} \sum_{i=1}^n || Q_i^{\delta} ||_2 \, \\
		& = \, O_{\mathbb{P}}( \gamma_n +  a_1(L) \Delta_n) \, \frac{1}{n} \sum_{i=1}^n || Q_i^{\delta} ||_2,
	\end{align*}
	where the first term in the inequality follows from $\hat{V}_i \in \mathcal{S}_V$ by design, and the second sum term follows from $g$ being continuously differentiable on a compact set, hence Lipschitz continuous on this set. The last equality follows from Equation (\ref{eq:V_suprate}) and Result (\ref{rslt:rate_g}). We assumed $\gamma_n \to 0$ and $a_1(L) \Delta_n \to 0$ as $n$ goes to infinity, thus we obtain $|| A_n || \, = \, o_{\mathbb{P}} (1)$.
\end{proof}

\clearpage

\section{Proofs of Results in Section \ref{sec:asymptoticnorm}}\label{app:proofAsnorm}

To comment further on the definitions in Section \ref{sec:trimming}, we note that if there exists a sequence of functions $(m_n)_{n \in \N}$ converging uniformly to $m^\varsigma$ on the extended support $\mathcal{S}_V^\varsigma$, the sequence of restrictions of $(m_n)_{n \in \N}$ on $\mathcal{S}_V$ converges uniformly to $m$.

As was the case with our previous definition of $\tau$, $|| \hat{V}_i - V_i || \leq ||\tilde{V}_i - V_i||$. This guaranteees that our results on the sup-norm convergence rates of the nonparametric two-step estimators $\hat{\mathcal{M}}$ and $\hat{k}$ and of their derivatives remain valid, provided some changes are made to the definition of the vector of basis functions $p^K(.)$ and to the approximation condition (\ref{assn:IN2stStep_approx}) of Assumption \ref{assn:IN2stStep}. First, $p^K(.)$ is defined on the extended support, and the bounds $b_1(K)$ and $b_2(K)$ are also defined as bounds on the sup norm over the extended support. Second, the approximation condition must be imposed on the extended functions $\mathcal{M}^\varsigma$ and $k^\varsigma$. Under these modified conditions, because the extended functions remain Lipschitz, the rates of convergence of the nonparametric two-step estimators to the extended functions are the same, the rate of convergence of $\hat{g}$ is unchanged and consistency of $\hat{\mu}$ holds.
As pointed out in Section \ref{sec:trimming}, we will not show that our asymptotic normality result applies to cases where the density of the regressors goes to zero on the boundaries of their support, as we did for consistency (see Assumption \ref{assn:NPV1stStep0Dsty} and \ref{assn:IN2stStep0Dsty}).
Indeed in contrast to the consistency proof, we will use rates on the sup norm of the nonparametric estimates as well as of their derivatives when the suprema are defined over the extended support. This rules out a direct application of the approach allowing the density of the regressors to go to zero on the boundaries of their support. This approach would require Condition (\ref{assn:IN2stStep0Dsty_density}) of Assumption \ref{assn:IN2stStep0Dsty} to hold on $\mathcal{S}_V^\varsigma$, which cannot be true if the density of the regressors is $0$ on the boundary of the original support. Computing rates on the extended support allowing for this case is beyond the scope of this paper. We therefore remain silent on the choice of the basis.

\subsection{Linearization}\label{sec:lineq_app}

Define a class of continuous functions $\mathcal{H}$ endowed with a pseudometric $||.||_{\mathcal{H}}$ such that $\mathcal{G} \in \mathcal{H}$. Arguments yielding asymptotic equivalence to its linearization typically require the following set of conditions. Define $\mathcal{H}_{\delta} = \{ \mathcal{G} \in \mathcal{H} : ||\mathcal{G} - \mathcal{G}_0|| \leq \delta \}$.

\begin{assumption} \label{assn:linearization}
	\
	\begin{enumerate}[noitemsep,nolistsep]
		\item \label{assn:linearization_SE} For all $\delta_n = o(1)$, $\sup_{||\mathcal{G}- \mathcal{G}_0||_{\mathcal{H}} \leq \delta_n } ||  \mathcal{X}_n(\mathcal{G}) - \mathcal{X}(\mathcal{G}) - \mathcal{X}_n(\mathcal{G}_0)|| = o_{\mathbb{P}}(n^{-1/2})$. 
		\item \label{assn:linearization_2ndorder} The pathwise derivative of $\mathcal{X}$ at $\mathcal{G}_0$ evaluated at $\mathcal{G} - \mathcal{G}_0$, $\mathcal{X}^{(G)}(\mathcal{G}_0) [\mathcal{G} - \mathcal{G}_0]$, exists in all directions $[\mathcal{G} - \mathcal{G}_0]$, and for all $\mathcal{G} \in \mathcal{H}_{\delta_n}$ with $\delta_n = o(1)$, $|| \mathcal{X}(\mathcal{G}) - \mathcal{X}^{(G)}(\mathcal{G}_0) [\mathcal{G} - \mathcal{G}_0]\, || \leq c ||\mathcal{G} - \mathcal{G}_0||_{\mathcal{H}}^2$, for some constant $c \geq 0$, 
		\item \label{assn:linearization_Rate} $||\hat{\mathcal{G}} - \mathcal{G}_0||_{\mathcal{H}} = o_{\mathbb{P}}(n^{-1/4})$.
	\end{enumerate}
\end{assumption} 

\begin{result}\label{rslt:X_linear}
	Under Assumption \ref{assn:linearization}, 
	$$ \sqrt{n} \big[ \mathcal{X}_n(\hat{\mathcal{G}}) -  \mathcal{X}_n(\mathcal{G}_0) \big] \, = \, \sqrt{n} \,  \mathcal{X}^{(G)}(\mathcal{G}_0) [\hat{\mathcal{G}} - \mathcal{G}_0] +\oP(1). $$
\end{result}

\begin{proof}[Proof of Result \ref{rslt:X_linear}]
	This proof is a special case of Theorem 2 in \cite{clvk03}. 
	By Assumption \ref{assn:linearization} (\ref{assn:linearization_Rate}), there exists $\delta_n = o(1)$ such that $\mathbb{P}( || \hat{\mathcal{G}} - \mathcal{G}_0 || > \delta_n) \to 0$. Take $\1_n = \1( || \hat{\mathcal{G}} - \mathcal{G}_0 || \leq \delta_n)$.
	$\1_n ||\mathcal{X}_n(\hat{\mathcal{G}}) -  \mathcal{X}_n(\mathcal{G}_0) - \mathcal{X}^{(G)}(\mathcal{G}_0) [\hat{\mathcal{G}} - \mathcal{G}_0]|| \leq \1_n||\mathcal{X}_n(\hat{\mathcal{G}}) - \mathcal{X}(\mathcal{G}) -  \mathcal{X}_n(\mathcal{G}_0) || + \1_n ||\mathcal{X}(\mathcal{G}) -  \mathcal{X}^{(G)}(\mathcal{G}_0) [\hat{\mathcal{G}} - \mathcal{G}_0]|| = \oP(n^{-1/2})$ by Assumption \ref{assn:linearization}. Hence $||\mathcal{X}_n(\hat{\mathcal{G}}) -  \mathcal{X}_n(\mathcal{G}_0) - \mathcal{X}^{(G)}(\mathcal{G}_0) [\hat{\mathcal{G}} - \mathcal{G}_0]|| = \oP(n^{-1/2})$.
\end{proof}

To apply Result \ref{rslt:X_linear} to our estimator, $\mathcal{H}$ and   $||.||_{\mathcal{H}}$ must be defined.
%\footnote{Should I insist ? Since not even Mammen et al (2016) mentions this}
In the definitions at the beginning of Section \ref{sec:as_miou_together}, the choice of $\mathcal{H}$ is driven by the stochastic equicontinuity condition, that is, Condition (\ref{assn:linearization_SE}) of Assumption \ref{assn:linearization}, following \cite{clvk03} and the choice of $||.||_{\mathcal{H}}$ will be driven by Condition (\ref{assn:linearization_2ndorder}).

To write the decomposition of  $\mathcal{X}^{(G)}(\mathcal{G}_0)[\mathcal{G} - \mathcal{G}_0]$ at the end of Section \ref{sec:lineariz}, the functions $\lambda_{M}$, $\lambda_{k}$ and $(\lambda_{bt})_{t}$ are obtained by first computing the partial pathwise derivatives of $\chi$. Write
$\chi_0^{(k)}(W_i)[\tilde{k}]$, $(\chi_0^{(bt)}(W_i)[\tilde{b}_t])_{t \leq T}$ and $\chi_0^{(M)}(W_i)[\tilde{\mathcal{M}}] )$ to be the partial pathwise derivatives of $\chi$ with respect to $k$, $b_t$ and $\mathcal{M}$ (respectively) at the true value $\mathcal{G}_0$, evaluated (respectively) at $\tilde{k}$, $\left(\tilde{b}_t\right)_{t \leq T}$ and $\tilde{\mathcal{M}}$. One can show
\begin{align*}
	\chi_0^{(M)}(W_i)[\tilde{\mathcal{M}}] \,  : & =  \, - [ g_0(V_i)' \otimes (Q_i^\delta \mathcal{M}_0(V_i)^{-1}) ] \, \VecM(\tilde{\mathcal{M}}(V_i)),\\
	\chi_0^{(k)}(W_i)[\tilde{k}] \,  : & = \,  \chi^{(k)}(W_i, \mathcal{G}_0)[\tilde{k}] \,  = \,  Q_i^\delta \mathcal{M}_0(V_i)^{-1}\tilde{k}(V_i),\\
	\chi_0^{(bt)}(W_i)[\tilde{b}_t] \,  :  & = \, \chi^{(2t)}(W_i, \mathcal{G}_0)[\tilde{b}_t]  \, = \, - Q_i^\delta  \frac{\partial g_0}{\partial v_t}(V_i) \tilde{b}_t(\xi_{it}).
\end{align*}
%I took off some computations, see older versions.

\subsection{Linear Application of a Nonparametric Two-Step Sieve Estimator}\label{app:lin_galcase}

The proof of asymptotic normality of $\hat{\mu}$ relies on the analysis of a linear function applied to nonparametric two-step sieve estimator. We focus on the general model (\ref{eq:reg_GalCase}) and derive here results for the generic case of estimation of a linear function $a$ evaluated at $h^W$ where $h^W(v) \, = \, \E(W|V=v)$. We use the nonparametric two-step sieve estimator $\hat{h}^W$. We do not specify the sieve basis.

Functionals of nonparametric estimators have been widely studied for different types of nonparametric estimators (see, e.g, \cite{n94b} for kernel estimators and \cite{n97} for series estimators). However the linear functional here is also evaluated at the more complicated two-step nonparametric estimators constructed in the previous sections. Its asymptotic distribution cannot be derived directly from the aforementioned results. \cite{hlr18} derive asymptotic normality results for nonlinear functionals of two-step nonparametric sieves estimators when the sieve estimators are from a general class of nonlinear sieve regression estimators. Characterizing the finite sample variance, they provide a practical estimator arguing that the asymptotic variance does not always have an analytical form. We show with a different set of calculations giving lower-level conditions on the primitives that the asymptotic variance of our estimator can be obtained for a class of models where the orthogonality condition between the first and second stage does not hold.
They however do not specify a formula for the asymptotic variance, arguing that it might not exist. This is not an issue in our case and we derive using a different type of proof the asymptotic normality and asymptotic variance of our sieve estimators.

The estimator of $a(h^W)$ will be $a(\hat{h}^W)$, and the purpose of this section is to write, under general conditions on the random variables $W,V,e^{W*}$ and the functions $h^W$ and $\rho^W$, the term $\sqrt{n} (a(h^W) - a(\hat{h}^W))$ as $\frac{1}{\sqrt{n}} \sum_{i = 1}^n  \, s_{i,n}^W + o_{\mathbb{P}}(1)$. We will then apply the derived results to $\mathcal{X}_0^{(k)} [\hat{k} - k_0]$ and $\mathcal{X}_0^{(M)}[\hat{\mathcal{M}} - \mathcal{M}_0 ]$.
We consider the case where $w \in \mathbb{R}$, $V \in \R^{Td_2}$, $a(h) \in \mathbb{R}^{d_a}$.
Define $\rho_i^W = \rho^W(X_i,Z_i)$ and the following matrices,
\begin{align*}
	\overbar{H}_t^W & =  \frac{1}{n} \sum_{i=1}^n \hat{p}_i \, (\partial h_\varsigma^W(V_i) / \partial v_t \otimes r_{it}') & A \, & =  \, (a(p_{1K}),\, .. \, , a(p_{KK})), \\
	& =  \frac{1}{n} \sum_{i=1}^n \hat{p}_i \, (\partial h^W(V_i) / \partial v_t \otimes r_{it}'), & \overbar{\text{d}P}_t^W & = \frac{1}{n} \sum_{i=1}^n \rho_i^W \left( \frac{\partial p^K (V_i)}{\partial v_t} \otimes r_{it}' \right),\\
	H_t^W  &= \E[p_i  \, (\partial h_\varsigma^W(V_i) / \partial v_t \otimes r_{it}')] & \text{d}P_t^W & = \E(\rho_i^W \left( \frac{\partial p^K (V_i)}{\partial v_t} \otimes r_{it}' \right)),\\
	& =  \E[p_i  \, (\partial h^W(V_i) / \partial v_t \otimes r_{it}')],
\end{align*}
where $h_\varsigma^W$ is the functional extension of $h^W$ and where the equalities on the first two matrices hold because $h^\varsigma$ and $h$ are equal on $\mathcal{S}_V$.
Recall that by $||\tau(v_1) - \tau(v_2)|| \leq ||v_1 - v_2 ||$, under Assumption \ref{assn:NPV1stStep} we have $\frac{1}{n} \, \sum_{i=1}^n || v_{i} - \hat{v}_{i} ||^2 = O_{\mathbb{P}} \left( L / n + L^{- 2 \gamma_1}  \right) = O_{\mathbb{P}} (\Delta_n^2).$ 

\begin{assumption}\label{assn:linearAsEquivlt}
	\
	\begin{enumerate}[noitemsep,nolistsep]
		\item The data $W_i$ is i.i.d,
		
		\item  \label{assn:LAE_ALipsch} $||a(g)|| \leq C |g|_0$,
		
		\item \label{assn:LAE_fctn} $h^W$ is twice continuously differentiable with bounded first and second derivatives, and $\rho^W$ is bounded, 
		
		\item \label{assn:LAE_approx} There exists $\gamma_1$ and $\beta_{t}^L$ such that for all $t \leq T$, $\sup_{\mathcal{S}_{\xi_t}} || b_t(\xi_t) - \beta_{t}^{L \, \prime} r^L(x_t^{ex},z_t) || \leq C L^{- \gamma_1}$. There exists $\gamma_2$, $\pi_{W}^K$ such that $\sup_{\mathcal{S}_{V}^\varsigma} || h_\varsigma^W(v) - p^K(v)' \, \pi_{W}^{K} || \leq C K^{- \gamma_2}$,
		
		\item \label{assn:LAE_sievebasis}  For all $t \leq T$, there exists $\Gamma_{1t}$, a $L \times L$ nonsingular matrix such that for $R_t^L(\xi_t) = \Gamma_{1t} r_t^L(\xi_t)$, $\E(R_t^L(\xi_t) R_t^L(\xi_t)')$ has smallest eigenvalue bounded away from $0$ uniformly in $L$.
		There exists $\Gamma_2$, a $K \times K$ nonsingular matrix such that for $P^K(V) = \Gamma_2 p^K(V)$, $\E(P^K(V) P^K(V)')$ has smallest eigenvalue bounded away from $0$ uniformly in $K$, 
		
		\item \label{assn:LAE_matrixbounded} $||A||$ is bounded, 
		
		\item \label{assn:LAE_rate} For $|R_t^L(\xi_t)|_0 \leq a_1(L)$, $|P^K(V)|_0^\varsigma \leq b_1(K)$, $|P^K(V)|_1^\varsigma \leq b_2(K)$, $|P^K(V)|_2^\varsigma \leq b_3(K)$, we have
		$\sqrt{n} K^{- \gamma_2} = o(1)$, $\max(\sqrt{K}, \sqrt{L} b_2(K)) \, a_1(L) \sqrt{L/n}  = o(1)$, 
		$b_2(K) \sqrt{n} L^{-\gamma_1} \allowbreak = o(1)$, $b_2(K)[\sqrt{L/n} + L^{- \gamma_1}][K + L]  = o(1)$, $ b_3(K) [L /\sqrt{n} + \sqrt{n} L^{- 2\gamma_1}] = o(1)$, $b_2(K)^2 \sqrt{K} [\sqrt{L/n} + L^{- \gamma_1}] = o(1)$,
		
		\item \label{assn:LAE_varbounded} $\E(||v_{t}||^2|\xi_{t})$ and $\Var(e^{W*}|X,Z)$ are bounded on $\mathcal{S}_{\xi t}$ and $\mathcal{S}_{X,Z}$ respectively.
	\end{enumerate}
\end{assumption}

\begin{lemma}\label{lem:LAE}
	Under Assumptions \ref{assn:NPV1stStep} and \ref{assn:linearAsEquivlt},
	\begin{align*}
		\sqrt{n} [ a(\hat{h}^W) & - a(h^W)] =  \\
		& \frac{1}{\sqrt{n}} \sum_{i=1}^n  A \, \E(p_i p_i')^{-1} \, \left[  p_i e_i^W  +  \sum_{t=1}^T  ( H_t^W - \text{d}P_t^W)  (I_{k_2} \otimes \E(r_{it} r_{it}')^{-1}) \, (v_{it} \otimes r_{it})  \right] + \oP(1).
	\end{align*}
\end{lemma}

We use the proof techniques of Lemma 2 of \cite{npv99} to obtain this approximation. However as mentioned previously, an essential orthogonal condition they assumed, namely the conditional mean independence of $w - \E(w|V)$ of $(X,Z)$, does not hold in our model. For this reason we obtain an extra term depending on $\rho^W$, the term $\text{d}P_t^W$ which would be zero if $\rho^W(X,Z)=0$. Another difference is the summation over $t$ of the $H_t^W$ and $\text{d}P_t^W$ terms due to vector of control variables being composed of $T$ estimated residuals coming from $T$ different cross-section regressions.

\begin{proof}[Proof of Lemma \ref{lem:LAE}]
	
	As argued in the proof of Result \ref{rslt:IN2stStep}, under Assumption \ref{assn:linearAsEquivlt} (\ref{assn:LAE_sievebasis}), we can impose without loss of generality the normalization $\E(p_i p_i') \, = \, I_K$ and $\E(r_i r_i') \, = \, I_L$.
	
	We define $\Delta_{\partial P} = b_2(K) \sqrt{L/n}$, 
	$\Delta_Q = b_2(K)^2 \Delta_n^2 + \sqrt{K} b_2(K) \Delta_n + b_1(K) \sqrt{K/n}$, $\Delta_{Q1} = a_1(L) \sqrt{L/n}$ and $\Delta_H = b_2(K) \Delta_n \sqrt{L} + a_1(L) \sqrt{K/n}$. Recall that $\Delta_n^2 = L / n + L^{- 2 \gamma_1}$ and note that $b_1(K) \leq b_2(K) \leq b_3(K)$
	Under Assumption \ref{assn:linearAsEquivlt} (\ref{assn:LAE_rate}), $\sqrt{n} K^{- \gamma_2} = o(1)$ and
	\begin{align*}
		& \sqrt{K} \Delta_Q  = O(\sqrt{K} [\sqrt{K} b_2(K) \Delta_n + b_1(K) \sqrt{K/n}]) = O(K b_2(K) [\sqrt{L/n} + L^{- \gamma_1}]) = o(1), \\
		&\Delta_{\partial P} \sqrt{L}  = b_2(K) L / \sqrt{n} = o(1), \qquad b_2(K) \Delta_{Q1} \sqrt{L}  = b_2(K) a_1(L) L / \sqrt{n} = o(1), \\
		&b_2(K) \Delta_n  = o(1), \qquad \sqrt{L} \Delta_H  = O(b_2(K) L (\sqrt{L/n} + L^{- \gamma_1})+ a_1(L) \sqrt{K L/n}) =o(1), \\
		&\Delta_Q b_2(K)  = O\left(b_2(K)^2 \sqrt{K} [ L^{- \gamma_1} + \sqrt{L/n}] + b_2(K)^2 \sqrt{K/n} \right) = o(1),\\
		&\sqrt{n} b_3(K) \Delta_n^2  = b_3(K) [L /\sqrt{n} + \sqrt{n} L^{- 2\gamma_1}] = o(1).
	\end{align*}
	These results imply in particular that $\Delta_{\partial P} = o(1)$, $\Delta_Q = o(1)$, $\Delta_{Q1} = o(1)$ and $\Delta_H = o(1)$. All these rates will be used in the steps of the proof.

	We define for all $t \leq T$, $Q_{1t} = R_t R_t'/n$. As in the proof of Result \ref{rslt:IN2stStep}, we obtain $||Q_{1t} - I_L|| \, = \, O_{\mathbb{P}}(\Delta_{Q1}) = \oP(1)$, $t \leq T$, and $||\hat{Q} - I_K|| \, = \, O_{\mathbb{P}}(\Delta_Q) = \oP(1)$. This implies, as argued in NPV99, that the eigenvalues of $\hat{Q}$ are bounded away from $0 \wpa1$, therefore $||B \hat{Q}^{-1}|| \leq ||B|| \OP(1)$ and  $||B \hat{Q}^{-1/2}|| \leq ||B|| \OP(1)$ for any matrix $B$.
	Using 
	\begin{align*}
		||\overbar{H}_t^W - H_t^W||  \leq  & ||  \frac{1}{n} \sum_{i=1}^n  (\hat{p}_i - p_i) (\partial h_\varsigma^W(V_i) / \partial v_t \otimes r_{it}') || \\
		& \qquad + || \frac{1}{n} \sum_{i=1}^n p_i \, (\partial h^W(V_i) / \partial v_t \otimes r_{it}') - \E[p_i  \, (\partial h^W(V_i) / \partial v_t \otimes r_{it}')]||, 
	\end{align*}
	with  
	$$|| \frac{1}{n} \sum_{i=1}^n  (\hat{p}_i - p_i) (\partial h_\varsigma^W(V_i) / \partial v_t \otimes r_{it}') || \leq b_2(K) \Delta_n \sup_{\mathcal{S}_V^\varsigma} || \partial h_\varsigma^W ||  \tr(R_t R_t')^{1/2}/\sqrt{n} = \OP(b_2(K) \Delta_n \sqrt{L}),$$
	and 
	\begin{align*}
		\E(|| & \frac{1}{n}  \sum_{i=1}^n p_i \, (\partial h^W(V_i) / \partial v_t \otimes r_{it}') - \E[p_i  \, (\partial h^W(V_i) / \partial v_t \otimes r_{it}')]||^2) \\
		& \leq \E \left[ \tr \left( \frac{1}{n^2}  \sum_{i=1}^n p_i \, (\partial h^W(V_i) / \partial v_t \otimes r_{it}') \, (\partial h^W(V_i) / \partial v_t \otimes r_{it}')' p_i' \right) \right] \\
		& \leq C \frac{1}{n} \E( \tr( p_i r_i' r_i p_i')) \leq C a_1(L)^2 K/n,
	\end{align*}
	we obtain $||\overbar{H}_t^W - H_t^W|| = \OP(\Delta_H)$. Thus by Assumption \ref{assn:linearAsEquivlt} (\ref{assn:LAE_rate}),  $||\overbar{H}_t^W - H_t^W|| = \oP(1)$.
	Moreover, since $\rho^W(.)$ is a bounded function and $\E(r_{it}' r_{it}) = \E(\tr(I_L)) = L$, we have for all $t \leq T$,
	\begin{align*}
		\E(||\overbar{\text{d}P}_t^W - \text{d}P_t^W ||^2) & 
		\leq  \frac{C}{n^2}  \sum_{i=1}^n  \E \left[ \tr \left( \frac{\partial p^K (V_i)}{\partial v_t} \left( \frac{\partial p^K (V_i)}{\partial v_t} \right)' \  r_{it}' r_{it} \right) \right] \leq C b_2(K)^2 L/n,
	\end{align*}
	which implies by M that $||\overbar{\text{d}P}_t^W - \text{d}P_t^W || = \OP(\Delta_{\partial P}) = \oP(1)$.
	
	Since $a(.)$ is linear, $a(\hat{h}^W) \, = \, A \hat{\pi}_W$.  Using Assumption \ref{assn:linearAsEquivlt} (\ref{assn:LAE_ALipsch}),
	$$
	||a(p^K(.)' \, \pi_W^{K}) - a(h^W(.))|| \leq  C \sup_{\mathcal{S}_V} ||p^K(.)' \, \pi_W^{K} - h^W(.)|| \leq C \sup_{\mathcal{S}_V^\varsigma} ||p^K(.)' \, \pi_W^{K} - h_\varsigma^W(.)|| \leq C K^{-\gamma_2}, 
	$$
	which implies using Assumptions \ref{assn:linearAsEquivlt} (\ref{assn:LAE_rate}) that
	\begin{align*}
		\sqrt{n} [a(\hat{h}^W) - a(h^W)] \, & = \, \sqrt{n}  A [\hat{\pi}_W - \pi_W^K] \, + \, o_{\mathbb{P}}(1), \\
		& = \,   A  \hat{Q}^{-1} \hat{P} ( W - \hat{P}' \pi_W^K) /\sqrt{n} \, + \, o_{\mathbb{P}}(1).
	\end{align*}
	Since 
	$$|| A  \hat{Q}^{-1} \hat{P} (\vec{\hat{h}}_\varsigma^W - \hat{P}' \pi_W^K) || \leq  || A  \hat{Q}^{-1} \hat{P}||  ||\vec{\hat{h}}_\varsigma^W - \hat{P}' \pi_W^K|| \leq \sqrt{n} ||A \hat{Q}^{-1/2}|| \, \sqrt{n} \sup_{\mathcal{S}_V^\varsigma} ||p^K(.)' \, \pi_W^{K} - h^W(.)||,$$
	then $ A  \hat{Q}^{-1} \hat{P} (\hat{h}_\varsigma^W - \hat{P}' \pi_W^K) /\sqrt{n} = \oP(1)$ by $||A||$ bounded. Therefore  we obtain as in NPV99,
	\begin{equation}
		\sqrt{n}  (a(\hat{h}^W) - a(h^W)) \, = \,   \underbrace{ A \hat{Q}^{-1} \hat{P} e^W / \sqrt{n}}_{\text{(B)}} \, + \, \underbrace{ A \hat{Q}^{-1} \hat{P} (\vec{h}_\varsigma^W - \vec{\hat{h}}_\varsigma^W) / \sqrt{n}}_{\text{(C)}}  \, + \, o_{\mathbb{P}}(1).
	\end{equation}
	We first focus on the term (B), which we decompose
	\begin{equation*}
		\text{(B)} = AP'e^W / \sqrt{n} \, + \, \underbrace{A (\hat{Q}^{-1} - I) P e^W/ \sqrt{n}}_{(B1)} \, + \,  \underbrace{A \hat{Q}^{-1} (\hat{P} - P) e^{W*} / \sqrt{n}}_{(B2)}\, + \,  \underbrace{A \hat{Q}^{-1} (\hat{P} - P) \vec{\rho}^W / \sqrt{n}}_{(B3)}.
	\end{equation*}
	Since 
	$$
	\E(|| P e^W / \sqrt{n} ||^2) =  \tr[\E(P e^W (e^W)' P')]/n =  \tr[\E( \E(e^W (e^W)'|\vec{V}) P'P)]/n  \leq C \tr[I_k] = O(K)
	$$ 
	by Assumption \ref{assn:linearAsEquivlt} (\ref{assn:LAE_varbounded}), by M we have  $|| P e^W / \sqrt{n} || = \OP(K^{1/2})$.
	Therefore,
	\begin{align*}
		& ||(B1)||   \leq || A \hat{Q}^{-1} || \, ||I_K - \hat{Q}|| \, || P e^W / \sqrt{n} || = || A || \OP(1) \, ||I_K - \hat{Q}|| \, || Pe^W / \sqrt{n} || = \OP( \Delta_Q K^{1/2})\\
		& \Rightarrow (B1) = \oP (1),
	\end{align*}
	under Assumption \ref{assn:linearAsEquivlt} (\ref{assn:LAE_varbounded}).
	We now look at the extra terms (B2) and (B3), where we decomposed $e^W$ as $\rho^W + e^{W*}$ since $e^W$ itself is not conditionally mean independent of $(\hat{P} - P)$, while $e^{W*}$ is. Indeed, since $ \E(e^{W*} |\vec{x},\vec{z}) = 0$,  
	\begin{align*}
		\E(|| (\hat{P} - P)e^{W*} / \sqrt{n} ||^2|\vec{x},\vec{z}) & = 1/n \tr[\E((\hat{P} - P)e^{W*} (e^{W*})' (\hat{P} - P)'|\vec{x},\vec{z})] \\
		& = 1/n \tr[(\hat{P} - P) \, \E(e^{W*} (e^{W*})'|\vec{x},\vec{z} \,) \, (\hat{P} - P)'] \, \leq C/n || \hat{P} - P||^2.
	\end{align*}
	By the proof of Result \ref{rslt:IN2stStep}, $||\hat{P} - P||^2/n = \OP(b_2(K)^2 \Delta_n^2)$ (the difference is that now $b_2(K)$ is defined as the sup rate over the extended support $\mathcal{S}_V^\varsigma$) hence by CM, $|| (\hat{P} - P)e^{W*} / \sqrt{n} || = \OP (b_2(K) \Delta_n)$.
	\begin{align*}
		||A \hat{Q}^{-1}(\hat{P} - P) e^{W*} /\sqrt{n} || & \leq ||A \hat{Q}^{-1}|| \, ||(\hat{P} - P) e^{W*} || \, /\sqrt{n} =  ||A ||  \OP(1) ||(\hat{P} - P) e^{W*} || /\sqrt{n}\\
		& \Rightarrow (B2) = \OP (b_2(K) \Delta_n) = \oP(1).
	\end{align*}
	We now focus on (B3). We have $(\hat{P} - P) \vec{\rho}^W / \sqrt{n} = 1/ \sqrt{n} \sum_{i=1}^n (\hat{p}_i^K - p_i^K) \rho_i^W$ and a second order Taylor expansion gives 
	$$||(\hat{p}_i^K - p_i^K) -  \frac{\partial p^K (\tau(V_i))}{\partial V}  \frac{\partial \tau (V_i)}{\partial V} (\tilde{V}_i - V_i)|| \leq C b_3(K) ||\tilde{V}_i - V_i||^2,  $$
	which can be rewritten as $||(\hat{p}_i^K - p_i^K) -  \frac{\partial p^K (V_i)}{\partial v} (\tilde{V}_i - V_i)|| \leq C b_3(K) ||\tilde{V}_i - V_i||^2,  $ since $V_i \in \mathcal{S}_V$ and we chose $\tau$ so that its Jacobian matrix is the identity matrix on $\mathcal{S}_V$. Hence 
	\begin{align*}
		||A \hat{Q}^{-1} \frac{1}{\sqrt{n}} [(\hat{P} - P) \vec{\rho}^W - \sum_{i=1}^n \frac{\partial p^K (V_i)}{\partial V} (\tilde{V}_i - V_i) \rho_i^W] \, || & \leq C \OP(1)  b_3(K) \sum_{i=1}^n ||\tilde{V}_i - V_i||^2 / \sqrt{n} \\
		& = \OP(\sqrt{n} b_3(K) \Delta_n^2) = \oP(1),
	\end{align*}
	by Assumption \ref{assn:linearAsEquivlt} (\ref{assn:LAE_rate}). Therefore $(B3) = A \hat{Q}^{-1} \sum_{i=1}^n \frac{\partial p^K (V_i)}{\partial V} (\tilde{V}_i - V_i) \rho_i^W / \sqrt{n} + \oP(1)$. We can decompose
	\begin{align*}
		(-1) (\tilde{v}_{it} - v_{it}) \, & = \, \hat{\beta}_{t}' \,  r_{it}  - b_{it} \, = \, (Q_{1t}^{-1}R_t [X_t^{en \, \prime} - R_t' \beta_{t}^{L}]/n )' r_{it} \, + \, \beta_{t}^{L \, \prime} r_{it}- b_{it},  \\
		& = [  (Q_{1t}^{-1} R_t \vec{v}_t^{\, \prime}/n )' r_{it}] \, + \, [(Q_{1t}^{-1} R_t [\vec{b}_{t}^{\, \prime} - R_t' \beta_{t}^{L}]/n )' r_{it}  ] \, + \, [ \beta_{t}^{L \, \prime} r_{it}- b_{it} ],
	\end{align*}
	and then apply this decomposition to $(B.3)$, $(B.3) = -[ (B3.1) + (B3.2) + (B3.3)] +  \oP(1)$, where
	\begin{align*}
		(B3.1) = 
		& \, \sum_{t=1}^T A  \hat{Q}^{-1}  \sum_{i=1}^n  \rho_i^W \frac{\partial p^K (V_i)}{\partial v_t} \left[ Q_{1t}^{-1} R_t [\vec{b}_{t}^{\, \prime} - R_t' \beta_{t}^{L}]/n  \right]' r_{it} / \sqrt{n}  \,\\
		(B3.2) = 
		&  \,   \sum_{t=1}^T A  \hat{Q}^{-1}  \sum_{i=1}^n  \rho_i^W \frac{\partial p^K (V_i)}{\partial v_t}  [ \beta_{t}^{L \, \prime} r_{it}- b_{it} ] / \sqrt{n}    \\
		(B3.3) = 
		&    \,  \sum_{t=1}^T A  \hat{Q}^{-1}  \sum_{i=1}^n \rho_i^W \frac{\partial p^K (V_i)}{\partial v_t} \left[ Q_{1t}^{-1} R_t \vec v_t'/n \right]' r_{it} / \sqrt{n}
	\end{align*}
	The first term in this expression of $(B.3)$ can be rewritten
	\begin{align*}
		(B3.1) & 
		= \sum_{t=1}^T A  \hat{Q}^{-1}  \sum_{i=1}^n  \rho_i^W \frac{\partial p^K (V_i)}{\partial v_t} \left[ Q_{1t}^{-1} R_t [\vec{b}_{t}^{\, \prime} - R_t' \beta_{t}^{L}]/n  \right]' r_{it} / \sqrt{n}  \\
		& 
		= \sum_{t=1}^T A  \hat{Q}^{-1} \frac{1}{n} \sum_{i=1}^n  \rho_i^W \left( \frac{\partial p^K (V_i)}{\partial v_t} \otimes r_{it}' \right) \VecM( Q_{1t}^{-1} R_t [\vec{b}_{t}^{\, \prime} - R_t' \beta_{t}^{L}]) /\sqrt{n} \\
		& 
		= \sum_{t=1}^T A  \hat{Q}^{-1} \overbar{\text{d}P}_t^W \, (I_{d_2} \otimes Q_{1t}^{-1} R_t)  \VecM( \vec{b}_{t}^{\, \prime} - R_t' \beta_{t}^{L}) /\sqrt{n},
	\end{align*}
	(see e.g p282 \cite{amagnus05}) where $||  \VecM( \vec{b}_{t}^{\, \prime} - R_t' \beta_{t}^{L})   || \leq \sqrt{n} \, \sup_{\mathcal{S}_{\xi_t}} || b_t(.) - \beta_{t}^{L \, \prime} r^L(.) || \leq \sqrt{n} L^{- \gamma_1} $. 
	Defining, for $d \leq d_2$, the matrix $\overbar{\text{d}P}_{td}^W  = \frac{1}{n} \sum_{i=1}^n \rho_i^W \frac{\partial p^K (V_i)}{\partial v_{td}} r_{it}'$ where $v_{td}$ is the d$^{th}$ component of $v_t$, then $\overbar{\text{d}P}_t^W = (\overbar{\text{d}P}_{t1}^W, \,.. \, , \overbar{\text{d}P}_{td_2}^W)$, and we can write
	\begin{align*}
		||A \hat{Q}^{-1} & \overbar{\text{d}P}_t^W (I_{d_2} \otimes Q_{1t}^{-1} R_t) ||^2 = \sum_{d =1}^{d_2} ||A \hat{Q}^{-1} \overbar{\text{d}P}_{td}^W  Q_{1t}^{-1} R_t) ||^2 \\
		& =  \sum_{d =1}^{d_2} \tr( A \hat{Q}^{-1}  \overbar{\text{d}P}_{td}^W  Q_{1t}^{-1} R_t R_t' Q_{1t}^{-1} \overbar{\text{d}P}_{td}^{W \, \prime} \hat{Q}^{-1} A'  )
		= n \sum_{d =1}^{d_2} \tr( A \hat{Q}^{-1} \overbar{\text{d}P}_{td}^W  Q_{1t}^{-1} \overbar{\text{d}P}_{td}^{W \, \prime} \hat{Q}^{-1} A' ),
	\end{align*}
	and $  \overbar{\text{d}P}_{td}^W  Q_{1t}^{-1} \overbar{\text{d}P}_{td}^{W \, \prime} \,  =  \, \left( \frac{1}{n} \sum_{i=1}^n \rho_i^W  \frac{\partial p^K (V_i)}{\partial v_{td}} r_{it}' \right) \left( \frac{1}{n} \sum_{i = 1}^n r_{it} r_{it}' \right)^{-1}  \left( \frac{1}{n} \sum_{i=1}^n \rho_i^W  \frac{\partial p^K (V_i)}{\partial v_{td}} r_{it}' \right)'$. 
	
	By $R_t' (R_t R_t')^{-1} R_t$ being an orthogonal projection matrix, 
	$$
	\overbar{\text{d}P}_{td}^W  Q_{t}^{-1} \overbar{\text{d}P}_{td}^{W \, \prime} \leq \frac{1}{n}  \sum_{i=1}^n (\rho_i^W)^2  \frac{\partial p^K (V_i)}{\partial  v_{td}} \left( \frac{\partial p^K (V_i)}{\partial  v_{td}} \right)',
	$$
	implying $|| \overbar{\text{d}P}_{td}^W  Q_{t}^{-1/2} || \leq b_2(K)$ and
	\begin{align*}
		\tr( A \hat{Q}^{-1} \overbar{\text{d}P}_{td}^W  Q_{t}^{-1} \overbar{\text{d}P}_{td}^{W \, \prime} \hat{Q}^{-1} A' ) & \leq \frac{1}{n} \sum_{i=1}^n (\rho_i^W)^2 \tr( A \hat{Q}^{-1} \frac{\partial p^K (V_i)}{\partial  v_{td}} \left( \frac{\partial p^K (V_i)}{\partial  v_{td}} \right)' \hat{Q}^{-1} A'  ) \\
		& \leq \frac{1}{n} ||A \hat{Q}^{-1}||^2 \, \sum_{i=1}^n  (\rho_i^W)^2  ||  \frac{\partial p^K (V_i)}{\partial  v_{td}} ||^2 = \OP(1) \, b_2(K)^2,
	\end{align*}
	since $\rho^W(.)$ is bounded. Hence $||A \hat{Q}^{-1} \overbar{\text{d}P}_t^W (I_{d_2} \otimes Q_{1t}^{-1} R_t) ||^2 = \OP(n b_2(K)^2)$
	and we obtain
	%\footnote{If there is a known bound on objects of the type $||Q^{-1/2} \partial P||$, a normalization of the first order partial derivatives, it could probably improve the rate.}
	by Assumption \ref{assn:linearAsEquivlt} (\ref{assn:LAE_rate}), 
	$$||(B3.1)|| = \OP(\sqrt{n} b_2(K) \sqrt{n} L^{-\gamma_1} / \sqrt{n}) = \OP(b_2(K) \sqrt{n} L^{-\gamma_1}) = \oP(1).$$
	Focusing now on the second term in the expression of $(B3)$,
	\begin{align*}
		||(B3.2)|| & = \big| \big| \sum_{t=1}^T A  \hat{Q}^{-1}  \sum_{i=1}^n  \rho_i^W \frac{\partial p^K (V_i)}{\partial v_t}  [ \beta_{t}^{L \, \prime} r_{it}- b_{it} ]  \big| \big| / \sqrt{n} \leq || A \hat{Q}^{-1} || n b_2(K) \, C L^{- \gamma_1}  / \sqrt{n} \\
		& = \OP(1) b_2(K)  \sqrt{n} L^{- \gamma_1} = \oP(1),
	\end{align*}
	again by Assumption \ref{assn:linearAsEquivlt} (\ref{assn:LAE_rate}). This implies that $(B3) = - (B3.3) + \oP(1)$, with
	\begin{align*}
		(B3.3) =  \sum_{t=1}^T A  \hat{Q}^{-1} \overbar{\text{d}P}_t^W  (I_{d_2} \otimes Q_{1t}^{-1} )  \VecM( R_t \vec v_t') /\sqrt{n}.
	\end{align*}
	First, by $\E(||v_{it}||^2|\xi_{it} = \xi_t)$ bounded, $||R_t \vec v_t' /\sqrt{n} || = \OP(\sqrt{L})$, which gives
	\begin{align*}
		|| \sum_{t=1}^T A  \hat{Q}^{-1} &\overbar{\text{d}P}_t^W  [ I_{d_2} \otimes Q_{1t}^{-1} - I_{d_2L}]  \VecM( R_t \vec v_t') /\sqrt{n} ||, \\
		& \leq ||A  \hat{Q}^{-1} || \, \sum_{t = 1}^T || \overbar{\text{d}P}_t^W  (I_{d_2} \otimes Q_{1t}^{-1/2})|| \,  ||I_{d_2} \otimes Q_{1t}^{-1/2}|| \, ||I_{d_2} \otimes (I_L - Q_{1t})|| \, ||R_t \vec v_t' /\sqrt{n} ||, \\
		& \leq \OP(1) C b_2(K) \OP(1) \Delta_{Q1} \sqrt{L} = \OP(b_2(K) \Delta_{Q1} \sqrt{L} ) = \oP(1),
	\end{align*}
	by Assumption \ref{assn:linearAsEquivlt} (\ref{assn:LAE_rate}). Similarly
	\begin{align*}
		|| \sum_{t=1}^T A  \hat{Q}^{-1} (\overbar{\text{d}P}_t^W - \text{d}P_t^W) \VecM( R_t \vec v_t') /\sqrt{n} || & \leq C ||A  \hat{Q}^{-1} || \, ||\overbar{\text{d}P}_t^W - \text{d}P_t^W || \,  ||R_t \vec v_t' /\sqrt{n} ||\\
		& = \OP(\Delta_{\partial P} \sqrt{L}) = \oP(1).   
	\end{align*}
	Finally, we write $\text{d}P_{td}^W  = \E\left( \rho_i^W \partial p^K (V_i) /\partial v_{td} \, r_{it}' \right)$, as well as  $v_{itd}$ the d$^{th}$ component of $v_{it}$ and $\vec{v}_{td} = (v_{1td}, \, .. , \, v_{ntd})'$. Then
	\begin{align*}
		\E(|| & \text{d}P_{td}^W R_t \vec{v}_{td}' / \sqrt{n} ||^2 )  = \tr \left( \text{d}P_{td}^W \E \left(  R_t \E( \vec{v}_{td} \vec{v}_{td}' | \vec{x}^{ex}, \vec{z}) R_t' \right) \text{d}P_{td}^{W \, \prime} \right) /n 
		\leq C \tr \left( \text{d}P_{td}^W \E \left(  R_t R_t' \right) \text{d}P_{td}^{W \, \prime} \right) /n \\
		& \leq C  \, || \text{d}P_{td}^W \text{d}P_{td}^{W \, \prime} ||  = C \E\left( \rho_i^W \partial p^K (V_i) /\partial v_{td} \, r_{it}' \right)\E(r_{it} r_{it}')^{-1} \E\left( \rho_i^W r_{it} \partial p^K (V_i)' /\partial v_{td}  \right) \\
		& \leq C \E( (\rho_i^W)^2 \partial p^K (V_i) /\partial v_{td} (\partial p^K (V_i) /\partial v_{td})') \leq C b_2(K)^2 ,
	\end{align*}
	where the second to last inequality follows from taking the orthogonal projection matrix argument to the limit.%\footnote{Surely, there is a better argument.} 
	This implies that $||\text{d}P_{td}^W R_t \vec{v}_{td}' / \sqrt{n} || = \OP(b_2(K))$, and 
	\begin{align*}
		|| \sum_{t=1}^T A  (\hat{Q}^{-1} - I_{d_2}) \text{d}P_t^W \VecM( R_t \vec{v}_t') /\sqrt{n} || & \leq \sum_{t=1}^T \sum_{d = 1}^{d_2}  ||A \hat{Q}^{-1} (I_{d_2} - \hat{Q}) \text{d}P_{td}^W  R_t \vec{v}_{td}') /\sqrt{n} || \\
		&  \leq  ||A \hat{Q}^{-1}|| \, \OP(\Delta_Q) \OP(b_2(K)) = \OP(\Delta_Q b_2(K)) = \oP(1),
	\end{align*}
	by Assumption \ref{assn:linearAsEquivlt} (\ref{assn:LAE_rate}). We can now write 
	\begin{align*}
		(B3.3) = \frac{1}{\sqrt{n}} A \sum_{t=1}^T  \text{d}P_t^W \VecM( R_t \vec{v}_t^{\, \prime}) + \oP(1) =  \frac{1}{\sqrt{n}} \sum_{i=1}^n  A \sum_{t=1}^T  \text{d}P_t^W \, v_{it} \otimes r_{it} \, + \, \oP(1),
	\end{align*}
	where the term appearing in the sum over $n$, were the weight matrices not normalized, would become 
	$  \sum_{t \leq T} A   \E(p_i p_i')^{-1} \allowbreak \text{d}P_t^W \, ( I_{d_2} \otimes \E(r_{it} r_{it}')^{-1} ) \, v_{it} \otimes r_{it}$. Adding all terms appearing in $(B)$, one obtains,
	\begin{align*}
		\text{(B)} = \frac{1}{\sqrt{n}} \sum_{i=1}^n  A \left[  p_i e_i^W  \, - \, \sum_{t=1}^T  \text{d}P_t^W \, (v_{it} \otimes r_{it})  \right] + \oP(1).
	\end{align*}
	
	The remaining term in the expression of $\sqrt{n} (a(\hat{h}^W) - a(h^W))$ is $(C) = A \hat{Q}^{-1} \hat{P} (\vec{h}_\varsigma^W - \vec{\hat{h}}_\varsigma^W) / \sqrt{n}.$
	This term is similar to the second term in equation $(A.16)$ of NPV99, p598, where the regression function is becomes $h_\varsigma^W$. Since $v \mapsto h_\varsigma^W(\tau(v))$ is by composition twice continuously differentiable and has bounded second order derivative on the extended support, one obtains using NPV99 
	\begin{equation*}
		\text{(C)} = \frac{1}{\sqrt{n}}  \sum_{i=1}^n A \sum_{t = 1}^T H_t^W (v_{it} \otimes r_{it})  + \oP (1),
	\end{equation*}
	adapting to the fact that $h_\varsigma^W$ is here function of $T$ generated covariates instead of one and using $\sqrt{n}L^{-\gamma_1}$, $\sqrt{n} b_1(K) \Delta_n^2$, $\sqrt{L} \Delta_{Q1}$, $\sqrt{K} \Delta_Q$ and $\sqrt{L} \Delta_H$ converge to zero as $n$ goes to infinity.
	Note that, absent the normalization of the weight matrices, the term  summed over $n$  in the previous equation would become
	$ A \, \E(p_i p_i')^{-1} \sum_{t = 1}^T H_t^W (I_{d_2} \otimes \E(r_{it} r_{it}')^{-1}) \,(v_{it} \otimes r_{it})  $.
\end{proof}

Assumption \ref{assn:linearAsEquivlt} (\ref{assn:LAE_matrixbounded}) is a condition on the generic functional $a$ applied to the elements of the approximating basis.  The functionals appearing in $\hat{\mu}$ are derived from the linearization of $\mathcal{X}$. They all take the form of an expectation $a(h^W) = \int \lambda_a(v) h^W(v) dF_V(v)$. This is exactly the mean square continuity condition of \cite{n97}, which he shows is sufficient to obtain $\sqrt{n}$ asymptotic normality of linear functionals of linear sieve estimators. We similarly exploit properties implied by this specification of $a$ 
%but instead obtain an intermediary result on the mean square convergence of the term in bracket on the RHS of Lemma \ref{lem:LAE}. We will use this in the following section 
to show that Condition (\ref{assn:LAE_matrixbounded}) of Assumption \ref{assn:linearAsEquivlt} holds and to then obtain the total asymptotic variance matrix of $\sqrt{n} \mathcal{X}_0^{(G)}  [\hat{\mathcal{G}} - \mathcal{G}_0]$. This will be used to show Result \ref{rslt:as.var}.

\begin{assumption}\label{assn:MSE}
	\
	\begin{enumerate}[noitemsep,nolistsep]
		\item \label{assn:MSE_lambda_int} There exists a function $\lambda_a : \R^{Td_2} \mapsto \R^{d_a}$ such that $a(h^W) = \int \lambda_a(v) h^W(v) dF_V(v)$,
		
		\item \label{assn:MSE_approx} There exists $\iota_{a}^K$ such that $|\lambda_a(V) - \iota_{a}^K p^K(V)|_1 =O(K^{-\gamma_3})$, and $(\iota_{aht}^L,\iota_{a\rho t}^L)$ such that as $L \to \infty$,
		$\E\left(||\E \left[ \lambda_a(V) \frac{\partial h^W(V)}{\partial v_{td}} | \xi_t \right] - \iota_{aht}^L r^L(\xi_t)||^2 \right) \to 0$ and
		$\E\left(||\E \left[ \rho_i^W \frac{\partial \lambda_a(V)}{\partial v_{td}} | \xi_t \right] - \iota_{a\rho t}^L r^L(\xi_t)||^2 \right) \to 0$, 
		
		\item \label{assn:MSE_rate} $b_2(K) K^{- \gamma_3} = o(1)$,
		
		\item \label{assn:MSE_var} $\Var(e_i^W | V)$ is bounded on $\mathcal{S}_{V}$.
	\end{enumerate}
\end{assumption}

Define $\tilde{\lambda}_a(v) = A \, \E(p_i p_i')^{-1} \, p^K(v)$, 
$\tilde{\lambda}_{a,td}^{\partial}(\xi_t) = A \, \E(p_i p_i')^{-1} H_{td}^W   \E(r_{it} r_{it}')^{-1} \, r^L(\xi_t)$ and $^{\partial}\tilde{\lambda}_{a,td}(\xi_t) =  A \, \E(p_i p_i')^{-1} \text{d}P_{td}^W  \allowbreak \E(r_{it} r_{it}')^{-1} \,  r^L(\xi_t)$. Assumption \ref{assn:linearAsEquivlt}' is Assumption \ref{assn:linearAsEquivlt} without Condition (\ref{assn:LAE_model_matrixbounded}).

\begin{lemma}\label{lem:MSE_cv}
	Under Assumption \ref{assn:linearAsEquivlt}' and \ref{assn:MSE}, as $K,L \to \infty$,
	$ \E( || e^W ( \tilde{\lambda}_a(V) - \lambda_a(V))||^2 ) \to 0$,
	
	\noindent $ \E\left( || v_{td} \left[\tilde{\lambda}_{a,td}^{\partial}(\xi_t) -  \E\left(\lambda_a(V) \frac{\partial h^W(V)}{\partial v_{td}}| \xi_t \right)\right] ||^2 \right) \to 0$ and $ \E\left( ||v_{td} \left( ^{\partial}\tilde{\lambda}_{a,td} -  \E \left[ \rho^W \frac{\partial \lambda_a(V)}{\partial v_{td}} | \xi_t \right]\right) ||^2\right) \to 0$.
\end{lemma}

Note that Condition (\ref{assn:MSE_approx}) is a sup norm rate condition on both $\lambda_a$ and $\partial \lambda_a / \partial V$, stronger than only assuming $\E(||\lambda_a(V) - \iota_{a}^K p^K(V)||^2) \to 0$ as is assumed in \cite{n97} (Assumption 7) to obtain $\sqrt{n}$ asymptotic normality. Loosely speaking, because the $\text{d}P_t^W$ includes derivatives of the vector of basis functions, left multiplication of $\text{d}P_t^W$ by the matrix $A$, which is the matrix of expectations of $\lambda_a$ multiplied by the functions, will yield under Condition  (\ref{assn:MSE_approx}) an approximation of the derivative of $\lambda_a$.
This will then appear in the asymptotic variance of our estimator when applied to the specific functionals.

\begin{proof}[Proof of Lemma \ref{lem:MSE_cv}]
	
	By Assumption \ref{assn:linearAsEquivlt} (\ref{assn:LAE_sievebasis}), we can assume wlog that $\E(p_i p_i') = I_K$ and $ \E(r_{it} r_{it}') = I_L$. Note that $\tilde{\lambda}_a(v) = A \, p^K(v) = \E\left(\lambda_a(V)p^K(V)'\right)\, p^K(v)$ is the mean square projection of $\lambda_a$ on the functional space spanned by $p^K$. As in the proof of Theorem 3 in Newey (1997), this implies that $\E( || \tilde{\lambda}_a(V) - \lambda_a(V)||^2 ) \leq \E( ||\iota_{a}^K p^K(V) - \lambda_a(V)||^2)$, which gives	
	$\E( || e^W ( \tilde{\lambda}_a(V) - \lambda_a(V))||^2 ) \leq C K^{- \gamma_3} \to 0$, using Assumption \ref{assn:MSE} (\ref{assn:MSE_rate}) and (\ref{assn:MSE_var}).

	Following NPV99, writing $\tilde{\lambda}_{a,td}^{\partial}(\xi_t) = A \, H_{td}^W r^L(\xi_t) = \E\left(\tilde{\lambda}_a(V) \frac{\partial h^W(V)}{\partial v_{td}} \otimes r^L(\xi_t)'\right)r^L(\xi_t)$, and since $\E\left( [\tilde{\lambda}_a(V) - \lambda_a(V)] \frac{\partial h^W(V)}{\partial v_{td}} \otimes r^L(\xi_t)'\right)r^L(\xi_t)$ is the mean square projection of the function $\E\left( \left(\tilde{\lambda}_a(V) - \lambda_a(V)\right) \frac{\partial h^W(V)}{\partial v_{td}}|\xi_t\right)$ on the functional space spanned by $r^L$, by properties of projection we have
	\begin{align*}
		\E\left( \Big| \Big| \tilde{\lambda}_{a,td}^{\partial}(\xi_t) \right. & \left. - \E \left( \lambda_a(V) \frac{\partial h^W(V)}{\partial v_{td}} \otimes r^L(\xi_t)'\right)r^L(\xi_t) \Big| \Big|^2 \right) \\
		& \leq  \E\left( ||[\tilde{\lambda}_a(V) - \lambda_a(V)] \frac{\partial h^W(V)}{\partial v_{td}} ||^2\right) \leq  C \E\left( ||[\tilde{\lambda}_a(V) - \lambda_a(V)] ||^2\right)  \to 0,
	\end{align*}
	where the last inequality holds by Assumption \ref{assn:linearAsEquivlt} (\ref{assn:LAE_fctn}). 
	
	Since $\E \left( \lambda_a(V) \frac{\partial h^W(V)}{\partial v_{td}} \otimes r^L(\xi_t)'\right)r^L(\xi_t)$ is the mean square projection of 
	$\E \left[ \lambda_a(V) \frac{\partial h^W(V)}{\partial v_{td}} | \xi_t \right]$, then
	\begin{align*}
		\E\left( \Big| \Big| \E \left( \lambda_a(V) \frac{\partial h^W(V)}{\partial v_{td}} \otimes r^L(\xi_t)'\right)r^L(\xi_t) \right.  &  \left. - \E \left[ \lambda_a(V) \frac{\partial h^W(V)}{\partial v_{td}} | \xi_t \right]\Big| \Big|^2 \right) \\
		&  \leq \E\left( || \iota_{aht}^L r^L(\xi_t) - \E \left[ \lambda_a(V) \frac{\partial h^W(V)}{\partial v_{td}} | \xi_t \right]||^2\right) \to 0.
	\end{align*}
	This implies $ \E\left( || \tilde{\lambda}_{a,td}^{\partial}(\xi_t) -  \E\left(\lambda_a(V) \frac{\partial h^W(V)}{\partial v_{td}}| \xi_t \right) ||^2 \right) \to 0$, and by Assumption \ref{assn:linearAsEquivlt} (\ref{assn:LAE_varbounded}), 
	$$\E\left( || v_{td} \left[ \tilde{\lambda}_{a,td}^{\partial}(\xi_t) -  \E\left(\lambda_a(V) \frac{\partial h^W(V)}{\partial v_{td}}| \xi_t \right) \right] ||^2 \right) \to 0.$$

	For the third result, we need
	\begin{align*}
		\E\left( || \frac{ \partial\tilde{ \lambda}_a(V)}{\partial v_{td}} -  \frac{\partial \lambda_a(V) }{\partial v_{td}} ||^2\right) \leq & 2  \E\left( ||  \E(\lambda_a(V) p^K(V)')\frac{\partial p^K(V)}{\partial v_{td}} \, - \, \iota_{a}^K \frac{\partial p^K(V)}{\partial v_{td}}  ||^2\right)\\
		& \qquad + 2 \E\left( || \iota_{a}^K \frac{\partial p^K(V)}{\partial v_{td}} \, - \, \frac{\partial \lambda_a(V) }{\partial v_{td}} ||^2\right),
	\end{align*}
	where the second term in the sum converges to $0$ by Assumption \ref{assn:MSE} (\ref{assn:MSE_approx}) and \ref{assn:MSE_rate}. The first term is
	\begin{align*}
		\E\left(\Big| \Big|  \E\left( [\lambda_a(V) - \iota_{a}^K p^K(V)] p^K(V)'\right)\frac{\partial p^K(V)}{\partial v_{td}} \,  \Big| \Big|^2 \right) & \leq b_2(K)^2  ||\E( [\lambda_a(V) - \iota_{a}^K p^K(V) ] p^K(V)') ||^2 \\
		& = O(b_2(K)^2 K^{-2 \gamma_3}) \to 0,
	\end{align*}
	where the last equality is obtained by the same argument as in the previous proof. %follows from taking the orthogonal projection matrix argument to the limit.
	This implies that $\E( || \frac{ \partial\tilde{ \lambda}_a(V)}{\partial v_{td}} -  \frac{\partial \lambda_a(V) }{\partial v_{td}} ||^2) \to 0$.
	Since $^{\partial}\tilde{\lambda}_{a,td}(\xi_t) =  A \,  \text{d}P_{td}^W   \,  r^L(\xi_t) = \E\left( A \rho^W \left( \frac{\partial p^K (v)}{\partial v_{td}} \otimes r_{t}(\xi_t)' \right)\right) r^L(\xi_t),$ we have by property of MSE projection,
	\begin{align*}
		\E\left( \Big| \Big| ^{\partial}\tilde{\lambda}_{a,td}(\xi_t) \right. & \left. - \E \left(  \rho^W  \frac{\partial \lambda_a(V) }{\partial v_{td}} \otimes r^L(\xi_t)'\right)r^L(\xi_t) \Big| \Big|^2 \right) \leq 
		\E\left( |\rho^W|^2 \ \Big| \Big|  \frac{ \partial\tilde{ \lambda}_a(V)}{\partial v_{td}} -  \frac{\partial \lambda_a(V) }{\partial v_{td}} \Big| \Big|^2\right) \to 0
	\end{align*}
	by Assumption \ref{assn:linearAsEquivlt} (\ref{assn:LAE_fctn}) and the result obtained above.
	Moreover, 
	\begin{align*}
		\E\left( \Big| \Big|  \E \left(  \rho^W  \frac{\partial \lambda_a(V) }{\partial v_{td}} \otimes r^L(\xi_t)'\right)r^L(\xi_t) \right. &  \left. - \E \left[\rho^W \frac{\partial \lambda_a(V)}{\partial v_{td}}  | \xi_t \right]\Big| \Big|^2 \right) \\
		&  \leq \E\left( || \iota_{a\rho t}^L r^L(\xi_t) - \E \left[ \rho^W \frac{\partial \lambda_a(V)}{\partial v_{td}}  | \xi_t \right]||^2\right) \to 0,
	\end{align*}
	which implies
	$ \E\left( || ^{\partial}\tilde{\lambda}_{a,td} -  \E \left[ \rho^W \frac{\partial \lambda_a(V)}{\partial v_{td}} | \xi_t \right] ||^2\right) \to 0$ and  $ \E\left( ||v_{td} ( ^{\partial}\tilde{\lambda}_{a,td} -  \E \left[ \rho^W \frac{\partial \lambda_a(V)}{\partial v_{td}} | \xi_t \right]) ||^2\right) \to 0$.
\end{proof}

\subsection{Proofs of Results in Section \ref{sec:asympt_lin_model}}

We first introduce some more notations. We define $\vec{b}_t = (b_{1t}, \, ... , b_{nt}) =  (b_t(\xi_{1t}), \, ... , b_t(\xi_{nt}))$, $\vec{v}_t = (v_{1t},\,.. , \, v_{nt})$, $\vec{V} = (V_1,.. \, , V_n)$, $\vec{x} = (X_1,..\, , X_n)$ and similarly $\vec{z}$. For the results in Section \ref{app:lin_galcase}, we also define the vector $\vec{h}_\varsigma^W  = (h_\varsigma^W(V_1),\, .. \, ,h_\varsigma^W(V_n)) = (h^W(V_1),\, .. \, , h^W(V_n))$ since $V_i \in \mathcal{S}_V \, \forall i \leq n$, and the vector $\vec{\hat{h}}_\varsigma^W = \allowdisplaybreaks (h_\varsigma^W(\hat{V}_1),\, .. \, , h_\varsigma^W(\hat{V}_n))$.

\begin{proof}[Proof of Result \ref{rslt:LAE_model}]
	
	We first focus on the functional $ \mathcal{X}_0^{(2t)}[b_t] = \int_\xi \lambda_{bt}(\xi_{t})b_{t}(\xi_{t}) d F_{\xi_t}(\xi_t) $. \cite{n97} shows in the proof of Theorem 2 (equation (A.7) p164 and the subsequent text) that if $||\mathcal{X}_0^{(bt)}[b_t] || \leq C |b_t|_0$, $\sqrt{n} L^{- \gamma_1} \to 0$, $\Delta_{Q1} = a_1(L) \sqrt{L/n} \to 0$, and $|| \Lambda^{bt}||$ is bounded, then
	$\sqrt{n} \mathcal{X}_0^{(bt)}(\hat{b}_t - b_{0t}) =  \Lambda^{bt} \sum_{i=1}^n r_{it} \otimes v_{it} / \sqrt{n} + \oP(1) $. 
	For all $t \leq T$, under Assumption \ref{assn:LAE_model} (\ref{assn:LAE_model_fctn}) and (\ref{assn:LAE_model_mmt}),  $||\mathcal{X}_0^{(bt)}[b_t] || = ||\E(Q_i  \frac{\partial g}{\partial v_t}(V_i)  b_t(\xi_{it})) || \leq C |b_t|_0$.
	The other required conditions hold by Assumption \ref{assn:LAE_model}.
	
	We now check that the conditions of Assumption \ref{assn:linearAsEquivlt} hold for the functionals $\mathcal{X}_0^{(k)}$ and $\mathcal{X}_0^{(M)}$ applied to the two-step estimators $\hat{k}$ and $\hat{\mathcal{M}}$.
	Under Assumption \ref{assn:cv_g} we have by Lemma \ref{lem:lambdaMinBound} that $\lambda_{\min}(\mathcal{M}(V)) \geq C$.  Together with Assumption \ref{assn:LAE_model} (\ref{assn:LAE_model_fctn}) and (\ref{assn:LAE_model_mmt}), this guarantees that $|| \mathcal{X}_0^{(k)}[\tilde{k}] || = || \E(Q_i \mathcal{M}_0(V_i)^{-1} \tilde{k}(V_i)) ||  \leq C |\tilde{k}|_0$ and $ ||\mathcal{X}_0^{(M)}[\tilde{\mathcal{M}}] || = || \E(Q_i \mathcal{M}_0(V_i)^{-1} \tilde{\mathcal{M}}(V_i) g_0(V_i) ||\leq C |\tilde{\mathcal{M}}|_0$. Hence Assumption \ref{assn:linearAsEquivlt} (\ref{assn:LAE_ALipsch}) holds for each functional. 
	
	Moreover  $\rho^M(X_i, Z_i)  = M_i - \mathcal{M}_0(V_i)$ where $M_i  =  I - \dot{X}_i (\dot{X}_i' \dot{X}_i)^{-1} \dot{X}_i' \, $ if $\dot{X}_i$ is of full rank, or $ M_i  =  I - \dot{X}_i  \dot{X}_i^{ +}$ if not, with $\dot{X}_i^{ +}$ is the Moore Penrose inverse. In either case, $||M_i||_2 \leq 1$ implying $||M_i||_F \leq C$ and $||\mathcal{M}_0(V_i)||_F = ||\E(M_i |V_i)||_F \leq C$, ensuring that $\rho^M$ is a bounded function.
	%\footnote{a bit clumsy here} 
	By the same argument, Assumption  \ref{assn:LAE_model} (\ref{assn:LAE_model_fctn}), (\ref{assn:LAE_model_mmt}) and (\ref{assn:LAE_model_varbounded}), $\rho^{k}(X_i, Z_i)  = [M_i - \mathcal{M}_0(V_i)] g_0(V_i)$ is uniformly bounded. 
	By Assumption \ref{assn:idCFA} and \ref{assn:id_M_GLn}, $k_0(V) =  \mathcal{M}_0^{-1}(V) g_0(V)$ therefore by  Assumption \ref{assn:LAE_model}  (\ref{assn:LAE_model_fctn}), $k$ is twice continuously differentiable, implying that Assumption  \ref{assn:linearAsEquivlt} (\ref{assn:LAE_fctn}) holds  for each functional. Also, $\E(|| e^{M*}||^2|X,Z) = 0$ for all $(X,Z)$, and by Assumption  \ref{assn:LAE_model} (\ref{assn:LAE_model_varbounded}), for all $(X,Z)$, $\E(|| e^{k*}||^2|X,Z) = \E( || M \dot{u} ||^2|X,Z) \leq C$, ensuring that Assumption \ref{assn:linearAsEquivlt} (\ref{assn:LAE_varbounded}) holds for each functional. The other conditions of Assumption \ref{assn:linearAsEquivlt} are direct consequences of Assumption \ref{assn:LAE_model}.
\end{proof}

\begin{proof}[Proof of Result \ref{rslt:as.var}]
	
	Under Assumptions \ref{assn:LAE_model}' and \ref{assn:MSE_model}, Conditions (\ref{assn:MSE_lambda_int}), (\ref{assn:MSE_approx}) and (\ref{assn:MSE_rate}) of Assumption \ref{assn:MSE} holds for $w_i = (M_i \dot{y}_i)_t$ and $w_i = (M_i)_{st}$ associated respectively with $\lambda_k^t$ and $\lambda_M^{s,t}$. We showed  that Assumption \ref{assn:LAE_model} implied that $\rho^M$ and $\rho^k$ are bounded, as well as  $\E(|| e^{M*}||^2|X,Z)$ and $\E(|| e^{k*}||^2|X,Z)$: this implies that for all $V$, $\Var(e^M | V) \leq C$ and $\Var(e^{k} | V) \leq C$. Condition (\ref{assn:MSE_var}) of Assumption \ref{assn:MSE} is also satisfied for our choices of $w_i$, hence we can apply Lemma \ref{lem:MSE_cv}.
	
	We now use Equation (\ref{eq:linearEq}) to construct the asymptotic variance. Define 
	\begin{align*}
		s_i =   [\delta_i \mu_i - \E(\mu \delta)] & +   Q_i^\delta  \dot{u}_i - \lambda_M(V_i)   \VecM(e_i^M) -  \lambda_{k}(V_i) e_i^{k} \\
		& + \sum_{t = 1}^T \left(\E \left[  \frac{\partial \lambda_M(V_i)}{\partial v_{t}} \right.\VecM(\rho_i^M) | \xi_{it} \right] -  \E\left[\lambda_M(V_i) \frac{\partial \mathcal{M}_0(V_i)}{\partial v_{t}}| \xi_{it} \right]  \\
		&  \qquad \qquad \left. + \E \left[ \frac{\partial \lambda_{k}(V_i)}{\partial v_{t}}  \rho_i^{k}| \xi_{it} \right] - \E\left[\lambda_{k}(V_i) \frac{\partial k_0(V_i)}{\partial v_{t}}| \xi_{it} \right]  - \lambda_{bt}(\xi_{it}) \right)\, v_{it},
	\end{align*}
	where, by a convenient abuse of notation, we denote with $  \frac{\partial \lambda_M(V)}{\partial v_{t}} \VecM(\rho_i^M)$ the sum $\sum_{j \leq (T-1)^2} \rho_{i,j}^M \frac{\partial \lambda_M^j(V)}{\partial v_{t}}$ with $\rho_{i,j}^M$ the $j^{th}$ component of the vector $\VecM(\rho^M(X_i, Z_i))$, and similarly for $\lambda_{k}$. We will, in a later step of this proof, simplify the formula for $s_i$.
	
	We conveniently decompose the difference $s_{i,n} - s_i$ as
	\begin{align*}
		s_{i,n} - s_i \ = & \ \lambda_M(V_i)  \VecM(e_i^M) - \Lambda^M (I_{(T-1)^2} \otimes \Theta) \,  \VecM(e_i^M)  \otimes p_i \\
		& + \lambda_{k}(V_i) e_i^{k} -   \Lambda^{k} (I_{T-1} \otimes \Theta) \, e_i^{k} \otimes p_i  \\
		& + \sum_{t = 1}^T \E\left[\lambda_M(V_i) \frac{\partial \mathcal{M}_0(V_i)}{\partial v_{t}}| \xi_{it} \right] v_{it} - \Lambda^M (I_{(T-1)^2} \otimes \Theta) H_t^M  (I_{k_2} \otimes \Theta_1) \, v_{it} \otimes r_{it}\\ 
		& + \sum_{t = 1}^T  \E\left[\lambda_{k}(V_i) \frac{\partial k_0(V_i)}{\partial v_{t}}| \xi_t \right] v_{it} - \Lambda^{k} (I_{T-1} \otimes \Theta) H_t^{k}(I_{k_2} \otimes \Theta_1) \, v_{it} \otimes r_{it} \\
		& + \sum_{t = 1}^T  \Lambda^M (I_{(T-1)^2} \otimes \Theta) \text{d}P_t^M  (I_{k_2} \otimes \Theta_1) \, v_{it} \otimes r_{it} -  \E \left[  \frac{\partial \lambda_M(V_i)}{\partial v_{t}} \VecM(\rho_i^M) | \xi_{it} \right] v_{it}  \\
		& + \sum_{t = 1}^T  \Lambda^{k} (I_{T-1} \otimes \Theta) \text{d}P_t^{k} (I_{k_2} \otimes \Theta_1) \, v_{it} \otimes r_{it} - \E \left[  \frac{\partial \lambda_{k}(V_i)}{\partial v_{t}} \rho_i^{k} | \xi_{it} \right]  \\
		& +  \sum_{t = 1}^T \lambda_{bt}(\xi_{it}) v_{it}  - \Lambda^{bt} (I_{k_2} \otimes \Theta_1) \, v_{it} \otimes r_{it},
	\end{align*}
	where each line in this sum is of one of the three types of elements analyzed in Lemma \ref{lem:MSE_cv}, except for the last line. $\E(||s_{i,n} - s_i||^2)$ is bounded by the sum of the expected squared norms of the elements of each line up to a multiplicative constant. To show that it converges to $0$ as $n$ goes to infinity, we use the fact that Assumption \ref{assn:MSE} holds for each $\lambda_a$, where $\lambda_a$ is a column of either $\lambda_{M}$ or $\lambda_{k}$.
	By Assumption \ref{assn:MSE_model} (\ref{assn:MSE_model_approx}) and Assumption \ref{assn:LAE_model} (\ref{assn:LAE_model_varbounded}), the expected squared norm of the term in the last line also converges to $0$ as $n \to \infty$.
	
	These arguments imply that $\E(||s_{i,n} - s_i||^2) \to 0$. By the proof of Result \ref{rslt:LAE_model} and Assumption \ref{assn:MSE_model} (\ref{assn:MSE_model_approx}), the functions multiplying the residuals appearing in the definition of $s_i$ are all bounded. Together with Assumption \ref{assn:MSE_model} (\ref{assn:MSE_model_var}), this guarantees $\E([s_i'c]^2) < \infty$. For a constant vector $c \in \R^{d_x}$, 
	$|c' [ \E(s_{in}s_{in}') - \E(s_{i}s_{i}')] c | \leq  \E([s_{in}'c - s_{i}'c]^2) + 2 \E([s_i'c]^2)^{1/2} \E([s_{in}'c - s_{i}'c]^2)^{1/2}$. Hence, $|c' [ \E(s_{in}s_{in}') - \E(s_{i}s_{i}')] c | \to 0$ for all $c$, implying $\E(s_{in}s_{in}') - \E(s_{i}s_{i}')  \to 0$. That is, $\Omega  \to \Var(s_i)$ as $n \to \infty$.
	
	We can now simplify the formula for $s_i$ using the primitives of the model. Indeed, note that
	\begin{align*}
		& \lambda_M(V_i)  \VecM(e_i^M) + \lambda_{k}(V_i) e_i^{k} \, = \, \E(Q_i^\delta | V_i) \mathcal{M}_0(V_i)^{-1} M_i \dot{u}_i ,\\
		& \lambda_M(V_i) \frac{\partial \mathcal{M}_0(V_i)}{\partial v_{t}}+ \lambda_{k}(V_i) \frac{\partial k_0(V_i)}{\partial v_{t}}  \, = \, \E(Q_i^\delta | V_i) \, \frac{ \partial g_0(V_i)}{\partial v_t},\\
		&   \frac{\partial \lambda_M(V_i)}{\partial v_{t}} \VecM(\rho_i^M) +  \frac{\partial \lambda_{k}(V_i)}{\partial v_{t}} \rho_i^{k}  \, = \, - \E(Q_i^\delta | V_i)  \mathcal{M}_0(V_i)^{-1} (M_i -  \mathcal{M}_0(V_i)) \frac{ \partial g_0(V_i)}{\partial v_t},
	\end{align*}
	and since
	$ \lambda_{bt}(\xi_{it}) = - \E\left(Q_i^\delta \frac{ \partial g_0(V_i)}{\partial v_t}  \,  | \, \xi_{it} = \xi_t \right)$,
	we obtain  
	\begin{align*}
		s_i  = &  [\delta_i \mu_i - \E(\mu \delta)] + Q_i^\delta  \dot{u}_i - \E(Q_i^\delta | V_i) \mathcal{M}_0(V_i)^{-1} M_i \dot{u}_i \\
		& + \sum_{t = 1}^T  \E\left( \left[ Q_i^\delta  -   \E(Q_i^\delta | V_i)  \mathcal{M}_0(V_i)^{-1} M_i \right] \frac{ \partial g_0(V_i)}{\partial v_t}  | \xi_{it} \right)   v_{it},\\
		= &   [\delta_i \mu_i - \E(\mu \delta)] + \tilde{Q}_i^\delta   \dot{u}_i +  \sum_{t = 1}^T  \E\left( \tilde{Q}_i^\delta \frac{ \partial g_0(V_i)}{\partial v_t}  | \, \xi_{it} \right)  \, v_{it}.
	\end{align*}
	Thus $\Var(s_i) = \Omega_0$ and $\Omega  \to_{n \to \infty} \Omega_0$. By $\Omega_0 \geq C I_{d_x}$, we obtain $\Omega^{-1/2}  \to_{n \to \infty} \Omega_0^{-1/2} \leq C^{-1/2} I_{k_x}$.
	
	\medskip
	
	We again use Lemma \ref{lem:MSE_cv} to prove that $||\Lambda^M||$, $||\Lambda^{k}||$, and $||\Lambda^{bt}||$ are bounded. Indeed, since wlog we can assume $\Theta_1 = I_L$, we have $|| \Lambda^{bt}||^2  = \tr( \Lambda^{bt}  \Lambda^{bt \, \prime} ) = \tr( \Lambda^{bt} (I_{k_2} \otimes \Theta_1)  \Lambda^{bt \, \prime}  )$. Using the notation of Lemma \ref{lem:MSE_cv} with $\tilde{\lambda}_{bt}^j(\xi) = \E(\lambda_{bt}^j(\xi_{it})r^L(\xi_{it})') \, r^l(\xi)$, this gives
	$
	|| \Lambda^{bt}||^2 
	= \tr \left( \sum_{j \leq T-1} \E\left[  \tilde{\lambda}_{bt}^j(\xi_{it})  \tilde{\lambda}_{bt}^j(\xi_{it})' \right]   \right)
	$.
	However, by Lemma \ref{lem:MSE_cv}, we know that $\E(||\tilde{\lambda}_{bt}^j(\xi_{it}) - \lambda_{bt}^j(\xi_{it})||^2)  \to 0$, under Assumption \ref{assn:MSE_model}. The same reasoning we used for $\E(s_{in}s_{in}') - \E(s_{i}s_{i}')$ applies, and since $\lambda_{bt}^j(.)$ is a bounded function, we obtain
	$\E\left(\tilde{\lambda}_{bt}^j(\xi_{it})\tilde{\lambda}_{bt}^j(\xi_{it})'\right) \to_{n \to \infty} \E\left(\lambda_{bt}^j(\xi_{it})\lambda_{bt}^j(\xi_{it})'\right)$.
	Therefore 
	$$
	|| \Lambda^{bt}||^2 \to_{n \to \infty} \tr\left(\sum_{j \leq T-1} \E\left(\lambda_{bt}^j(\xi_{it})\lambda_{bt}^j(\xi_{it})'\right) \right) \leq C.
	$$
	Hence $|| \Lambda^{bt}||^2$ is bounded. 
	
	The same arguments applied to the functions $\lambda_M$, $\lambda_{k}$, as well as to
	$$
	\E \left[  \frac{\partial \lambda_M(V_i)}{\partial v_{t}} \VecM(\rho_i^M) | \xi_{it} \right], \ \E\left[\lambda_M(V_i) \frac{\partial \mathcal{M}_0(V_i)}{\partial v_{t}}| \xi_{it} \right], \ \E \left[ \frac{\partial \lambda_{k}(V_i)}{\partial v_{t}}  \rho_i^{k}| \xi_{it} \right] \text{ and }\E\left[\lambda_{k}(V_i) \frac{\partial k_0(V_i)}{\partial v_{t}}| \xi_{it} \right],
	$$ 
	would imply that $||\Lambda^M||$, $|| \Lambda^{k}||$, $|| \Lambda^M (I_{(T-1)^2} \otimes \Theta) (H_t^M - \text{d}P_t^M)||$, and $|| \Lambda^{k} (I_{T-1} \otimes \Theta) (H_t^{k} - \text{d}P_t^{k})||$ are bounded.
\end{proof}

\subsection{Proofs of Results in Section \ref{sec:as_miou_together}}

The first remaining  step in obtaining asympotic normality is applying Result \ref{rslt:X_linear}, i.e, check that  Assumption \ref{assn:linearization} holds. Condition (\ref{assn:linearization_SE}) of  Assumption \ref{assn:linearization} is a stochastic equicontinuity condition. We follow Section 4 in \cite{clvk03} (CLVK thereafter) in our choice of the space $\mathcal{H}$, as they establish easy-to-check conditions implying stochastic equicontinuity in some spaces, and take $\mathcal{H}$ to be $\mathcal{H}_{c,c'}^\varrho$, defined in Section \ref{sec:as_miou_together}.
Our choice of the norm on $\mathcal{H}$, $||.||_{\mathcal{H}}$, is justified by Condition (\ref{assn:linearization_2ndorder}). The functional $\mathcal{X}$ is a function of $\mathcal{M}$, $k$ and $(b_t)_{t \leq T}$, where $\mathcal{M}$ and $k$ are composed with $(b_t)_{t \leq T}$. These compositions imply, as was clear in the computations, that the linearization will involve the first order partial derivatives of $\mathcal{M}_0$ and $k_0$. It also implies that the difference between $\mathcal{X}(\mathcal{G})$ and $\mathcal{X}(\mathcal{G}) - \mathcal{X}^{(G)}(\mathcal{G}_0) [\mathcal{G} - \mathcal{G}_0] $ can be easily controlled by, among other terms, the distance between first order partial derivatives of these functions.
%\footnote{Maybe more details here? Since not even Mammen et al 2016 mentions this.}
A natural norm on $\mathcal{H}_{c,c'}^\varrho$ is therefore $||\mathcal{G}||_{\mathcal{H}} \, = \,  \sum_{j=1}^{(T-1)^2}|\mathcal{M}_j - \mathcal{M}_{0,j}|_{1}^\varsigma + \sum_{j = 1}^{T-1}|k_j - k_{0,j}|_{1}^\varsigma   +  \sum_{t \leq T} \sum_{j =1}^{d_2} ||b_{t,j} - b_{0,t,j}||_{\infty} $ where by an abuse of notation $\mathcal{M}_j$ is the $j^{th}$ component of $\VecM(\mathcal{M})$.

\begin{result}\label{rslt:stochEqderiv}
	Defining $\mathcal{H} = \mathcal{H}_{c,c'}^\varrho$ and $||.||_{\mathcal{H}} $ as described, if Assumption \ref{assn:LAE_model}' (\ref{assn:LAE_model_fctn}) and (\ref{assn:LAE_model_mmt}), Assumption \ref{assn:MSE_model} (\ref{assn:MSE_model_var}) and Assumption \ref{assn:SE} hold, then Assumption \ref{assn:linearization} (\ref{assn:linearization_SE}) and (\ref{assn:linearization_2ndorder}) hold.
\end{result}

\begin{proof}[Proof of Result \ref{rslt:stochEqderiv}]
	We start by showing that the stochastic equicontinuity condition, Condition (\ref{assn:linearization_SE}), holds. Lemma 1 of CLVK shows that if $(W_i)_{i =1}^n$ is i.i.d, Assumption \ref{assn:linearization} (\ref{assn:linearization_SE}) holds if : (A) the class $\mathcal{F} \, = \,  \{\chi(W, \mathcal{G})\, : \, \mathcal{G} \in \mathcal{H}_{c,c'}^\varrho \}$ is $\mathbb{P}$-Donsker, i.e it satisfies $\int_0^{\infty} \sqrt{\log N_{[]}(\epsilon, \mathcal{F}, || .||_{L_2(P)})} d\epsilon < \infty$, where $N_{[]}(\epsilon, \mathcal{F}, || .||_{L_2(P)})$ is the covering number with bracketing,
	and if (B) $\chi(., \mathcal{G})$ is $L_2(P)$ continuous at $\mathcal{G}_0$, that is, $\E(|| \chi(W_i, \mathcal{G}) - \chi(W_i, \mathcal{G}_0)||^2) \to 0$ as $||\mathcal{G} - \mathcal{G}_0||_{\mathcal{H}} \to 0$. We now check that each of these conditions is satisfied under our assumptions.
	
	Condition (A):
	We use $j,l$ to index components of vectors. As in CLVK, it is enough to prove that $\mathcal{F}_l \, = \,  \{\chi_l(W, \mathcal{G})\, : \, \mathcal{G} \in \mathcal{H}_{c,c'}^\varrho \}$ is $\mathbb{P}$-Donsker for each component $l$ of $\chi(.)$.
	Recall that $\chi(W_i, \mathcal{G}) \, = \, Q_i^\delta \, \mathcal{M} \left( \tau \left[ (x_{it}^{en} - b_t(\xi_{it}))_{t \leq T}\right] \right)^{-1} \allowbreak k\left( \tau \left[  (x_{it}^{en} - b_t(\xi_{it}))_{t \leq T} \right] \right)$. We examine $\chi(W_i, \mathcal{G}) - \chi(W_i, \mathcal{G}_0)$ and write, by an abuse of notation and only in this proof, $V_i = (x_{it}^{en} - b_t(\xi_{it}))_{t \leq T}$ and $V_{0,i} = (x_{it}^{en} - b_{0,t}(\xi_{it}))_{t \leq T}$. Note that $V_0 = \tau(V_0)$. We decompose 
	\begin{align}
		\chi(W_i, \mathcal{G})&  - \chi(W_i, \mathcal{G}_0) = Q_i^\delta  \mathcal{M}(\tau[V])^{-1} [k(\tau[V]) - k_0(\tau[V])]  \nonumber \\
		&  + Q_i^\delta  \mathcal{M}(\tau[V])^{-1} [\mathcal{M}_0(\tau[V]) - \mathcal{M}(\tau[V]) ] \mathcal{M}_0(\tau[V])^{-1}k_0(V_0)   \label{eq:decompoz_SE} \\
		& +  Q_i^\delta  \mathcal{M}_0(\tau[V])^{-1} [\mathcal{M}_0(V_0) - \mathcal{M}_0(\tau[V]) ] \allowbreak \mathcal{M}_0(V_0)^{-1}k_0(V_0) \nonumber \\
		& +  Q_i^\delta \mathcal{M}(\tau[V])^{-1} [k_0(\tau[V]) - k_0(V_0)]. \nonumber
	\end{align}
	Since $(\mathcal{G}, \mathcal{G}_0) \in \mathcal{H}_{c,c'}^\varrho \times \mathcal{H}_{c,c'}^\varrho$, the norms of each functional and its first order derivatives are bounded. Moreover the derivatives of $\tau$ are bounded. This implies that $||\mathcal{M}_0(V_0) - \mathcal{M}_0(\tau[V])|| \leq c ||V_0 - V ||$, and the same result holds for $k$.  Hence, using (\ref{eq:decompoz_SE}), $|\chi_l(W_i, \mathcal{G}) - \chi_l(W_i, \mathcal{G}_0)| \leq ||\chi(W_i, \mathcal{G}) - \chi(W_i, \mathcal{G}_0)|| \leq  C ||Q_i^\delta|| \, ( \sum_{j=1}^{(T-1)^2}|\mathcal{M}_j - \mathcal{M}_{0,j}|_0^\varsigma + \sum_{j = 1}^{T-1}|k_j - k_{0,j}|_0^\varsigma  +  \sum_{t \leq T} \sum_{j =1}^{d_2} ||b_{t,j} - b_{0,t,j}||_{\infty} ) $, where the constant $C$ depends on $c$ and $c'$.
	By Assumption \ref{assn:SE}, $\E(||Q_i^\delta||^2) < \infty$ which implies by 
	the proof of Theorem 3 of CLVK that 
	$$ N_{[]}(\epsilon, \mathcal{F}, || .||_{L_2(P)}) \leq  N\left(\epsilon/c^Q, \mathcal{C}_c^{\varrho}(\mathcal{S}_{V}^\varsigma), || .||_{\infty}\right)^{(T-1)^2 + T-1}   \Pi_{t \leq T} N\left(\epsilon/c^Q, \mathcal{C}_c^{\varrho}(\mathcal{S}_{\xi_t}), || .||_{\infty}\right)^{d_2},$$
	where $N(\epsilon, \mathcal{C}_c^{\varrho}(\mathcal{S}_W), || .||_{\infty})$ denotes the covering number of the class $\mathcal{C}_c^{\varrho}(\mathcal{S}_W)$, and $c^Q = 2 [(T-1)^2 + T-1 + Td_2] \,  \E(||Q_i^\delta||^2)$, is the size of the brackets constructed in CLVK.
	
	It is known that for $\mathcal{S}_W$ a bounded subset of $\R^k$, $\log N(\epsilon, \mathcal{C}_c^{\varrho}(\mathcal{S}_W), || .||_{\infty}) \leq \epsilon^{- k/\varrho}$. By Assumption \ref{assn:SE}, $\varrho > \max(Td_2,d_z + d_1)/2$, which implies that $\mathcal{F}_j$ is $\mathbb{P}$-Donsker. Therefore, Condition (A) is satisfied.
	
	Condition (B) : By $\E(||Q_i^\delta||^2) \leq C$ and using once more the decomposition given by (\ref{eq:decompoz_SE}), 
	$\E(|| \chi(W_i, \mathcal{G}) - \chi(W_i, \mathcal{G}_0)||^2) \leq C ||\mathcal{G} - \mathcal{G}_0||_{\mathcal{H}}^2$ which gives the wanted result.
	
	\medskip
	
	We now show that Assumption \ref{assn:linearization} (\ref{assn:linearization_2ndorder}) holds. This condition is on the remainder of the linearization, $|| \mathcal{X}(\mathcal{G}) - \mathcal{X}^{(G)}(\mathcal{G}_0) [\mathcal{G} - \mathcal{G}_0]\, ||$. Note that 
	\begin{align*}
		\mathcal{X} &(\mathcal{G}) - \mathcal{X}^{(G)}(\mathcal{G}_0) [\mathcal{G} - \mathcal{G}_0] \\
		= & \E( Q_i^\delta  \mathcal{M}(\tau[V])^{-1} [k(\tau[V]) - k_0(\tau[V])]) - \E( Q_i^\delta \mathcal{M}_0(V_0)^{-1} [k(V_0) - k_0(V_0)])\\
		& + \E(Q_i^\delta \mathcal{M}(\tau[V])^{-1} [k_0(\tau[V]) - k_0(V_0)]) - \E( Q_i^\delta \mathcal{M}_0(V_0)^{-1} \frac{\partial k_0}{\partial V_0}(V_0) [V - V_0]) \\
		& + \E( Q_i^\delta  \mathcal{M}(\tau[V])^{-1} [\mathcal{M}_0(\tau[V]) - \mathcal{M}(\tau[V]) ] \mathcal{M}_0(\tau[V])^{-1}k_0(V_0)) \\
		& \qquad  \qquad \qquad \qquad \qquad \qquad -  \E(Q_i^\delta \mathcal{M}_0(V_0)^{-1} \, [\mathcal{M}_0(V_0) - \mathcal{M}(V_0)] \mathcal{M}_0(V_0)^{-1} k_0(V_0))
		\\
		& + \E(Q_i^\delta  \mathcal{M}_0(\tau[V])^{-1} [\mathcal{M}_0(V_0) - \mathcal{M}_0(\tau[V]) ] \mathcal{M}_0(V_0)^{-1}k_0(V_0) \\
		& \qquad \qquad \qquad \qquad \qquad \qquad - \E([ k_0(V_0)' \mathcal{M}_0(V_0)' \otimes (Q_i^\delta \mathcal{M}_0(V_0)^{-1}) ] \, \VecM(\frac{\partial \mathcal{M}_0}{\partial V_0}(V_0)) [V_0 - V]).
	\end{align*}
	We use this decomposition and we bound each line separately. We show here how to find upper bounds for the first and second lines, both of which will be less than $||\mathcal{G} - \mathcal{G}_0||_{\mathcal{H}}^2$ up to a multiplicative constant. The upper bounds for the third and forth lines of this decomposition can be obtained in a similar fashion. By the triangular inequality, this will give $|| \mathcal{X}(\mathcal{G}) - \mathcal{X}^{(G)}(\mathcal{G}_0) [\mathcal{G} - \mathcal{G}_0]\, || \leq C ||\mathcal{G} - \mathcal{G}_0||_{\mathcal{H}}^2$, as desired.
	First,
	\begin{align*}
		\big| \big| \E( &Q_i^\delta  \mathcal{M}(\tau(V))^{-1} [k(\tau[V]) - k_0(\tau(V))]) - \E( Q_i^\delta \mathcal{M}_0(V_0)^{-1} [k(V_0) - k_0(V_0)])  \big| \big|  \\
		& = \big| \big| \E( Q_i^\delta  \mathcal{M}(\tau[V])^{-1} [\mathcal{M}_0(\tau[V]) - \mathcal{M}(\tau[V]) ] \mathcal{M}_0(\tau[V])^{-1} [k(\tau[V]) - k_0(V)])  \big| \big| \\
		& \qquad \qquad  \qquad +  \big| \big| \E( Q_i^\delta  \mathcal{M}_0(\tau[V])^{-1} [\mathcal{M}_0(V_0) - \mathcal{M}_0(\tau[V]) ] \mathcal{M}_0(V_0)^{-1} [k(\tau[V]) - k_0(\tau[V])])  \big| \big|\\
		& \qquad \qquad  \qquad  \qquad \qquad \qquad \qquad + \big| \big| \E( Q_i^\delta \mathcal{M}_0(V_0)^{-1} [(k- k_0)(\tau[V]) - (k - k_0)(V_0)])  \big| \big| \\
		& \leq C \E(||Q_i^\delta||) \left( (\sum_{j=1}^{(T-1)^2}|\mathcal{M}_j - \mathcal{M}_{0,j}|_0^\varsigma)(\sum_{j = 1}^{T-1}|k_j - k_{0,j}||_0^\varsigma) \right. \\
		& \qquad \qquad  \qquad +   ( \sum_{t \leq T} \sum_{j =1}^{d_2} ||b_t - b_{0,t}||_{\infty} )( \sum_{j = 1}^{T-1}|k_j - k_{0,j}|_0^\varsigma  ) \\
		& \qquad \qquad  \qquad \qquad \qquad \qquad \qquad + \left. (  \sum_{j = 1}^{T-1}|k_j - k_{0,j}|_1^\varsigma )(  \sum_{t \leq T} \sum_{j =1}^{d_2} ||b_{t,j} - b_{0,t,j}||_{\infty}) \right) \leq C ||\mathcal{G} - \mathcal{G}_0||_{\mathcal{H}}^2.
	\end{align*}
	As for the second line of the decomposition of $\mathcal{X}(\mathcal{G}) - \mathcal{X}^{(G)}(\mathcal{G}_0) [\mathcal{G} - \mathcal{G}_0] $, we write
	\begin{align*}
		\big| \big| \E( & Q_i^\delta \mathcal{M}(\tau[V])^{-1} [k_0(\tau(V)) - k_0(V_0)]) - \E( Q_i^\delta \mathcal{M}_0(V_0)^{-1} \frac{\partial k_0}{\partial V_0}(V_0) [V - V_0])\big| \big|  \\
		& = \big| \big|  \E( Q_i^\delta  \mathcal{M}(\tau[V])^{-1} [\mathcal{M}(V_0) - \mathcal{M}(\tau[V]) ] \mathcal{M}(V_0)^{-1} [k_0(\tau[V]) - k_0(V_0)])  \big| \big|  \\
		& \qquad \qquad  \qquad + \big| \big|  \E( Q_i^\delta  \mathcal{M}(V_0)^{-1} [\mathcal{M}_0(V_0) - \mathcal{M}(V_0) ] \mathcal{M}_0(V_0)^{-1} [k_0(\tau(V)) - k_0(V_0)] ) \big| \big| \\
		&  \qquad \qquad  \qquad  \qquad \qquad \qquad \qquad + \big| \big|  \E( Q_i^\delta  \mathcal{M}_0(V_0)^{-1} [k_0(\tau[V]) - k_0(V_0)  - \frac{\partial k_0}{\partial V_0}(V_0) [V - V_0] ]) \big| \big|  \\
		& \leq C \E(||Q_i^\delta||) \left( ( \sum_{t \leq T} \sum_{j =1}^{d_2} ||b_{t,j} - b_{0,t,j}||_{\infty} )^2 \right. \\
		& \qquad \qquad  \qquad +   ( \sum_{t \leq T} \sum_{j =1}^{d_2} ||b_{t,j} - b_{0,t,j}||_{\infty} )(\sum_{j=1}^{(T-1)^2}|\mathcal{M}_j - \mathcal{M}_{0,j}|_0^\varsigma)\\
		& \qquad \qquad  \qquad \qquad \qquad \qquad \qquad + \left. ( \sum_{t \leq T} \sum_{j =1}^{d_2} ||b_{t,j} - b_{0,t,j}||_{\infty} )^2  \right) \leq C ||\mathcal{G} - \mathcal{G}_0||_{\mathcal{H}}^2,
	\end{align*}
	where the inequality for the third term in this equation holds by Assumption \ref{assn:LAE_model} (\ref{assn:LAE_model_fctn}), by the Jacobian of $\tau$ being the identity matrix when evaluated at $V_0$ (since $V_0 \in \mathcal{S}_V$) and by the second order derivative of $\tau$ being bounded.
\end{proof}

To apply Result \ref{rslt:X_linear}, it remains to check that  Condition (\ref{assn:linearization_Rate})  of Assumption \ref{assn:linearization} holds. It is a condition on the convergence rate of the estimators $\hat{b}_t$, $\hat{k}$ and $\hat{\mathcal{M}}$. The rate of convergence of $||\hat{b}_{t,j} - b_{0,t,j}||_{\infty}$ for all $(t,j)$ is given by Equation (\ref{eq:V_suprate}), see the Proof of Result \ref{rslt:NPV1stStep}. The rates of convergence of $|\hat{k}_j - k_{0,j}|_1^\varsigma$ and $|\hat{\mathcal{M}}_j - \mathcal{M}_{0,j}|_1^\varsigma$ are given by the following Corollary on  the rate of convergence of the first order partial derivative of the two-step nonparametric estimator.
Its proof follows from the proof of  Result \ref{rslt:IN2stStep}.
%\footnote{ Add comments on approximation of derivatives}

\begin{corollary}\label{coro:IN2stStep_deriv}
	Under Assumption \ref{assn:NPV1stStep} and \ref{assn:IN2stStep}, if $\sup_{\mathcal{S}_{V}} | \partial h^W(V) - \pi_W^{K \prime} \partial p^K(V) | \allowbreak \leq C K^{- \gamma_2}$,
	$$\sup_{V \in \mathcal{S}_{V}} | \partial \hat{h}^W(V) - \partial h^{W}(V) | = \OP \left( b_2(K) (K/n + K^{-2 \gamma_2} +  \Delta_n^2 b_2(K)^2)^{1/2} \right).$$
\end{corollary}

The conditions required to apply this corollary must be adapted to the extended support, as we did for other results. Assumption \ref{assn:LAE_model}' already includes most of these conditions, specifying an approximation rate of $\mathcal{M}_0$ and $k_0$ over the extended support and defining the rates $b_1$, $b_2$ and $b_3$ as bounds on sup-norms of derivatives of $P^K$ defined over the extended support. Assumption \ref{assn:LAE_model}'' includes the only needed modification, adding a control on the approximation rate of the first order partial derivatives. Result \ref{rslt:X_linear} thus applies under Assumptions  \ref{assn:idCFA}, \ref{assn:cv_g}, \ref{assn:LAE_model}'', \ref{assn:MSE_model} and \ref{assn:SE}.

\begin{proof}[Proof of Result \ref{rslt:as.normlt}]
	Define $\Phi =  \mathbb{P}(\det(\dot{X}_i' \dot{X}_i) > \delta_0)^{-1}$ and $\phi_n = \frac{1}{n} \sum_{i=1}^n \delta_i$. The estimator of the average effect $\E(\mu| \delta)$ is
	$\hat{\mu}  \,  =  \, \hat{\mu}^\delta / \phi_n$. We define $\Sigma_0 = \Var \left( (s_i',\delta_i)' \right)$. Since $\Omega_0 > 0$ by Assumption \ref{assn:MSE_model}, then Assumption \ref{assn:as.normlt} guarantees that $\Sigma_0 > 0$.
	We decompose
	\begin{align}
		\sqrt{n}  \phi_n [\hat{\mu} - \E(\mu|\delta)] & = \sqrt{n}   [\hat{\mu}^\delta -  \E(\mu \delta)] + \sqrt{n} \frac{ \E(\mu \delta)}{\Phi} [\Phi - \phi_n], \nonumber \\
		& = \sqrt{n} \, \left( I_{d_x} ; \E(\mu | \delta) \right) \,  \left[ \left( \begin{smallmatrix}\hat{\mu}^\delta  \\ \phi_n
		\end{smallmatrix} \right) -  \left( \begin{smallmatrix} \E(\mu \delta)  \\ \Phi
		\end{smallmatrix} \right) \right],
		\label{eq:mu_decompoz_ratio}
	\end{align}
	where $ \left( I_{d_x} ; \E(\mu | \delta) \right)$ is of size $d_x \times (d_x + 1)$. 
	
	We first show
	\begin{equation}\label{eq:as_norm_vector}
		\sqrt{n} \Sigma_0^{-1/2} \left[ \left( \begin{smallmatrix}\hat{\mu}^\delta  \\ \phi_n
		\end{smallmatrix} \right) -  \left( \begin{smallmatrix} \E(\mu \delta)  \\ \Phi
		\end{smallmatrix} \right) \right] \to^d  \mathcal{N}(0, I_{d_x +1}) .
	\end{equation}

	By Result \ref{rslt:linear_miou_num}, $\sqrt{n} \left[ \left( \begin{smallmatrix}\hat{\mu}^\delta  \\ \phi_n
	\end{smallmatrix} \right) -  \left( \begin{smallmatrix} \E(\mu \delta)  \\ \Phi
	\end{smallmatrix} \right) \right]   = \frac{1}{\sqrt{n}}  \sum_{i=1}^n \left( \begin{smallmatrix}  s_{i,n}  \\ \delta_i - \Phi
	\end{smallmatrix} \right) + o_{\mathbb{P}}(1)$.
	
	Define $\Sigma_n = \Var((s_{i,n}',\delta_i - \Phi)')$.
	We obtain the asymptotic distribution in two steps. We prove first that $\sqrt{n} \Sigma_n \left[ \left( \begin{smallmatrix}\hat{\mu}^\delta  \\ \phi_n
	\end{smallmatrix} \right) -  \left( \begin{smallmatrix} \E(\mu \delta)  \\ \Phi
	\end{smallmatrix} \right) \right]  \to^d  \mathcal{N}(0, I_{d_x +1}) $. We show in a second step that $\Sigma_n \to \Sigma_0$, which will yield the desired result. We follow \cite{npv99} in proving a Lindeberg condition for $ c ' \Sigma_n (s_{i,n}',\delta_i - \Delta)'$  for any constant vector $c \in \R^{d_x + 1}$ such that $||c|| = 1$. More precisely if for all such $c$, $\frac{1}{\sqrt{n}} c' \Sigma_n^{-1/2} \sum_{i=1}^n (s_{i,n}',\delta_i - \Phi)' \to^d \mathcal{N}(0,1)$, this first result will be a consequence the Cram\'er-Wold theorem.
	Write $S_{i,n} = c' \Sigma_n (s_{i,n}',\delta_i - \Phi)'$, then $\E(S_{i,n})=0$ and $\Var(S_{i,n}) = 1$.  Asymptotic normality is a consequence of the CLT, provided that the Lindeberg condition holds for $S_{i,n}$, i.e, for any $\epsilon > 0$, $\E( S_{i,n}^2 \1(|S_{i,n}| > \epsilon \sqrt{n})) \to 0 $.
	Note that by $\rho^M$ and $\rho^k$ bounded under Assumption \ref{assn:LAE_model},  $\E(\dot{u}_i |X_i, Z_i) = 0$ and $\E[|| \dot{u}_i||^4 | X_i = X, Z_i = Z] \leq C$ for all $(X,Z)$, then $\E[|| e_i^{k}||^4 | V_i = V] \leq C$ and $\E[|| \VecM(e_i^{M})||^4 | V_i = V] \leq C$.
	Fix $\epsilon > 0$. We normalize $\Theta = I_K$ and $\Theta_1 = I_L$, and obtain
	\begin{align*}
		n \epsilon^2 \E(& S_{i,n}^2  \1(|S_{i,n}| > \epsilon \sqrt{n})) \ \leq \ \E(S_{i,n}^4 \1(|S_{i,n}| > \epsilon \sqrt{n})) \ \leq \ \E(S_{i,n}^4 ), \\
		&\leq  C \left(  \, \E[||\mu_i - \E(\mu)||^4] +  \, \E[||Q_i^\delta  \dot{u}_i||^4] +  ||\Lambda^M ||^4 \, \E[||\VecM(e_i^M)  \otimes p_i||^4] + || \Lambda^{k} ||^4 \, \E[||e_i^{k}  \otimes p_i||^4] \right. \\
		& \ \ \ + \left. \sum_{t \leq T} || [\Lambda^M (H_t^M - \text{d}P_t^M) + \Lambda^{k} (H_t^{k} - \text{d}P_t^{k}) +  \Lambda^{bt}] \, ||^4 \, \E[||v_{it} \otimes r_{it}||^4] + \E[|| \delta_i - \Phi ||^4] \right).
	\end{align*}
	We can bound $\E(||v_{it} \otimes r_{it}||^4) = \E(||r_{it}||^4 ||v_{it}^4||) \leq C \E(||r_{it}||^4) $ by Assumption \ref{assn:as.normlt}, and $\E(||r_{it}||^4) \leq a_1(L)^2 \tr(\E(r_{it}' r_{it})) \allowbreak = a_1(L)^2 L$. Similarly,  by Assumption \ref{assn:as.normlt}, $\E(||e_i^{k}  \otimes p_i||^4) = O(b_1(K)^2 K)$ and $\E(||\VecM(e_i^M)  \otimes p_i||^4) = O(b_1(K)^2 K)$. Therefore, by Result \ref{rslt:as.var}, $n \epsilon^2 \E( S_{i,n}^2  \1(|S_{i,n}| > \epsilon \sqrt{n})) = O(b_1(K)^2 K +  a_1(L)^2 L)$.
	
	Assumption \ref{assn:LAE_model} (\ref{assn:LAE_model_rate}) implies $\Delta Q = o(1)$ and $\Delta Q_1 = o(1)$, in turn implying $\sqrt{K/n} \, b_1(K) \to 0$ and $\sqrt{L/n} \, a_1(L) \to 0$. Therefore the condition  $\E( S_{i,n}^2 \1(|S_{i,n}| > \epsilon \sqrt{n})) \to 0 $ holds.
	
	\medskip
	
	The second step to obtain (\ref{eq:as_norm_vector})  requires $\Sigma_n \to \Sigma_0$. This is a consequence ofof the proof of Result \ref{rslt:as.var}.	
	Now we can use (\ref{eq:mu_decompoz_ratio}) with  (\ref{eq:as_norm_vector}) to obtain by a delta method argument
	$$\sqrt{n}  \phi_n [\hat{\mu} - \E(\mu|\delta)] \to^d \mathcal{N} \left( 0,  \left( I_{d_x} ; \E(\mu | \delta) \right) \Sigma_0 \left( I_{d_x} ; \E(\mu | \delta) \right)' \right), $$
	hence 
	$\sqrt{n} [\hat{\mu} - \E(\mu|\delta)] \to^d \mathcal{N} \left( 0, \Phi^{-2} \left( I_{d_x} ; \E(\mu | \delta) \right) \Sigma_0 \left( I_{d_x} ; \E(\mu | \delta) \right)' \right). $
\end{proof}

\section{Monte Carlo Simulations}

We explore the properties of our multi-step estimator with Monte Carlo simulations when the model is a specific case of Model (\ref{eq:model_for_estim}) studied in the asymptotic analysis. All tables and figures are displayed in the Appendix.
More specifically, the outcome equation and specification for the random coefficients, covariates, instruments and time-varying disturbances of our data generating process (dgp) are
\begin{align*}
	&y_{it}  =  \, x_{it}^{ex} \, \mu_i^{ex} + x_{it}^{en} \, \mu_i^{en}  + \alpha_i + c \left[ \underbrace{f_t(V_i) + u_{it}}_{= \, \epsilon_{it}}\right] , \  i=1..n, \ t \leq T,\\
	&\mu_i^{ex} \, = \, \mathcal{U}[0,1],  \ \mu_i^{en} \, = \, \mathcal{U}[0,1], \ \mu_i^{ex} \indep \mu_i^{en}\\
	&\text{and } \, \forall \,  t \leq T, \ x_{it}^{ex} \, = \, \sqrt{\mu_i^{ex}} \, \tilde{x}_{it}^{ex}, \quad z_{it} \, = \, \sqrt{\mu_i^{en}} \, \tilde{z}_{it} , \\ 
	& \qquad \text{ with }\tilde{x}_{it}^{ex} \sim \, \mathcal{U}[0,10], \quad \tilde{z}_{it} \sim \, \mathcal{U}[0,10],  \quad \tilde{x}_{it}^{ex} \indep \tilde{z}_{it},\\
	&\ v_{it} \sim \, \mathcal{U}[-0.5,0.5], \ u_{it}  \sim \, \mathcal{U}[-0.5,0.5],\\
	& \ x_{it}^{en} \, = \, (x_{it}^{ex} + z_{it})^{1/2} + v_{it},
\end{align*}
where $(\tilde{x}_{it}^{ex},\tilde{z}_{it}, v_{it}, u_{it})$ is i.i.d over time. 
Since the identification argument requires $T$ to be at least $4$, we will thus use $T=4$ time periods in our simulations.
The dgp for $\alpha_i$ does not need to be specified as $\alpha_i$ is differenced out in the analysis and we focus here on the average partial effect of $(x_t^{ex}, x_t^{en})$, $\E(\mu) = (0.5,0.5)'$. Note that in this dgp, $x_t^{ex}$ is correlated with $\mu^{ex}$ while it is independent of $\mu^{en}$. On the other hand, $x_t^{en}$ is correlated with both $\mu^{ex}$ and $\mu^{en}$. This dependence on time invariant heterogeneity creates persistence in the regressors.
%could raise the point that this is the issue in GP12?

We impose $f_t$ to be a function of only $v_{it}$ and consider two choices, namely $s_t(v_{it})= \sin(3 v_{it})$ and $q_t(v_{it})= 8 (v_{it})^2 -1$. Both take values approximately between $-1$ and $1$: the image of $s_t$ on $[-0.5,0.5]$ is $[-0.9975,0.9975]$ and the image of $q_t$ is exactly $[-1,1]$. However these two functions have different shapes on their support: $s_t$ is a decreasing odd function on $[-0.5,0.5]$ while $q_t$ is symmetric and has a minimum at $v =0$. See a plot of $g_t$ on Figure \ref{fig:plotg}. The remaining parameter $c$ will be either $c_1 = 1$ or $c_2 =2$. The purpose of considering different cases for both $f_t$ and $c$ is to illustrate the validity of the nonparametric approach 
%here could mention JMR comment
for different function specifications and degrees of variation in the nuisance parameters. 
To understand better the variations allowed by these dgps, some calculations show that when $f_t = s_t$,  $\Var(\mu^{ex} x^{ex}) / \Var(c_2 \epsilon_{it}) = 7.75$, and when $f_t = q_t$ the ratio is $\Var(\mu^{ex} x^{ex}) / \Var(c_2 \epsilon_{it}) = 1.95$.

As we explain in Section \ref{sec:assn_CRC_matrix}, the parameter of interest $\E(\mu)$ is not necessarily regularly identified even in a case without endogeneity. For the design studied here, some simulation evidence\footnote{There are to our knowledge no simple dgps with closed-form expressions for $\det(\dot{X}'\dot{X})$ when $x_{t}$ is two-dimensional. See \cite{gp12}	for an example with a closed-form expression when $x_{t}$ is scalar.} indicates that it is likely that $\det(\dot{X}'\dot{X})$ does not have finite expectation. Computing an exact integral is not possible here so we instead approximate the expectation with a simulation average. Figure \ref{fig:nonregId} plots in log scale the sample averages of $\det(\dot{X}'\dot{X})$ for simulations of growing size, as well as the averages across simulations of same size. The simulation average noticeably increases with sample size, which suggests that $\E \left(\det(\dot{X}'\dot{X}) \right) = \infty$ and $\E(\mu)$ is not regularly identified. 
%Alternatively, assuming $\det(\dot{X}'\dot{X})$ has finite expectation and variance $\sigma$ and assuming the normality approximation is valid for a sample of size $10^7$, I could construct a confidence interval for $\sigma$.
% Or I could say I coul dtest for finite variance using the papers I cite in earlier sections and argue it is outside the scope of this paper.
We thus proceed by selecting a subpopulation with sufficient time-variation of the regressors, that is, panel units $i$ such that $\det(\dot{X_i}'\dot{X_i}) < \delta = 0.01$. The average partial effect for this subpopulation, $\E(\mu | \delta)$, is identified by (\ref{eq:Miou_delta}) and we obtain with simulations $\E(\mu | \delta) = (0.5012,0.5004)'$.

%a comment on the difficulty of obtaining dgps with regular identification:
%In gP12 when $T=3$ they can have finite expectation for unidimensional regressor. For bivariate regressor, having finite expectation is more rare and we couldn't find a simple dgp guaranteeing this (introducing a trend in one of the regressor e.g)

The identification argument and the asymptotic analysis require Assumptions (\ref{assn:id_M_GLn}) and (\ref{assn:cv_g}) respectively, on the invertibility of the matrix-valued function $\mathcal{M}$ on the support of the random variable $V$. Although there are no closed-form expression for this function, we use simulations to examine the behavior of $\det(\mathcal{M}(.))$ on the support of $V$.
Precisely, a grid of $[-0.7,0.7]^4$ is constructed. Because $V$ is independent of $(X^{ex}, Z)$, a sample of size $n=2000$ of $(X_i^{ex}, Z_i, \mu_i^{ex}, \mu_i^{en})$ is drawn and for each value $\tilde{V}=(\tilde{v}_t)_{t \leq T}$ on the grid, $X_i^{en}$ is constructed as $((x_{it}^{ex} + z_{it})^{1/2} + \tilde{v}_t)_{t \leq T}$ and the value of $\det(\mathcal{M})$ is computed as the determinant of the sample average of $I-X_i(X_i'X_i)^{-1}X_i'$. 
Note that we look at an extended domain. According to these simulations, the minimizers of $\det(\mathcal{M})$ on $[-0.5,0.5]^4$ are (up to simulation error) the two corners $(-0.5, 0.5, -0.5,0.5)$ and $(0.5,-0.5,0.5,-0.5)$. Because the argument of $\det(\mathcal{M})$ is 4-dimensional, examining graphically its behavior is a hard exercise. Nevertheless in Figure \ref{fig:detM}, we plot variations of $\det(\mathcal{M}(\tilde{V}))$ as $\tilde{V}$ varies around  $(-0.5, 0.5, -0.5,0.5)$  along each of the four dimensions separately, on  $[-0.7,0.7]^4$ thus allowing $\tilde{V}$ to be outside of the support of $V$. The results for the other minimizer are very similar. The graphs show the function to be strictly monotonic around the minimum points. Since the function has to be positive, being the determinant of a positive symmetric matrix, this simulation evidence suggests that the determinant is bounded away from $0$.

We analyze here results of $R = 500$ simulations of each of the dgps for two different sample sizes, $n = 2000$ and $n = 5000$. We use multivariate B-spline basis for both estimation of the conditional expectation of $x^{en}$ conditional on $(x^{ex},z)$, and estimation of the functions $\mathcal{M}(.) = \E(M_i |V_i = .)$ and $k(.) = \E(M_i \dot{y}_i | V_i = .)$. The conditional expectation of $x^{en}$ is used to construct the generated covariates $\hat{V}_i$.
Recall that $g(V) = \mathcal{M}^{-1}(V) k(V)$.  For each of the simulation draws $r$, an estimate $\hat{g}^r$ of
the function $g$ is computed and the estimators $\hat{\mu}^{ex\, r}$ and $\hat{\mu}^{en \, r}$ of the average partial effects $\E(\mu^{ex}|\delta)$ and $\E(\mu^{en}|\delta)$ are obtained following the second step of our estimation procedure.

To choose the number of elements in the sieve basis, the econometrician will compute the cross-validation (CV) criteria of both the nonparametric estimator of $\E(x_{it}^{en}|x_{it}^{ex},Z_{it})$ and the functions $k$ and $\mathcal{M}$. We give in Table \ref{tab:RMSE_CV} the values of the leave-one-out CV values for various combinations of nonparametric estimators and for each dgp. 
%comment for myself: doing cross validation because cannot do IMSE as I do not know the true $\mathcal{M}$ and the true k function. And is important to look at because the econometrician cannot do cross validation of $g$ itself, as it is a transformation of the two np regressions.
Estimation of $V$ has little impact on the estimation of $k$ and $\mathcal{M}$. 
A third order multivariate B-spline basis with $0$ knots is thus chosen to generate the variables $\hat{V}$. 
To estimate the function $\mathcal{M}(.) = \E(M_i |V_i = .)$, we use the second order multivariate B-spline basis. Because we consider two different control functions $q_t$ and $s_t$, three estimators for $k(.) = \E(M_i \dot{y}_i | V_i = .)$ will be examined: a second order multivariate B-spline basis with $0$ knots written Spl$(2,0)$, a second order multivariate B-spline basis with $1$ knots written Spl$(2,1)$ and a third order multivariate B-spline basis with $0$ knots written Spl$(3,0)$.
Note that only Spl$(3,0)$ includes quadratic terms such as $(v_{it})^{en}$.

Tables \ref{tab:mu1RMSE} and  \ref{tab:mu2RMSE} report the obtained values of the squared root mean squared error (RMSE) of $\hat{\mu}$ for each of the dgps, sample sizes $n$ and choice of estimators for $k$.
The values of the RMSE are generally higher for $\mu_2$, the random coefficient multiplying the endogenous regressor and correlated with it.
To illustrate our asymptotic result, Figures \ref{fig:histmu_1} and \ref{fig:histmu_2} show smoothed histograms of the obtained estimates of $\hat{\mu}$. %These plots are compatible with the asymptotic normality result of Section \ref{sec:asymptoticnorm}. Similarly, i
It is noticeable that the variance of the estimator of $\E(\mu_i^{en})$ is larger than the variance of the estimator of $\E(\mu_i^{ex})$.
Comparing the density of $\hat{\mu}^{en}$ for the two different control functions, one can see that when $f_t=s_t$ the bias is increased while the variance is smaller, which is consistent with the RMSE values of these two cases being similar for $n=5000$.
We also plot the estimation results for $\hat{g}$, to illustrate that the estimator is able to retrieve the different shapes imposed by the dgp when enough sieve terms are included. Figure \ref{fig:plotg} reports the first coordinate of the pointwise average of the estimates $\bar{g}_1(V) = \sum_{r \leq R} \hat{g}^r(V)/R$ as well as the $5^{\text{th}}$ and $95^{\text{th}}$ quantiles of the first coordinate of $\hat{g}^r$, written $g_1^{5}(V)$ and $g_1^{95}(V)$ respectively for each value of $V$. The nonparametric estimator for $k$ is Spl$(3,0)$.
Since $V$ is 4-dimensional, the plot is reported for a particular variation of $V$, namely $V = (v_1, 0, 0, 0)$ where $v_1$ varies  on 80 \% of the support as we note the presence of a boundary effect.

\newpage

\bigskip

\bigskip

\begin{table}[h]
	\centering
	$
	\begin{array}{ccccccccc}
		\hline &\multicolumn{4}{c}{f_t = s_t}  &  \multicolumn{4}{c}{f_t = q_t} \\
		& \multicolumn{2}{c}{c_1} & \multicolumn{2}{c}{c_2} & \multicolumn{2}{c}{c_1}& \multicolumn{2}{c}{c_2}\\
		& k & \mathcal{M} & k & \mathcal{M} & k & \mathcal{M} & k & \mathcal{M}\\
		\hline \hline \text{Spl}(3,0),\ \text{Spl}(2,0)\ & 0.426 & 0.284 & 0.848 & 0.281 & 0.483 & 0.285 & 0.918 & 0.281\\
		\text{Spl}(3,0),\ \text{Spl}(2,1)\ &0.433 & 0.289 & 0.859 & 0.286 & 0.434 & 0.289 & 0.836 & 0.287\\ 
		\text{Spl}(3,0), \ \text{Spl}(2,2)\ &0.471 & 0.314 & 0.937 & 0.31 & 0.47 & 0.312 & 0.942 & 0.321\\ 
		\text{Spl}(3,0),\ \text{Spl}(3,0)\ &0.435 & 0.291 & 0.862 & 0.287 & 0.43 & 0.289 & 0.835 & 0.287\\ 
		\text{Spl}(3,0),\ \text{Spl}(3,1)\ &0.474 & 0.329 & 0.97 & 0.316 & 0.479 & 0.315 & 0.909 & 0.311\\ 
		\text{Spl}(3,0),\ \text{Spl}(4,0)\ &0.49 & 0.341 & 1.05 & 0.353 & 0.494 & 0.321 & 0.93 & 0.312\\
		\text{Spl}(3,0),\ \text{Spl}(4,1)\ &1.96 & 1.72 & 4.09 & 1.32 & 1.07 & 0.626 & 2.22 & 0.624\\
		\text{Spl}(4,0),\ \text{Spl}(2,0)\ &0.426 & 0.284 & 0.847 & 0.281 & 0.483 & 0.285 & 0.919 & 0.281\\
		\text{Spl}(4,0),\ \text{Spl}(2,1)\ &0.433 & 0.289 & 0.861 & 0.287 & 0.435 & 0.289 & 0.84 & 0.288\\
		\text{Spl}(4,0),\ \text{Spl}(2,2)\ &0.469 & 0.317 & 0.941 & 0.31 & 0.462 & 0.309 & 0.905 & 0.306\\
		\text{Spl}(4,0),\ \text{Spl}(3,0)\ &0.435 & 0.29 & 0.863 & 0.288 & 0.431 & 0.29 & 0.839 & 0.288\\
		\text{Spl}(4,0),\ \text{Spl}(3,1)\ &0.475 & 0.317 & 0.98 & 0.317 & 0.47 & 0.314 & 0.919 & 0.308\\
		\text{Spl}(4,0),\ \text{Spl}(4,0)\ &0.497 & 0.323 & 1.05 & 0.332 & 0.496 & 0.32 & 0.942 & 0.31\\
		\text{Spl}(4,0),\ \text{Spl}(4,1)\ &1.33 & 0.631 & 3.24 & 0.676 & 1.03 & 0.73 & 1.68 & 0.535 \\
		\hline
	\end{array}
	$
	\caption{Cross validation RMSE values\\
		\small \sl Values obtained for one simulation of $n=2000$. The first column indicates first the sieve estimators used to contruct the generated covariates $\hat{V}$, then the sieve estimator of either $k$ or $\mathcal{M}$.}
	\label{tab:RMSE_CV}
\end{table}

\begin{table}
	\centering
	$
	\begin{array}{cccccc}
		\hline& & \multicolumn{2}{c}{f_t = s_t}  &  \multicolumn{2}{c}{f_t = q_t} \\
		& &c_1 & c_2 & c_1 & c_2\\
		\hline\hline \multirow{3}{*}{n = 2000}&\text{Spl}(2,0)& 0.58 & 0.53 & 0.061 & 0.084\\ 
		&\text{Spl}(2,1)&0.5 & 0.66 & 0.097 & 0.13\\ 
		&\text{Spl}(3,0)&0.52 & 0.63 & 0.28 & 0.15\\
		\hline \multirow{3}{*}{n = 5000}&\text{Spl}(2,0)& 0.069 & 0.13 & 0.035 & 0.064\\ 
		&\text{Spl}(2,1)&0.075 & 0.092 & 0.029 & 0.12\\ 
		&\text{Spl}(3,0)&0.072 & 0.1 & 0.037 & 0.11\\
		\hline
	\end{array}
	$
	\caption{RMSE table for $\mu_1$}
	\label{tab:mu1RMSE}
\end{table}

\begin{table}
	\centering
	$
	\begin{array}{cccccc}
		\hline& &\multicolumn{2}{c}{f_t = s_t}  &  \multicolumn{2}{c}{f_t = q_t} \\
		& &c_1 & c_2 & c_1 & c_2\\
		\hline\hline \multirow{3}{*}{n = 2000}&\text{Spl}(2,0)&2.4 & 0.64 & 0.25 & 0.21\\ 
		&\text{Spl}(2,1)&2.1 & 0.59 & 0.45 & 0.33\\
		&\text{Spl}(3,0)&2.3 & 0.58 & 1.4 & 0.34\\
		\hline \multirow{3}{*}{n = 5000}&\text{Spl}(2,0)&0.16 & 0.31 & 0.18 & 0.092\\ 
		&\text{Spl}(2,1)&0.16 & 0.22 & 0.088 & 0.17\\ 
		&\text{Spl}(3,0)&0.16 & 0.23 & 0.11 & 0.18\\
		\hline
	\end{array}
	$
	\caption{RMSE table for $\mu_2$}
	\label{tab:mu2RMSE}
\end{table}

\begin{figure}[h]
	\centering
	\captionsetup{justification=centering}
	\captionsetup{width=.95\linewidth}
	\includegraphics[scale=0.5]{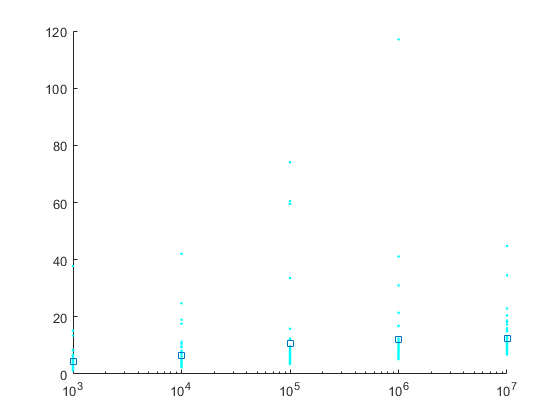}
	\caption[First]{ Sample averages of $\det(\dot{X}'\dot{X})$\\
		\small \sl The x-axis indicates sample sizes, there are 50 simulations for each sample size. Averages across simulations for each sample size are in dark blue, their values are $4.28$, $6.63$, $10.6$, $12.1$ and $12.5$. }
	\label{fig:nonregId}
\end{figure}

\begin{figure}[h]
	\centering
	\captionsetup{justification=centering}
	\captionsetup{width=.95\linewidth}
	\includegraphics[scale=0.41]{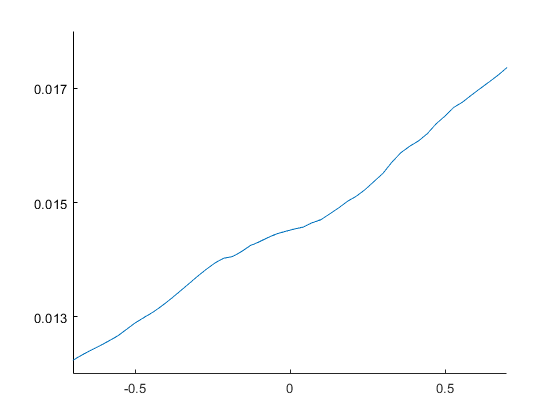}
	\includegraphics[scale=0.41]{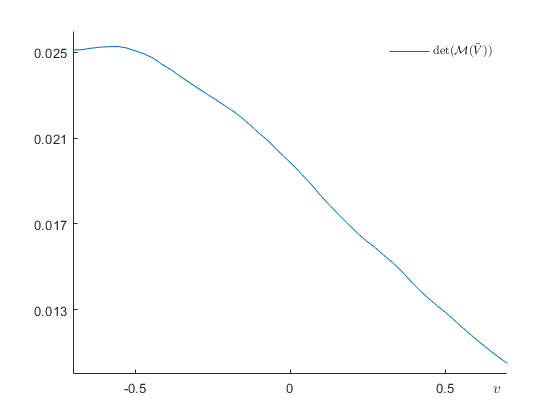}\\
	\includegraphics[scale=0.41]{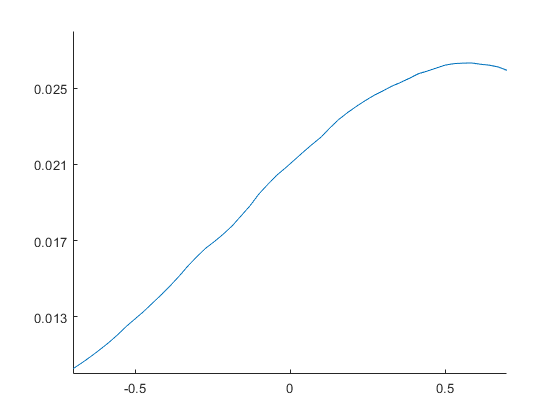}
	\includegraphics[scale=0.41]{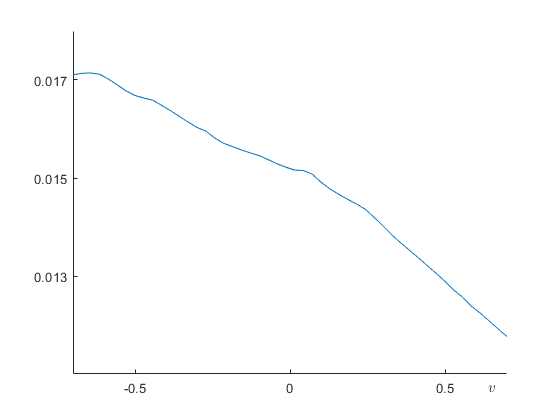}	
	\caption[First]{ Variations of $\det(\mathcal{M}(\tilde{V}))$  \\
		\small \sl Variations of $\det(\mathcal{M}(\tilde{V}))$ for $v \in [-0.7, 0.7]$ and $\tilde{V} = (v,0.5,-0.5,0.5)$ (top left), $\tilde{V} = (-0.5,v,-0.5,0.5)$ (top right),  $\tilde{V} = (-0.5,0.5,v,0.5)$ (bottom left), and  $\tilde{V} = (-0.5,0.5,-0.5,v)$ (bottom right).}
	\label{fig:detM}
\end{figure}

\begin{figure}[h]
	\centering
	\captionsetup{justification=centering}
	\captionsetup{width=.95\linewidth}
	\includegraphics[scale=0.35]{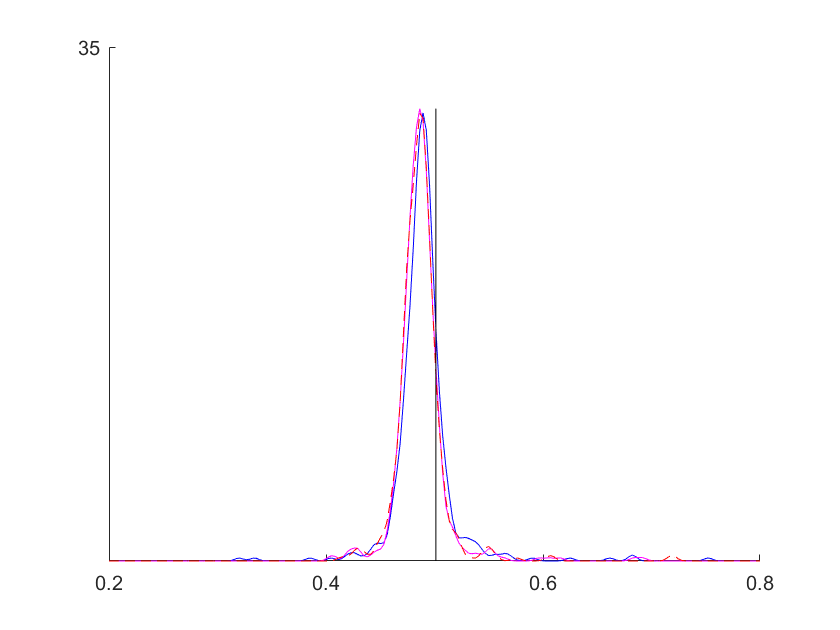}
	\includegraphics[scale=0.35]{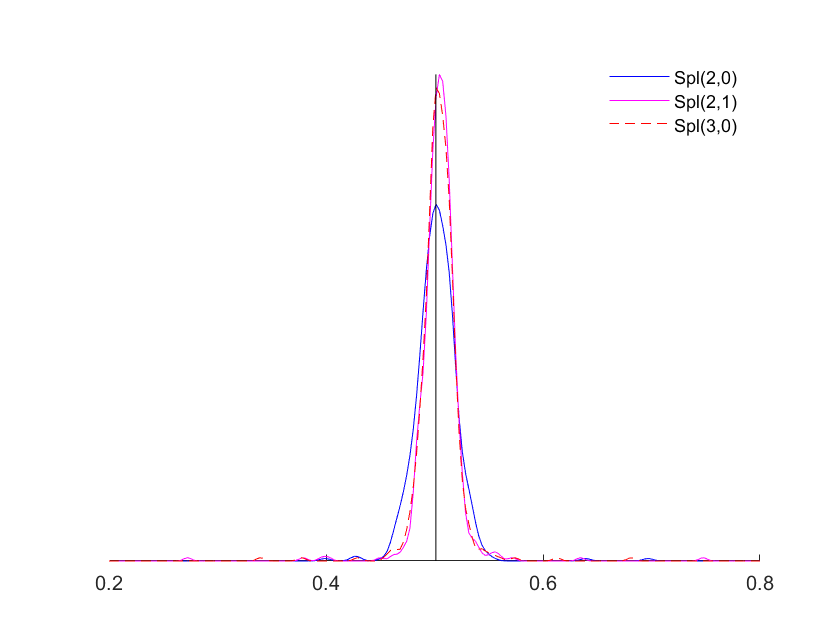}
	\caption[First]{  Density of the estimates for $\mu_1$  \\
		\small \sl Smoothed histogram of the simulated $\hat{\mu}^{ex}$, for $c=c_1$ and $n=5000$. On the left hand side, $f_t = s_t$ and the right hand side $f_t = q_t$. The true value $\E(\mu^{ex}|\delta)$ is shown in black.}
	\label{fig:histmu_1}
\end{figure}

\begin{figure}[h]
	\centering
	\captionsetup{justification=centering}
	\captionsetup{width=.95\linewidth}
	\includegraphics[scale=0.35]{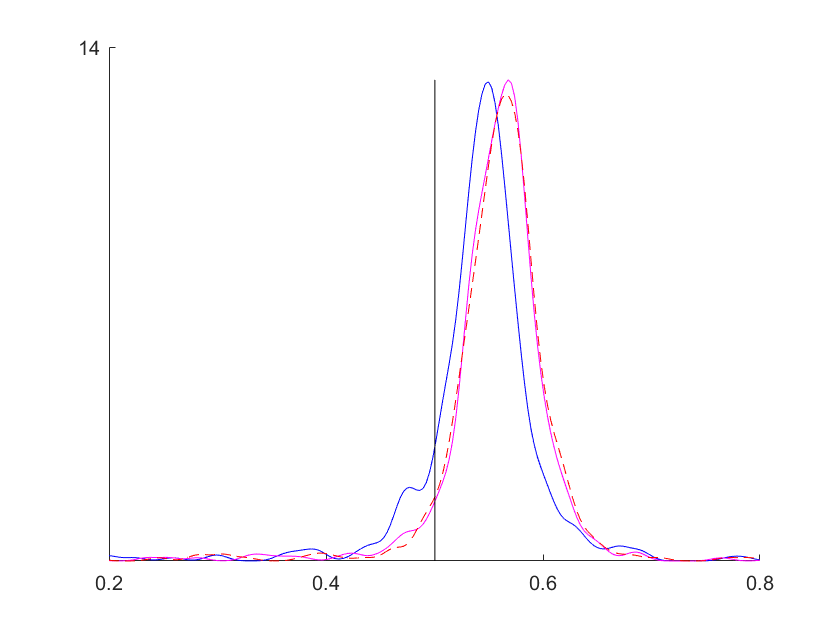}
	\includegraphics[scale=0.35]{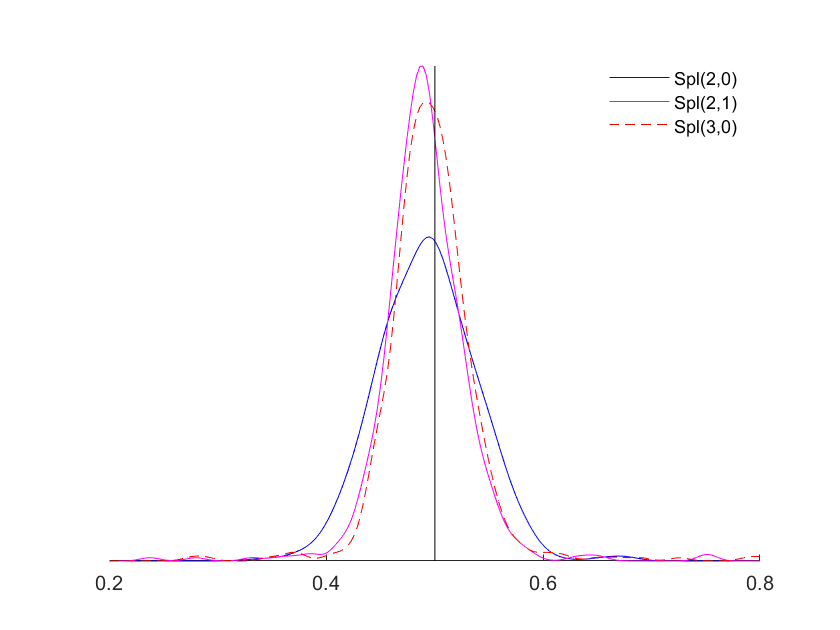}
	\caption[First]{  Density of the estimates for $\mu_2$  \\
		\small \sl Smoothed histogram of the simulated $\hat{\mu}^{en}$, for $c=c_1$ and $n=5000$. On the left hand side, $f_t = s_t$ and the right hand side $f_t = q_t$. The true value $\E(\mu^{en}|\delta)$ is shown in black.}
	\label{fig:histmu_2}
\end{figure}

\begin{figure}[t]
	\centering
	\captionsetup{justification=centering}
	\captionsetup{width=.95\linewidth}
	\includegraphics[scale=0.35]{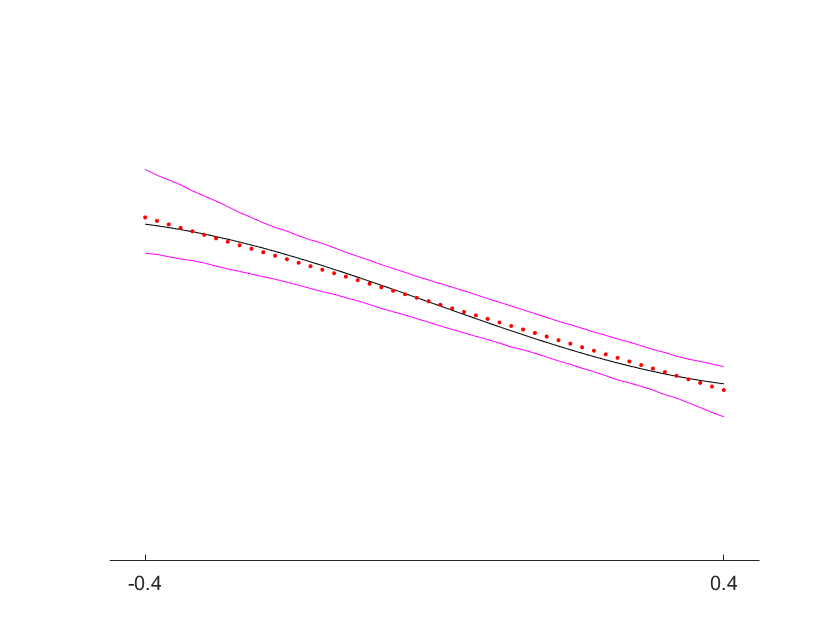}
	\includegraphics[scale=0.35]{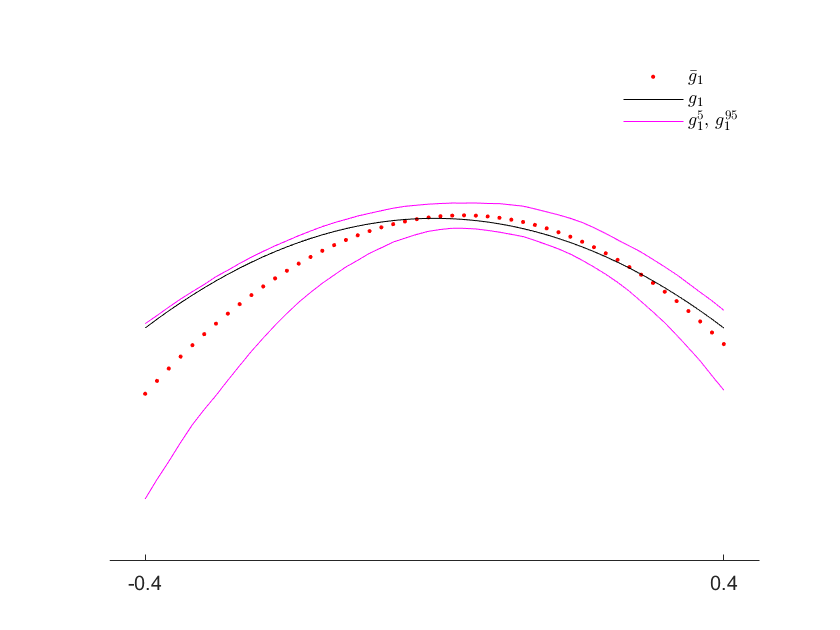}
	\caption[First]{ Nonparametric estimation of $g_1$  \\
		\small \sl Plot of the pointwise average $\bar{g}_1$, the 90 percent MC confidence bands $g_1^{5}$ and $g_1^{95}$ and  the true value $g_1$, for $c=c_1$ and $n=5000$. On the left hand side, $f_t = s_t$ and the right hand side $f_t = q_t$.}
	\label{fig:plotg}
\end{figure}

\end{document}